\documentclass{theoretics}

\allowdisplaybreaks

\usepackage{amsmath,amsthm,amsfonts,amssymb}
\usepackage{graphicx,tikz}

\usetikzlibrary{shapes,decorations,arrows.meta}

\makeatletter
\DeclareRobustCommand{\Udots}{%
  \vcenter{\offinterlineskip
    \halign{%
      \hbox to .8em{##}\cr
      \hfil.\cr\noalign{\kern.2ex}
      \hfil.\hfil\cr\noalign{\kern.2ex}
      .\hfil\cr}%
  }%
}
\makeatother

\newcommand{\eq}[1]{\hyperref[eq:#1]{(\ref*{eq:#1})}}
\renewcommand{\sec}[1]{\hyperref[sec:#1]{Section~\ref*{sec:#1}}}
\newcommand{\thm}[1]{\hyperref[thm:#1]{Theorem~\ref*{thm:#1}}}
\newcommand{\lem}[1]{\hyperref[lem:#1]{Lemma~\ref*{lem:#1}}}
\newcommand{\cor}[1]{\hyperref[cor:#1]{Corollary~\ref*{cor:#1}}}
\newcommand{\app}[1]{\hyperref[app:#1]{Appendix~\ref*{app:#1}}}
\newcommand{\tabl}[1]{\hyperref[tab:#1]{Table~\ref*{tab:#1}}}
\newcommand{\defin}[1]{\hyperref[def:#1]{Definition~\ref*{def:#1}}}
\newcommand{\fig}[1]{\hyperref[fig:#1]{Figure~\ref*{fig:#1}}}
\newcommand{\clm}[1]{\hyperref[clm:#1]{Claim~\ref*{clm:#1}}}
\newcommand{\conj}[1]{\hyperref[conj:#1]{Conjecture~\ref*{conj:#1}}}
\newcommand{\rem}[1]{\hyperref[rem:#1]{Remark~\ref*{rem:#1}}}
\newcommand{\para}[1]{\hyperref[para:#1]{Paragraph~\ref*{para:#1}}}

\def\ket#1{{\lvert}#1\rangle}
\def\bra#1{{\langle}#1\rvert}

\def\braket#1#2{{{\langle}#1\vert}#2\rangle}
\def\abs#1{\left| #1 \right|}

\def\norm#1{\left\| #1 \right\|}

\newcommand{\eps}{\varepsilon}

\def\w{{\sf w}}

\title{Multidimensional Quantum Walks, with Application to \texorpdfstring{$k$}{k}-Distinctness}
\ThCSshorttitle{Multidimensional Quantum Walks, with Application to $k$-Distinctness}
\ThCSauthor{Stacey Jeffery}{smjeffery@gmail.com}[0000-0003-0046-5089]
 \ThCSauthor{Sebastian Zur}{zursebastian@gmail.com}[0000-0002-1376-0362]
\ThCSaffil{CWI \& QuSoft, the Netherlands}
\ThCSshortnames{S.~Jeffery, S.~Zur}
\ThCSthanks{This work is supported by ERC STG grant 101040624-ASC-Q, NWO Klein project number OCENW.Klein.061, and ARO contract no W911NF2010327. SJ is a CIFAR Fellow in the Quantum Information Science Program. A preliminary version of this article appeared  at STOC 2023~\cite{JefferyZ23}.} 
\ThCSkeywords{quantum walk, random walk, k-distinctness, element distinctness, welded trees, quantum algorithm}

\ThCSyear{2025}
\ThCSarticlenum{7}
\ThCSreceived{Sep 27, 2023}
\ThCSrevised{Aug 28, 2024}
\ThCSaccepted{Oct 28, 2024}
\ThCSpublished{Mar 4, 2025}
\ThCSdoicreatedtrue

\begin{document}

\maketitle

\begin{abstract}
While the quantum query complexity of $k$-distinctness is known to be $O(n^{\frac{3}{4}-\frac{1}{4}\frac{1}{2^k-1}})$ for any constant $k\geq 4$ [Belovs, FOCS 2012], the best previous upper bound on the time complexity was $\widetilde{O}(n^{1-1/k})$. We give a new upper bound of $\widetilde{O}(n^{\frac{3}{4}-\frac{1}{4}\frac{1}{2^k-1}})$ on the time complexity, matching the query complexity up to polylogarithmic factors. In order to achieve this upper bound, we give a new technique for designing quantum walk search algorithms, which is an extension of the electric network framework. We also show how to solve the welded trees problem in $O(n)$ queries and $O(n^2)$ time using this new technique, showing that the new quantum walk framework can achieve exponential speedups.  
\end{abstract}

\section{Introduction}

In the problem of \emph{element distinctness}, the input is a list of $n$ integers, and the output is a bit indicating whether the integers are all distinct, or there exists a pair of integers that are the same, called a \emph{collision}. This problem has been studied as a fundamental problem in query complexity, but also for its relationship to other more practical problems, such as sorting, or \emph{collision finding}, which is similar, but one generally assumes there are many collisions and one wants to find one. In the worst case, element distinctness requires $\Theta(n)$ classical queries~\cite{ajtai2005EDlowerbound}. 

The first quantum algorithm to improve on this was a $O(n^{3/4})$ query algorithm~\cite{buhrman2001ElementDistinctness}, which is a variation of an optimal quantum algorithm for collision finding~\cite{brassard1997collision}, whose main technique is amplitude amplification~\cite{brassard2002AmpAndEst}. The algorithm of~\cite{buhrman2001ElementDistinctness} could also be implemented time efficiently, in $\widetilde{O}(n^{3/4})$ steps, with a log factor overhead from storing large subsets of the input in a sorted data structure. This was later improved to $O(n^{2/3})$ queries, and $\widetilde{O}(n^{2/3})$ time by Ambainis~\cite{ambainis2004QWalkForElementDist}, which is optimal~\cite{aaronson2004QLowerBndCollisionAndElementDistinct}. Ambainis' algorithm has been modified to solve other problems in various domains, from $k$-sum~\cite{childs2005quantum}, to path finding in isogeny graphs~\cite{tani2009claw,
costello2016isogeny}. Moreover, this algorithm was a critical step in our understanding of quantum query complexity, and quantum algorithms in general, as the algorithm used a new technique that was later generalised by Szegedy into a generic speedup for random walk search algorithms of a particular form~\cite{szegedy2004QMarkovChainSearch}. 

For any constant integer $k\geq 2$, the problem \emph{$k$-distinctness} is to decide if an input list of integers contains $k$ copies of the same integer. When $k=2$, this is exactly element distinctness. Ambainis~\cite{ambainis2004QWalkForElementDist} actually gave a quantum algorithm for $k$-distinctness for any $k\geq 2$, with query complexity $O(n^{1-1/(k+1)})$, and time complexity $\widetilde{O}(n^{1-1/(k+1)})$. For $k\geq 3$, Belovs gave an improved quantum query upper bound of $O(n^{3/4-\frac{1}{4}\frac{1}{2^k-1}})$~\cite{belovs2012kDist}, however, this upper bound was not constructive. Belovs proved this upper bound by exhibiting a dual adversary solution, which can be turned into a quantum algorithm that relies on controlled calls to a particular unitary. This unitary can be implemented in one query, but actually implementing this algorithm requires giving an efficient circuit for the unitary, which is not possible in general. This is analogous to being given a classical table of values, but no efficient circuit description. While it seems reasonable to guess that the time complexity of $k$-distinctness should not be significantly higher than the query complexity -- what could one possibly do aside from querying and sorting well-chosen sets of inputs? -- the problem of finding a matching time upper bound was open for ten years.

In the meantime, lower bounds of $\Omega(n^{\frac{3}{4}-\frac{1}{2k}})$ for $k\geq 3$~\cite{bun2018PolyMethodStrikesBack} and $\Omega(n^{\frac{3}{4}-\frac{1}{4k}})$ for $k\geq 4$~\cite{mande2020kDistLB} were exhibited. Progress was also made for the $k=3$ case. Two simultaneous works,~\cite{belovs2013ElectricWalks} and~\cite{childs2013arXivTimeEfficientQW3Distintness} (published together as~\cite{belovs2013TimeEfficientQW3Distintness}),
gave a $\widetilde{O}(n^{5/7})$ time upper bound for 3-distinctness. Ref.~\cite{belovs2013ElectricWalks} achieved this bound using a generalization of Szegedy's quantum walk framework, called the \emph{electric network framework}. Ref.~\cite{childs2013arXivTimeEfficientQW3Distintness} used the MNRS quantum walk framework~\cite{magniez2006SearchQuantumWalk}, and could also be generalised to give a slight improvement on the time upper bound to $\widetilde{O}(n^{1-1/k})$ for any $k>3$~\cite{jeffery2014thesis}. 

In this work, we give an upper bound of $\widetilde{O}(n^{\frac{3}{4}-\frac{1}{4}\frac{1}{2^k-1}})$ on the time complexity of $k$-distinctness, matching the best known query upper bound up to polylogarithmic factors. We do this using ideas from Belovs' query upper bound in a new framework for quantum walk algorithms, the \emph{multidimensional quantum walk framework}, which is an extension of the electric network framework -- the most general of the quantum walk frameworks~\cite{apers2019UnifiedFrameworkQWSearch}. We give a high-level overview of this extension in \sec{intro-QW}. 

Quantum walk search frameworks, discussed more in \sec{intro-QW}, are important because they allow one to design a quantum algorithm by first designing a classical random walk algorithm of a particular form, which can be compiled into an often faster quantum algorithm.  
While quantum walk frameworks make it extremely easy to design quantum algorithms, even without an in-depth knowledge of quantum computing, as evidenced by their wide application across domains, the major drawback is that they can achieve at most a quadratic speedup over the best classical algorithm. This is because a quantum walk search algorithm essentially takes a classical random walk algorithm, and produces a quantum algorithm that is up to quadratically better. 

This drawback does not hold for the multidimensional quantum walk framework. We give a quantum algorithm in our framework that solves the \emph{welded trees} problem in $O(n)$ queries and $O(n^2)$ time, which is an exponential speedup over the classical lower bound of $2^{\Omega(n)}$~\cite{childs2003ExpSpeedupQW}. While a poly$(n)$ quantum algorithm based on continuous-time quantum walks was already known, this proof-of-concept application shows that our framework is capable of exponential speedups. We emphasise that unlike the quantum walk search frameworks mentioned here that give generic speedups over classical random walk algorithms, continuous-time quantum walks are not easily designed and analysed, and their applications have been limited (with some exceptions based on converting quantum walk search algorithms into continuous-time quantum walks, such as~\cite{apers2022quadratic}). Our multidimensional quantum walk framework, as a generalization of the electric network framework, is in principle similarly easy to apply, but with the potential for significantly more dramatic speedups.

\subsection{Quantum Walks}\label{sec:intro-QW}

We give a brief overview of previous work on quantum walk search algorithms, with sufficient detail to understand, at a high level, the improvements we make, before describing these improvements at the end of this section.

The first quantum walk search framework is due to Szegedy~\cite{szegedy2004QMarkovChainSearch}, and is a generalization of the technique used by Ambainis in his element distinctness algorithm~\cite{ambainis2004QWalkForElementDist}. The framework can be described in analogy to a classical random walk algorithm that first samples an initial vertex according to the stationary distribution $\pi$ of some random walk (equivalently, reversible Markov process) $P$, and repeatedly takes a step of the random walk by sampling a neighbour of the current vertex, checking each time if the current vertex belongs to some \emph{marked set} $M$. Let $HT(P,M)$ be the hitting time, or the expected number of steps needed by a walker starting from $\pi$ to reach a vertex in $M$.
If ${\sf S}$ is the cost of sampling from $\pi$, ${\sf U}$ is the cost of sampling a neighbour of any vertex, ${\sf C}$ is the cost of checking if a vertex is marked, and $H$ is an upper bound on $HT(P,M)$ assuming $M\neq\emptyset$, then this classical algorithm finds a marked vertex with bounded error in complexity:
$$O({\sf S}+H({\sf U}+{\sf C})).$$
Szegedy showed that given such a $P$ and $M$, if ${\sf S}$ is the cost of coherently\footnote{Technically the classical ${\sf S}$ and ${\sf U}$ might be different from the quantum ones, but in practice they are often similar.} sampling from $\pi$, i.e.~generating $\sum_u\sqrt{\pi(u)}\ket{u}$, and ${\sf U}$ is the cost of generating, for any $u$, the superposition over its neighbours $\sum_v\sqrt{P_{u,v}}\ket{v}$, then there is a quantum algorithm that detects if $M\neq \emptyset$ with bounded error in complexity:
$$O({\sf S}+\sqrt{H}({\sf U}+{\sf C})).$$
This result was extended to the case of \emph{finding} a marked vertex, rather than just \emph{detecting} a marked vertex in~\cite{ambainis2019QuadSpeedupFindingMarkedQW}. This framework and subsequent related frameworks have been widely applied, because this is a very simple way to design a quantum algorithm. 

Belovs generalised this framework to the \emph{electric network framework}, by allowing the initial state to be $\ket{\sigma}=\sum_u\sqrt{\sigma(u)}\ket{u}$ for \emph{any} distribution $\sigma$, analogous to starting a random walk in some arbitrary initial distribution. Then if ${\sf S}_\sigma$ is the cost to generate $\ket{\sigma}$, there is a quantum algorithm that detects a marked vertex with bounded error in complexity:
$$O({\sf S}_\sigma+\sqrt{C}({\sf U}+{\sf C})),$$
where $C$ is a quantity that may be the same, or much larger than the hitting time of the classical random walk starting at $\sigma$. For example, if $\sigma=\pi$, then $C=H$ as above, but when $\sigma$ is supported on a single vertex $s$, and $M=\{t\}$, $C$ is the \emph{commute time} from $s$ to $t$~\cite{chandra1996ElectricalResAndCommute}, which is the expected number of steps needed to get from $s$ to $t$, and then back to $s$. If the hitting time from $s$ to $t$ is the same as the hitting time from $t$ to $s$, this is just twice that hitting time. However, in some cases the hitting time from $t$ to $s$ may be significantly larger than the hitting time from $s$ to $t$. 

A second incomparable quantum walk search framework that is similarly easy to apply is the MNRS framework~\cite{magniez2006SearchQuantumWalk}. Loosely speaking, this is the quantum analogue of a classical random walk that does not check if the current vertex is marked at every step, but rather, only after sufficiently many steps have been taken so that the current vertex is independent of the previously checked vertex. Ref.~\cite{apers2019UnifiedFrameworkQWSearch} extended the electric network framework to be able to \emph{find} a marked vertex, and also showed that the MNRS framework can be seen as a special case of the resulting framework. Thus, the finding version of the electric network framework captures all quantum walk search frameworks in one unified framework. 

We now discuss, at a high level, how a quantum walk search algorithm works -- particularly in the electric network framework (but others are similar)\footnote{We discuss the classic construction of such algorithms, without modifications that were more recently made in \cite{ambainis2019QuadSpeedupFindingMarkedQW} and \cite{apers2019UnifiedFrameworkQWSearch} to not only detect, but find.}. We will suppose for simplicity that $\sigma$ is supported on a single vertex $s$, and either $M=\emptyset$ or $M=\{t\}$. 
Fix a graph $G$, possibly with weighted edges, such that $s,t\in V(G)$. It is simplest if we imagine that $G$ is bipartite, so let $V(G)=V_{\cal A}\cup V_{\cal B}$ be a bipartition, with $s\in V_{\cal A}$. Let $G'$ be the graph $G$ with a single extra vertex $v_0$, connected to $s$, and connected to $t$ if and only if $t\in M$. 
For $u\in V_{\cal A}$, define \emph{star states}:
$$\ket{\psi_\star^{G'}(u)} = \sum_{v\in V_{\cal B}\cup\{v_0\}:\{u,v\}\in E(G')}\sqrt{\w_{u,v}}\ket{u,v},$$
where $\w_{u,v}$ is the weight of the edge $\{u,v\}$. If we normalise this state, we get $\sum_v\sqrt{P_{u,v}}\ket{u,v}$, where $P$ is the transition matrix of the random walk on $G'$. For $v\in V_{\cal B}$, define:
$$\ket{\psi_\star^{G'}(v)} = \sum_{u\in V_{\cal A}\cup\{v_0\}:\{u,v\}\in E(G')}\sqrt{\w_{u,v}}\ket{u,v}.$$
Let
\begin{align*}
{\cal A} &:= \mathrm{span}\{\ket{\psi_\star^{G'}(u)}: u\in V_{\cal A}\}
\;\mbox{ and }\;{\cal B} := \mathrm{span}\{\ket{\psi_\star^{G'}(v)}: v\in V_{\cal B}\}.
\end{align*}
Then a quantum walk algorithm works by performing phase estimation of the unitary 
$$U_{\cal AB}:=(2\Pi_{\cal A}-I)(2\Pi_{\cal B}-I)$$
on initial state $\ket{s,v_0}$ to some sufficiently high precision -- this precision determines the complexity of the algorithm. Let us consider why this algorithm can distinguish $M=\emptyset$ from $M=\{t\}$. 

First suppose $M=\{t\}$. Assume there is a path from $s$ to $t$ in $G$ (otherwise a random walk from $s$ will never find $t$), which means there is a cycle in $G'$ containing the edge $(v_0,s)$, obtained by adding $(t,v_0)$ and $(v_0,s)$ to the $st$-path in $G$. We can define a cycle state for a cycle $u_1,\dots,u_d=u_1$ as:
$$\sum_{i=1}^{d-1}\frac{\ket{e_{u_i,u_{i+1}}}}{\sqrt{\w_{u_i,u_{i+1}}}}
\mbox{ where }
\ket{e_{u,v}}:= \left\{\begin{array}{ll}
\ket{u,v} & \mbox{if }(u,v)\in V_{\cal A}\times V_{\cal B}\mbox{ or }v=v_0\\
-\ket{v,u} & \mbox{if }(u,v)\in V_{\cal B}\times V_{\cal A}\mbox{ or }u=v_0.
\end{array}\right.$$
A cycle state is orthogonal to all star states: if the cycle goes through a vertex $u$, it is supported on 2 of the edges adjacent to $u$: one contributing $-1$ because it goes into $u$, and the other $+1$ because it comes out of $u$. Thus, a cycle state is in the $(+1)$-eigenspace of $U_{\cal AB}$. 
If there is a cycle that uses the edge $(v_0,s)$, then it has non-zero overlap with the initial state $\ket{s,v_0}$, and so the initial state has non-zero overlap with the $(+1)$-eigenspace of $U_{\cal AB}$, and so the phase estimation algorithm will have a non-zero probability of outputting a phase estimate of 0. The shorter the cycle (i.e.~the shorter the $st$-path) the greater this overlap is relative to the size of the cycle state. We can make a similar argument if we take not just a single $st$-path in $G$, but a superposition of paths called an $st$-flow. Then the \emph{energy} of this flow (see \defin{flow}) controls the probability of getting Ta phase estimate of 0. The minimum energy of a unit flow from $s$ to $t$ is called the \emph{effective resistance} between $s$ and $t$, denoted ${\cal R}_{s,t}(G)$.

On the other hand, suppose $M=\emptyset$. Then we claim that 
$$\ket{s,v_0} = \sum_{u\in V_{\cal A}}\ket{\psi_\star^{G'}(u)} - \sum_{v\in V_{\cal B}}\ket{\psi_\star^{G'}(v)}\in {\cal A}+{\cal B} = ({\cal A}^\bot\cap {\cal B}^\bot)^\bot.$$
Since $\ket{s,v_0}$ also only overlaps with the star state of $s \in V_{\cal A}$, it is orthogonal to ${\cal B}$, and Thus, to ${\cal A}\cap {\cal B}$. Combined, this means that our initial state has no overlap with the $(+1)$-eigenspace of $U_{\cal AB}$, which is exactly $({\cal A}\cap{\cal B})\oplus ({\cal A}^\bot\cap {\cal B}^\bot)$, so if we could do phase estimation with infinite precision, the probability we would measure a phase estimate of 0 would be 0. Our precision is not infinite, but using a linear algebraic tool called the \emph{effective spectral gap lemma}, we can show that precision proportional to
$$\textstyle\norm{{\displaystyle\sum}_{u\in V_{\cal A}}\ket{\psi_\star^{G'}(u)}}^2 = {\displaystyle\sum}_{e\in G'}\w_e=:{\cal W}(G)$$
is sufficient. 

Combining these two analyses for the $M=\{t\}$ and $M=\emptyset$ case yield (in a non-obvious way) that approximately $\sqrt{\cal RW}$ steps of the quantum walk is sufficient, if ${\cal R}$ is an upper bound on ${\cal R}_{s,t}(G)$ whenever $M=\{t\}$, and ${\cal W}$ is an upper bound on ${\cal W}(G)$ whenever $M=\emptyset$. A nice way to  interpret this is that the quantity ${\cal R}_{s,t}(G){\cal W}(G)$ is equal to the \emph{commute time} from $s$ to $t$ -- the expected number of steps a random walker starting from $s$ needs to reach $t$, and then return to $s$. For a discussion of how to interpret this quantity in the case of more general $\sigma$ and $M$, see \cite{apers2019UnifiedFrameworkQWSearch}. 

\paragraph{The Multidimensional Quantum Walk Framework:} We extend this algorithm in two ways:
\begin{description}
\item[Edge Composition] To implement the unitary $U_{\cal AB}$, we perform a mapping that acts, for any $u\in V_{\cal A}$, as $\ket{u,0}\mapsto \ket{\psi_\star^{G'}(u)}$ (up to normalization), and a similar mapping for $v\in V_{\cal B}$. Loosely speaking, what this usually means is that we have a labelling of the edges coming out of $u$, and some way of computing $(u,v)$ from $(u,i)$, where $v$ is the $i$-th neighbour of $u$. If this computation costs ${\sf T}_{u,i}$ steps, then it takes $O(\max_{u,i}{\sf T}_{u,i})$ steps to implement $U_{\cal AB}$. However, in case this cost varies significantly over different $u,i$, we can do much better. We show how we can obtain a unitary with polylogarithmic cost, and essentially consider, in the analysis of the resulting algorithm, a quantum walk on a modified graph in which an edge $\{u,v\}$, where $v$ is the $i$-th neighbour of $u$, is replaced by a path of length ${\sf T}_{u,i}$.  A similar thing was already known for \emph{learning graphs}, when a transition could be implemented with ${\sf T}_{u,i}$ \emph{queries}~\cite{belovs2012LG}. This is an extremely useful, if not particularly surprising, feature of the framework, which we use in our application to $k$-distinctness. 
\item[Alternative Neighbourhoods] The more interesting way we augment the electric network framework is to allow the use of \emph{alternative neighbourhoods}. In order to generate the star state of a vertex $u$, which is a superposition of the edges coming out of $u$, one must, in some sense, know the neighbours of $u$, as well as their relative weights. In certain settings, the algorithm will know that the star state for $u$ is one of a small set of easily preparable states $\Psi_\star(u)=\{\ket{\psi_\star^1(u)},\ket{\psi^2_\star(u)},\dots\}$, but computing precisely which one of these is the correct state would be computationally expensive. In that case, we include all of $\Psi_\star(u)$ when constructing the spaces ${\cal A}$ and ${\cal B}$. In the case when $M=\emptyset$, the analysis is the same -- by increasing ${\cal A}+{\cal B}$, we have only made the analysis easier. However, in the case $M\neq \emptyset$, the analysis has become more constrained. For the analysis of this case, we used a circulation, because it is orthogonal to all star states. However, now there are some extra states in ${\cal A}+{\cal B}$, and we need to take extra care to find a circulation that is also orthogonal to these. 
\end{description}
The alternative neighbourhoods technique is best understood through examples, of which we shortly describe two. We first remark on the unifying idea from which both these techniques follow. 

If we let $\{\ket{\psi_\star(u)}\}_{u\in V}$ be \emph{any} set of states, we can make a graph $G$ on $V$ by letting $u$ and $v$ be adjacent if and only if $\braket{\psi_\star(u)}{\psi_\star(v)}\neq 0$. Then, if this graph is bipartite, and we can reflect around the span of each state individually, we can reflect around $\mathrm{span}\{\ket{\psi_\star(u)}:u\in V\}$. Quantum walk search algorithms can be seen as a special case of this, where we additionally exploit the structure of the graph to analyse the complexity of this procedure. One way of viewing alternative neighbourhoods is that we extend this reasoning to the case where we have \emph{spaces} $\{\mathrm{span}\{\Psi_\star(u)\}\}_{u\in V}$, each of which we can efficiently reflect around, and $G$ is now a bipartite graph encoding the overlap of the \emph{spaces}, hence the qualifier \emph{multidimensional}.

Edge composition also exploits this picture. We can define a sequence of subspaces $\{\Psi_t^{u,v}\}_{t=1}^{{\sf T}_{u,i}}$ that only overlap for adjacent $t$, and such that the subroutine computing $\ket{v,j}$ from $\ket{u,i}$ can be seen as moving through these spaces. Now the overlap graph of all these spaces will look like $G$, except with each edge $(u,v)$ replaced by a path of length ${\sf T}_{u,i}$.
See \fig{spaces-graph} and \fig{welded-spaces} for examples of such overlap graphs.

Before moving on to our examples, we comment that unlike the finding version of the electric network framework~\cite{apers2019UnifiedFrameworkQWSearch}, our extension does not allow one to find a marked vertex, but only to detect if there is one or not. We leave extending our framework to finding as future work. 

\subsection{Welded Trees}

We motivate the alternative neighbourhoods modification by an application to the welded trees problem~\cite{childs2003ExpSpeedupQW}. In the welded trees problem, the input is an oracle $O_G$ for a graph $G$ with $s,t\in V(G)\subset \{0,1\}^{2n}$. Each of $s$ and $t$ is the root of a full binary tree with $2^n$ leaves, and we connect these leaves with a pair of random matchings. This results in a graph in which all vertices except $s$ and $t$ have degree 3, and $s$ and $t$ each have degree 2. Given a string $u\in\{0,1\}^{2n}$, the oracle $O_G$ returns $\bot$ if $u\not\in V(G)$, which is true for all but at most a $2^{-n+2}$ fraction of strings, and otherwise it returns a list of the 2 or 3 neighbours of $u$. We assume $s=0^{2n}$, so we can use $s$ as our starting point, and the goal is to find $t$, which we can recognise since it is the only other vertex with only 2 neighbours. The classical query complexity of this problem is $2^{\Omega(n)}$~\cite{childs2003ExpSpeedupQW}. Intuitively that is because this problem is set up so that a classical algorithm has no option but to do a random walk, starting from $s$, until it hits $t$. However, this takes $2^{\Omega(n)}$ steps, because wherever a walker is in the graph, the probability of moving towards the centre, where the leaves of the two trees are connected, is twice the probability of moving away from the centre, towards $s$ or $t$. So a walker quickly moves from $s$ to the centre, but then it takes exponential time to escape to $t$. 

While we know there is a quantum algorithm that solves this problem in ${\sf poly}(n)$ queries\footnote{The best previous query complexity was $O(n^{1.5})$~\cite{atia2021welded}, although it is likely that continuous time quantum walks could also be used to solve this problem in $O(n)$ queries.} to $O_G$~\cite{childs2003ExpSpeedupQW}, if we try to reproduce this result in the electric network framework, we will get an exponential-time algorithm, essentially because the total weight of the graph is exponential. 

Suppose we could add weights to the edges of $G$, so that at any vertex $u$, the probability of moving towards the centre or away from the centre were the same: that is, if $\w$ is the weight on the edge from $u$ to its \emph{parent}, then the other two edges should have weight $\w/2$. This would already be very helpful for a classical random walk, however, a bit of thought shows that this is not possible to implement. By querying $u$, we learn the labels of its three neighbours, $v_1,v_2,v_3$,  which are random $2n$-bit strings, but we get no indication which is the parent. However, we know that the correct star state in the weighted graph that we would like to be able to walk on is proportional to one of the following:
\begin{align*}
&\ket{u,v_1}+\frac{1}{2}\ket{u,v_2}+\frac{1}{2}\ket{u,v_3},
\quad
\ket{u,v_2}+\frac{1}{2}\ket{u,v_1}+\frac{1}{2}\ket{u,v_3},
\quad\mbox{or}\quad
\ket{u,v_3}+\frac{1}{2}\ket{u,v_1}+\frac{1}{2}\ket{u,v_2}.
\end{align*}
Thus, we add all three states (up to some minor modifications) to $\Psi_\star(u)$, which yields an algorithm that can learn any bit of information about $t$ in $O(n)$ queries. By composing this with the Bernstein-Vazirani algorithm we can find $t$. For details, see \sec{welded}. 

We emphasise that our application to the welded trees problem does not use the edge composition technique. It would be trivial to embed any known exponential speedup in our framework by simply embedding the exponentially faster quantum algorithm in one of the edges of the graph, but we are able to solve the welded trees problem using only the alternative neighbourhoods idea.

\subsection{3-Distinctness}\label{sec:intro-3-dist}

We describe an attempt at a quantum walk algorithm for 3-distinctness, how it fails, and how the Multidimensional Quantum Walk Framework comes to the rescue. While our result for $k=3$ is not new, our generalization to $k>3$ is, and the case of $k=3$ is already sufficient to illustrate our techniques. 
Formally, the problem of 3-distinctness is: given a string $x\in [q]^n$, output a 1 if and only if there exist distinct $a_1,a_2,a_3\in [n]$ such that $x_{a_1}=x_{a_2}=x_{a_3}$. We make the standard simplifying assumptions (without loss of generality) that if such a 3-collision exists, it is unique, and moreover, there is an equipartition $[n]=A_1\cup A_2\cup A_3$ such that $a_1\in A_1$, $a_2\in A_2$ and $a_3\in A_3$. 

We now describe a graph that will be the basis for a quantum walk attempt. A vertex $v_{R_1,R_2}$ is described by a pair of sets $R_1\subset A_1$ and $R_2\subset A_2$.  $v_{R_1,R_2}$ stores these sets, as well as input-dependent \emph{data} consisting of the following:
\begin{itemize}
\item Queried values for all of $R_1$: $D_1(R):=\{(i,x_i):i\in R_1\}$.
\item Queried values for those elements of $R_2$ that have a match in $R_1$:
\[D_2(R):=\{(i_1,i_2,x_{i_1}): i_1\in R_1, i_2\in R_2, x_{i_1}=x_{i_2}\}.\]
\end{itemize}By only keeping track of the values in $R_2$ that have a match in $R_1$, we save the cost of initially querying the full set $R_2$.
The vertices will be in 4 different classes, for some parameters $r_1$ and $r_2$ with $r_1\ll r_2$:
\begin{align*}
V_0 &=\{v_{R_1,R_2}: |R_1|=r_1, |R_2|=r_2\}\\
V_1 &= \{v_{R_1,R_2}: |R_1|=r_1+1, |R_2|=r_2\}\\
V_2 &= \{v_{R_1,R_2}: |R_1|=r_1+1, |R_2|=r_2+1\}\\
V_3 &= \{v_{R_1,R_2,i_3}: |R_1|=r_1+1, |R_2|=r_2+1, i_3\in A_3\}.
\end{align*}
The vertices $v_{R_1,R_2,i_3}\in V_3$ are just like the vertices in $V_2$, except there is an additional index $i_3\in A_3$ stored. We connect vertices in $V_{\ell}$ and $V_{\ell+1}$ in the obvious way: $v_{R_1,R_2}\in V_{\ell}$ is adjacent to $v_{R_1',R_2'}\in V_{\ell+1}$ if and only if $R_1\subseteq R_1'$ and $R_2\subseteq R_2'$ (exactly one of these inclusions is proper); and $v_{R_1,R_2}\in V_2$ is adjacent to $v_{R_1,R_2,i_3}\in V_3$ for any $i_3\in A_3$ (see \fig{3-dist-graph}).

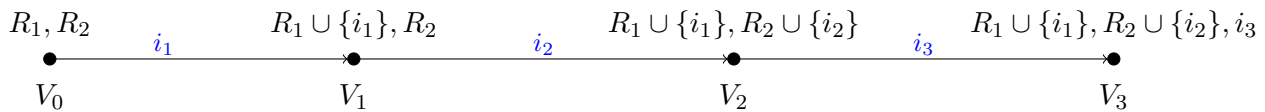
\begin{figure}
\centering
\begin{tikzpicture}[scale=1.2]
\filldraw (1,0) circle (.08);	\draw[-{Latex[length=2mm, width=2mm]}] (1,0) -- (4.92,0);
\filldraw (5,0) circle (.08);	\draw[-{Latex[length=2mm, width=2mm]}] (5,0) -- (9.92,0);
\filldraw (10,0) circle (.08);	\draw[-{Latex[length=2mm, width=2mm]}] (10,0) -- (14.92,0);
\filldraw (15,0) circle (.08);	

\node at (1,.45) {$R_1,R_2$};
\node at (1,-.5) {$V_0$};

	\node at (2.5,.2) {\small\color{blue}$i_1$};

\node at (5,.45) {$R_1\cup\{i_1\},R_2$};
\node at (5,-.5) {$V_1$};

	\node at (7.5,.2) {\small\color{blue}$i_2$};

\node at (10,.45) {$R_1\cup\{i_1\},R_2\cup\{i_2\}$};
\node at (10,-.5) {$V_2$};

	\node at (12.5,.2) {\small\color{blue}$i_3$};

\node at (15,.45) {$R_1\cup\{i_1\},R_2\cup\{i_2\},i_3$};
\node at (15,-.5) {$V_3$};

\end{tikzpicture}
\caption{A sample path from $V_0$ to $V_3$ in our first attempt at a quantum walk for 3-distinctness. The indices shown in blue can be seen to label the edges.}\label{fig:k-dist-graph}\label{fig:3-dist-graph}
\end{figure}

We say a vertex $v_{R_1,R_2,i_3}\in V_3$ is marked if $a_1\in R_1$, $a_2\in R_2$, and $a_3=i_3$, where $(a_1,a_2,a_3)$ is the unique 3-collision. Thus, a quantum walk that decides if there is a marked vertex or not decides 3-distinctness. 

We imagine a quantum walk that starts in a uniform superposition over $V_0$. To construct this initial state, we first take a uniform superposition over all sets $R_1$ of $r_1$ indices, and query them. Next we take a uniform superposition over all sets $R_2$ of size $r_2$, but rather than query everything in $R_2$, we search for all indices in $R_2$ that have a match in $R_1$. This saves us the cost of querying all $r_2$ elements of $R_2$, which is important because we will set $r_2$ to be larger than the total complexity we aim for (in this case, $r_2 \gg n^{5/7}$), so we could not afford to spend so much time. However, we do not only care about query complexity, but also the total time spent on non-query operations, so we also do not want to spend time writing down the set $R_2$, even if we do not query it, which is the first problem with this approach:
\begin{description}
\item[Problem 1:] Writing down $R_2$ would take too long. 
\end{description}
The fix for Problem 1 is rather simple: we will not let $R_2$ be a uniform random set of size $r_2$. Instead, we will assume that $A_2$ is partitioned into $m_2$ blocks, each of size $n/(3m_2)$, and $R_2$ will be made up of $t_2:=3m_2r_2/n$ of these blocks. This also means that when we move from $V_1$ to $V_2$, we will add an entire block, rather than just a single index. The main implication of this is that when we move from $V_1$ to $V_2$, we will have to search the new block of indices that we are adding to $R_2$ for any index that collides with $R_1$. This means that transitions from $V_1$ to $V_2$ have a non-trivial cost, $n^{\eps}$ for some small constant $\eps$, unlike all other transitions, which have polylogarithmic cost. Naively we would incur a multiplicative factor of $n^\eps$ on the whole algorithm, but we avoid this because the edge composition technique essentially allows us to only incur the cost $n^{\eps}$ on the edges that actually incur this cost, and not on every edge in the graph. Otherwise, our solution to Problem 1 is technical, but not deep, and so we gloss over Problem 1 and its solution for the remainder of this high-level synopsis. This is the only place we use the edge composition part of the framework in our applications, but we suspect it can be used in much more interesting ways.

Moving on, in order to take a step from a vertex $v_{R_1,R_2}\in V_0$ to a vertex $v_{R_1\cup\{i_1\},R_2}\in V_1$, we need to select a uniform new index $i_1$ to add to $R_1$, and then also update the data we store with each vertex. That means we have to query $i_1$ and add $(i_1,x_{i_1})$ to $D_1(R)$, which is simple, and can be done in $O(\log n)$ basic operations as long as we use a reasonable data structure to store $D_1(R)$; and we also have to update $D_2(R)$ by finding anything in $R_2$ that collides with $i_1$. Since $R_2$ has not been queried, this latter update would require an expensive search, which we do not have time for, so we want to avoid this. However, if we do not search $R_2$ for any $i_2$ such that $x_{i_2}=x_{i_1}$, then whenever we add some $i_1$ that has a match in $R_2$, the data becomes incorrect, and we have introduced what is referred to in~\cite{belovs2012kDist} as a \emph{fault}. This is a serious issue, because if $i_1$ is the unique index in $R_1$ such that there exists $i_2\in R_2$ with $x_{i_1}=x_{i_2}$, but this is not recorded in $D_2(R)$, then $i_1$ is ``remembered'' as having been added after $i_2$. That is, the resulting vertex does not only depend on $R_1\cup\{i_1\},R_2$, but on $i_1$ as well. For quantum interference to happen, it is crucial that when we are at a vertex $v$, the state does not remember anything about how we got there. 
\begin{description}
\item[Problem 2:] When we add $i_1$ to $R_1$ without searching for a match in $R_2$, we may introduce a \emph{fault}. 
\end{description}Our handling of this is inspired by the solution to an analogous problem in the query upper bound of~\cite{belovs2012kDist}. We partition $R_1$ into three sets: $R_1(\{1\})$, $R_1(\{2\})$, and $R_1(\{1,2\})$; and $R_2$ into two sets $R_1(1)$ and $R_1(2)$. Then $D_2(R)$ will only store collisions $(i_1,i_2,x_{i_1})$ such that $x_{i_1}=x_{i_2}$ if $i_1\in R_1(S)$ and $i_2\in R_2(s)$ for some $s\in S$. This is shown in \fig{set-partitions}.

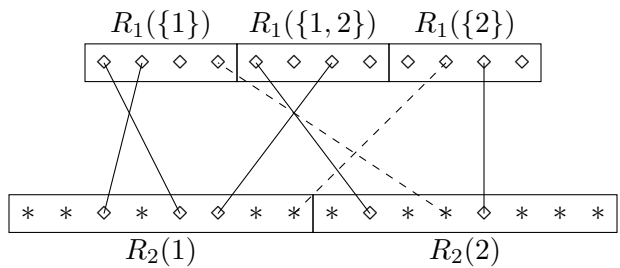
\begin{figure}
\centering
\begin{tikzpicture}[scale =1.2]
\node at (0,2) {$\diamond$};
\node at (.5,2) {$\diamond$};
\node at (1,2) {$\diamond$};
\node at (1.5,2) {$\diamond$};
\node at (2,2) {$\diamond$};
\node at (2.5,2) {$\diamond$};
\node at (3,2) {$\diamond$};
\node at (3.5,2) {$\diamond$};
\node at (4,2) {$\diamond$};
\node at (4.5,2) {$\diamond$};
\node at (5,2) {$\diamond$};
\node at (5.5,2) {$\diamond$};

\draw (-.25,2.25) rectangle (1.75,1.75);
\draw (1.75,2.25) rectangle (3.75,1.75);
\draw (3.75,2.25) rectangle (5.75,1.75);

\node at (.75,2.5) {$R_1(\{1\})$};
\node at (2.75,2.5) {$R_1(\{1,2\})$};
\node at (4.75,2.5) {$R_1(\{2\})$};

\node at (-1,0) {$*$};
\node at (-.5,0) {$*$};
\node at (0,0) {$\diamond$}; \draw (0,0)--(.5,2);
\node at (.5,0) {$*$};
\node at (1,0) {$\diamond$}; \draw (1,0)--(0,2);
\node at (1.5,0) {$\diamond$}; \draw (1.5,0)--(3,2);
\node at (2,0) {$*$};
\node at (2.5,0) {$*$}; \draw[dashed] (2.5,0)--(4.5,2);
\node at (3,0) {$*$};
\node at (3.5,0) {$\diamond$}; \draw (3.5,0)--(2,2);
\node at (4,0) {$*$};
\node at (4.5,0) {$*$};\draw[dashed] (4.5,0)--(1.5,2);
\node at (5,0) {$\diamond$}; \draw (5,0)--(5,2);
\node at (5.5,0) {$*$};
\node at (6,0) {$*$};
\node at (6.5,0) {$*$};

\draw (2.75,.25) rectangle (6.75,-.25);
\draw (-1.25,.25) rectangle (2.75,-.25);

\node at (.75,-.5) {$R_2(1)$};
\node at (4.75,-.5) {$R_2(2)$};
\end{tikzpicture}
\caption{The data we keep track of for a vertex $v_{R_1,R_2}$. $\diamond$ represents a queried index. $*$ represents an index whose query value is not stored. We only store the query value of an index in $R_2(s)$ if it collides with something in $R_1(\{s\})\cup R_1(\{1,2\})$, shown here by a solid line. If $i_2\in R_2(1)$ collides with some value in $R_1(\{2\})$, shown here by a dashed line, we do not record that, and do not store $x_{i_2}$.}\label{fig:set-partitions}
\end{figure}

Now when we add $i_1$ to $R_1$, we have three choices: we can add it to $R_1(\{1\})$, $R_1(\{2\})$, or $R_1(\{1,2\})$. Importantly, at least one of these choices does not introduce a fault. To see this, suppose there is some $i_2\in R_2$ such that $x_{i_1}=x_{i_2}$. We claim there can be at most one such index, because otherwise there would be a 3-collision in $A_1\cup A_2$, and we are assuming the unique 3-collision has one part in $A_3$. This leads to three possibilities:
\begin{description}
\item[Type 1:] $i_2\in R_2(2)$, in which case, adding $i_1$ to $R_1(\{1\})$ does not introduce a fault.
\item[Type 2:] $i_2\in R_2(1)$, in which case, adding $i_1$ to $R_1(\{2\})$ does not introduce a fault.
\item[Type 0:] There is no such $i_2$, in which case, adding $i_1$ to $R_1(\{1\})$ or $R_1(\{2\})$ or $R_1(\{1,2\})$ does not introduce a fault.
\end{description}
We modify the graph so that we first move from $v_{R_1,R_2}\in V_0$ to $v_{R_1,R_2,i_1}\in V_0^+$ by selecting a new $i_1\in A_1\setminus R_1$, and then move from $v_{R_1,R_2,i_1}$ to $v_{R_1\cup\{i_1\},R_2}\in V_1$ -- here there are three possibilities for $R_1\cup\{i_1\}$, depending on to which of the three parts of $R_1$ we add $i_1$. However, we will only add $i_1$ to a part of $R_1$ that does not introduce a fault. Thus, a vertex $v_{R_1,R_2,i_1}$ in $V_0^+$ has one edge leading back to $V_0$, and either one or three edges leading forward to $V_1$, as shown in \fig{alt-neighbourhoods-S}.

\begin{cfigure}
\centering
\begin{tikzpicture}[scale=1.2]
\node at (0,0) {
\begin{tikzpicture}[scale=1.2]
\draw[->] (0,0)--(.75,0); \draw (.75,0)--(1.5,0);		\draw[->] (1.5,0)--(2.25,.5); \draw (2.25,.5)--(3,1);
										\draw (1.5,0)--(2.25,0); \draw[<-] (2.25,0)--(3,0);
										\draw[->] (1.5,0)--(2.25,-.5); \draw (2.25,-.5)--(3,-1);

\filldraw (0,0) circle (.05);		\filldraw (1.5,0) circle (.08);	\filldraw (3,1) circle (.05);
												\filldraw (3,0) circle (.05);
												\filldraw (3,-1) circle (.05);

\node at (-.25,.2) {$v$};
\node at (1.5,.25) {$u$};
\node at (3.5,1) {$v^{\{1\}}$};
\node at (3.5,0) {$v^{\{1,2\}}$};
\node at (3.5,-1) {$v^{\{2\}}$};

\node at (1.9,.5) {\color{blue}\small ${}_{\{1\}}$};
\node at (2.2,.15) {\color{blue}\small ${}_{\{1,2\}}$};
\node at (1.9,-.5) {\color{blue}\small ${}_{\{2\}}$};
\end{tikzpicture}
};

\node at (0,-1.5) {Type 0};

\node at (6,0) {
\begin{tikzpicture}[scale=1.2]
\draw[->] (0,0)--(.75,0); \draw (.75,0)--(1.5,0);		\draw[->] (1.5,0)--(2.25,.5); \draw (2.25,.5)--(3,1);

\filldraw (0,0) circle (.05);		\filldraw (1.5,0) circle (.08);	\filldraw (3,1) circle (.05);

\node at (-.25,.2) {$v$};
\node at (1.5,.25) {$u$};
\node at (3.5,1) {$v^{\{1\}}$};
\node at (3.5,-1) {\color{white} $v_{\{2\}}$};

\node at (1.9,.5) {\color{blue}\small ${}_{\{1\}}$};
\end{tikzpicture}
};

\node at (6,-1.5) {Type 1};

\node at (12,0) {
\begin{tikzpicture}[scale=1.2]
\draw[->] (0,0)--(.75,0); \draw (.75,0)--(1.5,0);		
										\draw[->] (1.5,0)--(2.25,-.5); \draw (2.25,-.5)--(3,-1);

\filldraw (0,0) circle (.05);		\filldraw (1.5,0) circle (.08);	
												\filldraw (3,-1) circle (.05);

\node at (-.25,.2) {$v$};
\node at (1.5,.25) {$u$};
\node at (3.5,1) {\color{white}$v^{\{1\}}$};
\node at (3.5,-1) {$v^{\{2\}}$};

\node at (1.9,-.5) {\color{blue}\small ${}_{\{2\}}$};
\end{tikzpicture}
};

\node at (12,-1.5) {Type 2};

\end{tikzpicture}
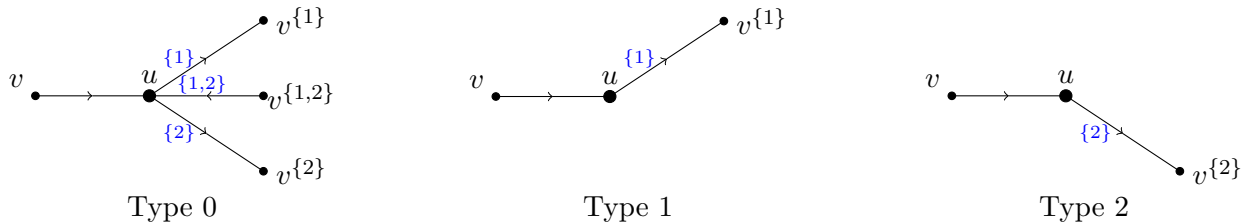
\caption{The possible neighbourhoods of $u=v_{R_1,R_2,i_1}\in V_0^+$, depending on the type of vertex. $v^{S}\in V_1$ is obtained from $v$ by adding $i_1$ to $R_1(S)$. The backwards neighbour $v=v_{R_1,R_2}\in V_0$ is always the same.}\label{fig:alt-neighbourhoods-S}
\end{cfigure}

On its own, this is not a solution, because for a given $v_{R_1,R_2,i_1}$, in order to determine its type, we would have to search for an $i_2\in R_2$ such that $x_{i_1}=x_{i_2}$, which is precisely what we want to avoid. However, this is exactly the situation where the alternative neighbourhood technique is useful. For all $u\in V_0^+$, we will let $\Psi_\star(u)$ contain all three possibilities shown in \fig{alt-neighbourhoods-S}, of which exactly one is the correct state. We are then able to carefully construct a flow that is orthogonal to all three states, in our analysis. The idea is that all incoming flow from $v$ must leave along the edge $(u,v^{\{1\}})$ so that the result is a valid flow in case of Type 1. However, in order to be a valid flow in case of Type 2, all incoming flow from $v$ must leave along the edge $(u,v^{\{2\}})$. But now to ensure that we also have a valid flow in case of Type 0, we must have negative flow on the edge $(u,v^{\{1,2\}})$, or equivalently, flow from $v^{\{1,2\}}$ to $u$. This is indicated by the arrows on the edges in \fig{alt-neighbourhoods-S}. For details, see \sec{3-dist}. 

\paragraph{Model of Computation:} Our $k$-distinctness algorithm works in the same model as previous $k$-distinctness algorithms, which we try to make more explicit than has been done in previous work. In addition to arbitrary 1- and 2-qubit gates, we assume \emph{quantum random access} to a large \emph{quantum} memory (QRAM). This version of QRAM is fully quantum, whereas some previous works have used ``QRAM'' to refer to classical memory that can be read in superposition by a quantum machine.  We describe precisely what we mean by QRAM in \sec{model}.

\subsection{Organization}

The remainder of this article is organised as follows. In \sec{prelim}, we give preliminaries on graph theory, quantum subroutines, quantum data structures, and probability theory, including several non-standard definitions, which we encourage the experienced reader not to skip. In \sec{fwk}, we present the Multidimensional Quantum Walk Framework, which is stated as \thm{full-framework}. In \sec{welded}, we present our first application to the welded trees problem. This section is mostly self-contained, explicitly constructing and analysing an algorithm rather than referring to our new framework, which we feel gives an intuitive demonstration of the framework. In \sec{k-dist-full}, we present our new application to $k$-distinctness.

\section{Preliminaries}\label{sec:prelim}

\subsection{Graph Theory}\label{sec:prelim-graph}

In this section, we define graph theoretic concepts and notation.

\begin{definition}[Network]\label{def:network}
A network is a weighted graph $G$ with an (undirected) edge set $E(G)$, vertex set $V(G)$, and some weight function $\w:E(G)\rightarrow\mathbb{R}_{>0}$. Since edges are undirected, we can equivalently describe the edges by some set $\overrightarrow{E}(G)$ such that for all $\{u,v\}\in E(G)$, exactly one of $(u,v)$ or $(v,u)$ is in $\overrightarrow{E}(G)$. The choice of edge directions is arbitrary. Then we can view the weights as a function $\w:\overrightarrow{E}(G)\rightarrow\mathbb{R}_{> 0}$, and for all $(u,v)\in \overrightarrow{E}$, define $\w_{v,u}=\w_{u,v}$. For convenience, we will define $\w_{u,v}=0$ for every pair of vertices such that $\{u,v\}\not\in E(G)$. 
The \emph{total weight} of $G$ is 
$$\textstyle{\cal W}(G):={\displaystyle\sum}_{e\in \overrightarrow{E}(G)}\w_e.$$ 
\end{definition}

\noindent For an implicit network $G$, and $u\in V(G)$, we will let $\Gamma(u)$ denote the \emph{neighbourhood} of $u$:
$$\Gamma(u):=\{v\in V(G):\{u,v\}\in E(G)\}.$$
We use the following notation for \emph{the out- and in-neighbourhoods} of $u\in V(G)$:
\begin{equation}
\begin{split}
\Gamma^+(u) &:= \{v\in\Gamma(u):(u,v)\in \overrightarrow{E}(G)\}\\
\Gamma^-(u) &:= \{v\in\Gamma(u):(v,u)\in \overrightarrow{E}(G)\},
\end{split}\label{eq:neighbourhoods}
\end{equation}

\begin{definition}[Flow, Circulation]\label{def:flow}
A \emph{flow} on a network $G$ is a real-valued function $\theta:\overrightarrow{E}(G)\rightarrow\mathbb{R}$, extended to edges in both directions by $\theta(u,v)=-\theta(v,u)$ for all $(u,v)\in\overrightarrow{E}(G)$. 
For any flow $\theta$ on $G$, and vertex $u\in V(G)$, we define $\theta(u)=\sum_{v\in \Gamma(u)}\theta(u,v)$ as the flow coming out of $u$. If $\theta(u)=0$, we say flow is conserved at $u$. If flow is conserved at every vertex, we call $\theta$ a \emph{circulation}. 
If $\theta(u)>0$, we call $u$ a \emph{source}, and if $\theta(u)<0$ we call $u$ a \emph{sink}. A flow with unique source $s$ and unique sink $t$ is called an \emph{$st$-flow}. 
The \emph{energy} of $\theta$ is 
$${\cal E}(\theta):=\textstyle{\displaystyle\sum}_{(u,v)\in\overrightarrow{E}(G)}\displaystyle\frac{\theta(u,v)^2}{\w_{u,v}}.$$
\end{definition}

\paragraph{Accessing $G$:} In computations involving a (classical) random walk on a graph $G$, it is usually assumed that for any $u\in V(G)$, it is possible to sample a neighbour $v\in\Gamma(u)$ according to the distribution 
$$\Pr[v] = \frac{\w_{u,v}}{\w_u}\mbox{ where }\w_u:=\textstyle{\displaystyle\sum}_{ v'\in\Gamma(u)}\w_{u,v'}.$$
It is standard to assume this is broken into two steps: (1) sampling some $i\in [d_u]$, where $d_u:=|\Gamma(u)|$ is the degree of $u$, and (2) computing the $i$-th neighbour of $u$. That is, we assume that for each $u\in V(G)$, there is an efficiently computable function $f_u:[d_u]\rightarrow V(G)$ such that $\mathrm{im}(f_u)=\Gamma(u)$, and we call $f_u(i)$ the \emph{$i$-th neighbour of $u$}. In the quantum case (see \defin{QW-access} below), we assume that the sample (1) can be done coherently, and we use a reversible version of the map $(u,i)\mapsto f_u(i)$. We will also find it convenient to suppose the indices $i$ of the neighbours of $u$ come from some more general set $L(u)$, which may equal $[d_u]$, or some other convenient set, which we call the \emph{edge labels of $u$}. It is possible to have $|L(u)|>|\Gamma(u)|=d_u$, meaning that some elements of $L(u)$ do not label an edge adjacent to $u$ (these labels should be sampled with probability 0). 
We assume we have a partition of $L(u)$ into disjoint $L^+(u)$ and $L^-(u)$ such that:
\begin{equation*}
\begin{split}
L^+(u) &\supseteq \{i\in L(u): (u,f_u(i))\in\overrightarrow{E}(G)\} = \{i\in L(u):f_u(i)\in \Gamma^+(u)\}\\
L^-(u) &\supseteq \{i\in L(u): (f_u(i),u)\in\overrightarrow{E}(G)\} = \{i\in L(u):f_u(i)\in \Gamma^-(u)\}.
\end{split}
\end{equation*}
Note that for any $(u,v)\in\overrightarrow{E}(G)$, with $i=f_u^{-1}(v)$ and $j=f_v^{-1}(u)$, any of $(u,v)$, $(v,u)$, $(u,i)$, or $(v,j)$ fully specify the edge. Thus, it will be convenient to denote the weight of the edge using any of the alternatives:
$$\w_{u,v}=\w_{v,u}=\w_{u,i}=\w_{v,j}.$$
For any $i\in L(u)$, we set $\w_{u,i}=0$ if and only if $\{u,f_u(i)\}\not\in E(G)$.

\begin{definition}[Quantum Walk access to $G$]\label{def:QW-access}
For each $u\in V(G)$, let $L(u)=L^+(u)\cup L^-(u)$ be some finite set of \emph{edge labels}, and $f_u:L(u)\rightarrow V(G)$ a function such that $\Gamma(u)\subseteq \mathrm{im}(f_u)$. 
A quantum algorithm has \emph{quantum walk access} to $G$ if it has access to the following subroutines:
\begin{itemize}
\item A subroutine that ``samples'' from $L(u)$ by implementing a map $U_\star$ in cost ${\sf A}_\star$ that acts as:
\begin{equation*}
U_\star\ket{u,0} \propto \sum_{i\in L^+(u)}\sqrt{\w_{u,i}}\ket{u,i}-\sum_{i\in L^-(u)}\sqrt{\w_{u,i}}\ket{u,i} =: \ket{\psi_\star^G(u)}.
\end{equation*}
\item A subroutine that implements the \emph{transition map}
$
\ket{u,i}\mapsto \ket{v,j}
$
(possibly with some error) where $i=f^{-1}_u(v)$ and $j=f^{-1}_v(u)$, with costs $\{{\sf T}_{u,i}={\sf T}_{u,v}\}_{(u,v)\in\overrightarrow{E}(G)}$. 

\item Query access to the total vertex weights $\w_u=\sum_{v\in\Gamma(u)}\w_{u,v}$. 
\end{itemize}
We call $\{{\sf T}_e\}_{e\in\overrightarrow{E}(G)}$ the set of \emph{transition costs} and ${\sf A}_\star$ the \emph{cost of generating the star states}. 
\end{definition}

\begin{definition}[Networks with lengths]\label{def:nwk-length}
If $G$ is a network, and $\ell:\overrightarrow{E}(G)\rightarrow\mathbb{Z}_{\geq 1}$ a positive-integer-valued function on the edges of $G$, we define $G^{\ell}$ to be the graph obtained from replacing each edge $(u,v)\in \overrightarrow{E}(G)$ of $G$ with a path from $u$ to $v$ of length $\ell_{u,v}$, and giving each edge in the path the weight $\w_{u,v}$. We define:
\begin{equation*}
{\cal W}^{\ell}(G):={\cal W}(G^\ell)=\textstyle{\displaystyle\sum}_{e\in \overrightarrow{E}(G)}\w_e\ell_e,
\end{equation*}
and for any flow $\theta$ on $G$, we let $\theta^{\ell}$ be the flow on $G^{\ell}$ obtained by assigning flow $\theta(u,v)$ to any edge in the path from $u$ to $v$, and define:
\begin{equation*}
{\cal E}^{\ell}(\theta):={\cal E}(\theta^{\ell})=\textstyle{\displaystyle\sum}_{e\in \overrightarrow{E}(G)}\displaystyle\frac{\theta(e)^2}{\w_e}\ell_e.
\end{equation*}
\end{definition}

\subsection{Model of Computation and Quantum Subroutines}\label{sec:model}

We will work in the (fully quantum) QRAM model, which we now describe. By QRAM, we mean \emph{quantum} memory, storing an arbitrary quantum state, to which we can apply random access gates. By this, we mean we can implement, for $i\in [n]$, $b\in\{0,1\}$, $x\in \{0,1\}^n$, a random access read:
$$\text{READ}:\ket{i}\ket{b}\ket{x} \mapsto \ket{i}\ket{b\oplus x_i}\ket{x},$$
or a random access write:
$$\text{WRITE}:\ket{i}\ket{b}\ket{x} \mapsto \ket{i}\ket{b}\ket{x_1,\dots,x_{i-1},x_i\oplus b,x_{i+1},\dots,x_n},$$
on any superposition.
By applying $\text{READ}\cdot\text{WRITE}\cdot\text{READ}$, we can implement a controlled swap:
$$\sum_{i\in [n]}\ket{i}\bra{i}\otimes \text{SWAP}_{0,i}(\ket{i}\ket{b}\ket{x}) = \ket{i}\ket{x_i}\ket{x_1,\dots,x_{i-1},b,x_{i+1},\dots,x_n}.$$
Aside from these operations, we count the number of elementary gates, by which we mean arbitrary unitaries that act on $O(1)$ qubits.

We will be interested in running different iterations of a subroutine on different branches of a superposition, for which we use the concept of a quantum subroutine. We note that \defin{variable-time} is \emph{not} the most general definition, but it is sufficient for our purposes. 

\begin{definition}[Quantum Subroutine]\label{def:variable-time}
A \emph{quantum subroutine} is a sequence of unitaries $U_0$,\dots, $U_{{\sf T}_{\max}-1}$ on 
$H_{\cal Z}=\mathrm{span}\{\ket{z}:z\in{\cal Z}\}$
for some finite set ${\cal Z}$. For $X,Y\subseteq{\cal Z}$, we say the subroutine computes an injective function $f:X\rightarrow Y$ in times $\{{\sf T}_x\leq {\sf T}_{\max}\}_{x\in X}$ with errors $\{\epsilon_x\}_{x\in X}$ if:
\begin{enumerate}
\item The map $\sum_{t=0}^{{\sf T}_{\max}-1}\ket{t}\bra{t}\otimes U_t$ can be implemented in ${\sf polylog}({\sf T}_{\max})$ complexity.\label{item:select}
\item For all $x\in X$, $\norm{\ket{f(x)} - U_{{\sf T}_x-1}\dots U_0(\ket{x})}^2\leq \epsilon_x$. \label{item:error}
\item The maps $x\mapsto {\sf T}_x$ and $y\mapsto {\sf T}_{f^{-1}(y)}$ can both be implemented in ${\sf polylog}({\sf T}_{\max})$ complexity.\label{item:known-times}
\item There exists a decomposition ${\cal Z}=\bigcup_{x\in X} {\cal Z}_x$ such that $x,f(x)\in {\cal Z}_x$, and for all $t\in \{0,\dots,{\sf T}_{\max}-1\}$, $U_t\dots U_0\ket{x}\in \mathrm{span}\{\ket{z}:z\in {\cal Z}_x\}$. \label{item:non-confusion}

\end{enumerate}
\end{definition}

While not all of our assumptions are general, they are reasonable in our setting. Item \ref{item:select} is standard in subroutines that will be run in superposition (see~e.g.~\cite{ambainis2010VTSearch}), and is reasonable, for example, in settings where the algorithm is sufficiently structured to compute $U_t$ from a standard gate set on the fly, which we formalise in \lem{uniform-alg} below (see also the discussion in \cite[Section 2.2]{cornelissen2020SpanProgramTime}). 

Item \ref{item:known-times} is not always necessary, but it is often true, and simplifies things considerably. It means, in particular, that one can decide, based on the input, how many steps of the algorithm should be applied, and then, based on the output, uncompute this information.

Item \ref{item:non-confusion}
is not a standard assumption, but it is also not unreasonable. For example, if 
$X=X'\times\{0\}$ and $f(x,0)=(x,g(x))$ for some function $g$, the algorithm may simply use $x$ as a control, and so its state always encodes $x$, and therefore remains orthogonal for different $x$. 

\begin{lemma}\label{lem:uniform-alg}
Call unitaries $U_0,\dots,U_{{\sf T}_{\max}-1}$ on $H$ a \emph{uniform quantum algorithm} if
there exists $\ell={\sf polylog}({\sf T}_{\max})$, unitaries $W_1,\dots,W_{\ell}$, and maps $g:\{0,\dots,{\sf T}_{\max}-1\}\rightarrow[\ell]$ and $g':\{0,\dots,{\sf T}_{\max}-1\}\rightarrow 2^{[\log\dim H]}$ such that:
\begin{enumerate}
\item For each $j\in [\ell]$, $W_j$ can be implemented by ${\sf polylog}({\sf T}_{\max})$ gates from some implicit gate set (and therefore acts on $m={\sf polylog}({\sf T}_{\max})$ qubits).
\item $g$ and $g'$ can be computed in ${\sf polylog}({\sf T}_{\max})$ complexity.
\item For all $t\in\{0,\dots,{\sf T}_{\max}-1\}$, $U_t = W_{g(t)}(g'(t))$, where $W_{\ell}(S)$ denotes $W_{\ell}$ applied to the qubits specified by $S$.
\end{enumerate}
Then $\sum_{t=0}^{{\sf T}_{\max}-1}\ket{t}\bra{t}\otimes U_t$ can be implemented in ${\sf polylog}({\sf T}_{\max})$ gates. 
\end{lemma}
\begin{proof}
We describe how to implement $\sum_{t=0}^{{\sf T}_{\max}-1}\ket{t}\bra{t}\otimes U_t$ on $\ket{t}\ket{z}$ for $\ket{z}\in H$. 
Append registers $\ket{0}_A\ket{0}_{A'}\in\mathrm{span}\{\ket{j,S}:j\in \{0,\dots,\ell\}, S\in {\cal S}\}$, where ${\cal S}$ is the set of subsets of $[\log\dim H]$ of size at most $m$. 
Compute $g(t)$ and $g'(t)$ to get:
$\ket{t}\ket{z}\ket{0}_A\ket{0}_{A'}\mapsto \ket{t}\ket{z}\ket{g(t)}_A\ket{g'(t)}_{A'}.$
Controlled on $g'(t)$, we can swap the qubits acted on by $U_t$ into the first $m$ positions. Then we can implement $\sum_{j=1}^{\ell}\ket{j}\bra{j}\otimes W_j+\ket{0}\bra{0}\otimes I$ by decomposing it into a sequence of $\ell$ controlled operations:
$$\textstyle{\displaystyle\prod}_{j=1}^{\ell}\left(\ket{j}\bra{j}\otimes W_j+(I-\ket{j}\bra{j})\otimes I\right).$$
The result follows from noticing that each of these $\ell={\sf polylog}({\sf T}_{\max})$ operations can be implemented with ${\sf polylog}({\sf T}_{\max})$ controlled gates.
\end{proof}

\begin{lemma}\label{lem:uniform-combine}
Fix a constant integer $c$, and for $j\in [c]$, let ${\cal S}_j$ be a quantum subroutine on $H_j=\mathrm{span}\{\ket{j}\}\otimes H$ for some space $H$ that takes time $\{{\sf T}_x={\sf T}_j\}_{x\in X_j}$ with errors $\{\epsilon_x=\epsilon_j\}_{x\in X_j}$. Then there is a quantum algorithm that implements $\sum_{j=1}^c\ket{j}\bra{j}\otimes {\cal S}_j$ in variable times ${\sf T}_{j,x}=O({\sf T}_j)$ and errors $\epsilon_{j,x}=\epsilon_j$ for all $x\in X_j$.
\end{lemma}
\begin{proof}
Pad each algorithm with identities so that they all have the same number, ${\sf T}_{\max}=\max_{j\in [c]}{\sf T}_{\max}^{(c)}$ of unitaries. Then for $t=\{0,\dots,c{\sf T}_{\max}-1\}$, with $t=qc+r$ for $r\in\{0,\dots,c\}$, let $U_t= \ket{r}\bra{r}\otimes U_q^{(r)}+(I-\ket{r}\bra{r})\otimes I$. 
\end{proof}

\subsection{Quantum Data Structures}\label{sec:data}

We will assume we have access to a data structure that can store a set of keyed items,
$S\subset {\cal I}\times {\cal K},$
for finite sets ${\cal K}$ and ${\cal I}$. For such a stored set $S$, we assume the following can be implemented in ${\sf polylog}(|{\cal I}\times {\cal K}|)$ complexity:
\begin{enumerate}
\item For $(i,k)\in {\cal I}\times {\cal K}$, insert $(i,k)$ into $S$.
\item For $(i,k)\in S$, remove $(i,k)$ from $S$.
\item For $k\in {\cal K}$, query the number of $i\in {\cal I}$ such that $(i,k)\in S$.
\item For $k\in{\cal K}$, return the smallest $i$ such that $(i,k)\in S$.
\item Generate a uniform superposition over all $(i,k)\in S$.  
\end{enumerate}

In addition, for quantum interference to take place, we assume the data structure is coherent, meaning it depends only on $S$, and not on, for example, the order in which elements were added. See~\cite[Section 3.1]{buhrman2022limits} for an example of such a data structure.

\subsection{Probability Theory}\label{sec:prelim-prob}

\paragraph{Hypergeometric Distribution:} In the hypergeometric distribution with parameters $(N,K,d)$, we draw $d$ objects uniformly without replacement from a set of $N$ objects, $K$ of which are marked, and consider the number of marked objects that are drawn.  We will use the following:
\begin{lemma}[Hypergeometric Tail Bounds~\cite{Janson2011RandomGraphs}]\label{lem:hypergeo}
Let $Z$ be a hypergeometric random variable with parameters $(N,K,d)$, and $\mu=\frac{Kd}{N}$. Then for every $B\geq 7\mu$, $\Pr[Z\geq B]\leq e^{-B}$. Furthermore, for every $\epsilon>0$, 
$$\Pr[|Z-\mu|\geq \epsilon\mu]\leq 2\exp\{-((1+\epsilon)\log(1+\epsilon)-\epsilon)\mu\}.$$
\end{lemma}

\noindent We will make use of the second bound from \lem{hypergeo} in the following form, when $\mu=o(1)$:
\begin{corollary}\label{cor:hypergeo}
Let $Z$ be a hypergeometric random variable with parameters $(N,K,d)$, and $\mu=\frac{Kd}{N}$. Then $\Pr[Z\geq c]\leq 2e^c(c/\mu)^{-c}$. 
\end{corollary}

\paragraph{$d$-wise Independence} It will be convenient to divide the input into blocks, which we will argue are random. In order to avoid the $\widetilde{\Theta}(n)$ cost of sampling a uniform random permutation to define these blocks, we use a $d$-wise independent family of permutations.
\begin{definition}\label{def:d-wise}
Let $\{\tau_s:[n]\rightarrow[n]\}_{s\in {\cal S}}$ for some finite seed set ${\cal S}$. For $d\in \mathbb{N}$ and $\delta\in (0,1)$, we say that $\tau_s$ is a $d$-wise $\delta$-independent permutation (family) if for $s$ chosen uniformly at random from ${\cal S}$, for any distinct $i_1,\dots,i_d\in [n]$ and distinct $i_1',\dots,i_d'\in [n]$,
$$\abs{\Pr[\tau_s(i_1)=i_1',\dots,\tau_s(i_d)=i_d'] - \frac{1}{n(n-1)\dots (n-d+1)}}\leq \delta.$$
\end{definition}

For $d\in {\sf polylog}(n)$, and any $\delta \in (0,1)$, there exist families of $d$-wise $\delta$-independent permutations with the following properties (see, for example, \cite{kaplan2009derandomized}):
\begin{itemize}
\item $s$ can be sampled in $O(d\log n \log\frac{1}{\delta})$ time and space.
\item For any $s\in {\cal S}$, $i\in [n]$, $\tau_s(i)$ can be computed in time ${\sf poly}(d\log n\log\frac{1}{\delta})$.
\item For any $s\in {\cal S}$, $i'\in [n]$, $\tau_s^{-1}(i')$ can be computed in time ${\sf poly}(d\log n\log\frac{1}{\delta})$.\footnote{For example, in $d$-wise independent permutation families based on Feistel networks applied to $d$-wise independent functions $h$, inverting $\tau_h$ is as easy as computing $h$.}
\end{itemize}
We will design our algorithms assuming such a construction for $\delta=0$, although this is not known to exist. By taking $\delta$ to be a sufficiently small inverse polynomial, our algorithm will not notice the difference.

\section{Framework}\label{sec:fwk}

In this section, we present the Multidimensional Quantum Walk Framework, which defines a quantum algorithm from a network, and certain subroutines. In \sec{phase-est}, we describe the type of algorithm that will be used to prove our main theorem. In \sec{full-framework}, we state and prove our main theorem, \thm{full-framework}. 

\subsection{Phase Estimation Algorithms}\label{sec:phase-est}

In this section, we formally define a particular kind of quantum algorithm that uses phase estimation~\cite{kitaev1996PhaseEst}, and describe ingredients sufficient to analyse such an algorithm. All algorithms in this paper are of this specific form. 

\begin{definition}[Parameters of a Phase Estimation Algorithm]\label{def:lin-alg}
For an implicit input $x\in\{0,1\}^*$, fix a finite-dimensional complex inner product space $H$, a unit vector $\ket{\psi_0}\in H$, and sets of vectors
$\Psi^{\cal A},\Psi^{\cal B}\subset H$. We further assume that $\ket{\psi_0}$ is orthogonal to every vector in $\Psi^{\cal B}$.  Let $\Pi_{\cal A}$ be the orthogonal projector onto ${\cal A}=\mathrm{span}\{\Psi^{\cal A}\}$, and similarly for $\Pi_{\cal B}$.
\end{definition}
Then $(H,\ket{\psi_0},\Psi^{\cal A},\Psi^{\cal B})$ defines a quantum algorithm as follows. Let 
\begin{equation}
U_{\cal AB}=(2\Pi_{\cal A}-I)(2\Pi_{\cal B}-I).\label{eq:U-AB}
\end{equation}
Do phase estimation\footnote{For what is meant precisely by ``phase estimation'', refer to the proof of \thm{lin-alg-fwk}.} of $U_{\cal AB}$ on initial state $\ket{\psi_0}$ to a certain precision, measure the phase register, and output 1 if the measured phase is 0, and output 0 otherwise. \thm{lin-alg-fwk} at the end of this section describes what precision is sufficient, and when we can expect the output to be 1 and when 0. 

In practice, unitaries like $U_{\cal AB}$ that are the product of two reflections are nice to work with because if each of $\Psi^{\cal A}$ and $\Psi^{\cal B}$ is a pairwise orthogonal set, implementing $U_{\cal AB}$ can be reduced to generating the states in $\Psi^{\cal A}$ and $\Psi^{\cal B}$ respectively, and a product of reflections has sufficient structure to analyse the relevant eigenspaces, as will become clear throughout this section.

\paragraph{Negative Analysis:}
The first of the two cases we want to distinguish with a phase estimation algorithm is the \emph{negative case}, in which there exists a negative witness, defined as follows.

\begin{definition}[Negative Witness]\label{def:neg-witness}
A $\delta$-\emph{negative witness} for $(H,\ket{\psi_0},\Psi^{\cal A},\Psi^{\cal B})$ is a pair of vectors $\ket{w_{\cal A}},\ket{w_{\cal B}}\in H$ such that $\ket{\psi_0}=\ket{w_{\cal A}}+\ket{w_{\cal B}}$; and $\ket{w_{\cal A}}$ is mostly in the space $\cal A$, and $\ket{w_{\cal B}}$ is mostly in the space $\cal B$, in the sense that $\norm{(I-\Pi_{\cal A})\ket{w_{\cal A}}}^2\leq \delta$ and $\norm{(I-\Pi_{\cal B})\ket{w_{\cal B}}}^2\leq \delta$. 
\end{definition}

For intuition, it is useful to think of the case when $\delta=0$. There exists a 0-negative witness precisely when $\ket{\psi_0}\in {\cal A}+{\cal B} = ({\cal A}^\bot\cap {\cal B}^\bot)^\bot$. For the rest of this section, we write  $\Lambda_\Theta$ for the orthogonal projector onto the span of the $e^{i\theta}$-eigenspaces of $U_{\cal AB}$ with $|\theta|\leq \Theta$. The negative analysis relies on the effective spectral gap lemma:
\begin{lemma}[Effective Spectral Gap Lemma~\cite{lee2011QQueryCompStateConv}]\label{lem:effective-spectral-gap}
Fix $\Theta\in (0,\pi)$. If $\ket{\psi_{\cal A}}\in {\cal A}$, then 
$$\norm{\Lambda_\Theta(I-\Pi_{\cal B})\ket{\psi_{\cal A}}} \leq \frac{\Theta}{2}\norm{\ket{\psi_{\cal A}}}.$$
\end{lemma}

\begin{lemma}[Negative Analysis]\label{lem:lin-alg-negative-approx}
Fix $\delta\geq 0$ and $\Theta\in (0,\pi)$. Suppose there exists a $\delta$-negative witness, $\ket{w_{\cal A}}$, $\ket{w_{\cal B}}$, for $(H,\ket{\psi_0},\Psi^{\cal A},\Psi^{\cal B})$. Then we have:
\begin{equation*}
\norm{\Lambda_\Theta\ket{\psi_0}} \leq \frac{\Theta}{2}\norm{\ket{w_{\cal A}}}+2\sqrt{\delta}.
\end{equation*}
\end{lemma}
\begin{proof}
We can apply the effective spectral gap lemma to $\Pi_{\cal A}\ket{w_{\cal A}}\in {\cal A}$, to get:
\begin{align*}
&\frac{\Theta}{2}\norm{\Pi_{\cal A}\ket{w_{\cal A}}} \geq \norm{\Lambda_\Theta(I-\Pi_{\cal B})\Pi_{\cal A}\ket{w_{\cal A}}} \\
&\frac{\Theta}{2}\norm{\ket{w_{\cal A}}} \geq \norm{\Lambda_\Theta\left( I-\Pi_{\cal B} -(I-\Pi_{\cal B})(I-\Pi_{\cal A}) \right)\ket{w_{\cal A}}}\\
&\geq \norm{\Lambda_\Theta( I-\Pi_{\cal B})\ket{w_{\cal A}}} - \norm{\Lambda_\Theta(I-\Pi_{\cal B})(I-\Pi_{\cal A}) \ket{w_{\cal A}}} & \mbox{by the triangle ineq.}\\
&\geq \norm{\Lambda_\Theta( I-\Pi_{\cal B})(\ket{\psi_0}-\ket{w_{\cal B}})} - \norm{(I-\Pi_{\cal A}) \ket{w_{\cal A}}} & \mbox{since }\ket{\psi_0}=\ket{w_{\cal A}}+\ket{w_{\cal B}}.\\
&\geq \norm{\Lambda_\Theta(I-\Pi_{\cal B})\ket{\psi_0}} - \norm{\Lambda_\Theta(I-\Pi_{\cal B})\ket{w_{\cal B}}} - \norm{(I-\Pi_{\cal A}) \ket{w_{\cal A}}}. & \mbox{by the triangle ineq.}
\end{align*}
Since $\ket{\psi_0}$ is orthogonal to ${\cal B}$, and $\norm{(I-\Pi_{\cal A})\ket{w_{\cal A}}}\leq \sqrt{\delta}$ and similarly for ${\cal B}$, the result follows. \end{proof}

\paragraph{Positive Analysis:} 
We want to distinguish the case where there exists a negative witness (the negative case) from the \emph{positive case}, which is the case where there exists a positive witness, defined as follows.

\begin{definition}[Positive Witness]\label{def:pos-witness}
A $\delta$-positive witness for $(H,\ket{\psi_0},\Psi^{\cal A},\Psi^{\cal B})$ is a vector $\ket{w}\in H$ such that $\braket{\psi_0}{w}\neq 0$ and $\ket{w}$ is almost orthogonal to all $\ket{\psi}\in \Psi^{\cal A}\cup \Psi^{\cal B}$, in the sense that $\norm{\Pi_{\cal A}\ket{w}}^2\leq \delta\norm{\ket{w}}^2$ and $\norm{\Pi_{\cal B}\ket{w}}^2\leq \delta\norm{\ket{w}}^2$.\footnote{We note that for technical reasons, positive witness error is defined multiplicatively (relative error), whereas negative witness error was defined additively. We could also have defined negative witness error as relative error, but it would have been relative to $\norm{\ket{w_{\cal A}}}^2$ for both $\norm{(I-\Pi_{\cal A})\ket{w_{\cal A}}}^2$ (makes perfect sense) and $\norm{(I-\Pi_{\cal B})\ket{w_{\cal B}}}^2$ (confusing).} 
\end{definition}

Again, for intuition, we consider the case where $\delta=0$. A 0-positive witness is exactly a component of $\ket{\psi_0}$ in $({\cal A}+{\cal B})^\bot$, which exists precisely when $\ket{\psi_0}\not\in {\cal A}+{\cal B}$. Thus, the case where there exists a 0-positive witness is the complement of the case where there exists a 0-negative witness, so it is theoretically possible to distinguish these two cases. When $\delta>0$, the two cases may or may not be distinct, depending on $\delta$, and the overlap between ${\cal A}$ and ${\cal B}$. 

When $\ket{w}$ is a $0$-positive witness, it is straightforward to see that 
\begin{equation*}
\norm{\Lambda_0\ket{\psi_0}} \geq \frac{|\braket{w}{\psi_0}|}{\norm{\ket{w}}},
\end{equation*}
where $\Lambda_0$ is the orthogonal projector onto the $(+1)$-eigenspace of $U_{\cal AB}$. For the case of $\delta>0$, we need the following lemma, analogous to the effective spectral gap lemma.

\begin{lemma}[Effectively Zero Lemma]\label{lem:effectively-zero}
Fix $\delta\geq0$ and $\Theta\in (0,\pi)$. For $\ket{\psi}\in H$ such that $\norm{\Pi_{\cal A}\ket{\psi}}^2\leq \delta\norm{\ket{\psi}}^2$ and  $\norm{\Pi_{\cal B}\ket{\psi}}^2\leq \delta\norm{\ket{\psi}}^2$, 
\begin{equation*}
\norm{(I-\Lambda_\Theta)\ket{\psi}}^2  \leq \frac{4\pi^2\delta\norm{\ket{\psi}}^2}{\Theta^2}.
\end{equation*}
\end{lemma}
\begin{proof}
Let $\{\theta_j\}_{j\in J}\subset (-\pi,\pi]$ be the set of eigenphases of $U_{\cal AB}$, and let $\Pi_j$ be the orthogonal projector onto the $e^{i\theta_j}$-eigenspace of $U_{\cal AB}$, so we can write:
\begin{equation}
U_{\cal AB} = \textstyle{\displaystyle\sum}_{j\in J} e^{i\theta_j}\Pi_j.\label{eq:U-AB-form}
\end{equation}
We have (see \eq{U-AB})
\begin{align}
U_{\cal AB}\ket{\psi} &= \ket{\psi} + 4\Pi_{\cal A}\Pi_{\cal B}\ket{\psi}-2\Pi_{\cal A}\ket{\psi}-2\Pi_{\cal B}\ket{\psi},\label{eq:eff-0-a}
\end{align}
and using the triangle inequality, $\norm{\Pi_{\cal A}\ket{\psi}}^2\leq \delta\norm{\ket{\psi}}^2$ and $\norm{\Pi_{\cal B}\ket{\psi}}^2 \leq \delta\norm{\ket{\psi}}^2$, we can compute
\begin{equation}
\begin{split}
\norm{4\Pi_{\cal A}\Pi_{\cal B}\ket{\psi}-2\Pi_{\cal A}\ket{\psi}-2\Pi_{\cal B}\ket{\psi}}^2 
&= \norm{2(2\Pi_{\cal A}-I)\Pi_{\cal B}\ket{\psi}-2\Pi_{\cal A}\ket{\psi}}^2\\
&\leq \left(\norm{2(2\Pi_{\cal A}-I)\Pi_{\cal B}\ket{\psi}} + \norm{2\Pi_{\cal A}\ket{\psi}}\right)^2 
\leq 16\delta\norm{\ket{\psi}}^2.
\end{split}\label{eq:eff-0-b}
\end{equation}
Thus, by \eq{eff-0-a} and \eq{eff-0-b}:
\begin{align*}
\norm{U_{\cal AB}\ket{\psi}-\ket{\psi}}^2 &\leq 16\delta\norm{\ket{\psi}}^2\\
\sum_{j\in J}|e^{i\theta_j}-1|^2\norm{\Pi_j\ket{\psi}}^2 & \leq 16\delta\norm{\ket{\psi}}^2 & \mbox{by \eq{U-AB-form}}\\
\sin^2\frac{\Theta}{2}\sum_{j\in J:|\theta_j|>\Theta}\norm{\Pi_j\ket{\psi}}^2 & \leq 4\delta\norm{\ket{\psi}}^2 & \mbox{since }|e^{i\theta_j}-1|^2 = 4\sin^2\frac{\theta_j}{2}\\
\norm{(I-\Lambda_\Theta)\ket{\psi}}^2 & \leq \frac{4\delta\norm{\ket{\psi}}^2}{\sin^2\frac{\Theta}{2}}.
\end{align*}
Then since $\sin^2\frac{\Theta}{2}\geq \frac{4}{\pi^2}\frac{\Theta^2}{4}$ whenever $\Theta\in (-\pi,\pi)$, the result follows.
\end{proof}

\begin{lemma}[Positive Analysis] \label{lem:lin-alg-positive-approx}
Fix $\delta\geq 0$ and $\Theta\in (0,\pi)$.
Suppose there exists a $\delta$-positive witness $\ket{w}$ for $(H,\ket{\psi_0},\Psi^{\cal A},\Psi^{\cal B})$. Then
\begin{equation*}
\norm{\Lambda_\Theta\ket{\psi_0}} \geq \frac{|\braket{\psi_0}{w}|}{\norm{\ket{w}}} - \frac{2\sqrt{\delta}\pi}{\Theta}.
\end{equation*}
\end{lemma}
\begin{proof}
We compute:
\begin{align*}
|\bra{\psi_0}\Lambda_\Theta\ket{w}| &\geq |\braket{\psi_0}{w}| - |\bra{\psi_0}(I-\Lambda_\Theta)\ket{w}| & \mbox{by the triangle ineq.}\\
&\geq |\braket{\psi_0}{w}| - \norm{\ket{\psi_0}}\cdot\norm{(I-\Lambda_\Theta)\ket{w}} & \mbox{by Cauchy-Schwarz}\\
&\geq |\braket{\psi_0}{w}| - \frac{2\sqrt{\delta}\pi}{\Theta}\norm{\ket{w}},
\end{align*}
where we used $\norm{\ket{\psi_0}}=1$ and \lem{effectively-zero}.
Then:
\begin{align*}
\norm{\Lambda_\Theta\ket{\psi_0}} &\geq \norm{\frac{\Lambda_\Theta\ket{w}\bra{w}\Lambda_\Theta}{\norm{\Lambda_\Theta\ket{w}}^2}\ket{\psi_0}}
= \frac{|\bra{w}\Lambda_\Theta\ket{\psi_0}|}{\norm{\Lambda_\Theta\ket{w}}}
\geq \frac{|\braket{\psi_0}{w}|-\frac{2\sqrt{\delta}\pi}{\Theta}\norm{\ket{w}}}{\norm{\ket{w}}}
= \frac{|\braket{\psi_0}{w}|}{\norm{\ket{w}}} - \frac{2\sqrt{\delta}\pi}{\Theta}.\qedhere
\end{align*}
\end{proof}

\paragraph{Phase Estimation Algorithm:}

By \lem{lin-alg-positive-approx}, if there exists a $\delta$-positive witness, which happens precisely when there is some component of $\ket{\psi_0}$ that is nearly orthogonal to ${\cal A}+{\cal B}$, then $\ket{\psi_0}$ overlaps the $e^{i\theta}$-eigenspaces of $U_{\cal AB}$ for small $\theta$, say with $|\theta|\leq \Theta_0$ for some small-ish choice of $\Theta_0$. The precise overlap depends on the size of this component, and allows us to lower bound the probability that phase estimation of $U_{\cal AB}$ on $\ket{\psi_0}$ will result in a 0 in the phase register. On the other hand, if $\ket{\psi_0}$ is actually in ${\cal A}+{\cal B}$, then \lem{lin-alg-negative-approx} upper bounds the overlap of $\ket{\psi_0}$ with small phase spaces, where ``small'' is determined by the parameter $\Theta>\Theta_0$. This allows us to upper bound the probability that phase estimation of $U_{\cal AB}$ on $\ket{\psi_0}$, to precision $\Theta$, will result in a 0 in the phase register. The key is then to choose the parameter $\Theta$ small enough so that there is a constant gap between the lower bound on the probability of a 0 phase in the positive case, and the upper bound on the probability of a 0 phase in the negative case. This leads to the following theorem.

\begin{theorem}\label{thm:lin-alg-fwk}
Fix $(H,\ket{\psi_0},\Psi^{\cal A},\Psi^{\cal B})$ as in \defin{lin-alg}.
Suppose we can generate the state $\ket{\psi_0}$ in cost ${\sf S}$, and implement $U_{\cal AB}=(2\Pi_{\cal A}-I)(2\Pi_{\cal B}-I)$ in cost ${\sf A}$.

\noindent Let $c_+\in [1,50]$ be some constant, and let ${\cal C}_-\geq 1$ be a positive real number that may scale with $|x|$. Let $\delta$ and $\delta'$ be positive real parameters such that 
\begin{equation*}
\delta\leq \frac{1}{(8c_+)^{3}\pi^8{{\cal C}_-}}\quad\mbox{ and }\quad\delta'\leq \frac{3}{4}\frac{1}{\pi^4c_+}.
\end{equation*}
Suppose we are guaranteed that exactly one of the following holds:
\begin{description}
\item[Positive Condition:] There is a $\delta$-positive witness $\ket{w}$ (see \defin{pos-witness}), s.t.~$\frac{|\braket{w}{\psi_0}|^2}{\norm{\ket{w}}^2}\geq \frac{1}{c_+}$. 
\item[Negative Condition:] There is a $\delta'$-negative witness $\ket{w_{\cal A}},\ket{w_{\cal B}}$ (\defin{neg-witness}), s.t.~$\norm{\ket{w_{\cal A}}}^2 \leq {\cal C}_-$.
\end{description}
Suppose we perform $T= \sqrt{8}\pi^4 c_+\sqrt{{\cal C}_-}$ steps of phase estimation of $U_{\cal AB}$ on initial state $\ket{\psi_0}$, and output 1 if and only if the measured phase is 0, otherwise we output 0. Then
\begin{description}
\item[Positive Case:] If the positive condition holds, the algorithm outputs 1 with probability 
$
\geq\frac{2.25}{\pi^2c_+}\geq \frac{2.25}{50\pi^2}.
$ 
\item[Negative Case:] If the negative condition holds, the algorithm outputs 1 with probability 
$
\leq\frac{2}{\pi^2c_+}.
$ 
\end{description}
Thus, the algorithm distinguishes between these two cases with bounded error, in cost
$$O\left({\sf S}+\sqrt{{\cal C}_-}{\sf A}\right).$$
\end{theorem}
\begin{proof}
Let $\{\theta_j\}_{j\in J}\subset (-\pi,\pi]$ be the set of phases of $U_{\cal AB}$, and let $\Pi_j$ be the orthogonal projector onto the $e^{i\theta_j}$-eigenspace of $U_{\cal AB}$, so we can write:
\begin{equation*}
U_{\cal AB} = \textstyle{\displaystyle\sum}_{j\in J} e^{i\theta_j}\Pi_j.
\end{equation*}

After making a superposition over $t$ from $0$ to $T-1$ in the phase register, and applying $U_{\cal AB}^t$ to the input register conditioned on the phase register,  as one does in phase estimation~\cite{kitaev1996PhaseEst}, we have the state:
\begin{equation*}
\begin{split}
\sum_{t=0}^{T-1}\frac{1}{\sqrt{T}}\ket{t}U_{\cal AB}^t\ket{\psi_0} &=  \sum_{j\in J}\sum_{t=0}^{T-1}\frac{1}{\sqrt{T}}\ket{t}e^{it\theta_j}\Pi_j\ket{\psi_0}.
\end{split}
\end{equation*}
The phase estimation algorithm then proceeds by applying an inverse Fourier transform, $F_T^\dagger$, to the first register, and then measuring the result, to obtain some $t\in \{0,\dots,T-1\}$. We choose the output bit based on whether $t=0$ or not. The probability of measuring 0 is:
\begin{equation}
\begin{split}
p_0&:=  \norm{\bra{0}F_T^\dagger \otimes I\left( \sum_{j\in J} \sum_{t=0}^{T-1}\frac{1}{\sqrt{T}}\ket{t}e^{it\theta_j}\Pi_j\ket{\psi_0} \right)}^2 
= \norm{\sum_{t=0}^{T-1}\frac{1}{\sqrt{T}}\bra{t}\otimes I \left( \sum_{j\in J} \sum_{t=0}^{T-1}\frac{1}{\sqrt{T}}\ket{t}e^{it\theta_j}\Pi_j\ket{\psi_0} \right)}^2\\
&= \frac{1}{T^2}\norm{\sum_{j\in J}\sum_{t=0}^{T-1}e^{it\theta_j}\Pi_j\ket{\psi_0}}^2
= \frac{1}{T^2}\sum_{j\in J:\theta_j\neq 0}\left| \frac{1-e^{i\theta_j T}}{1-e^{i\theta_j}} \right|^2 \norm{\Pi_j\ket{\psi_0}}^2 + \norm{\Lambda_0\ket{\psi_0}}^2\\
&=\frac{1}{T^2}\sum_{j\in J:\theta_j\neq 0}\frac{\sin^2(T\theta_j/2)}{\sin^2(\theta_j/2)}\norm{\Pi_j\ket{\psi_0}}^2+\norm{\Lambda_0\ket{\psi_0}}^2,
\end{split}\label{eq:phase-prob}
\end{equation}
since $\abs{\sum_{t=0}^{T-1}e^{i t\theta}} = \abs{\frac{1-e^{i\theta T}}{1-e^{i\theta}}}$,
and $|1-e^{i\theta}|^2=4\sin^2\frac{\theta}{2}$ for any $\theta\in\mathbb{R}$.
We will analyse the positive and negative cases one-by-one. 

\paragraph{Positive Case:} Assume the positive condition, which allows us to apply \lem{lin-alg-positive-approx}. In the following, we will use the identities $\sin^2\theta\leq \theta^2$ for all $\theta$, and $\sin^2\theta\geq 4\theta^2/\pi^2$ whenever $|\theta|\leq \pi/2$. 
Let $\Theta_0=\pi/T$. 
Continuing from \eq{phase-prob}, 
we can lower bound the probability of measuring a 0 in the phase register by:
\begin{align*}
p_0 &\geq \frac{1}{T^2}\sum_{j\in J:0<|\theta_j|\leq\Theta_0}\frac{\sin^2(T\theta_j/2)}{\sin^2(\theta_j/2)}\norm{\Pi_j\ket{\psi_0}}^2+\norm{\Lambda_0\ket{\psi_0}}^2\\
&\geq \frac{1}{T^2} \sum_{j\in J:0<|\theta_j|\leq\Theta_0}\frac{4(T\theta_j/2)^2/\pi^2}{(\theta_j/2)^2}\norm{\Pi_j\ket{\psi_0}}^2+\norm{\Lambda_0\ket{\psi_0}}^2 & \mbox{since }|T\theta_j/2|\leq T\Theta_0/2= \pi/2 \\
&\geq \frac{4}{\pi^2} \left(\sum_{j\in J:0<|\theta_j|\leq\Theta_0}\norm{\Pi_j\ket{\psi_0}}^2+\norm{\Lambda_0\ket{\psi_0}}^2 \right) \\
&\geq \frac{4}{\pi^2}\norm{\Lambda_{\Theta_0}\ket{\psi_0}}^2 \geq \frac{4}{\pi^2}\left(\frac{|\braket{\psi_0}{w}|}{\norm{\ket{w}}}-\frac{2\sqrt{\delta}\pi}{\Theta_0}\right)^2 & \mbox{by \lem{lin-alg-positive-approx}}\\
&\geq \frac{4}{\pi^2}\left(\frac{1}{\sqrt{c_+}}-\frac{2\pi T}{\pi}\frac{1}{(8c_+)^{3/2}\pi^4\sqrt{{\cal C}_-}}\right)^2 & \sqrt{\delta}\leq \frac{1}{(8c_+)^{3/2}\pi^4\sqrt{{\cal C}_-}}\\
&= \frac{4}{\pi^2}\left(\frac{1}{\sqrt{c_+}}-\frac{2\sqrt{8}\pi^4c_+\sqrt{{\cal C}_-}}{(8c_+)^{3/2}\pi^4\sqrt{{\cal C}_-}}\right)^2
= \frac{4}{\pi^2}\left(\frac{3}{4}\right)^2\frac{1}{c_+}
= \frac{2.25}{\pi^2}\frac{1}{c_+} \geq \frac{2.25}{50\pi^2}.\!\!\!\!\!\!\!\!
\end{align*}

\paragraph{Negative Case:} Assume the negative condition, which allows us to apply \lem{lin-alg-negative-approx}. In the following, we will use the identities $\sin^2\theta\leq \min\{1,\theta^2\}$ for all $\theta$, and $\sin^2(\theta/2)\geq \theta^2/\pi^2$ whenever $|\theta|\leq \pi$. Let $\Theta = {\pi^{-2}({c_+{\cal C}_-})^{-1/2}}$.
Continuing from \eq{phase-prob}, 
we can \emph{upper} bound the probability of measuring a 0 in the phase register by:
\begin{align*}
p_0
&=\frac{1}{T^2}\sum_{j\in J:0<|\theta_j|\leq \Theta}\frac{\sin^2(T\theta_j/2)}{\sin^2(\theta_j/2)}\norm{\Pi_j\ket{\psi_0}}^2
+\frac{1}{T^2}\sum_{j\in J:|\theta_j|> \Theta}\frac{\sin^2(T\theta_j/2)}{\sin^2(\theta_j/2)}\norm{\Pi_j\ket{\psi_0}}^2+\norm{\Lambda_0\ket{\psi_0}}^2\!\!\!\!\!\!\!\!\!\!\!\!\!\!\!\!\!\!\!\!\!\!\!\!\!\!\!\!\!\!\!\!\!\!\!\!\!\!\!\!\!\!\!\!\!\!\!\!\!\!\!\!\!\!\!\!\!\!\!\!\!\!\!\!\!\!\!\!\!\!\!\!\!\!\!\!\!\!\!\!\!\!\!\!\\
&\leq \frac{1}{T^2}\sum_{j\in J:0<|\theta_j|\leq \Theta}\frac{(T\theta_j/2)^2}{(\theta_j/\pi)^2}\norm{\Pi_j\ket{\psi_0}}^2+\frac{1}{T^2}\sum_{j\in J:|\theta_j|>\Theta}\frac{1}{(\theta_j/\pi)^2}\norm{\Pi_j\ket{\psi_0}}^2+\norm{\Lambda_0\ket{\psi_0}}^2\!\!\!\!\!\!\!\!\!\!\!\!\!\!\!\!\!\!\!\!\!\!\!\!\!\!\!\!\!\!\!\!\!\!\!\!\!\!\!\!\!\!\!\!\!\!\!\!\!\!\!\!\!\!\!\!\!\!\!\!\!\!\!\!\!\!\!\!\!\!\!\!\!\!\\
&\leq \frac{\pi^2}{4}\norm{\Lambda_\Theta\ket{\psi_0}}^2+\frac{1}{T^2}\frac{\pi^2}{\Theta^2}\\
&\leq \frac{\pi^2}{4}\left(\frac{\Theta}{2}\norm{\ket{w_{\cal A}}}+2\sqrt{\delta'}\right)^2+\frac{1}{8\pi^8c_+^2{\cal C}_-}\frac{\pi^2}{\Theta^2} & \mbox{by \lem{lin-alg-negative-approx}}\\
&\leq \frac{\pi^2}{4}\left(\frac{1}{\pi^2\sqrt{c_+{\cal C}_-}}\sqrt{\cal C}_-+2\frac{\sqrt{3}}{2\pi^2\sqrt{c_+}}\right)^2 + \frac{\pi^2}{8\pi^8c_+^2{\cal C}_-}\pi^4 c_+{\cal C}_- & \sqrt{\delta'}\leq \frac{\sqrt{3}}{2}\frac{1}{\pi^2\sqrt{c_+}}\mbox{ and } \norm{\ket{w_{\cal A}}}^2\leq {\cal C}_-\\
&\leq \frac{1}{4\pi^2c_+}\left(1+\sqrt{3}\right)^2 + \frac{1}{8\pi^2c_+}
 \leq \frac{2}{\pi^2c_+}.
\end{align*}
To complete the proof, it is easily verified that the described algorithm has the claimed cost.
\end{proof}

\subsection{Multidimensional Quantum Walks}\label{sec:full-framework}

Our new framework extends the electric network framework, so for intuition, we first describe this framework and sketch how it works, before stating our main theorem extending the electric network framework in \sec{alternative}.

\subsubsection{Sketch of the Electric Network Framework}\label{sec:sketch}

We begin by sketching the electric network framework of Belovs~\cite{belovs2013ElectricWalks}, on which our new framework is based. This explanation is for intuition, and the expert reader may skip it. 
Fix a network $G$ that can depend on an input $x\in\{0,1\}^n$, as in \defin{network}. Let $V_0,V_{\sf M}\subset V(G)$ be disjoint sets, $\sigma$ an \emph{initial distribution} on $V_0$, and $M\subseteq V_{\sf M}$ a \emph{marked set}. For simplicity, let us assume that $G$ is bipartite, which is always possible to ensure, for example, by replacing each edge with a path of length two.\footnote{In fact, we essentially do this in the proof of our new framework, since we replace each edge with a sort of ``algorithm gadget'', which is analogous to a path, and always has even length.}  Let $(V_{\cal A},V_{\cal B})$ be the bipartition, and assume for convenience that $V_0\subseteq V_{\cal A}$, and $V_{\sf M}\subseteq V_{\cal B}$.

Then the electric network framework is a way of designing a quantum algorithm that detects if $M\neq \emptyset$ with bounded error. 
To explain how this algorithm works conceptually, we will modify $G$ by adding a vertex $v_0$, which is connected to each vertex $u\in V_0$ by an edge of weight $\w_0\sigma(u)$ for some parameter $\w_0$, and to each vertex in $M$ by an edge of weight $\w_{\sf M}$.
Call this new graph $G'$. We assume the new edges are pointing into $v_0$, so 
$$V(G')=V(G)\cup\{v_0\}
\mbox{ and }
\overrightarrow{E}(G')=\overrightarrow{E}(G)\cup\{(u,v_0):u\in V_0\cup M\}.$$

\paragraph{The Algorithm:} We describe a phase estimation algorithm, of the form given in \sec{phase-est}. Let 
$$H=\mathrm{span}\{\ket{u,v}:(u,v)\in\overrightarrow{E}(G')\}.$$
For any $(u,v)\in \overrightarrow{E}(G')$, $H$ does not contain a vector $\ket{v,u}$, so we define:
$$\ket{v,u} := -\ket{u,v}.$$
It is conceptually convenient to think of negation as reversing the direction of an edge, but note that this means that $(u,v)$ and $(v,u)$ are \emph{not} labelling orthogonal states. 
For each $u\in V(G)\setminus (V_0\cup M)$, we define a \emph{star state}\footnote{To simplify things here, these are slightly different from the star states we use in \thm{full-framework} and \defin{alternative}. They are the same, up to sign, if $\overrightarrow{E}(G)\subset V_{\cal A}\times V_{\cal B}$.}:
\begin{equation*}
\ket{\psi_\star^{G'}(u)}:=\textstyle{\displaystyle\sum}_{v\in\Gamma(u)}\sqrt{\w_{u,v}}\ket{u,v}=\textstyle{\displaystyle\sum}_{v\in\Gamma^+(u)}\sqrt{\w_{u,v}}\ket{u,v} - \textstyle{\displaystyle\sum}_{v\in\Gamma^-(u)}\sqrt{\w_{u,v}}\ket{v,u},
\end{equation*}
where the last expression shows how to express $\ket{\psi_\star^{G'}(u)}$ in the standard basis of $H$, which only includes $(u,v)\in\overrightarrow{E}(G')$. Recall that $\Gamma^+(u)$ and $\Gamma^-(u)$ are the out- and in-neighbourhoods of $u$, defined in \eq{neighbourhoods}.
Similarly, for $u\in V_0$, we define
\begin{equation*}
\ket{\psi_\star^{G'}(u)}:=\textstyle{\displaystyle\sum}_{v\in\Gamma(u)}\sqrt{\w_{u,v}}\ket{u,v}+\sqrt{\w_0\sigma(u)}\ket{u,v_0},
\end{equation*}
and for $u\in M$,
\begin{equation*}
\ket{\psi_\star^{G'}(u)}:=\textstyle{\displaystyle\sum}_{v\in\Gamma(u)}\sqrt{\w_{u,v}}\ket{u,v}+\sqrt{\w_{\sf M}}\ket{u,v_0},
\end{equation*}
which simply includes the new edges we add when going from $G$ to $G'$. The star states are not normalised, but if we normalise $\ket{\psi_\star^{G'}(u)}$, we get a \emph{quantum walk state}: a state that, if measured, would allow one to sample from the neighbours of $u$ as in a random walk on $G'$. 

We use an initial state based on the initial distribution $\sigma$:
$$\ket{\psi_0}=\textstyle{\displaystyle\sum}_{u\in V_0}\sqrt{\sigma(u)}\ket{u,v_0}=\ket{\sigma}\ket{v_0}.$$

Finally, let 
$$\Psi^{\cal A}=\{\ket{\psi_\star^{G'}(u)}:u\in V_{\cal A}\}
\mbox{ and }
\Psi^{\cal B}=\{\ket{\psi_\star^{G'}(u)}:u\in V_{\cal B}\}.$$
Since $V_{\cal A}$ and $V_{\cal B}$ are each independent sets, each set is a pairwise orthogonal set of states. Thus, being able to generate these states\footnote{Note that $v_0\not\in V_{\cal A}\cup V_{\cal B}$. This is important, because generating the star state for $v_0$ would require knowing precisely which vertices of $V_{\cal B}$ are in $M$.} is sufficient to be able to reflect around ${\cal A}=\mathrm{span}\{\Psi^{\cal A}\}$ and ${\cal B}=\mathrm{span}\{\Psi^{\cal B}\}$, in order to implement:
$$U_{\cal AB}=(2\Pi_{\cal A}-I)(2\Pi_{\cal B}-I).$$
It can then be verified that $({\cal A}+{\cal B})^\bot$ is the span of all \emph{circulation} states on $G'$:
$$\ket{C}=\textstyle{\displaystyle\sum}_{(u,v)\in \overrightarrow{E}(G')}\displaystyle\frac{C(u,v)}{\sqrt{\w_{u,v}}}\ket{u,v},$$
where $C$ is a circulation (see \defin{flow}). Thus, a 0-positive witness for a phase estimation algorithm with these parameters (see \defin{pos-witness}) is always a circulation. 

\paragraph{Positive Case:} Suppose that whenever $M\neq\emptyset$, there exists a flow $\theta$ on $G$ (see \defin{flow}) with sources in $V_0$ and sinks in $M$, and suppose $\theta(u)=\sigma(u)$ for all $u\in V_0$. Then we can extend $\theta$ to a circulation, $C_\theta$, on $G'$ by sending all excess flow from $M$ into $v_0$, and then sending the flow out from $v_0$ to $V_0$, distributed according to $\sigma$. The state corresponding to $C_\theta$ will include a term $\sum_{u\in V_0}\frac{\sigma(u)}{\sqrt{\sigma(u) \w_0}}\ket{u,v_0}=\frac{1}{\sqrt{\w_0}}\ket{\psi_0}$ -- the part that distributes flow from $v_0$ to $V_0$ according to $\sigma$. Thus, $\ket{C_\theta}$ is a positive witness: a state in $({\cal A}+{\cal B})^\bot$ that has non-zero overlap with $\ket{\psi_0}$. In particular, $\braket{\psi_0}{C_\theta}=\frac{1}{\sqrt{\w_0}}$, and one can check that 
$\norm{\ket{C_\theta}}^2 \approx \frac{1}{\w_0}+{\cal E}(\theta),$
where ${\cal E}(\theta)$ is the energy of $\theta$ (see \defin{flow}). So in particular, 
$$\frac{\norm{\ket{C_\theta}}^2}{|\braket{\psi_0}{C_\theta}|^2} \approx 1+\w_0 {\cal E}(\theta).$$

\paragraph{Negative Case:} On the other hand, whenever $M=\emptyset$, if we add up \emph{all} star states $\ket{\psi_\star^{G'}(u)}$ for $u\in V(G)=V_{\cal A}\cup V_{\cal B}$ (this does not include $v_0$), for every edge $(u,v)\in \overrightarrow{E}(G)$, we will get a contribution of $\sqrt{\w_{u,v}}\ket{u,v}$ from $\ket{\psi_\star^{G'}(u)}$, and a contribution of $\sqrt{\w_{u,v}}\ket{v,u}=-\sqrt{\w_{u,v}}\ket{u,v}$ from $\ket{\psi_\star^{G'}(v)}$, which will add up to 0. However, the edges in $\overrightarrow{E}(G')\setminus\overrightarrow{E}(G)$, which are precisely the edges from $u\in V_0$ to $v_0$ (as $M=\emptyset$) will only appear in $\ket{\psi_\star^{G'}(u)}$, which contributes $\sqrt{\w_0\sigma(u)}\ket{u,v_0}$. Thus, adding up all star states results in the vector $\sqrt{\w_0}\ket{\psi_0}$, so if we let 
$$\ket{w_{\cal A}}=\frac{1}{\sqrt{\w_0}}\sum_{u\in V_{\cal A}}\ket{\psi_\star^{G'}(u)}
\mbox{ and }
\ket{w_{\cal B}}=\frac{1}{\sqrt{\w_0}}\sum_{u\in V_{\cal B}}\ket{\psi_\star^{G'}(u)}
$$
then these form a 0-negative witness (see \defin{neg-witness}), with 
$$\norm{\ket{w_{\cal A}}}^2 = \frac{1}{\w_0}\textstyle{\displaystyle\sum}_{e\in \overrightarrow{E}(G')}\w_e = \displaystyle\frac{1}{\w_0}{\cal W}(G') \approx \frac{1}{\w_0}{\cal W}(G).$$

\paragraph{Electric Network Framework:} By applying \thm{lin-alg-fwk}, we can get the following. Let 
${\cal R}$ be an upper bound on $\min_\theta {\cal E}(\theta)$ where $\theta$ runs over all flows from $\sigma$ to $M$, whenever $M\neq \emptyset$. Let ${\cal W}$ be an upper bound on ${\cal W}(G)$. Define:
$$c_+ = 1+\w_0{\cal R}
\mbox{ and }
{\cal C}_- = \frac{1}{\w_0}{\cal W},$$
and let $\w_0={\cal R}^{-1}$ so that $c_+=O(1)$. 
If ${\sf S}$ is the complexity of generating the state $\ket{\sigma}$, and ${\sf A}$ is the cost of generating the star states (which requires checking if a vertex is marked), then there is a quantum algorithm that decides if $M=\emptyset$ with bounded error in complexity 
$$O({\sf S}+\sqrt{{\cal C}_-}{\sf A}) = O({\sf S}+\sqrt{{\cal R}{\cal W}}{\sf A}).$$

\subsubsection{The Multidimensional Quantum Walk Framework}\label{sec:alternative}

We now state our extension of the electric network framework. 
In the electric network framework, we assume we can generate all star states in some cost ${\sf A}$, which implicitly assumes that for any vertex $u$, we can compute the neighbours of $u$ in time at most ${\sf A}$. However, in some cases, the actual \emph{transition cost} from $u$ to $v$, ${\sf T}_{u,v}$ may vary significantly for different $u$, and different neighbours $v$ of $u$. Our modified framework takes this variation into account, avoiding incurring a factor of the maximum ${\sf T}_{u,v}$. Instead, the complexity will scale as if we replaced each edge $\{u,v\}$ in $G$ by a path of length ${\sf T}_{u,v}$. 

The second modification we make to the electric network framework is to allow for \emph{alternative neighbourhoods}, which we now formally define.

\begin{definition}[Alternative Neighbourhoods]\label{def:alternative}
For a network $G$, as in \defin{network} and \defin{QW-access}, a set of \emph{alternative neighbourhoods} is a collection of states:
$$\Psi_\star=\{\Psi_\star(u)\subset\mathrm{span}\{\ket{u,i}:i\in L(u)\}: u\in V(G)\}$$
such that for all $u\in V(G)$, 
$$\ket{\psi_\star^G(u)}:=\sum_{i\in L^+(u)}\sqrt{\w_{u,i}}\ket{u,i}-\sum_{i\in L^-(u)}\sqrt{\w_{u,i}}\ket{u,i}\in \Psi_\star(u).$$ 
We view the states of $\Psi_\star(u)$ as different possibilities for $\ket{\psi_\star^G(u)}$, only one of which is ``correct.'' Let $d_{\max}=\max\{|L(u)|:u\in V(G)\}$. We say we can \emph{generate $\Psi_\star$ in complexity ${\sf A}_\star$}, for some ${\sf A}_{\star}=\Omega(\log d_{\max})$, if there is a map $U_\star$ such that:
\begin{itemize} 
\item for each $u\in V(G)$, there is an orthonormal basis $\overline{\Psi}(u)=\{\ket{\overline{\psi}_{u,0}},\dots,\ket{\overline{\psi}_{u,a_u-1}}\}$ for the span of $\Psi_\star(u)$, such that for all $k\in [a_u]$,
$U_{\star}\ket{u,k}=\ket{\bar\psi_{u,k}}$, and 
\item $U_\star$ can be implemented with complexity ${\sf A}_\star$.
\end{itemize}
\end{definition}

It may be possible to implement a set of alternative neighbourhoods $\Psi_\star$ for $G$ faster than it would be possible to generate the star states of $G$. This happens when, given $u$, it is expensive to determine the correct form of $\ket{\psi_\star^G(u)}$, but we do know that it is one of a set of easily generated states, say $\ket{\psi_\star^1(u)}$ or $\ket{\psi_\star^2(u)}$ (see the discussion in \sec{intro-QW}).

We now state the main result of this paper, from which the applications in \sec{welded} and \sec{k-dist-full} follow. 

\begin{theorem}[Multidimensional Quantum Walk Framework]\label{thm:full-framework}
Fix a family of networks $G$ that may depend on some implicit input $x$, with disjoint sets $V_0,V_{\sf M}\subset V(G)$ such that for any vertex, checking if $v\in V_0$ (resp. if $v\in V_{\sf M}$) can be done in at most ${\sf A}_\star$ complexity. Let $M\subseteq V_{\sf M}$ be the \emph{marked set}, and $\sigma$ an \emph{initial distribution} on $V_0$. Let $\Psi_\star=\{\Psi_\star(u):u\in V(G)\}$ be a set of alternative neighbourhoods for $G$ (see \defin{alternative}). 
For all $u\in V_0\cup V_{\sf M}$, assume that $\Psi_\star(u)=\{\ket{\psi_\star^G(u)}\}$. 
Fix some positive real-valued ${\cal W}^{\sf T}$ and ${\cal R}^{\sf T}$, that may scale with $|x|$.
Suppose the following conditions hold.
\begin{description}
\item[Setup Subroutine:] The state $\ket{\sigma}=\sum_{u\in V_0}\sqrt{\sigma(u)}\ket{u}$ can be generated in cost ${\sf S}$, and furthermore, for any $u\in V_0$, $\sigma(u)$ can be computed in $O(1)$ complexity. 
\item[Star State Generation Subroutine:] We can generate $\Psi_\star$ in complexity ${\sf A}_\star$.
\item[Transition Subroutine:] There is a quantum subroutine (see \defin{variable-time}) that implements the transition map of $G$ (see \defin{QW-access}) 
with errors $\{\epsilon_{u,v}\}_{(u,v)\in\overrightarrow{E}(G)}$ and costs $\{{\sf T}_{u,v}\}_{(u,v)\in\overrightarrow{E}(G)}$. We make the following assumptions on the errors $\epsilon_{u,v}$, where $\tilde{E}\subset\overrightarrow{E}(G)$ is some (possibly unknown) set of edges on which we allow the subroutine to fail:
	\begin{description}
	\item[TS1] For all $(u,v)\in\overrightarrow{E}(G)\setminus\tilde{E}$, $\epsilon_{u,v}\leq \epsilon$, where $\displaystyle\epsilon = o\left(\frac{1}{{\cal W}^{\sf T}{\cal R}^{\sf T}}\right)$. 
	\item[TS2] For all $(u,v)\in \tilde{E}$, there is no non-trivial upper bound on $\epsilon_{u,v}$, but $\displaystyle\widetilde{\cal W}:=\sum_{e\in \tilde{E}}\w_e=o\left(\frac{1}{{\cal R}^{\sf T}}\right)$. 
	\end{description}
\item[Checking Subroutine:] There is an algorithm that checks, for any $u\in V_{\sf M}$, if $u\in M$, in cost ${\sf A}_\star$.\footnote{This is without loss of generality. Suppose the checking cost is some higher value ${\sf C}>{\sf A}_\star$. Then we can simply put an outgoing edge on each vertex $u\in V_{\sf M}$ that ends at a new vertex $(u,b)$ that encodes whether $u\in M$ in the bit $b$. Such an edge can be implemented with transition cost ${\sf C}$.} 
\item[Positive Condition:] Interpreting ${\sf T}_{u,v}$ as a length function on $\overrightarrow{E}(G)$, $G^{\sf T}$ is the graph obtained by replacing each edge $(u,v)$ of $G$ with a path of length ${\sf T}_{u,v}$ (see \defin{nwk-length}). If $M\neq \emptyset$, then there exists a flow $\theta$ on $G$ (see \defin{flow}) such that
	\begin{description}
	\item[P1] For all $(u,v)\in \tilde{E}$, $\theta(u,v)=0$.
	\item[P2] For all $u\in V(G)\setminus (V_0\cup M)$ and $\ket{\psi_\star(u)}\in \Psi_\star(u)$, 
$$\sum_{i\in L^+(u)}\frac{\theta(u,f_u(i))\braket{\psi_\star(u)}{u,i}}{\sqrt{\w_{u,i}}}-\sum_{i\in L^-(u)}\frac{\theta(u,f_u(i))\braket{\psi_\star(u)}{u,i}}{\sqrt{\w_{u,i}}}=0.\footnote{Whenever $\ket{\psi_\star(u)}=\ket{\psi_\star^G(u)}$, this condition is simply saying that $\theta(u)=0$ (i.e.~flow is conserved at $u$), but we require that the condition also hold for all other states in $\Psi_\star(u)$ as well.}$$ \newpage 
	\item[P3] $\sum_{u\in V_0}\theta(u)=1$.\footnote{Intuitively, we want to think of the flow as coming in at $V_0$, and exiting at $M$. While we do not make it a strict requirement that all sources are in $V_0$ and all sinks in $M$, this condition implies that we do not simply have all the flow coming in at vertices in $V_0$ and then leaving again through other vertices in $V_0$.}
	\item[P4] $\sum_{u\in V_0}\frac{|\theta(u)-\sigma(u)|^2}{\sigma(u)}\leq 1$.\footnote{Intuitively, $\theta$ should be a $\sigma$-$M$ flow, meaning that for all $u\in V_0$, $\theta(u)=\sigma(u)$. We don't make this a strict requirement, but this condition means it should hold in some approximate sense.}
	\item[P5] ${\cal E}^{\sf T}(\theta)\leq {\cal R}^{\sf T}$.
	\end{description}
\item[Negative Condition:] If $M=\emptyset$, then ${\cal W}(G^{\sf T})\leq {\cal W}^{\sf T}$. 
\end{description}                                                             
Then there is a quantum algorithm that decides if $M=\emptyset$ or not with bounded error in complexity:
$$O\left({\sf S}+\sqrt{{\cal R}^{\sf T}{\cal W}^{\sf T}}\left({\sf A}_{\star}+{\sf polylog}({\sf T}_{\max})\right)\right).$$
\end{theorem}

In the remainder of this section, we prove \thm{full-framework} by describing (parameters of) a phase estimation algorithm and analysing it using \thm{lin-alg-fwk}. 

\begin{remark}\label{rem:orphan-vertices}
For an edge $(u,v)\in\tilde{E}$, we may without loss of generality assume that $v\not\in V(G)$. Suppose $i=f_u^{-1}(v)$. Then since we don't actually implement the transition $\ket{u,i}\rightarrow\ket{v,j}$ correctly anyway, we can assume that $v=(u,i)$, which is distinct from all vertices in $V(G)$, and so we can consider it an almost isolated vertex with the single backwards neighbour $u$. We can equivalently think of these as dangling edges, without an endpoint.
\end{remark}

\subsubsection{The Transition Subroutine}\label{sec:fwk-transitions}

Recall from \defin{variable-time} that a quantum subroutine is given by a sequence $U_0,\dots,U_{{\sf T}_{\max}-1}$ of unitaries on
$H=\mathrm{span}\{\ket{z}:z\in{\cal Z}\},$
such that we can implement $\sum_{t=0}^{{\sf T}_{\max}-1}\ket{t}\bra{t}\otimes U_t$ in cost ${\sf polylog}({\sf T}_{\max})$. 
In our case, the subroutine computes the \emph{transition map} (see \defin{QW-access}), $\ket{u,i}\mapsto \ket{v,j}$, so we assume 
$$\{(u,i): u\in V(G),i\in L(u)\}\subseteq{\cal Z}.$$
Then by conditions \textbf{TS1} and \textbf{TS2} from \thm{full-framework}, for any $(u,v)\in \overrightarrow{E}(G)$, with $i=f_u^{-1}(v)$ and $j=f_v^{-1}(u)$, we have  
\begin{equation}
\norm{\ket{v,j} - U_{{\sf T}_{u,v}-1}\dots U_0\ket{u,i}}^2 = \epsilon_{u,v},\label{eq:fwk-alg-error}
\end{equation}
where $\epsilon_{u,v}\leq \epsilon$ whenever $(u,v)\in\overrightarrow{E}(G)\setminus\tilde{E}$. Otherwise, we only have the trivial upper bound $\epsilon_{u,v}\leq 4$.  

We will assume that in $O(1)$ time, we can check, for any $z\in {\cal Z}$, if $z=(u,i)$ for some $u\in V(G)$ and $i\in L(u)$, and further, whether $i\in L^+(u)$ or $i\in L^-(u)$. This is without loss of generality, by the following construction. Assume that for all $u\in V(G)$, every label in $L^+(u)$ ends with the symbol $\rightarrow$, and every label in $L^-(u)$ ends with the symbol $\leftarrow$. Further assume that no other $z\in {\cal Z}$ ends with these symbols. Then it is sufficient to check a single constant-dimensional register.

We will assume that ${\sf T}_{u,v}$ is always even. This assumption incurs at most a small constant slowdown. We will also assume that after exactly ${\sf T}_{u,v}$ steps, the algorithm sets an \emph{internal flag register} to 1, and we will let this 1-flag be part of the final state $(v,j)$ by letting each $i\in L(u)$ contain an extra bit set to 1. This also ensures that the state of the algorithm is never $\ket{v,j}$ before ${\sf T}_{u,v}$ steps have passed. This assumption is without loss of generality, because we can simply let the algorithm use an internal timer in order to decide to set a flag after exactly ${\sf T}_{u,v}$ steps, and uncompute this timer using our ability to compute ${\sf T}_{u,v}$ from the final correct state $\ket{v,j}$. 

Recall from \defin{variable-time} that for any $u\in V(G)$, $i\in L(u)$ and $t\in\{0,\dots,{\sf T}_{\max}-1\}$, 
$$U_t\dots U_0\ket{u,i}\in \mathrm{span}\{\ket{z}:z\in {\cal Z}_{u,i}\}.$$ 
For convenience, we will let ${\cal Z}_{u,v}={\cal Z}_{u,i}$, where $v=f_u(i)$.
For $b\in\{0,1\}$, let ${\cal Z}_{u,v}^b\subset {\cal Z}_{u,v}$ be the subset of states in which the algorithm's internal flag register is set to $b$. So by the above discussion, we have $(v,j)\in {\cal Z}_{u,v}^1$,
$$\forall t<{\sf T}_{u,v},\; U_t\dots U_0\ket{u,i} \in {\cal Z}_{u,v}^0,
\mbox{ and }
\forall t\geq {\sf T}_{u,v},\;
U_t\dots U_0\ket{u,i} \in {\cal Z}_{u,v}^1.$$

\subsubsection{Parameters of the Phase Estimation Algorithm}

Our phase estimation algorithm will work on the space:
\begin{equation}
\begin{split}
H'&=\mathrm{span}\{\ket{u,i}\ket{0}:u\in V(G),i\in L^+(u)\cup\{0\}\}\oplus\mathrm{span}\{\ket{v,j}\ket{0}:v\in V(G),j\in L^-(v)\}\\
&\qquad\oplus\bigoplus_{(u,v)\in\overrightarrow{E}(G)}\mathrm{span}\{\ket{z}\ket{t}:z\in{\cal Z}_{u,v}^0,t\in[{\sf T}_{u,v}-1]\}\cup\{\ket{z}\ket{{\sf T}_{u,v}}:z\in{\cal Z}_{u,v}^1\}.
\end{split}\label{eq:full-fwk-H}
\end{equation}
We now define sets of states $\Psi^{\cal A}$ and $\Psi^{\cal B}$ in $H'$.

\paragraph{Star States:} We slightly modify the star states to get states in $H'$. To all star states $\ket{\psi_\star(u)}\in \Psi_\star(u)$ (see \defin{alternative}), we append a register $\ket{0}$, but for $u\in V_0\cup M$, we make a further modification. 
Conceptually, we modify the graph $G$, by adding a new vertex, $v_0$, to get a new graph $G'$ (see Figure~\ref{fig:v0-graph}). The new vertex is connected to every $u\in V_0$ by an edge of weight $\w_0\sigma(u)$, for some $\w_0$ to be assigned later; and it is connected to every $u\in M$ by an edge of weight $\w_{\sf M}$, for some $\w_{\sf M}$ to be assigned later.  So for $u\in V_0\cup M$, we modify the star state $\Psi_\star(u)=\{\ket{\psi_\star^G(u)}\}$ by adding the extra register $\ket{0}$, but we also account for the additional edge to $v_0$. We assume that for all $u\in V_0\cup M$, the edge to $v_0$ is labelled by $0\not\in L(u)$. With this intuition, we define, for $u\in V_0$:
\begin{equation}
\ket{\psi_{\star}^{G'}(u)}\ket{0}:=\underbrace{
\sum_{i\in L^+(u)}\sqrt{\w_{u,i}}\ket{u,i}\ket{0} - \sum_{i\in L^-(u)}\sqrt{\w_{u,i}}\ket{u,i}\ket{0}  
}_{=\ket{\psi_\star^G(u)}\ket{0}}
+\sqrt{\w_0\sigma(u)}\ket{u,0}\ket{0}\label{eq:psi-star-F-1}
\end{equation}
and for $u\in M$:
\begin{equation}
\ket{\psi_{\star}^{G'}(u)}\ket{0}:=\underbrace{
\sum_{i\in L^+(u)}\sqrt{\w_{u,i}}\ket{u,i}\ket{0} - \sum_{i\in L^-(u)}\sqrt{\w_{u,i}}\ket{u,i}\ket{0}  
}_{=\ket{\psi_\star^G(u)}\ket{0}}
+\sqrt{\w_{\sf M}}\ket{u,0}\ket{0}.\label{eq:psi-star-F-M}
\end{equation}

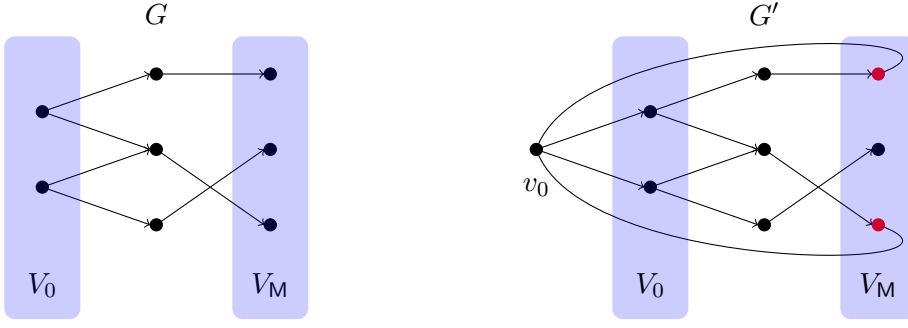
\begin{figure}
	\renewcommand{\arraystretch}{1.25}
	\centering
	\begin{tikzpicture}[scale=1.2]
		\filldraw (1.5,-0.5) circle (.08);	\draw[-{Latex[length=2mm, width=2mm]}] (1.5,-0.5) -- (2.92,-1);	\draw[-{Latex[length=2mm, width=2mm]}] (1.5,-0.5) -- (2.92,0);
		\filldraw (1.5,0.5) circle (.08);	\draw[-{Latex[length=2mm, width=2mm]}] (1.5,0.5) -- (2.92,1);	\draw[-{Latex[length=2mm, width=2mm]}] (1.5,0.5) -- (2.92,0);
		\filldraw (3,-1) circle (.08);	\draw[-{Latex[length=2mm, width=2mm]}] (3,-1) -- (4.42,0);
		\filldraw (3,0) circle (.08);	\draw[-{Latex[length=2mm, width=2mm]}] (3,0) -- (4.42,-1);
		\filldraw (3,1) circle (.08);	\draw[-{Latex[length=2mm, width=2mm]}] (3,1) -- (4.42,1);
		\filldraw (4.5,-1) circle (.08);
		\filldraw (4.5,0) circle (.08);
		\filldraw (4.5,1) circle (.08);
		
		\node at (1.5,-1.8) {$V_0$};
		\node at (4.5,-1.8) {$V_{\sf M}$};
		\node at (3,1.8) {$G$};
		
		\filldraw (8,0) circle (.08);	\draw[{Latex[length=2mm, width=2mm]}-](8.07,-0.03) -- (9.42,-0.5); 	\draw[{Latex[length=2mm, width=2mm]}-] (8.07,0.03) -- (9.42,0.5);
		\filldraw (9.5,-0.5) circle (.08);	\draw[-{Latex[length=2mm, width=2mm]}] (9.5,-0.5) -- (10.92,-1);	\draw[-{Latex[length=2mm, width=2mm]}] (9.5,-0.5) -- (10.92,0);
		\filldraw (9.5,0.5) circle (.08);	\draw[-{Latex[length=2mm, width=2mm]}] (9.5,0.5) -- (10.92,1);	\draw[-{Latex[length=2mm, width=2mm]}] (9.5,0.5) -- (10.92,0);
		\filldraw (11,-1) circle (.08);	\draw[-{Latex[length=2mm, width=2mm]}] (11,-1) -- (12.42,0);
		\filldraw (11,0) circle (.08);	\draw[-{Latex[length=2mm, width=2mm]}] (11,0) -- (12.42,-1);
		\filldraw (11,1) circle (.08);	\draw[-{Latex[length=2mm, width=2mm]}] (11,1) -- (12.42,1);
		\draw[-{Latex[length=2mm, width=2mm]}] (12.5,-1) to[out=-20,in=-70]  (8,-.08);
		\filldraw[red] (12.5,-1) circle (.08); 
		\filldraw (12.5,0) circle (.08);
		\draw[-{Latex[length=2mm, width=2mm]}] (12.5,1) to[out=20,in=70]  (8,0.08);
		\filldraw[red] (12.5,1) circle (.08); 
		\fill[rounded corners, fill=blue, opacity=0.2] (1, 1.5) -- (2, 1.5) -- (2, -2.25) -- (1, -2.25) -- cycle;
		\fill[rounded corners, fill=blue, opacity=0.2] (4, 1.5) -- (5, 1.5) -- (5, -2.25) -- (4, -2.25) -- cycle;
		\fill[rounded corners, fill=blue, opacity=0.2] (9, 1.5) -- (10, 1.5) -- (10, -2.25) -- (9, -2.25) -- cycle;
		\fill[rounded corners, fill=blue, opacity=0.2] (12, 1.5) -- (13, 1.5) -- (13, -2.25) -- (12, -2.25) -- cycle;

		\node at (8,-.5) {$v_0$};
		\node at (9.5,-1.8) {$V_0$};
		\node at (12.5,-1.8) {$V_{\sf M}$};
		\node at (11,1.8) {$G'$};

	\end{tikzpicture}
	\caption{Example of a graph $G$ with $V_0,V_{\sf M} \subseteq V(G)$ and the induced graph $G'$ that is obtained from $G$ by adding a new vertex $v_0$. This new vertex is connected to all vertices in $V_0$ and only connected to those vertices in $V_{\sf M}$ which are marked (visualised by the red vertices).} \label{fig:v0-graph}
\end{figure}

For $u\in V(G)\setminus (V_0\cup M)$, the neighbours and weights in $G'$ are the same as $G$, so we let $\ket{\psi_\star^{G'}(u)}=\ket{\psi_\star^{G}(u)}$, which we know is in $\Psi_\star(u)$ (possibly among other states).  
We let:
\begin{equation}
\begin{split}
\Psi_{\star}'&:=\bigcup_{u\in V(G)\setminus (V_0\cup M)}\underbrace{\left\{\ket{\psi_{\star}(u)}\ket{0}:\ket{\psi_\star(u)}\in \Psi_\star(u)\right\}}_{=:\Psi_\star'(u)}
\cup
\bigcup_{u\in V_0\cup M}\underbrace{\left\{\ket{\psi_{\star}^{G'}(u)}\ket{0}\right\}}_{=:\Psi_\star'(u)}.\label{eq:stars-full-fwk}
\end{split}
\end{equation}

\paragraph{Algorithm States:} 
For each $u\in V(G)$ and $i\in L^+(u)$, define a state
\begin{equation}
\ket{\psi_{\rightarrow}^{u,i}}:=\ket{u,i}\ket{0} - U_0\ket{u,i}\ket{1}.\label{eq:psi-ui}
\end{equation}
These represent a transition from an outgoing edge to the first step of the algorithm implementing that edge transition. 
For each $(u,v)\in \overrightarrow{E}(G)$, and $t\in [{\sf T}_{u,v}-1]$, define states:
\begin{equation}
\Psi_t^{u,v}:=\left\{\ket{\psi_{t}^z}:=\ket{z}\ket{t} - U_t\ket{z}\ket{t+1}:
z\in {\cal Z}_{u,v}^0\right\}.\label{eq:psi-z-t}
\end{equation}
These represent steps of the edge transition subroutine.
For each $v\in V(G)$ and $j\in L^-(v)$, with $u=f_v(j)$, define a state:
\begin{equation}
\ket{\psi_{\leftarrow}^{v,j}}:=\ket{v,j}\ket{{\sf T}_{u,v}}-\ket{v,j}\ket{0}.\label{eq:psi-vj}
\end{equation}
These represent exiting the algorithm to an edge going into vertex $v$. 
Letting $\Psi_\star'$ be as in \eq{stars-full-fwk}, define
\begin{equation}
\begin{split}
\Psi^{\cal A} &= \Psi_\star'\cup \bigcup_{(u,v)\in\overrightarrow{E}(G)}\bigcup_{\substack{t=1:\\ t\;\mathrm{odd}}}^{{\sf T}_{u,v}-1}\Psi_t^{u,v}\\
\Psi^{\cal B} &=  \{\ket{\psi_{\rightarrow}^{u,i}}:u\in V(G),i\in L^+(u)\}\cup \{\ket{\psi_{\leftarrow}^{v,j}}:v\in V(G), j\in L^-(v)\}\cup \!\bigcup_{(u,v)\in\overrightarrow{E}(G)}\bigcup_{\substack{t=1:\\ t\;\mathrm{even}}}^{{\sf T}_{u,v}-1}\Psi_t^{u,v}.
\end{split}\label{eq:full-fwk-states}
\end{equation}
The reason we have divided the states in this way between $\Psi^{\cal A}$ and $\Psi^{\cal B}$ is so that if we replace each $\Psi_\star(u)$ with an orthonormal basis, all states in $\Psi^{\cal A}$ (or $\Psi^{\cal B}$) are pairwise orthogonal. We leave it up to the reader to verify that this is the case (it is implicitly proven in \sec{fwk-unitary}), but we note that this fact relies on the assumption that ${\sf T}_{u,v}$ is always even. This ensures that for even $t$, $\braket{t+1}{{\sf T}_{u,v}}=0$, so $\braket{\psi_{t}^z}{\psi_{\leftarrow}^{v,j}}=0$. \fig{spaces-graph} shows a graph of the overlap between various sets of states, and we can observe that the sets in $\Psi^{\cal A}$ and the sets in $\Psi^{\cal B}$ form a bipartition of this overlap graph into independent sets.

\begin{figure}
\centering
\begin{tikzpicture}[scale=1.2]
\node at (0,0) {\begin{tikzpicture}[scale=1.2]
\draw[-{Latex[length=2mm, width=2mm]}] (0,0)--(1.9,0);
\draw[-{Latex[length=2mm, width=2mm]}] (1.92,.08)--(1.08,.92);
\draw[{Latex[length=2mm, width=2mm]}-] (.92,.92)--(0,0);
\filldraw (0,0) circle (.1);
\filldraw (2,0) circle (.1);
\filldraw (1,1) circle (.1);

\node at (-.25,0) {$u$};
\node at (1,1.25) {$v$};
\node at (2.25,0) {$w$};

\node at (1,-1) {$G$};
\end{tikzpicture}};

\node at (8.5,0) {\begin{tikzpicture}[scale=1.2]
\draw (0,0)--(12,0)--(6,6)--(0,0);

\node[rectangle, rounded corners, draw, thick, fill=white] at (0,0) {$\Psi_\star'(u)$};
	\node[rectangle, rounded corners, draw, thick, fill=white] at (2,0) {$\ket{\psi_{\rightarrow}^{u,i'}}$};
	\node[rectangle, rounded corners, draw, thick, fill=white] at (4,0) {$\Psi_1^{u,w}$};
	\node[rectangle, fill=white] at (6,0) {$\dots$};
	\node[rectangle, rounded corners, draw, thick, fill=white] at (8,0) {$\Psi_{{\sf T}_{u,w}}^{u,w}$};
	\node[rectangle, rounded corners, draw, thick, fill=white] at (10,0) {$\ket{\psi_{\leftarrow}^{w,j'}}$};

\node[rectangle, rounded corners, draw, thick, fill=white] at (12,0) {$\Psi_\star'(w)$};
	\node[rectangle, rounded corners, draw, thick, fill=white] at (11,1) {$\ket{\psi_{\rightarrow}^{w,i''}}$};
	\node[rectangle, rounded corners, draw, thick, fill=white] at (10,2) {$\Psi_1^{w,v}$};
	\node[rectangle, fill=white, rotate=135] at (9,3) {$\dots$};
	\node[rectangle, rounded corners, draw, thick, fill=white] at (8,4) {$\Psi_{{\sf T}_{w,v}}^{w,v}$};
	\node[rectangle, rounded corners, draw, thick, fill=white] at (7,5) {$\ket{\psi_{\leftarrow}^{v,j''}}$};

\node[rectangle, rounded corners, draw, thick, fill=white] at (6,6) {$\Psi_\star'(v)$};
	\node[rectangle, rounded corners, draw, thick, fill=white] at (1,1) {$\ket{\psi_{\rightarrow}^{u,i}}$};
	\node[rectangle, rounded corners, draw, thick, fill=white] at (2,2) {$\Psi_1^{u,v}$};
	\node[rectangle, fill=white, rotate = 45] at (3,3) {$\dots$};
	\node[rectangle, rounded corners, draw, thick, fill=white] at (4,4) {$\Psi_{{\sf T}_{u,v}}^{u,v}$};
	\node[rectangle, rounded corners, draw, thick, fill=white] at (5,5) {$\ket{\psi_{\leftarrow}^{v,j}}$};

\end{tikzpicture}};

\end{tikzpicture}
\caption{A graph showing the overlap of various sets of states, for an example graph $G$. With the exception of the spaces $\Psi_\star'(u)$ (which we will replace with orthonormal bases in \sec{fwk-unitary}), each node represents an orthonormal set. There is an edge between two nodes if and only if the sets contain overlapping vectors.}\label{fig:spaces-graph}
\end{figure}
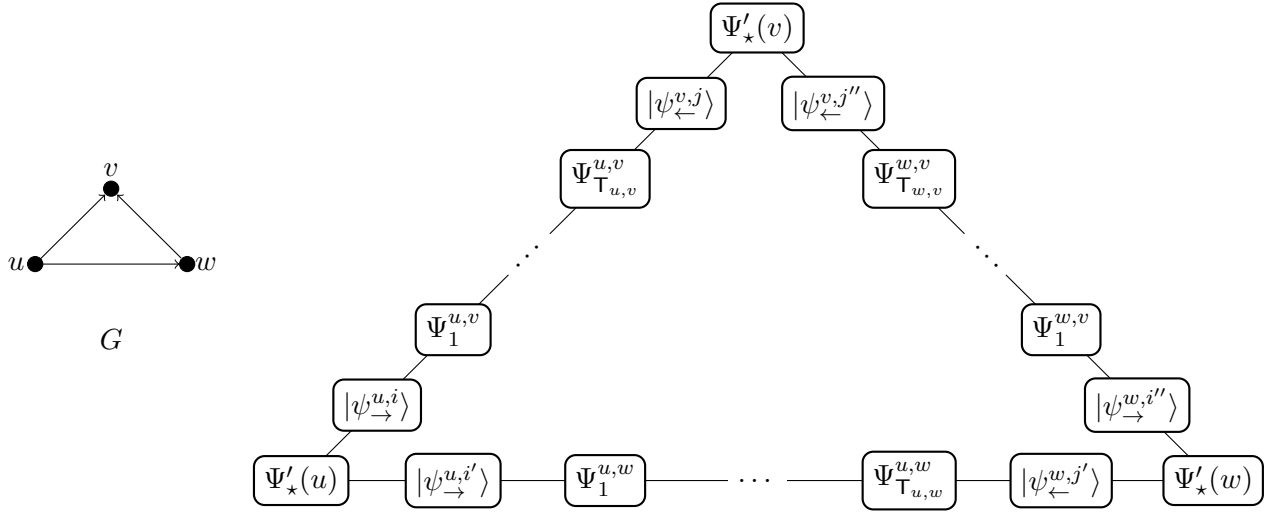

Finally, we define the initial state of the algorithm:
\begin{equation}
\ket{\psi_0}:=\ket{\sigma}\ket{0}\ket{0}=\sum_{u\in V_0}\sqrt{\sigma(u)}\ket{u,0}\ket{0}.\label{eq:fwk-psi-init}
\end{equation}

\subsubsection{Implementing the Unitary}\label{sec:fwk-unitary}

Let ${\cal A}=\mathrm{span}\{\Psi^{\cal A}\}$ and ${\cal B}=\mathrm{span}\{\Psi^{\cal B}\}$ (see \eq{full-fwk-states}), and let $\Pi_{\cal A}$ and $\Pi_{\cal B}$ be the orthogonal projectors onto ${\cal A}$ and ${\cal B}$. 
In this section we will prove:
\begin{lemma}\label{lem:update-cost}
The unitary $U_{{\cal AB}}=(2\Pi_{{\cal A}}-I)(2\Pi_{{\cal B}}-I)$ on $H'$ can be implemented in complexity \mbox{$O\left({\sf A}_\star+{\sf polylog}({\sf T}_{\max})\right)$}.
\end{lemma}

\noindent This essentially follows\footnote{For a simple example of how reflecting around a set of states reduces to generating the set, see \clm{leftarrow-refl}.} from the fact that we can efficiently generate orthonormal bases for each of $\Psi^{\cal A}$ and $\Psi^{\cal B}$, since:
\begin{itemize}
\item By the \textbf{Star State Generation Subroutine} condition of \thm{full-framework}, we can generate an orthonormal basis for $\bigcup_{u\in V(G)}\Psi_\star(u)$. Since we can also efficiently check if a vertex is in $V_0$ or $M$, we can generate orthonormal bases for $\Psi_\star'=\bigcup_{u\in V(G)}\Psi'_\star(u)$ (see \clm{fwk-star-unitary}).
\item Generating the states $\ket{\psi_t^z}=\ket{z}\ket{t}-U_t\ket{z}\ket{t+1}$ for odd $t$ can be done using $\sum_{t}\ket{t}\bra{t}\otimes U_t$ (see \clm{fwk-odd-unitary}). The same is true for even $t$ (\clm{fwk-even-unitary}), also including the states $\ket{\psi_{\rightarrow}^{u,i}}=\ket{u,i}\ket{0}-U_0\ket{u,i}\ket{1}$.
\item Generating the states $\ket{\psi_{\leftarrow}^{v,j}}=\ket{v,j}(\ket{{\sf T}_{u,v}}-\ket{0})$ can be done efficiently because we can compute ${\sf T}_{u,v}$ from $(v,j)$ (see \clm{leftarrow-refl}). 
\end{itemize}
There is nothing conceptually new in this proof, and the reader may skip ahead to \sec{fwk-positive} with no loss of understanding.

\begin{claim}\label{clm:fwk-star-unitary}
Let $R_\star =2\Pi_{\star}-I$, where $\Pi_\star$ is the orthogonal projector onto $\mathrm{span}\{\Psi_\star'\}$. Then $R_\star$ can be implemented in complexity $O({\sf A}_\star+\log{\sf T}_{\max})$. 
\end{claim}
\begin{proof}
By the \textbf{Star State Generation Subroutine} condition of \thm{full-framework}, we can generate $\Psi_\star$ in cost ${\sf A}_\star$, which means (see \defin{alternative}) that for each  $u\in V(G)$, there is an orthonormal basis $\overline{\Psi}(u)=\{\ket{\overline{\psi}_{u,1}},\dots,\ket{\overline{\psi}_{u,a_u}}\}$ for $\mathrm{span}\{\Psi_\star(u)\}$, and a unitary $U_\star$ with complexity ${\sf A}_\star$, such that for all $u\in V(G)$ and $k\in \{0,\dots,a_u-1\}$,
$U_\star\ket{u,k}=\ket{\overline{\psi}_{u,k}}.$
Then for all $u\in V(G)\setminus (V_0\cup V_{\sf M})$, $\overline{\Psi}'(u):=\{\ket{\overline{\psi}_{u,1}}\ket{0},\dots,\ket{\overline{\psi}_{u,a_u}}\ket{0}\}$ is an orthonormal basis for $\Psi'(u)$ (see \eq{stars-full-fwk}).
For $u\in V_0\cup V_{\sf M}$, we have $\Psi_{\star}'(u)=\{\ket{\psi_\star^{G'}(u)}\ket{0}\}$, so $\overline{\Psi}'(u)$ is just the single normalization of this state. 

We will first define a unitary $U_\star'$ that acts, for $u\in V(G)$, $k\in\{0,\dots,a_u-1\}$, as
$
U_\star'\ket{u,k}\ket{0}=\ket{\overline{\psi}_{u,k}'}.
$
We define $U_\star'$ by its implementation. To begin we will append a qutrit register, $\ket{0}_A$ (this will be uncomputed, so that the action described is indeed unitary), and set it to $\ket{1}_A$ if $u\in V_0$, and $\ket{2}_A$ if $u\in M$. We are assuming we can check if $u\in V_0$ or $u\in M$ in at most ${\sf A}_{\star}$ complexity. We proceed in three cases, controlled on the value of the ancilla.

\noindent First, controlled on $\ket{0}_A$, we apply $U_{\star}$, in cost ${\sf A}_\star$, to get:
$$\ket{u,k}\ket{0}\ket{0}_A\mapsto \ket{\overline{\psi}_{u,k}}\ket{0}\ket{0}_A
= \ket{\overline{\psi}_{u,k}'}\ket{0}_A.$$

\noindent Next, controlled on $\ket{1}_A$, we implement, on the last register, a single qubit rotation that acts as 
$$\ket{0}\mapsto \sqrt{\frac{\w_u}{\w_u+\w_0\sigma(u)}}\ket{1}+\sqrt{\frac{\w_0\sigma(u)}{\w_u+\w_0\sigma(u)}}\ket{0}.$$
This requires that we can query $\w_u$ and $\sigma(u)$ ($\w_0$ is a parameter of the algorithm).
Controlled on $\ket{1}$ in the last register (and also still $\ket{1}_A$ in the third), we apply $U_\star$ to get (when $u\in V_0$ we only care about the behaviour for $k=0$):
\begin{align*}
\ket{u,0}\ket{0}\ket{1}_A &\mapsto \left(\sqrt{\frac{\w_u}{\w_u+\w_0\sigma(u)}}\frac{\ket{\psi_\star^G(u)}}{\sqrt{\w_u}}\ket{1}+\sqrt{\frac{\w_0\sigma(u)}{\w_u+\w_0\sigma(u)}}\ket{u,0}\ket{0}\right)\ket{1}_A.
\end{align*}
Above we have used the fact that when $u\in V_0$,
$$\ket{\overline{\psi}_{u,0}}=\frac{\ket{\psi_{\star}^G(u)}}{\norm{\ket{\psi_{\star}^G(u)}}}=\frac{\ket{\psi_{\star}^G(u)}}{\sqrt{\w_u}}.$$
To complete the map for the case $u\in V_0$, note that $\ket{\psi_\star^G(u)}$ is supported on $\ket{u,i}$ for $i\neq 0$, so we can uncompute the second register to get:
\begin{align*}
\frac{\ket{\psi_\star^G(u)}\ket{0}+\sqrt{\w_0\sigma(u)}\ket{u,0}\ket{0}}{\sqrt{\w_u+\w_0\sigma(u)}}\ket{1}_A
=\frac{\ket{\psi_\star^{G'}(u)}\ket{0}}{\norm{\ket{\psi_\star^{G'}(u)}}}\ket{1}_A.
\end{align*}

Finally, controlled on $\ket{2}_A$, we do something very similar, but now the single qubit map we use is
$$\ket{0}\mapsto \sqrt{\frac{\w_u}{\w_u+\w_{\sf M}}}\ket{1}+\sqrt{\frac{\w_{\sf M}}{\w_u+\w_{\sf M}}}\ket{0}.$$
This is possible given query access to $\w_u$ ($\w_{\sf M}$ is a parameter of the algorithm). 

Since all states still have $\ket{u}$ in the first register, controlled on $u$, we can uncompute the ancilla. Thus, we can implement $U_\star'$ in complexity $O({\sf A}_\star)$, and $U_\star'$ maps the subspace
$${\cal L}_\star:=\mathrm{span}\{\ket{u,k}\ket{0}:u\in V(G),k\in\{0,\dots,a_u-1\}\}$$                                                                                                                                                                                        
of $H'$ to the $\mathrm{span}\{\Psi_\star'\}$. Thus
$(2\Pi_\star-I) = U_\star' (2\Pi_{{\cal L}_\star}-I) {U_\star'}^\dagger,$
so we complete the proof by describing how to implement $2\Pi_{{\cal L}_\star}-I$. Initialise two ancillary flag qubits, $\ket{0}_{F_1}\ket{0}_{F_2}$. For a computational basis state $\ket{z}\ket{t}$, if $t\neq 0$, flip $F_1$ to $\ket{1}_{F_1}$. This check costs $\log{\sf T}_{\max}$. If $t=0$, we can assume that $z$ has the form $(u,k)$, and interpret $k$ as an integer. If $k<a_u$, which can be checked in $O(\log d_{\max})=O({\sf A}_\star)$, flip $F_2$ to $\ket{1}_{F_2}$. Reflect conditioned on either of the flags being set to 1, and then uncompute both flags. 
\end{proof}

\begin{claim}\label{clm:fwk-odd-unitary}
Let $R_{\sf odd} =2\Pi_{\sf odd}-I$, where $\Pi_{\sf odd}$ is the orthogonal projector onto $\mathrm{span}\{\Psi_t^{u,v}:(u,v)\in\overrightarrow{E}(G), t\in\{1,\dots,{\sf T}_{u,v}-1\}\mbox{ odd}\}$. Then $R_{\sf odd}$ can be implemented in complexity ${\sf polylog}({\sf T}_{\max})$. 
\end{claim}
\begin{proof}
We first describe the implementation of a unitary $U_{\sf odd}$ such that:
\begin{equation*}
\forall (u,v)\in\overrightarrow{E}(G),z\in{\cal Z}_{u,v}^0, t\in[{\sf T}_{u,v}-1] \text{ odd}, \; U_{\sf odd}\ket{z}\ket{t} = \frac{1}{\sqrt{2}}\ket{z}\ket{t} - \frac{1}{\sqrt{2}}U_t\ket{z}\ket{t+1}    
= \frac{1}{\sqrt{2}}\ket{\psi_{z}^{t}}.\label{eq:U-S-0}
\end{equation*}
We begin by decrementing the $\ket{t}$ register, which costs $\log {\sf T}_{\max}$.
Next we apply an ${\sf X}$ gate, followed by a Hadamard gate, to the last qubit of $\ket{t-1}$. If $t$ is odd, $t-1$ is even and the last qubit is $\ket{0}\overset{{\sf HX}}{\mapsto} (\ket{0}-\ket{1})/\sqrt{2}$, so we get: 
$$\ket{z}\ket{t}\mapsto \ket{z}\ket{t-1}\mapsto \left(\ket{z}\ket{t-1}-\ket{z}\ket{t}\right)/\sqrt{2}.$$
Then controlled on the last qubit of $\ket{t}$ being $\ket{1}$ (i.e.~on odd parity of $t$) apply $\sum_{t=0}^{{\sf T}_{\max}-1}\ket{t}\bra{t}\otimes U_t$, which can be done in cost ${\sf polylog}({\sf T}_{\max})$ by assumption, to get:
$
\left(\ket{z}\ket{t-1}-U_t\ket{z}\ket{t}\right)/\sqrt{2}.
$
Complete the operation by incrementing the $\ket{t}$ register. Thus, $U_{\sf odd}$ maps the subspace:
$${\cal L}_{\sf odd}:=\textstyle{\displaystyle\bigoplus}_{(u,v)\in\overrightarrow{E}(G)}\mathrm{span}\{\ket{z}\ket{t}:z\in {\cal Z}_{u,v}^0,t\in\{1,\dots,{\sf T}_{u,v}-1\},\;\mbox{odd}\}$$
of $H'$ to the support of $\Pi_{\sf odd}$, and so
$R_{\sf odd}=U_{\sf odd}(2\Pi_{{\cal L}_{\sf odd}}-I)U_{\sf odd}^\dagger.$
We complete the proof by describing how to implement $2\Pi_{{\cal L}_{\sf odd}}-I$. For $\ket{z}\ket{t}$, we can check if $t$ is odd in $O(1)$, and if not, set an ancillary flag $F_1$. Next, we will ensure that $z\in {\cal Z}_{u,v}^0$ for some $(u,v)\in\overrightarrow{E}(G)$, which also ensures that $t<{\sf T}_{u,v}$, by the structure of $H'$, and if not, set a flag $F_2$. Reflect if either $F_1$ or $F_2$ is set, and then uncompute both of them. 
\end{proof}

\begin{claim}\label{clm:fwk-even-unitary}
Let $R_{\sf even} =2\Pi_{\sf even}-I$, where $\Pi_{\sf even}$ is the orthogonal projector onto the span of
$$\bigcup_{\substack{(u,v)\in\overrightarrow{E}(G)\\t\in\{1,\dots,{\sf T}_{u,v}-1\}: t\;\mathrm{ even}}}\Psi_t^{u,v}\cup\{\ket{\psi_{\rightarrow}^{u,i}}:u\in V(G),i\in L^+(u)\}.$$
Then $R_{\sf even}$ can be implemented in complexity ${\sf polylog}({\sf T}_{\max})$. 
\end{claim}
\begin{proof}
We describe the implementation of a unitary $U_{\sf even}$ such that for all $(u,v)\in\overrightarrow{E}(G)$ with $i=f_u^{-1}(v)$ and $j=f_v^{-1}(u)$:
\vskip-18pt
\begin{equation*}
\begin{split}
U_{\sf even}\ket{u,i}\ket{0} &= \frac{1}{\sqrt{2}}\left(\ket{u,i}\ket{0} - U_0\ket{u,i}\ket{1} \right)
= \frac{1}{\sqrt{2}}\ket{\psi_{\rightarrow}^{u,i}}\\
\forall z\in {\cal Z}_{u,v}^0, t\in [{\sf T}_{u,v}-1]\text{ even}, \;
U_{\sf even}\ket{z}\ket{t} &= \frac{1}{\sqrt{2}}\left(\ket{z}\ket{t} -U_t\ket{z}\ket{t+1}\right) 
= \frac{1}{\sqrt{2}}\ket{\psi_{t}^z}.
\end{split}\label{eq:U-S-1}
\end{equation*}
We can implement such a mapping nearly identically to the proof of \clm{fwk-odd-unitary}, except the decrementing of $t$ happens after the Hadamard is applied. Thus, $U_{\sf even}$ maps the subspace:
$${\cal L}_{\sf even}:=\bigoplus_{u\in V(G), i\in L^+(u)}\mathrm{span}\{\ket{u,i}\ket{0}\}\oplus\bigoplus_{(u,v)\in\overrightarrow{E}(G)}\mathrm{span}\{\ket{z}\ket{t}:z\in{\cal Z}_{u,v}^0,t\in\{1,\dots,{\sf T}_{u,v}-1\},\;\mbox{even}\}$$
of $H'$ to the support of $\Pi_{\sf even}$, and so
$R_{\sf even}=U_{\sf even}(2\Pi_{{\cal L}_{\sf even}}-I)U_{\sf even}^\dagger.$
We complete the proof by describing how to implement $2\Pi_{{\cal L}_{\sf even}}-I$. For $\ket{z}\ket{t}$, we can check if $t$ is even in $O(1)$ steps, and if not, set an ancillary flag $F_1$. Next, we check if $z\in {\cal Z}_{u,v}^0$ by checking the subroutine's internal flag, which also ensures that $t<{\sf T}_{u,v}$, and if not, set an ancillary flag $F_2$. Note that if $t=0$, $z$ has the form $(u,i)$ for some $i\in L(u)$, by the structure of $H'$, and by the discussion in \sec{fwk-transitions}, we can check if $i\in L^+(u)$ in $O(1)$ time, and otherwise, set a flag $F_3$. Reflect if either $F_1$, $F_2$ or $F_3$ is set, and then uncompute all three flags.
\end{proof}

\begin{claim}\label{clm:leftarrow-refl}
Let $R_{\leftarrow} =2\Pi_{\leftarrow}-I$, where $\Pi_{\leftarrow}$ is the orthogonal projector onto the span of
$\{\ket{\psi_{\leftarrow}^{v,j}}:v\in V(G),j\in L^-(v)\}.$
Then $R_{\leftarrow}$ can be implemented in complexity ${\sf polylog}({\sf T}_{\max})$. 
\end{claim}
\begin{proof}
We describe the implementation of a unitary $U_{\leftarrow}$ that acts, for all $v\in V(G)$ and $j\in L^-(v)$, with $u=f_v(j)$, as:
$$U_{\leftarrow}\ket{v,j}\ket{0} = \frac{1}{\sqrt{2}}\left(
	\ket{v,j}\ket{0}-\ket{v,j}\ket{{\sf T}_{u,v}}
\right)
=-\frac{1}{\sqrt{2}}\ket{\psi_{\leftarrow}^{v,j}}.$$
First, append an ancilla $\ket{-}_A$. Controlled on this ancilla, compute ${\sf T}_{u,v}$ from $(v,j)$, which we can do in ${\sf polylog}({\sf T}_{\max})$ basic operations, by the assumptions of \defin{variable-time}, to get:
$$\ket{v,j}\ket{0}\ket{-}_A \mapsto \ket{v,j}\left(\ket{0}\ket{0}_A - \ket{{\sf T}_{u,v}}\ket{1}_A\right)/\sqrt{2}.$$
Uncompute the ancilla by adding 1 into register $A$ conditioned on the time register having a value greater than 0. Thus, $U_{\leftarrow}$ maps the subspace
$${\cal L}_{\leftarrow}:=\mathrm{span}\{\ket{v,j}\ket{0}:v\in V(G), j\in L^-(v)\}$$
of $H'$ to the support of $\Pi_{\leftarrow}$, and so
$R_{\leftarrow} = U_{\leftarrow}(2\Pi_{{\cal L}_\leftarrow}-I)U_{\leftarrow}^\dagger.$
We complete the proof by describing how to implement $2\Pi_{{\cal L}_\leftarrow}-I$. Append two ancillary qubits, $\ket{0}_{F_1}$ and $\ket{0}_{F_2}$. For a computational basis state $\ket{z}\ket{t}$, if $t\neq 0$, which can be checked in $O(\log {\sf T}_{\max})$ time, flip $F_1$ to get $\ket{1}_{F_1}$. By the discussion in \sec{fwk-transitions}, we can check if $z$ has the form $(v,j)$ for some $v\in V(G)$ and $j\in L^-(v)$ in $O(1)$ time, and if not, flip $F_2$ to get $\ket{1}_{F_2}$. Reflect the state if either flag is set to 1, and then uncompute both flags.
\end{proof}

\begin{proof}[Proof of Lemma~\ref{lem:update-cost}]
We can see that $\Pi_\star\Pi_{\sf odd}=0$, since $\Pi_\star$ is supported on states with $0$ in the last register, and $\Pi_{\sf odd}$ is the span of states with an odd $t\in \{1,\dots,{\sf T}_{u,v}-1\}$ in the first term, and an even $t\in\{2,\dots,{\sf T}_{u,v}\}$ in the second term. Thus, 
$$(2\Pi_{\sf even}-I)(2\Pi_\star-I)=-(2(\Pi_{\sf even}+\Pi_\star)-I) = -(2\Pi_{\cal A}-I),$$
where the last equality is because the support of $\Pi_{\cal A}$ is the direct sum of the supports of $\Pi_\star$ and $\Pi_{\sf even}$, by their definitions. 
By a similar argument, $(2\Pi_{\sf odd}-I)(2\Pi_{\leftarrow}-I)=(2\Pi_{\cal B}-I)$, and Thus, the result follows from \clm{fwk-star-unitary}, \clm{fwk-even-unitary}, \clm{fwk-odd-unitary} and \clm{leftarrow-refl}.
\end{proof}

\subsubsection{Positive Analysis}\label{sec:fwk-positive}

Suppose there is a flow $\theta$ on $G$ satisfying conditions \textbf{P1}-\textbf{P5} of \thm{full-framework}. We will use it to make a positive witness as follows. For each $(u,v)\in \overrightarrow{E}(G)$, with $i=f^{-1}_u(v)$ and $j=f_v^{-1}(u)$, define:
\begin{equation}
\begin{split}
\ket{w_{u,v}^0} &:= \ket{u,i}\\
\forall t\in \{1,\dots,{\sf T}_{u,v}\},\; \ket{w_{u,v}^t} &:= U_{t-1}\ket{w_{u,v}^{t-1}}\\
\ket{w_{u,v}} &:= \sum_{t=0}^{{\sf T}_{u,v}}\ket{w_{u,v}^t}\ket{t}+\ket{v,j}\ket{0}.
\end{split}\label{eq:full-fwk-w_uv}
\end{equation}
Then $\ket{w_{u,v}}$ is a kind of \emph{history state} \cite{kitaev1999quantum} for the algorithm on input $(u,i)$. We first show it is almost orthogonal to all algorithm states, defined in \eq{psi-ui}, \eq{psi-z-t} and \eq{psi-vj}.

\begin{claim}\label{clm:ortho1}
For all $(u,v)\in \overrightarrow{E}(G)$, letting $j=f_v^{-1}(u)$:
\begin{enumerate}
\item For all $u'\in V(G)$ and $i'\in L^+(u)$, $\braket{\psi_{\rightarrow}^{u',i'}}{w_{u,v}}=0$.
\item For all $(u',v')\in \overrightarrow{E}(G)$, $z\in{\cal Z}_{u',v'}^0$ and $t\in \{1,\dots,{\sf T}_{u,v}-1\}$, $\braket{\psi_t^{z}}{w_{u,v}}=0$. 
\item For all $v'\in V(G)$ and $j'\in L^-(u)$, 
$|\braket{\psi_{\leftarrow}^{v',j'}}{w_{u,v}}|^2\leq
\delta_{(v,j),(v',j')}{\epsilon_{u,v}}$, where $\delta_{xy}$ denotes the Kronecker delta function.
\end{enumerate}
\end{claim}
\begin{proof}
\textbf{Item 1:}~Recalling that $\ket{\psi_{\rightarrow}^{u',i'}}=\ket{u',i'}\ket{0}-U_0\ket{u',i'}\ket{1}$, we have
\begin{align*}
\braket{\psi_{\rightarrow}^{u',i'}}{w_{u,v}} &= \braket{u',i'}{w_{u,v}^0} - \bra{u',i'}U_0^\dagger\ket{w_{u,v}^{1}}
= \braket{u',i'}{u,i} - \bra{u',i'}U_0^\dagger U_0\ket{u,i} = 0.
\end{align*}

\noindent\textbf{Item 2:}~Recall that $\ket{\psi_{t}^z}=\ket{z}\ket{t}-U_t\ket{z}\ket{t+1}$. This is always orthogonal to the last term of $\ket{w_{u,v}}$, since $t>0$. We also note that if $z\in {\cal Z}_{u',v'}$ for $(u',v')\neq (u,v)$, we have $\braket{\psi_t^z}{w_{u,v}}=0$, since $\ket{w_{u,v}}$ is only supported on $z\in{\cal Z}_{u,v}$. Thus, we can assume $(u,v)=(u',v')$, and so $t<{\sf T}_{u,v}$. Thus:
\begin{align*}
\braket{\psi_t^z}{w_{u,v}} &= \braket{z}{w_{u,v}^t} - \bra{z}U_t^\dagger\ket{w_{u,v}^{t+1}}
= \braket{z}{w_{u,v}^t} - \bra{z}U_t^\dagger U_t\ket{w_{u,v}^{t}} = 0.
\end{align*}

\noindent\textbf{Item 3:} Recall that $\ket{\psi_{\leftarrow}^{v',j'}}=\ket{v',j'}\ket{{\sf T}_{u',v'}}-\ket{v',j'}\ket{0}$. Again, if $(v',j')\neq (v,j)$, $(v',j')\not\in {\cal Z}_{u,v}$, so $\braket{\psi_{\leftarrow}^{v',j'}}{w_{u,v}}=0$. So supposing $(v',j')=(v,j)$, we have:
\begin{align*}
\braket{\psi_{\leftarrow}^{v',j'}}{w_{u,v}} &=  \braket{v,j}{w_{u,v}^{{\sf T}_{u,v}}}-\braket{v,j}{w_{u,v}^0}
-\braket{v,j}{v,j}
= \braket{v,j}{w_{u,v}^{{\sf T}_{u,v}}}-1,
\end{align*}
since $\ket{w_{u,v}^0}=\ket{u,i}$, so the middle term is 0. Then, using 
$$\abs{1-\braket{v,j}{w_{u,v}^{{\sf T}_{u,v}}}}^2\leq\norm{\ket{v,j}-\ket{w_{u,v}^{{\sf T}_{u,v}}}}^2= {\epsilon_{u,v}},$$
by \eq{fwk-alg-error}, 
we have
$|\braket{\psi_{\leftarrow}^{v',j'}}{w_{u,v}}|^2\leq {\epsilon_{u,v}}.$
\end{proof}

Next let $\theta$ be a flow satisfying conditions \textbf{P1}-\textbf{P5} of \thm{full-framework}, which can only exist if $M\neq\emptyset$. Then
we define:
\begin{equation}
\ket{w}=\sum_{(u,v)\in \overrightarrow{E}(G)}\frac{\theta(u,v)}{\sqrt{\w_{u,v}}}\ket{w_{u,v}}-\sum_{u\in V_0}\frac{\theta(u)}{\sqrt{\w_0\sigma(u)}}\ket{u,0}\ket{0}-\sum_{u\in M}\frac{\theta(u)}{\sqrt{\w_{\sf M}}}\ket{u,0}\ket{0},\label{eq:full-fwk-w}
\end{equation}
which we show is a positive witness, in the sense of \defin{pos-witness}.

\begin{lemma}\label{lem:full-fwk-pos-1}
Let $\w_{\sf M}=|V(G)|$, $\w_0=1/{\cal R}^{\sf T}$, and $c_+=7$. Then $\displaystyle \frac{\norm{\ket{w}}^2}{|\braket{w}{\sigma}|^2}  \leq 7 = c_+$.
\end{lemma}
\begin{proof}
To analyse the positive witness, we compute (referring to \eq{fwk-psi-init}): 
\begin{align}
\braket{\psi_0}{w}=\sum_{u\in V_0}\sqrt{\sigma(u)}\braket{u,0,0}{w} = -\sum_{u\in V_0}\sqrt{\sigma(u)}\frac{\theta(u)}{\sqrt{\w_0\sigma(u)}}=-\frac{1}{\sqrt{\w_0}} = -\sqrt{{\cal R}^{\sf T}},\label{eq:fwk-pos-ip}
\end{align}
by condition \textbf{P3} of \thm{full-framework}.
Since this is non-zero, $\ket{w}$ is a positive witness, though it may have some error. To continue, we compute:
\begin{align}
\norm{\ket{w}}^2 &=\sum_{(u,v)\in \overrightarrow{E}(G)}\frac{\theta(u,v)^2}{\w_{u,v}}\norm{\ket{w_{u,v}}}^2+\sum_{u\in{V_0}}\frac{\theta(u)^2}{\w_0\sigma(u)}+\sum_{u\in M}\frac{\theta(u)^2}{\w_{\sf M}}.\label{eq:full-fwk-pos-complexity}
\end{align}
To upper bound the first term of \eq{full-fwk-pos-complexity}, we have $\norm{\ket{w_{u,v}}}^2={\sf T}_{u,v}+2$, so we have
\begin{align}
\sum_{(u,v)\in\overrightarrow{E}(G)}\frac{\theta(u,v)^2}{\w_{u,v}}\norm{\ket{w_{u,v}}}^2 = \sum_{(u,v)\in\overrightarrow{E}(G)}\frac{\theta(u,v)^2}{\w_{u,v}}({\sf T}_{u,v}+2)={\cal E}^{\sf T}(\theta)+2{\cal E}(\theta)\leq 2{\cal R}^{\sf T},\label{eq:full-fwk-pos-complexity-1st}
\end{align}
by condition \textbf{P5} of \thm{full-framework}, and using $2{\cal E}(\theta)\leq {\cal E}^{\sf T}(\theta)$, since each ${\sf T}_{u,v}\geq 2$. 

\noindent To upper bound the second term of \eq{full-fwk-pos-complexity}, we have
\begin{equation}
\begin{split}
\sum_{u\in{V_0}}\frac{\theta(u)^2}{\w_0\sigma(u)} & \leq \frac{1}{\w_0}2\sum_{u\in V_0}\frac{\sigma(u)^2+(\theta(u)-\sigma(u))^2}{\sigma(u)}=
2{\cal R}^{\sf T}\left(1+\sum_{u\in V_0}\frac{(\theta(u)-\sigma(u))^2}{\sigma(u)}\right)
\leq 4{\cal R}^{\sf T}
\end{split}\label{eq:full-fwk-pos-complexity-2nd}
\end{equation}
by condition \textbf{P4} of \thm{full-framework}. 
Finally, we can upper bound the last term of \eq{full-fwk-pos-complexity} as:
\begin{equation}
\sum_{u\in M}\frac{\theta(u)^2}{\w_{\sf M}} = \frac{1}{|V(G)|}\sum_{u\in M}\left(\sum_{v\in \Gamma(u)}\theta(u,v)\right)^2 
\leq \frac{1}{|V(G)|}\sum_{u\in M}d_u\sum_{v\in \Gamma(u)}\theta(u,v)^2
\leq {\cal E}(\theta)\leq {\cal E}^{\sf T}(\theta)\leq {\cal R}^{\sf T}.\label{eq:full-fwk-pos-complexity-3rd}
\end{equation}

\noindent Plug \eq{full-fwk-pos-complexity-1st}, \eq{full-fwk-pos-complexity-2nd} and \eq{full-fwk-pos-complexity-3rd} into \eq{full-fwk-pos-complexity} to get 
$\norm{\ket{w}}^2 \leq 7{\cal R}^{\sf T}$, which, with \eq{fwk-pos-ip}, completes the proof.
\end{proof}

Next, we analyse the error of $\ket{w}$ as a positive witness, by upper bounding its overlap with the various states in $\Psi^{\cal A}\cup\Psi^{\cal B}$. 
First, we have the following corollary to \clm{ortho1}. 
\begin{corollary}\label{cor:ortho1}
\begin{enumerate}
\item For all $u\in V(G)$, $i\in L^+(u)$, $\braket{\psi_{\rightarrow}^{u,i}}{w}=0$.
\item For all $(u,v)\in\overrightarrow{E}(G)$, $z\in{\cal Z}_{u,v}^0$ and $t\in \{1,\dots,{\sf T}_{u,v}-1\}$, $\braket{\psi_{t}^z}{w}=0$.
\item For all $v\in V(G)$ and $j\in L^-(v)$, letting $u=f_v(j)$, we have: 
$|\braket{\psi_{\leftarrow}^{v,j}}{w}|\leq \frac{\theta(u,v)}{\sqrt{\w_{u,v}}}\sqrt{\epsilon_{u,v}}.$
\end{enumerate}
\end{corollary}

\noindent Next we show that the states in $\Psi_\star'$ are orthogonal to $\ket{w}$. 

\begin{claim}\label{clm:ortho2}
For all $u'\in V(G)$, and any $\ket{\psi_{\star}(u')}\ket{0}\in \Psi_{\star}'(u')$, $\braket{\psi_{\star}(u'),0}{w}=0$.
\end{claim}
\begin{proof}
If $u'\not\in V_0\cup M$, we have:
\begin{align*}
\braket{\psi_{\star}(u'),0}{w} &= \sum_{(u,v)\in\overrightarrow{E}(G)}\frac{\theta(u,v)}{\sqrt{\w_{u,v}}}\bra{\psi_{\star}(u'),0}(\ket{u,f_u^{-1}(v)}\ket{0}+\ket{v,f_v^{-1}(u)}\ket{0})\\
&= \sum_{v\in\Gamma^+(u')}\frac{\theta(u',v)}{\sqrt{\w_{u',v}}}\braket{\psi_\star(u')}{u',f_{u'}^{-1}(v)}
+\sum_{u\in\Gamma^-(u')}\frac{\theta(u,u')}{\sqrt{\w_{u,u'}}}\braket{\psi_\star(u')}{u',f_{u'}^{-1}(u)}=0,
\end{align*}
by condition \textbf{P2} of \thm{full-framework}, using $\theta(u,u')=-\theta(u',u)$, and the fact that $\ket{\psi_\star(u')}$ is supported on $\ket{u',i}$ such that $i\in L(u')$.

\noindent If $u'\in M$, the only state in $\Psi_{\star}'(u')$ is $\ket{\psi_{\star}^{G'}(u')}\ket{0}$ (see \eq{psi-star-F-M}), so we have:
\begin{align*}
\braket{\psi_{\star}^{G'}(u'),0}{w} &= \sum_{(u,v)\in\overrightarrow{E}(G)}\frac{\theta(u,v)}{\sqrt{\w_{u,v}}}\braket{\psi_{\star}^{G'}(u'),0}{w_{u,v}}
-\sum_{u\in M}\frac{\theta(u)}{\sqrt{\w_{\sf M}}}\braket{\psi_{\star}^{G'}(u'),0}{u,0,0}\\
&= \sum_{v\in\Gamma^+(u')}\frac{\theta(u',v)}{\sqrt{\w_{u',v}}}\sqrt{\w_{u',v}}
-\sum_{u\in\Gamma^-(u')}\frac{\theta(u,u')}{\sqrt{\w_{u,u'}}}\sqrt{\w_{u,u'}}-\frac{\theta(u')}{\sqrt{\w_{\sf M}}}\sqrt{\w_{\sf M}}\\
&= \sum_{u\in\Gamma(u')}\theta(u,u') - \theta(u') = 0.
\end{align*}

\noindent Similarly, if $u'\in V_0$, 
\begin{align*}
\braket{\psi_{\star}^{G'}(u'),0}{w} 
&= \!\!\!\!\sum_{v\in\Gamma^+(u')}\!\!\theta(u',v)
-\sum_{u\in\Gamma^-(u')}\!\!\theta(u,u')-\frac{\theta(u')}{\sqrt{\w_0\sigma(u')}}\sqrt{\w_0\sigma(u')}
= \sum_{u\in\Gamma(u')}\!\!\theta(u,u') - \theta(u') = 0.\qedhere
\end{align*}
\end{proof}

\noindent We can combine these results in the following lemma:
\begin{lemma}\label{lem:full-fwk-pos-2}
When $M\neq\emptyset$, $\ket{w}$ as defined in \eq{full-fwk-w} is an $\epsilon/2$-positive witness (see \defin{pos-witness}).
\end{lemma}
\begin{proof}
Note that $\ket{w}$ is only defined when $M\neq \emptyset$, as it is constructed with a flow from $V_0$ to $M$.
For $\ket{w}$ to be a positive witness, we require that $\braket{w}{\psi_0}\neq 0$, which follows from \eq{fwk-pos-ip}. All that remains is to show that $\norm{\Pi_{\cal A}\ket{w}}^2$ and $\norm{\Pi_{\cal B}\ket{w}}^2$ are both at most $\frac{\epsilon}{2}\norm{\ket{w}}^2$. 
By \clm{ortho1} and \clm{ortho2}, we have 
\begin{align*}
\norm{\Pi_{\cal A}\ket{w}}^2=0
\mbox{ and }
\norm{\Pi_{\cal B}\ket{w}}^2 &= \sum_{v\in V(G),j\in L^-(v)}\frac{|\braket{\psi_{\leftarrow}^{v,j}}{w}|^2}{\norm{\ket{\psi_{\leftarrow}^{v,j}}}^2}
\leq \sum_{v\in V(G),j\in L^-(v)}\frac{\theta(v,f_v(j))^2 \epsilon_{f_v(j),v}/\w_{v,f_v(j)}}{2}.
\end{align*}
Since $\theta(e)=0$ for all $e\in \tilde{E}$ (condition \textbf{P1} of \thm{full-framework}) and for all $e\in \overrightarrow{E}(G)\setminus\tilde{E}$, $\epsilon_{u,v}\leq \epsilon$ (condition \textbf{TS1} of \thm{full-framework}), we can continue:
\begin{align}
\norm{\Pi_{\cal B}\ket{w}}^2 &\leq \frac{1}{2}\sum_{(u,v)\in\overrightarrow{E}(G)\setminus \tilde{E}}\frac{\theta(u,v)^2 \epsilon_{u,v}}{\w_{u,v}}
\leq \frac{\epsilon}{2}{\cal E}(\theta) < \frac{\epsilon}{2}{\cal E}^{\sf T}(\theta). \label{eq:pos-2-approx}
\end{align}
We have used the fact that the energy of the flow in $G$ (see \defin{flow}), ${\cal E}(\theta)$, is at most
the energy of that flow extended to the graph $G^{\sf T}$ in which we replace the edges by paths of positive lengths determined by ${\sf T}$ (see \defin{nwk-length}). For the final step of the proof, we know due to \eq{full-fwk-pos-complexity-1st} that
\begin{align*}
	\norm{\ket{w}}^2 \geq \sum_{(u,v)\in\overrightarrow{E}(G)}\frac{\theta(u,v)^2}{\w_{u,v}}\norm{\ket{w_{u,v}}}^2 \geq {\cal E}^{\sf T}(\theta), 
\end{align*}
and therefore, it follows from \eq{pos-2-approx} that $\norm{\Pi_{\cal B}\ket{w}}^2\leq \frac{\epsilon}{2}\norm{\ket{w}}^2$, so $\ket{w}$ is an $\epsilon/2$-positive witness.
\end{proof}

\subsubsection{Negative Analysis} 

Let ${\cal A}=\mathrm{span}\{\Psi^{\cal A}\}$ and ${\cal B}=\mathrm{span}\{\Psi^{\cal B}\}$ (see \eq{full-fwk-states}).
In this section, we will define a negative witness, which is some $\ket{w_{\cal A}},\ket{w_{\cal B}}\in H'$, such that $\ket{\psi_0}=\ket{w_{\cal A}}+\ket{w_{\cal B}}$ and $\ket{w_{\cal A}}$ (resp. $\ket{w_{\cal B}}$) is almost in ${\cal A}$ (resp. ${\cal B}$) (see \defin{neg-witness}). 
We first define, for all $(u,v)\in\overrightarrow{E}(G)$ with $i=f_u^{-1}(v)$, 
\begin{equation}
\begin{split}
\ket{w_{u,v}^{\cal A}} &= \sum_{t\in [{\sf T}_{u,v}-1]:t\text{ odd}}\sum_{z\in {\cal Z}_{u,v}^0}\braket{z}{w_{u,v}^t}\ket{\psi_t^{z}}\in{\cal A}\\
\ket{w_{u,v}^{\cal B}} &= \ket{\psi_{\rightarrow}^{u,i}}+\sum_{t\in [{\sf T}_{u,v}-1]:t\text{ even}}\sum_{z\in {\cal Z}_{u,v}^0}\braket{z}{w_{u,v}^t}\ket{\psi_t^{z}}\in{\cal B},
\end{split}\label{eq:w_uv-AB}
\end{equation}
where $\ket{w_{u,v}^t}$ is defined in \eq{full-fwk-w_uv}.

\begin{lemma}\label{lem:alg-states-collapse}
For all $(u,v)\in\overrightarrow{E}(G)$ with $i=f^{-1}_u(v)$, $\ket{w_{u,v}^{\cal A}}+\ket{w_{u,v}^{\cal B}}=\ket{u,i}\ket{0}-\ket{w_{u,v}^{{\sf T}_{u,v}}}\ket{{\sf T}_{u,v}}$. 
\end{lemma}
\begin{proof}
Below we use the fact that for $t<{\sf T}_{u,v}$, $\ket{w_{u,v}^t}\in\mathrm{span}\{\ket{z}:z\in{\cal Z}_{u,v}^0\}$ (see \sec{fwk-transitions}), and that $\ket{w_{u,v}^0}=\ket{u,i}$ (see \eq{full-fwk-w_uv}). 
\begin{align*}
\ket{w_{u,v}^{\cal A}}+\ket{w_{u,v}^{\cal B}} &= \ket{\psi_{\rightarrow}^{u,i}}+
\sum_{t\in [{\sf T}_{u,v}-1]}\sum_{z\in {\cal Z}_{u,v}^0}\braket{z}{w_{u,v}^t}\ket{\psi_t^{z}}\\
&= \ket{\psi_{\rightarrow}^{u,i}}+\sum_{t=1}^{{\sf T}_{u,v}-1}\sum_{z\in{\cal Z}_{u,v}^0}\braket{z}{w_{u,v}^{t}}\ket{z,t}-\sum_{t=1}^{{\sf T}_{u,v}-1}\sum_{z\in{\cal Z}_{u,v}^0}\braket{z}{w_{u,v}^{t}}U_t\ket{z}\ket{t+1} & \mbox{see \eq{psi-z-t}}\\
&= \ket{\psi_{\rightarrow}^{u,i}}+\sum_{t=1}^{{\sf T}_{u,v}-1}\ket{w_{u,v}^{t}}\ket{t}-\sum_{t=1}^{{\sf T}_{u,v}-1}U_t\ket{w_{u,v}^{t}}\ket{t+1} \\
&=\ket{\psi_{\rightarrow}^{u,i}}+\sum_{t=1}^{{\sf T}_{u,v}-1}\ket{w_{u,v}^{t}}\ket{t}-\sum_{t=1}^{{\sf T}_{u,v}-1}\ket{w_{u,v}^{t+1}}\ket{t+1} & \mbox{see \eq{full-fwk-w_uv}}\\
&= \ket{u,i}\ket{0}-U_0\ket{u,i}\ket{1}+\ket{w_{u,v}^1}\ket{1}
-\ket{w_{u,v}^{{\sf T}_{u,v}}}\ket{{\sf T}_{u,v}} & \mbox{see \eq{psi-ui}}\\
&= \ket{u,i}\ket{0}-\ket{w_{u,v}^{{\sf T}_{u,v}}}\ket{{\sf T}_{u,v}},
\end{align*}
since $\ket{w_{u,v}^1}=U_0\ket{u,i}$ (see \eq{full-fwk-w_uv}).
\end{proof}

For $v\in V(G)$ and $j\in L^-(v)$, with $u=f_v(j)$, define
\begin{equation}
\ket{\tilde\psi_{\leftarrow}^{v,j}}:=\ket{w_{u,v}^{{\sf T}_{u,v}}}\ket{{\sf T}_{u,v}}-\ket{v,j}\ket{0}.\label{eq:tilde-psi}
\end{equation}
We now define our negative witness:
\begin{align*}
\ket{w_{\cal A}} &= \frac{1}{\sqrt{\w_0}}\sum_{u\in V(G)}\ket{\psi_{\star}^{G'}(u)}\ket{0}-\frac{1}{\sqrt{\w_0}}\sum_{(u,v)\in\overrightarrow{E}(G)}\sqrt{\w_{u,v}}\ket{w_{u,v}^{\cal A}}    \in{\cal A}\\
\ket{w_{\cal B}} &= -\frac{1}{\sqrt{\w_0}}\sum_{(u,v)\in\overrightarrow{E}(G)}\sqrt{\w_{u,v}}\left(\ket{w_{u,v}^{\cal B}}+\ket{\tilde\psi_{\leftarrow}^{v,f_v^{-1}(u)}}\right).
\end{align*}
We first show that this is indeed a negative witness in \lem{full-fwk-negative} and then analyse its error and complexity in \lem{full-fwk-negative-error}.
\begin{lemma}\label{lem:full-fwk-negative}
Let $\ket{\psi_0}$ be as in \eq{fwk-psi-init}. Then if $M=\emptyset$, $\ket{w_{\cal A}}+\ket{w_{\cal B}}=\ket{\psi_0}$.
\end{lemma}
\begin{proof}
We have:
\begin{align}
\sqrt{\w_0}\left(\ket{w_{\cal A}}+\ket{w_{\cal B}}\right) &= \sum_{u\in V(G)}\ket{\psi_{\star}^{G'}(u)}\ket{0}-\sum_{(u,v)\in\overrightarrow{E}(G)}\sqrt{\w_{u,v}}(\ket{w_{u,v}^{\cal A}}+\ket{w_{u,v}^{\cal B}}+\ket{\tilde\psi_{\leftarrow}^{v,f_v^{-1}(u)}}).\label{eq:wA-wB}
\end{align}
Letting $i=f_u^{-1}(v)$ and $j=f_v^{-1}(u)$, we have:
\begin{equation}
\begin{aligned}
\ket{w_{u,v}^{\cal A}}+\ket{w_{u,v}^{\cal B}}+\ket{\tilde\psi_{\leftarrow}^{v,j}}
&=\ket{u,i}\ket{0}-\ket{w_{u,v}^{{\sf T}_{u,v}}}\ket{{\sf T}_{u,v}}+\ket{\tilde\psi_{\leftarrow}^{v,j}}    & \mbox{by \lem{alg-states-collapse}}\\
 &= \ket{u,i}\ket{0}-\ket{v,j}\ket{0} & \mbox{by \eq{tilde-psi}}.
\end{aligned}\label{eq:wA-wB-1}
\end{equation}
Next we recall from \eq{psi-star-F-1} and \eq{psi-star-F-M} that for $u\in V(G)\setminus M=V(G)$ (since $M=\emptyset$), we have, letting $\delta_{u,V_0}=1$ iff $u\in V_0$:
\begin{equation}
\begin{split}
\ket{\psi_{\star}^{G'}(u)}\ket{0} 
&=\sum_{i\in L^+(u)}\sqrt{\w_{u,i}}\ket{u,i}\ket{0} - \sum_{j\in L^-(u)}\sqrt{\w_{u,j}}\ket{u,j}\ket{0}+\delta_{u,V_0}\sqrt{\w_0\sigma(u)}\ket{u,0}\ket{0}\\
\sum_{u\in V(G)}\ket{\psi_{\star}^{G'}(u)}\ket{0} 
&= \sum_{\substack{(u,v)\\ \in\overrightarrow{E}(G)}}\left(\sqrt{\w_{u,v}}\ket{u,f_u^{-1}(v)}-\sqrt{\w_{u,v}}\ket{v,f_v^{-1}(u)}\right)\ket{0}+\sum_{u\in V_0}\sqrt{\w_0\sigma(u)}\ket{u,0}\ket{0}.
\end{split}\label{eq:wA-wB-2}
\end{equation}
Plugging \eq{wA-wB-1} and \eq{wA-wB-2} back into \eq{wA-wB}, we get:
\begin{align*}
\sqrt{\w_0}\left(\ket{w_{\cal A}}+\ket{w_{\cal B}}\right) &= \sum_{u\in V_0}\sqrt{\w_0\sigma(u)}\ket{u,0}\ket{0}=\sqrt{\w_0}\ket{\psi_0}.\qedhere
\end{align*}
\end{proof}

\begin{lemma}\label{lem:full-fwk-negative-error}
Let $\w_0=1/{\cal R}^{\sf T}$, and $\delta'=\epsilon {\cal R}^{\sf T}{\cal W}+4{\cal R}^{\sf T}\widetilde{\cal W}$. Then $\ket{w_{\cal A}},\ket{w_{\cal B}}$ is a $\delta'$-negative witness (see \defin{neg-witness}), and 
$$\norm{\ket{w_{\cal A}}}^2\leq 2{\cal R}^{\sf T}{\cal W}^{\sf T}+1.$$
\end{lemma}
\begin{proof}
By construction, $\ket{w_{\cal A}}\in {\cal A}$, and the only part of $\ket{w_{\cal B}}$ that is not made up of states in $\Psi^{\cal B}$ (which are in ${\cal B}=\mathrm{span}\{\Psi^{\cal B}\}$) are the $\ket{\tilde\psi_{\leftarrow}^{v,j}}$ parts. Since $(I-\Pi_{\cal B})\ket{\psi_{\leftarrow}^{v,j}}=0$ for all $(v,j)$, we have:
\begin{align*}
(I-\Pi_{\cal B})\ket{w_{\cal B}} &= -\frac{1}{\sqrt{\w_0}}\sum_{v\in V(G),j\in L^-(v)}\sqrt{\w_{v,j}}(I-\Pi_{\cal B})\ket{\tilde\psi_{\leftarrow}^{v,j}}\\
&=-\frac{1}{\sqrt{\w_0}}\sum_{v\in V(G),j\in L^-(v)}\sqrt{\w_{v,j}}(I-\Pi_{\cal B})(\ket{\tilde\psi_{\leftarrow}^{v,j}}-\ket{\psi_{\leftarrow}^{vj}})\\
&=-\frac{1}{\sqrt{\w_0}}\sum_{(u,v)\in\overrightarrow{E}(G)}\sqrt{\w_{u,v}}(I-\Pi_{\cal B})(\ket{w_{u,v}^{{\sf T}_{u,v}}}\ket{{\sf T}_{u,v}} - \ket{v,f_v^{-1}(u)}\ket{{\sf T}_{u,v}}) & \mbox{by \eq{tilde-psi} and \eq{psi-vj}}.
\end{align*}
Then by \eq{full-fwk-w_uv} and \eq{fwk-alg-error} we have:
\begin{equation*}
\norm{\ket{w_{u,v}^{{\sf T}_{u,v}}} - \ket{v,f_v^{-1}(u)}}^2 = \epsilon_{u,v}.
\end{equation*}
Furthermore, the terms $\ket{w_{u,v}^{{\sf T}_{u,v}}} - \ket{v,f_v^{-1}(u)}\in\mathrm{span}\{\ket{z}:z\in {\cal Z}_{u,v}\}$ for different $(u,v)\in \overrightarrow{E}(G)$ are pairwise orthogonal. 
Let $\epsilon$ be the upper bound on $\epsilon _{u,v}$ for all $(u,v)\in\overrightarrow{E}(G)\setminus\tilde{E}$, from \thm{full-framework}. Then:
\begin{align*}
\norm{(I-\Pi_{\cal B})\ket{w_{\cal B}}}^2 &\leq\frac{1}{{\w_0}}\sum_{(u,v)\in\overrightarrow{E}(G)}{\w_{u,v}}\epsilon_{u,v}\\
&\leq {\cal R}^{\sf T}\Bigg(\sum_{(u,v)\in\overrightarrow{E}(G)\setminus \tilde{E}}{\w_{u,v}}\epsilon + 4\sum_{(u,v)\in\tilde{E}}\w_{u,v}\Bigg) \leq {\cal R}^{\sf T}\left(\epsilon {\cal W}+4\widetilde{\cal W}\right),
\end{align*}
where we used $\w_0=1/{\cal R}^{\sf T}$, and the trivial upper bound $\epsilon_{u,v}\leq 4$ when $(u,v)\in \tilde{E}$. 
To complete the proof, we give an upper bound on $\norm{\ket{w_{\cal A}}}^2$. We first note:
\begin{align*}
\norm{\ket{w_{\cal A}}}^2 &= \frac{1}{\w_0}\sum_{u\in V(G)}\norm{\ket{\psi_{\star}^{G'}(u)}}^2+\frac{1}{\w_0}\sum_{(u,v)\in\overrightarrow{E}(G)}\w_{u,v}\norm{\ket{w_{u,v}^{\cal A}}}^2.
\end{align*}
Then we can compute, for any $u\in V(G)$, letting $\delta_{u,V_0}=1$ iff $u\in V_0$,
\begin{align*}
\norm{\ket{\psi_{\star}^{G'}(u)}}^2 &= \textstyle{\displaystyle\sum}_{v\in\Gamma(u)}\w_{u,v}+\delta_{u,V_0} \w_0\sigma(u)
\end{align*}
and for any $(u,v)\in\overrightarrow{E}(G)$, since for all $t<{\sf T}_{u,v}$, $\sum_{z\in {\cal Z}_{u,v}^0}|\braket{z}{w_{u,v}^t}|^2=\norm{\ket{w_{u,v}^t}}^2 =1$, 
\begin{align*}
\norm{\ket{w_{u,v}^{\cal A}}}^2 &= \sum_{t\in\{0,\dots,{\sf T}_{u,v}-1\}:t\text{ odd}}\sum_{z\in {\cal Z}_{u,v}^0}|\braket{z}{w_{u,v}^t}|^2\norm{\ket{\psi_t^{z}}}^2 
= \left\lfloor\frac{{\sf T}_{u,v}}{2}\right\rfloor  \cdot 2 \leq {\sf T}_{u,v}.
\end{align*}
Thus
\begin{align*}
\norm{\ket{w_{\cal A}}}^2 &\leq \frac{1}{\w_0}\sum_{u\in V(G)}\sum_{v\in\Gamma(u)}\w_{u,v}+\frac{1}{\w_0}\sum_{u\in V_0} \w_0\sigma(u)+\frac{1}{\w_0}\sum_{(u,v)\in\overrightarrow{E}(G)}\w_{u,v}{\sf T}_{u,v}\\
&= \frac{1}{\w_0}{\cal W}(G)+1+\frac{1}{\w_0}{\cal W}^{\sf T}(G).
\end{align*}
Note that ${\cal W}(G)$ is always less than ${\cal W}^{\sf T}(G)$ (see \defin{nwk-length}). We Thus, complete the proof by substituting $\w_0=1/{{\cal R}^{\sf T}}$ and using the \textbf{Negative Condition} of \thm{full-framework} that ${\cal W}^{\sf T}(G)\leq {\cal W}^{\sf T}$. 
\end{proof}

\subsubsection{Conclusion of Proof of Theorem \ref{thm:full-framework}}

We now give the proof of \thm{full-framework}, by appealing to \thm{lin-alg-fwk}, using $\ket{\psi_0}$ as defined in \eq{fwk-psi-init}, and $\Psi^{\cal A},\Psi^{\cal B}$ as defined in \eq{full-fwk-states}.
By the \textbf{Setup Subroutine} condition of \thm{full-framework}, we can generate $\ket{\sigma}$ in cost ${\sf S}$. It follows that we can generate $\ket{\psi_0}=\ket{\sigma}\ket{0}\ket{0}$ in cost ${\sf S}'={\sf S}+\log{\sf T}_{\max}$, since the last register is $\log {\sf T}_{\max}$ qubits. By \lem{update-cost}, we can implement $U_{\cal AB}$ in cost ${\sf A}_\star+{\sf polylog}({\sf T}_{\max})$. 

We use $c_+=7$, so $1\leq c_+\leq 50$, as desired. We use  
\begin{equation*}
{\cal C}_-=2{\cal R}^{\sf T}{\cal W}^{\sf T}+1,
\quad
\delta = \frac{\epsilon}{2}
\quad\mbox{and}\quad
\delta' = \epsilon {\cal R}^{\sf T}{\cal W}+4{\cal R}^{\sf T}\widetilde{\cal W}.
\end{equation*}

To apply \thm{lin-alg-fwk}, we require that $\delta\leq \frac{1}{(8c_+)^{3}\pi^8{{\cal C}_-}}$, which follows, for sufficiently large $|x|$, from condition \textbf{TS1} of \thm{full-framework}:
\begin{align*}
\delta = \frac{\epsilon}{2}=o\left(\frac{1}{{\cal R}^{\sf T}{\cal W}^{\sf T}}\right).
\end{align*}
We also require that $\delta'\leq \frac{3}{4}\frac{1}{\pi^4c_+}=\frac{3}{28\pi^4}$. The bound on $\epsilon$ implies that $\epsilon {\cal R}^{\sf T}{\cal W} = o(1)$, since ${\cal W}\leq {\cal W}^{\sf T}$. The bound $\widetilde{\cal W}=o(1/{\cal R}^{\sf T})$ from \textbf{TS2} of \thm{full-framework} implies that $4{\cal R}^{\sf T}\widetilde{\cal W}=o(1)$. 
Together these ensure that $\delta'=o(1)$.  
We verify the remaining conditions of \thm{lin-alg-fwk} as follows.
\begin{description}
\item[Positive Condition:] By \lem{full-fwk-pos-1} and \lem{full-fwk-pos-2}, if $M\neq \emptyset$, there is a $\delta$-positive witness $\ket{w}$ such that $\frac{|\braket{w}{\psi_0}|^2}{\norm{\ket{w}}^2}\geq \frac{1}{c_+}=\frac{1}{7}$. 
\item[Negative Condition:] By \lem{full-fwk-negative-error}, if $M=\emptyset$, there is a $\delta'$-negative witness with $\norm{\ket{w_{\cal A}}}^2\leq {\cal C}_-$. 
\end{description}

\noindent Thus, the algorithm described in \thm{lin-alg-fwk} distinguishes between the cases $M\neq \emptyset$ and $M=\emptyset$ with bounded error in complexity:
$$O\left({\sf S}+\log{\sf T}_{\max}+\sqrt{{\cal C}_-}\left({\sf A}_\star+{\sf polylog}({\sf T}_{\max})\right)\right)=O\left({\sf S}+\sqrt{{\cal R}^{\sf T}{\cal W}^{\sf T}}\left({\sf A}_\star+{\sf polylog}({\sf T}_{\max})\right)\right)$$
which completes the proof of \thm{full-framework}.

\section{Welded Trees}\label{sec:welded}

A straightforward application of our technique is to the welded trees problem of \cite{childs2003ExpSpeedupQW}, illustrating the power of the framework to achieve exponential speedups over classical algorithms. This application also serves as a pedagogic demonstration of the alternative neighbourhoods technique, as it does not make use of a non-trivial edge transition subroutine, and so the resulting algorithm is in that sense rather simple. Although it would be possible to apply our framework without looking ``under the hood'' at the underlying algorithm, to give intuition about the framework, we instead describe and analyse the full algorithm explicitly, proving our upper bound without appealing to \thm{full-framework}. 

\paragraph{The Welded Trees Problem:} In the welded trees problem, the input is a graph $G$ with $2^{n+2}-2$ vertices from the set $\{0,1\}^{2n}$, consisting of two full binary trees of depth $n$ (the $2^n$ leaves are at edge-distance $n$ from the root), which we will refer to as the left and right trees, with additional edges connecting the leaves of one tree to another. Specifically, we assume there are two disjoint perfect matchings from the leaves of the left tree to the leaves of the right tree. Every vertex of this graph has degree 3 except for the roots of the two trees, which we denote by $s$ and $t$. The graph's structure is shown in \fig{welded-trees-weights}.

We are promised that $s=0^{2n}$ is the root of the left tree, but other than $s$, it is difficult to even find a vertex in the graph, since less than a $2^{-n+2}$ fraction of strings in $\{0,1\}^{2n}$ labels an actual vertex. We assume we have access to an oracle $O_G$ that tells us the neighbours of any vertex.
That is, for any string $\sigma\in\{0,1\}^{2n}$, we can query $O_G(\sigma)$ to learn either $\bot$, indicating it is not a vertex label, or a list of three neighbours (or in case of $s$ and $t$, only two neighbours). 

The welded trees problem is: given such an oracle $O_G$, output the label of $t$. We assume we can identify $t$ when we see it, for example by querying it to see that it only has two neighbours. Classically, this problem requires $2^{\Omega(n)}$ queries~\cite{childs2003ExpSpeedupQW}, which is intuitively because the problem is set up to ensure that the only thing a classical algorithm can do is a random walk on $G$, starting from $s$. The hitting time from $s$ to $t$ is $2^{\Omega(n)}$ because a walker is always twice as likely to move towards the centre of the graph than away from it, and so a walker starting at $s$ will quickly end up in the centre of the graph, but then will be stuck there for a long time. 
On the other hand a quantum algorithm can solve this problem in poly$(n)$ time~\cite{childs2003ExpSpeedupQW}, with the best known upper bound being $O(n^{1.5}\log n)$ queries~\cite{atia2021welded}. We show how to solve this problem in our new framework, with $O(n)$ queries and ${O}(n^2)$ time. Specifically, in the remainder of this section we show:

\begin{theorem}\label{thm:welded}
Let $g:\{0,1\}^{2n}\rightarrow\{0,1\}$ be any function. Then there is a quantum algorithm that, given an oracle $O_G$ for a welded trees graph $G$ as above, decides if $g(t)=1$ with bounded error in $O(n)$ queries to $O_G$. If $g$ can be computed in $O(n)$ complexity, then the time complexity of this algorithm is ${O}(n^2)$. 
\end{theorem}

From this it immediately follows that we can solve the welded trees problem with $2n$ applications of this algorithm, letting $g(t)=t_i$ -- the $i$-th bit of $t$ -- for $i=1,\dots,2n$. However, we can also do slightly better by composing it with the Bernstein-Vazirani algorithm, which recovers a string $t$ in a single quantum query to an oracle that computes $\ket{z}\mapsto (-1)^{z\cdot t}\ket{z}$ for any string $z\in\{0,1\}^{2n}$.
\begin{corollary}
There is a quantum algorithm that can solve the welded trees problem in $O(n)$ queries and ${O}(n^2)$ time. 
\end{corollary}
\begin{proof}
For any $z\in\{0,1\}^{2n}$, define $g_z(t)=z\cdot t = \sum_{i=1}^{2n}z_it_i\mod 2$. Clearly $g_z$ can be computed in complexity $O(n)$. To compute $g_z(t)$, we simply run the algorithm from \thm{welded}. The Bernstein-Vazirani algorithm~\cite{bernstein1997BValg} outputs $t$ using a single such query, and $O(n)$ additional gates. 
\end{proof}

Previous quantum algorithms for this problem are quantum walk algorithms in the sense that they construct a Hamiltonian based on the structure of the graph and simulate it, but this technique has not been replicated for many other problems, unlike quantum walk search algorithms described in \sec{intro-QW}\footnote{Where similar frameworks have been developed in the setting of continuous quantum walks~\cite{apers2022quadratic}, they are also limited to a quadratic speedup.}. Our hope is that our new quantum walk search framework bridges the gap between a general and easily applied technique (quantum walk search algorithms) and exponential speedups.

\paragraph{$G$ as a Weighted Network:} Assume that $n$ is even (this greatly simplifies notation, the proof is equivalent for the case where $n$ is odd). We partition $V(G)$ into $V_0\cup V_1\cup\dots\cup V_{2n+1}$, where $V_k$ is the set of vertices at distance $k$ from $s$, so $V_0=\{s\}$, and $V_{2n+1}=\{t\}$. We first prove \thm{welded} under the assumption that it is possible to check, for any vertex, whether it is in $V_{\mathrm{even}}:=V_0\cup V_2\cup\dots \cup V_{2n}$, or $V_{\mathrm{odd}}:=V_1\cup V_3\cup \dots\cup V_{2n+1}$. At the end of this section, we will explain how to remove this assumption. Define $M=\{t\}$ if $g(t)=1$ and otherwise $M=\emptyset$.

For $k\in \{1,\dots,2n+1\}$, define
\begin{equation}
\begin{split}
E_k&=\{(u,v)\in V_{k-1}\times V_k: \{u,v\}\in E(G)\}\\
\mbox{so }|E_k| &= \left\{\begin{array}{ll}
2^k & \mbox{if }k\in\{1,\dots,n+1\}\\
2^{2n+2-k} &\mbox{if }k\in\{n+1,\dots,2n+1\}.
\end{array}\right.
\end{split}\label{eq:welded-E_k}
\end{equation}
We define the set of directed edges as follows (see \defin{network}):
\begin{equation} 
\overrightarrow{E}(G)=\bigcup_{\substack{k\in\{1,\dots,2n+1\}:\\k\;\mathrm{mod}\; 4\in \{0,1\}}}\{(u,v):(u,v)\in E_k\}
\cup \bigcup_{\substack{k\in\{1,\dots,2n+1\}:\\k\;\mathrm{mod}\; 4\in \{2,3\}}}\{(u,v):(v,u)\in E_k\}.\label{eq:welded-EG}
\end{equation}
Note that $E_k\subset \overrightarrow{E}(G)$ only holds when $k\mod 4 \in \{0,1\}$, so we do not always set the default directions left to right. 
At the moment it is not clear why we set the directions this way, but one thing this accomplishes is that the direction of edges switches at every layer of $V_{\mathrm{odd}}$. \fig{welded-trees-weights} illustrates how the directions of the edges change layer by layer. 
We now assign weights to all edges in $G$. We assign all edges in $E_k$ the same weight, $\w_k$, defined:
\begin{equation}
\w_k=\left\{\begin{array}{ll}
2^{-2\lceil k/2\rceil} & \mbox{if }k\in\{1,\dots,n\}\\
2^{-2(n+2-\lceil k/2\rceil)} & \mbox{if }k\in\{n+1,\dots,2n+1\}.
\end{array}\right.\label{eq:welded-w_k}
\end{equation}
It should be somewhat clear why this might be a useful weighting: we have increased the probability of moving away from the centre. Finally, we add a vertex $v_0$ connected to $s$ by an edge of weight $\w_0$, and connected to $t$ by an edge of weight $\w_{\sf M}$ if and only if $t$ is marked, and call this resulting graph $G'$. We remark that we do not need to account for $v_0$ explicitly if we just want to appeal directly to \thm{full-framework}, but we are going to explicitly construct an algorithm for the sake of exemplification.

\begin{figure}
\centering
\begin{tikzpicture}[scale=1.35]
\filldraw (-7.5,0) circle (.1); 		\node at (-7.5,.25) {$v_0$};

	\draw[-{Latex[length=2mm, width=2mm]}] (-7.4,0)--(-6.6,0);	\node at (-7,.25) {$\w_0$};

\filldraw[fill=blue] (-6.5,0) circle (.1); 	\node at (-6.5,.25) {$s$};

	\draw[-{Latex[length=2mm, width=2mm]}] (-6.4,0) -- (-4.6,2); 	\node at (-5.5,1.35) {$\frac{1}{4}$};
	\draw[-{Latex[length=2mm, width=2mm]}] (-6.4,0) -- (-4.6,-2);	\node at (-5.5,-1.35) {$\frac{1}{4}$};

\filldraw (-4.5,2) circle (.1); 
\filldraw (-4.5,-2) circle (.1);

	\draw[{Latex[length=2mm, width=2mm]}-] (-4.4,2) -- (-3.1,3);	\node at (-3.75,2.85) {$\frac{1}{4}$};
	\draw[{Latex[length=2mm, width=2mm]}-] (-4.4,2) -- (-3.1,1);
	\draw[{Latex[length=2mm, width=2mm]}-] (-4.4,-2) -- (-3.1,-3);
	\draw[{Latex[length=2mm, width=2mm]}-] (-4.4,-2) -- (-3.1,-1);

\filldraw[fill=blue] (-3,3) circle (.1);
\filldraw[fill=blue] (-3,1) circle (.1);
\filldraw[fill=blue] (-3,-1) circle (.1);
\filldraw[fill=blue] (-3,-3) circle (.1);

	\draw[{Latex[length=2mm, width=2mm]}-] (-2.9,3) -- (-2.1,3.75);		\node at (-2.6,3.65) {$\frac{1}{16}$};
	\draw[{Latex[length=2mm, width=2mm]}-] (-2.9,3) -- (-2.1,2.25);
	  \draw[{Latex[length=2mm, width=2mm]}-] (-2.9,1) -- (-2.43,1.5);
	  \draw[{Latex[length=2mm, width=2mm]}-] (-2.9,1) -- (-2.43,.5);
	  \draw[{Latex[length=2mm, width=2mm]}-] (-2.9,-1) -- (-2.43,-.5);
	  \draw[{Latex[length=2mm, width=2mm]}-] (-2.9,-1) -- (-2.43,-1.5);
	\draw[{Latex[length=2mm, width=2mm]}-] (-2.9,-3) -- (-2.1,-3.75);
	\draw[{Latex[length=2mm, width=2mm]}-] (-2.9,-3) -- (-2.1,-2.25);

\filldraw (-2,3.75) circle (.1);
\filldraw (-2,2.25) circle (.1);
\filldraw (-2,-2.25) circle (.1);
\filldraw (-2,-3.75) circle (.1);

	\draw[-{Latex[length=2mm, width=2mm]}] (-1.9,3.75) -- (-1.1,4.125);	\node at (-1.5,4.25) {$\frac{1}{16}$};
	\draw[-{Latex[length=2mm, width=2mm]}] (-1.9,3.75) -- (-1.1,3.375);
	\draw[-{Latex[length=2mm, width=2mm]}] (-1.9,2.25) -- (-1.1,2.625);
	\draw[-{Latex[length=2mm, width=2mm]}] (-1.9,2.25) -- (-1.1,1.875);
	\draw[-{Latex[length=2mm, width=2mm]}] (-1.9,-3.75) -- (-1.1,-4.125);
	\draw[-{Latex[length=2mm, width=2mm]}] (-1.9,-3.75) -- (-1.1,-3.375);
	\draw[-{Latex[length=2mm, width=2mm]}] (-1.9,-2.25) -- (-1.1,-2.625);
	\draw[-{Latex[length=2mm, width=2mm]}] (-1.9,-2.25) -- (-1.1,-1.875);

\filldraw[fill=blue] (-1,4.125) circle (.1);
\filldraw[fill=blue] (-1,3.375) circle (.1);
\filldraw[fill=blue] (-1,2.625) circle (.1);
\filldraw[fill=blue] (-1,1.875) circle (.1);
\filldraw[fill=blue] (-1,-4.125) circle (.1);
\filldraw[fill=blue] (-1,-3.375) circle (.1);
\filldraw[fill=blue] (-1,-2.625) circle (.1);
\filldraw[fill=blue] (-1,-1.875) circle (.1);


\draw[-{Latex[length=2mm, width=2mm]}] (-.9,4.125) -- (.9,4.125);		\node at (0,4.4) {$\frac{1}{2^{n+2}}$};
\draw[-{Latex[length=2mm, width=2mm]}] (-.9,4.125) -- (.9,1.875);
\draw[-{Latex[length=2mm, width=2mm]}] (-.9,3.375) -- (.9,2.625);
\draw[-{Latex[length=2mm, width=2mm]}] (-.9,3.375) -- (.9,-2.625);


\filldraw (6.5,0) circle (.1); \node at (6.5,.25) {$t$};

	\draw[{Latex[length=2mm, width=2mm]}-] (6.4,0) -- (4.6,2);	\node at (5.5,1.35) {$\frac{1}{4}$};
	\draw[{Latex[length=2mm, width=2mm]}-] (6.4,0) -- (4.6,-2);

\filldraw[fill=blue] (4.5,2) circle (.1); 
\filldraw[fill=blue] (4.5,-2) circle (.1);

	\draw[{Latex[length=2mm, width=2mm]}-] (4.4,2) -- (3.1,3);	\node at (3.75,2.85) {$\frac{1}{16}$};
	\draw[{Latex[length=2mm, width=2mm]}-] (4.4,2) -- (3.1,1);
	\draw[{Latex[length=2mm, width=2mm]}-] (4.4,-2) -- (3.1,-3);
	\draw[{Latex[length=2mm, width=2mm]}-] (4.4,-2) -- (3.1,-1);

\filldraw (3,3) circle (.1);
\filldraw (3,1) circle (.1);
\filldraw (3,-1) circle (.1);
\filldraw (3,-3) circle (.1);

	\draw[-{Latex[length=2mm, width=2mm]}] (2.9,3) -- (2.1,3.75);	\node at (2.6,3.65) {$\frac{1}{16}$};
	\draw[-{Latex[length=2mm, width=2mm]}] (2.9,3) -- (2.1,2.25);
	  \draw[-{Latex[length=2mm, width=2mm]}] (2.9,1) -- (2.43,1.5);
	  \draw[-{Latex[length=2mm, width=2mm]}] (2.9,1) -- (2.43,.5);
	  \draw[-{Latex[length=2mm, width=2mm]}] (2.9,-1) -- (2.43,-.5);
	  \draw[-{Latex[length=2mm, width=2mm]}] (2.9,-1) -- (2.43,-1.5);
	\draw[-{Latex[length=2mm, width=2mm]}] (2.9,-3) -- (2.1,-3.75);
	\draw[-{Latex[length=2mm, width=2mm]}] (2.9,-3) -- (2.1,-2.25);

\filldraw[fill=blue] (2,3.75) circle (.1);
\filldraw[fill=blue] (2,2.25) circle (.1);
\filldraw[fill=blue] (2,-2.25) circle (.1);
\filldraw[fill=blue] (2,-3.75) circle (.1);

	\draw[-{Latex[length=2mm, width=2mm]}] (1.9,3.75) -- (1.1,4.125);	\node at (1.7,4.25) {$\frac{1}{2^{n+2}}$};
	\draw[-{Latex[length=2mm, width=2mm]}] (1.9,3.75) -- (1.1,3.375);
	\draw[-{Latex[length=2mm, width=2mm]}] (1.9,2.25) -- (1.1,2.625);
	\draw[-{Latex[length=2mm, width=2mm]}] (1.9,2.25) -- (1.1,1.875);
	\draw[-{Latex[length=2mm, width=2mm]}] (1.9,-3.75) -- (1.1,-4.125);
	\draw[-{Latex[length=2mm, width=2mm]}] (1.9,-3.75) -- (1.1,-3.375);
	\draw[-{Latex[length=2mm, width=2mm]}] (1.9,-2.25) -- (1.1,-2.625);
	\draw[-{Latex[length=2mm, width=2mm]}] (1.9,-2.25) -- (1.1,-1.875);

\filldraw (1,4.125) circle (.1);
\filldraw (1,3.375) circle (.1);
\filldraw (1,2.625) circle (.1);
\filldraw (1,1.875) circle (.1);
\filldraw (1,-4.125) circle (.1);
\filldraw (1,-3.375) circle (.1);
\filldraw (1,-2.625) circle (.1);
\filldraw (1,-1.875) circle (.1);

\end{tikzpicture}
\caption{The weights of the graph $G'$ (obtained from adding $v_0$ to $G$), and default edge directions.}\label{fig:welded-trees-weights}
\end{figure}
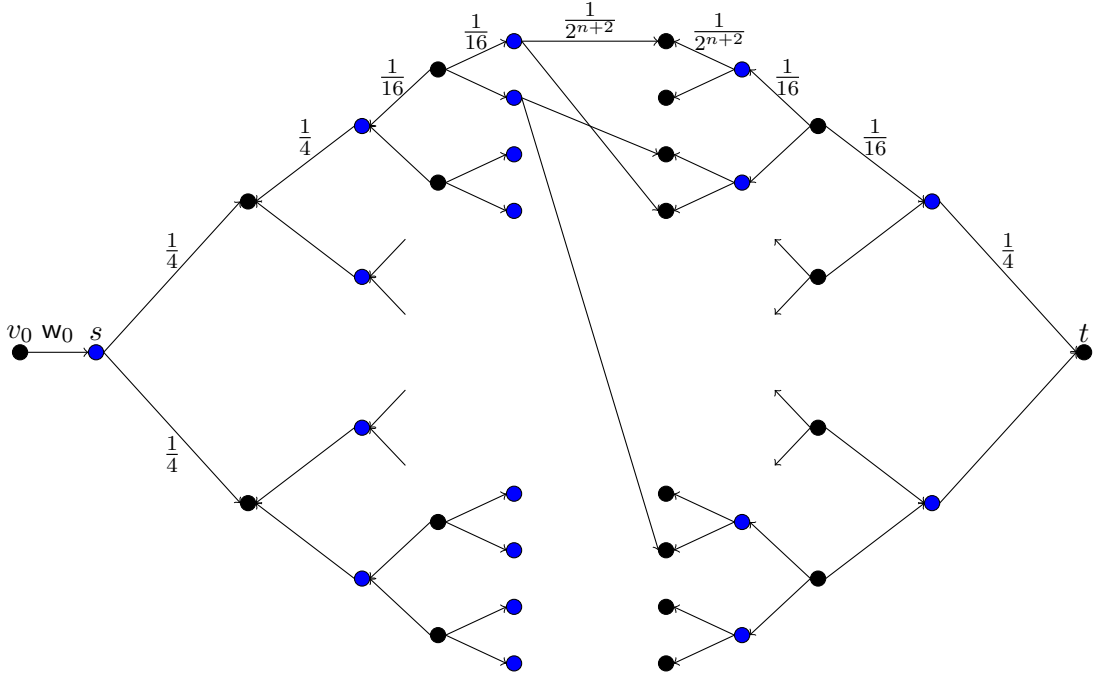

\thm{full-framework} assumes we label the outgoing edges from a vertex $u$ by indices from some set $L(u)$, and then implement a map $\ket{u,i}\mapsto \ket{v,j}$ for any $i\in L(u)$ (see also \defin{QW-access}). Here, since we assume we can simply query the set of the three neighbours of $u$, $\Gamma(u)=\{v_1,v_2,v_3\}$, in unit time, we can let $L(u)=\Gamma(u)$. In that case, the map $\ket{u,i}\mapsto\ket{v,j}$ is actually just $\ket{u,v}\mapsto \ket{v,u}$, which can be accomplished by swapping the two registers. The decomposition of $L(u)$ into $L^+(u)=\Gamma^+(u)$ and $L^-(u)=\Gamma^-(u)$ depends on the directions we assigned to the edges coming out of $u$ in \eq{welded-EG}. 

Our algorithm will be based on phase estimation of a unitary acting on:
\begin{equation}
H=\mathrm{span}\{\ket{u,v}:u\in V(G),v\in\Gamma(u)\}.\label{eq:welded-H}
\end{equation}
Here we let $\Gamma(u)=\Gamma_{G'}(u)$ refer to the neighbours of $u$ in $G'$, meaning that $\Gamma(s)$ and $\Gamma(t)$ both include $v_0$. We emphasise that for all $(u,v)\in \overrightarrow{E}(G)$, we have included both $\ket{u,v}$ and $\ket{v,u}$ in $H$, as orthonormal vectors. This is different from the $H$ defined in \sec{sketch} (in which $\ket{u,v}=-\ket{v,u}$), and instead we should think of $\ket{u,v}$ as analogous to $\ket{u,i}\ket{0}$ in \eq{full-fwk-H}: $v$ takes the place of the label $i$, and there is no $\ket{0}$ because there is no transition subroutine steps to count. Alternatively, we can think of $\ket{u,v}$ and $\ket{v,u}$ as labelling distinct edges on a path of length two connecting $u$ and $v$: one adjacent to $u$, and the other adjacent to $v$ (see also \fig{welded-spaces}).

For any $u\in V(G)$ and $v\in \Gamma(u)$, let $\Delta_{u,v}=0$ if $v\in\Gamma^+(u)$ (i.e.~$(u,v)\in\overrightarrow{E}(G')=\overrightarrow{E}(G)\cup\{(s,v_0),(t,v_0)\}$) and $\Delta_{u,v}=1$ if $v\in\Gamma^-(u)$ (i.e.~$(v,u)\in\overrightarrow{E}(G')$). 
We can then define star states in $H$ as follows, for all $u\in V(G)$ with neighbours (in $G'$) $\Gamma(u)=\{v_1,v_2,v_3\}$:  
\begin{equation}
\ket{\psi_\star^{G'}(u)} := \sqrt{\w_{u,v_1}}(-1)^{\Delta_{u,v_1}}\ket{u,v_1}+\sqrt{\w_{u,v_2}}(-1)^{\Delta_{u,v_2}}\ket{u,v_2}+\sqrt{\w_{u,v_3}}(-1)^{\Delta_{u,v_3}}\ket{u,v_3}.\label{eq:welded-stars}
\end{equation}
We cannot efficiently generate these star states. For $\ell\in \{0,\dots,n-1\}$ and $v\in V_{2\ell+1}\subset V_{\mathrm{odd}}$ with neighbours $\Gamma(v)=\{u_1,u_2,u_3\}$, we have, referring to \eq{welded-EG} and \eq{welded-w_k}:
\begin{equation}
\begin{split}
\ket{\psi_\star^{G'}(v)} &= \left\{\begin{array}{ll}
(-1)^{\ell+1}\frac{1}{2^{\ell+1}}\left(\ket{v,u_1}+\ket{v,u_2}+\ket{v,u_3} \right) & \mbox{if }\ell\in\{0,\dots,n/2-1\}\\
(-1)^{\ell+1}\frac{1}{2^{n-\ell+1}}\left(\ket{v,u_1}+\ket{v,u_2}+\ket{v,u_3} \right) & \mbox{if }\ell\in \{n/2,\dots,n\}.
\end{array}\right.
\end{split}\label{eq:welded-stars-B}
\end{equation}
Even though we don't know which layer $v$ is in, and therefore we do not know the precise scaling or direction of its edges, by querying the neighbours of $v$ to learn the set $\{u_1,u_2,u_3\}$, we know that:
$$\ket{\psi_\star^{G'}(v)}\in\mathrm{span}\{\ket{v,u_1}+\ket{v,u_2}+\ket{v,u_3}\}.$$
On the other hand, for $\ell\in\{1,\dots,n\}$ and $u\in V_{2\ell}\subset V_{\mathrm{even}}$, though we can compute $\Gamma(u)=\{v_1,v_2,v_3\}$, we do not know which neighbour is the parent -- the unique neighbour of $u$ that is further from the centre of the graph than $u$. Let $p(u)\in\{v_1,v_2,v_3\}$ be the parent of $u$, and $c_1(u)$, $c_2(u)$ the other two vertices in $\{v_1,v_2,v_3\}$. Then, referring to \eq{welded-EG}, \eq{welded-w_k} and \eq{welded-stars}, the star state of $u$ has the form:
\begin{align*}
\ket{\psi_\star^{G'}(u)}
&=\left\{\begin{array}{ll}
(-1)^{\ell+1}\frac{1}{2^{\ell}}\left(\ket{u,p(u)}-\frac{1}{2}\ket{u,c_1(u)}-\frac{1}{2}\ket{u,c_1(u)}\right) & \mbox{if }\ell\in\left\{0,\dots,\frac{n}{2}\right\}\\
 (-1)^{\ell+1}\frac{1}{2^{n-\ell+1}}\left(\ket{u,p(u)}-\frac{1}{2}\ket{u,c_1(u)}-\frac{1}{2}\ket{u,c_1(u)}\right)& \mbox{if }\ell\in\left\{\frac{n}{2}+1,\dots,n\right\}.
\end{array}\right.
\end{align*}
Generating this state would require knowing which of $\{v_1,v_2,v_3\}$ is the parent, $p(u)$, which is not something that can be learned from simply querying the neighbours of $u$. However, if we were to weight everything uniformly, our quantum walk would, like a classical random walk, suffer from the fact that the centre of the graph has exponential weight, and most time will be spent there. Thus, we employ the alternative neighbourhoods technique. 
For $u\in V_{\mathrm{even}}\setminus\{s\}$ with neighbours $v_1<v_2<v_3$, define:
\begin{equation}
\begin{split}
\forall j\in\{1,2,3\},\; \ket{\psi_\star^j(u)} &= \sqrt{\frac{2}{3}}\left(\ket{u,v_j}-\frac{1}{2}\ket{u,v_{j+1}}-\frac{1}{2}\ket{u,v_{j+2}}\right),
\end{split}\label{eq:welded-alt-neighbourhood}
\end{equation}
where the indices add modulo 3. Then we know that 
$$\frac{\ket{\psi_\star^G(u)}}{\norm{\ket{\psi_\star^G(u)}}}\in \{\ket{\psi_\star^1(u)},\ket{\psi_\star^2(u)},\ket{\psi_\star^3(u)}\}=:\Psi_\star(u),$$
though we do not know which one. 

For $s$ and $t$, suppose the neighbours are $\Gamma(s)=\{v_0,v_1,v_2\}$ and $\Gamma(t)=\{v_0,u_1,u_2\}$ -- meaning that when we query $s$, we learn $\{v_1,v_2\}$, and when we query $t$ we learn $\{u_1,u_2\}$, and then we add to each neighbourhood the additional special vertex $v_0$ (although if $t$ is not marked, we set the weight $\w_{t,v_0}=0$). 
Then the star states are, respectively:
\begin{equation}
\begin{split}
\ket{\psi_\star^{G'}(s)} &= \sqrt{\w_0}\ket{s,v_0}+\frac{1}{2}\ket{s,v_1}+\frac{1}{2}\ket{s,v_2}\\
\ket{\psi_\star^{G'}(t)} &= \delta_{g(t),1}\sqrt{\w_{\sf M}}\ket{t,v_0}-\frac{1}{2}\ket{t,u_1}-\frac{1}{2}\ket{t,u_2},
\end{split}\label{eq:welded-st-stars}
\end{equation}
where $\delta_{g(t),1}=1$ if and only if $t\in M$. To see that this follows from \eq{welded-stars}, note that $(s,v_0),(t,v_0)\in \overrightarrow{E}(G')$ by definition, and for $i\in\{1,2\}$, $(s,v_i)\in E_1$, so $(s,v_i)\in \overrightarrow{E}(G)$ since $1=1\mod 4$, and $(u_i,t)\in E_{2n+1}$, so $(u_i,t)\in\overrightarrow{E}(G)$, since $2n+1=1\mod 4$ (we are assuming $n$ is even), which is why we have minus signs in front of the $\ket{t,u_i}$
(see \eq{welded-EG}).
We can generate $\ket{\psi_\star^{G'}(s)}$ and $\ket{\psi_\star^{G'}(t)}$, because we can recognise $s$ and~$t$. Thus, for all $v\in  V_{\mathrm{odd}}\cup\{s\}$ (the vertices with easy to generate star states), we let $\Psi_\star(v):=\{\ket{\psi_\star^{G'}(v)}\}$. 

Finally, we define states corresponding to the transitions $\ket{u,v}\mapsto \ket{v,u}$:
\begin{equation}
\forall (u,v)\in \overrightarrow{E}(G),\; \ket{\psi_{u,v}} := \ket{u,v}-\ket{v,u}.\label{eq:welded-transition-states}
\end{equation}

\begin{figure}
\centering
\begin{tikzpicture}[scale=1.2]
\node at (0,0) {\begin{tikzpicture}[scale=1.2]
\draw (-1.5,-.75) -- (1.5,.75);
\draw (0,0) -- (1.5,-.75);
\draw (-1,-.5) -- (-.5,-.75); 	\draw (1,.5)--(1.5,.25);	\draw (1,-.5)--(1.5,-.25);

\filldraw (-1,-.5) circle (.1);
\filldraw (0,0) circle (.1);
\filldraw (1,.5) circle (.1);
\filldraw (1,-.5) circle (.1);

\node at (-1.25,-.5) {$u$};
\node at (0,.25) {$v$};
\node at (1,.8) {$u'$};
\node at (1,-.8) {$u''$};

\node at (0,-1.25) {$G$};
\end{tikzpicture}};

\node at (8.5,0) {\begin{tikzpicture}[scale=1.2]
\draw (-5,-2.5)--(5,2.5);
\draw (0,0)--(5,-2.5);
\draw (-4,-2)--(-3,-2.5);	\draw (4,2)--(5,1.5); 	\draw (4,-2)--(5,-1.5);

\node[rectangle, rounded corners, draw, thick, fill=white] at (-4,-2) {$\Psi_\star(u)$};
	\node[rectangle, rounded corners, draw, thick, fill=white] at (-2,-1) {$\ket{\psi_{u,v}}$};

\node[rectangle, rounded corners, draw, thick, fill=white] at (0,0) {$\Psi_\star(v)$};
	\node[rectangle, rounded corners, draw, thick, fill=white] at (2,1) {$\ket{\psi_{v,u'}}$};
	\node[rectangle, rounded corners, draw, thick, fill=white] at (2,-1) {$\ket{\psi_{v,u''}}$};

\node[rectangle, rounded corners, draw, thick, fill=white] at (4,2) {$\Psi_\star(u')$};

\node[rectangle, rounded corners, draw, thick, fill=white] at (4,-2) {$\Psi_\star(u'')$};

\end{tikzpicture}};

\end{tikzpicture}
\caption{A piece of the graph $G$ (left) and the corresponding piece of the overlap graph of the spaces $\mathrm{span}\{\Psi_\star(u)\}$ and $\mathrm{span}\{\ket{\psi_{u,v}}\}$. There is an edge between two nodes in the overlap graph if and only if the sets contain overlapping states. Compare with \fig{spaces-graph}.}\label{fig:welded-spaces}
\end{figure}
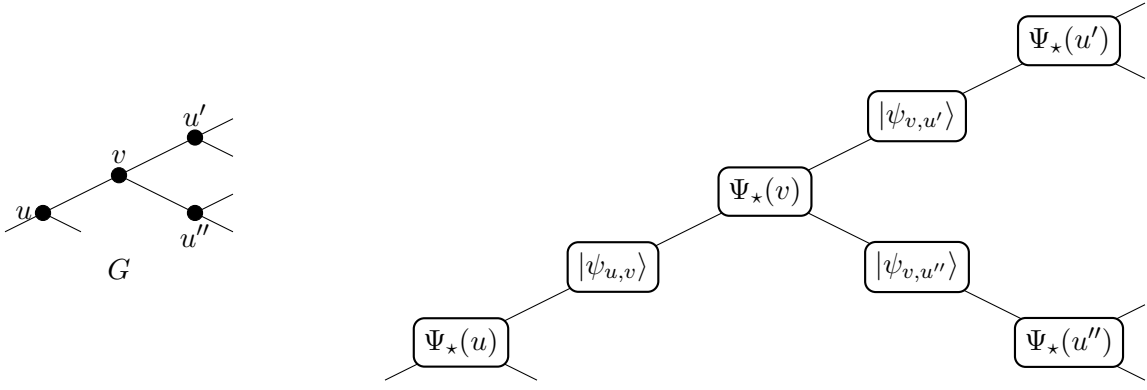

The star states  and the transition states \eq{welded-transition-states} will comprise all states of $\Psi^{\cal A}\cup \Psi^{\cal B}$ as follows:
\begin{equation}
\begin{split}
\Psi^{\cal A} &:= \bigcup_{u\in V(G)}\Psi_\star(u)=\bigcup_{u\in V_{\mathrm{odd}}\cup\{s\}}\{\ket{\psi_\star^{G'}(u)}\}\cup\bigcup_{u\in V_{\mathrm{even}}\setminus\{s\}}\{\ket{\psi^1_\star(u)},\ket{\psi^2_\star(u)},\ket{\psi^3_\star(u)}\}\\
\Psi^{\cal B} &:= \{\ket{\psi_{u,v}}:(u,v)\in\overrightarrow{E}(G)\}.
\end{split}\label{eq:welded-Psi}
\end{equation}
Then $\Psi^{\cal B}$ is a pairwise orthogonal set, and if we replace each $\Psi_\star(u)$ in $\Psi^{\cal A}$ with an orthonormal basis for $\mathrm{span}\{\Psi_\star(u)\}$ we get a pairwise orthogonal set. \fig{welded-spaces} shows that the overlap graph for the sets $\Psi_\star(u)$ for $u\in V(G)$ and $\{\ket{\psi_{u,v}}\}$ for $(u,v)\in\overrightarrow{E}(G)$ is bipartite, and we have chosen $\Psi^{\cal A}$ and $\Psi^{\cal B}$ according to this bipartition. 

Let $U_{\cal AB}=(2\Pi_{\cal A}-I)(2\Pi_{\cal B}-I)$, where $\Pi_{\cal A}$ and $\Pi_{\cal B}$ are orthogonal projectors on ${\cal A}:=\mathrm{span}\{\Psi^{\cal A}\}$ and ${\cal B}:=\mathrm{span}\{\Psi^{\cal B}\}$ respectively.
In the remainder of this section, we will show that we can solve the welded trees problem with bounded error by performing phase estimation of $U_{\cal AB}$ on initial state $\ket{\psi_0}=\ket{s,v_0}$, as described in \thm{lin-alg-fwk}. 

\paragraph{Implementing the Unitary:} In order to implement $U_{\cal AB}$, we need to be able to generate an orthonormal basis for each of ${\cal A}$ and ${\cal B}$, for which we use the following fact. 
\begin{claim}\label{clm:welded-F-basis}
Let $\omega_3=e^{2\pi i/3}$ be a third root of unity.
For a vertex $u\in V(G)$ with neighbours $v_1<v_2<v_3$, define for $j\in\{0,1,2\}$:
$$\ket{\widehat\psi^j(u)}:=\frac{1}{\sqrt{3}}\left(\ket{u,v_1}+\omega_3^j\ket{u,v_2}+\omega_3^{2j}\ket{u,v_3}\right).$$
Then these three vectors are an orthonormal set, and for $u\in V_{\mathrm{even}}\setminus\{s\}$,
$$\mathrm{span}\{\ket{\psi_\star^1(u)},\ket{\psi_\star^2(u)},\ket{\psi_\star^3(u)}\}=\mathrm{span}\{\ket{\widehat\psi^1(u)},\ket{\widehat\psi^2(u)}\}.$$
For $v\in V_{\mathrm{odd}}\setminus \{t\}$, 
$$\ket{\psi_\star^{G'}(v)}\in\mathrm{span}\{\ket{\widehat\psi^0(u)}\}.$$
\end{claim}
\begin{proof}
Note that the states $\ket{\widehat{\psi}^0(u)},\ket{\widehat{\psi}^1(u)},\ket{\widehat{\psi}^2(u)}$ form an orthonormal basis for the space $\mathrm{span}\{\ket{u,v_1},\ket{u,v_2},\ket{u,v_3}\}$ -- it is the Fourier basis. Thus, the first part of the statement is simply proven by observing that for each $j\in\{1,2,3\}$, $\ket{\psi_\star^j(u)}\in\mathrm{span}\{\ket{u,v_1},\ket{u,v_2},\ket{u,v_3}\}$ and $\braket{\psi_\star^j(u)}{\widehat{\psi}^0(u)}=0$; and that $\mathrm{span}\{\ket{\psi_\star^j(u)}\}_{j=1}^3$ has dimension greater than 1.
The second statement follows easily from \eq{welded-stars-B}.
\end{proof}

\begin{lemma}\label{lem:welded-U}
The unitary $U_{\cal AB}=(2\Pi_{\cal A}-I)(2\Pi_{\cal B}-I)$ can be implemented in $O(1)$ queries to $O_G$, and $O(n)$ elementary operations.
\end{lemma}
\begin{proof}
Let 
$$H'=\mathrm{span}\{\ket{j}\ket{u,v}:j\in\{0,1,2\},u\in V(G),v\in \Gamma(u)\cup\{\bot\}\},$$
so in particular $\ket{0}\otimes H\subset H'$ (where $H$ is as in \eq{welded-H}).

We first describe how to implement $2\Pi_{\cal A}-I$. 
We describe a unitary $U_\star$ on $H'$, and in particular, its behaviour on states of the form $\ket{j}\ket{u,\bot}$, where $j=0$ whenever $u\in V_{\mathrm{odd}}\cup\{s\}$, and $j\in\{1,2\}$ whenever $u\in V_{\mathrm{even}}\setminus \{s\}$. We begin by querying the neighbours of $u$ in an ancillary register, $Q$, initialised to $\ket{0}$ using $O_G$:
$$\ket{j}\ket{u,\bot}\ket{0}_Q\mapsto \ket{j}\ket{u,\bot}\ket{v_1,v_2,v_3}_Q$$
where if $u\in \{s,t\}$, $v_1<v_2$ and $v_3=\bot$ (which we can interpret as $v_0$), and otherwise, since we assume $u\in V(G)$, $v_1<v_2<v_3$ are the neighbours of $u$. 
We initialise an ancilla qubit $A$, and compute a trit $\ket{a}_A$ for $a\in\{0,1,2\}$ as follows, to determine what happens next. If $v_3\neq \bot$, then $a=0$. Else if $u=0^{2n}=s$, we let $a=1$. Else if $v_3=\bot$ but $u\neq 0^{2n}$, so $u=t$, we let $a=2$. 

Controlled on $\ket{0}_A$, apply a Fourier transfrom $F_3$ to $\ket{j}$ to get $\ket{\hat{j}}=(\ket{1}+\omega_3^j\ket{2}+\omega_3^{2j}\ket{3})/\sqrt{3}$. Then, still conditioned on $\ket{0}_A$, swap the first and third registers, so now the first register contains $\ket{\bot}$, and then perform $\ket{\bot}\mapsto \ket{0}$ on the first register to get:
$\ket{0}\ket{u}\ket{\hat{j}}\ket{0}\ket{v_1,v_2,v_3}_Q\ket{0}_A.$
Then, conditioned on the value in the $\ket{\hat{j}}$ register, we can copy over the first, second or third value in the $\ket{v_1,v_2,v_3}$ register to get:
$$\frac{1}{\sqrt{3}}\ket{0}\ket{u}\left(\ket{1}\ket{v_1}+\omega_3^j\ket{2}\ket{v_2}+\omega_3^{2j}\ket{3}\ket{v_3}\right)\ket{v_1,v_2,v_3}_Q\ket{0}_A.$$
This requires $O(n)$ basic operations. 
We can uncompute the value $\ket{i}$ in $\ket{i}\ket{v_i}$ by referring to the last register to learn $v_i$'s position, and then we are left with:
$\ket{0}\ket{\widehat{\psi}^j(u)}\ket{v_1,v_2,v_3}_Q\ket{0}_A.$

Next, we control on $\ket{1}_A$, meaning $u=s$. In that case, we assume that $j=0$. 
Using $v_1$ and $v_2$ in the last register, we can map $\ket{\bot}$ to a state proportional to 
$\sqrt{\w_0}\ket{v_0}+\frac{1}{2}\ket{v_1}+\frac{1}{2}\ket{v_2}$
to get
$$\ket{0}\frac{\ket{{\psi}^{G'}_\star(s)}}{\norm{\ket{{\psi}^{G'}_\star(s)}}^2}\ket{v_1,v_2,v_3}_Q\ket{1}_A.$$

Lastly, we control on $\ket{2}_A$, meaning $u=t$. We can compute $g(t)$ in a separate register, and using $g(t)$, $v_1$, and $v_2$, map $\ket{\bot}$ to a state proportional to:
$\delta_{g(t),1}\sqrt{\w_{\sf M}}\ket{v_0}-\frac{1}{2}\ket{v_1}-\frac{1}{2}\ket{v_2}$
to get
$$\ket{0}\frac{\ket{{\psi}^{G'}_\star(t)}}{\norm{\ket{{\psi}^{G'}_\star(t)}}^2}\ket{v_1,v_2,v_3}_Q\ket{2}_A.$$

We can uncompute the ancilla $A$, since the registers containing $u$, and $v_1,v_2,v_3$ haven't changed.
Since the register containing $u$ has not changed, we can uncompute the register $\ket{v_1,v_2,v_3}_Q$ using another call to $O_G$. Then, removing ancillae, we have performed a map, $U_\star$ that acts, for $j=0$ when $u\in V_{\mathrm{odd}}\cup\{s\}$ and $j\in\{1,2\}$ when $u\in V_{\mathrm{even}}\setminus\{s\}$, as:
$\ket{j}\ket{u,\bot}\mapsto \ket{0}\ket{\widehat{\psi}^j(u)},$
where, using \clm{welded-F-basis}, for all $u\in V_{\mathrm{odd}}\cup\{s\}$,
$$\mathrm{span}\{\ket{\widehat{\psi}^0(u)}\}=\mathrm{span}\{\ket{\psi_\star^{G'}(u)}\}=\mathrm{span}\{\Psi_\star(u)\}$$
and for all $u\in V_{\mathrm{even}}\setminus \{s\}$:
$$\mathrm{span}\{\ket{\widehat{\psi}^1(u)},\ket{\widehat{\psi}^2(u)}\}=\mathrm{span}\{\ket{\psi_\star^1(u)},\ket{\psi_\star^2(u)},\ket{\psi_\star^3(u)}\}=\mathrm{span}\{\Psi_\star(u)\}.$$
Thus, $U_\star$ maps the subspace
$${\cal L}:=\mathrm{span}\{\ket{0,u,\bot}:u\in V_{\mathrm{odd}}\cup\{s\}\}\cup\{\ket{1,u,\bot},\ket{2,u,\bot}:u\in V_{\mathrm{even}}\setminus\{s\}\}$$
of $H'$ to $\ket{0}\otimes \mathrm{span}\{\Psi^{\cal A}\}\cong {\cal A}$, and thus 
$2\Pi_{\cal A}-I = U_\star \left(2\Pi_{\cal L}  -I\right) U_\star^\dagger.$
We describe how to implement $2\Pi_{\cal L}-I$. Initialise ancillary flag qubits $\ket{0}_{F_1}\ket{0}_{F_2}\ket{0}_{F_3}$. For a computational basis state $\ket{j}\ket{u,v}$ of $H'$, by assumption (which is removed at the end of this section) we can efficiently check whether $u$ is in $V_{\mathrm{odd}}$ or $V_{\mathrm{even}}$, and we can check whether $u=s=0^{2n}$ in $O(n)$ cost. If $u\in V_{\mathrm{odd}}\cup\{s\}$, we check if the first register is 0, and if not, flip $F_1$ to get $\ket{1}_{F_1}$. If $u\in V_{\mathrm{even}}\setminus \{s\}$, we check if the first register is 1 or 2, and if not, flip $F_2$ to get $\ket{1}_{F_2}$. If the last register is not $\bot$, flip $F_3$ to get $\ket{1}_{F_3}$. Reflect if any flag is set, and then uncompute all flags. This can all be done in $O(n)$ basic operations.

Next, we describe how to implement $2\Pi_{\cal B}-I$. We describe a unitary $U_{\cal S}$ on $H'$, and in particular, its behaviour on states of the form $\ket{1}\ket{u,v}$ for $\{u,v\}\in E(G)$ with $u<v$. First, apply a Hadamard gate to the first register, and then, controlled on its value, swap the second two registers to get:
$$\left(\ket{0}\ket{u,v} - \ket{1}\ket{v,u}\right)/\sqrt{2}.$$
We can uncompute the first register by adding in a bit indicating if the last two registers are in sorted order, to get:
$$\ket{0}\frac{1}{\sqrt{2}}\left(\ket{u,v} - \ket{v,u}\right)\in\left\{\begin{array}{ll}
\mathrm{span}\{\ket{0}\ket{\psi_{u,v}}\} & \mbox{if }(u,v)\in\overrightarrow{E}(G)\\
\mathrm{span}\{\ket{0}\ket{\psi_{v,u}}\} & \mbox{if }(v,u)\in\overrightarrow{E}(G).
\end{array}\right.$$
Thus, $U_{\cal S}$ maps
$${\cal L}':=\mathrm{span}\{\ket{1}\ket{u,v}:\{u,v\}\in E(G), u<v\}$$
to $\mathrm{span}\{\ket{0}\ket{\psi_{u,v}}:(u,v)\in\overrightarrow{E}(G)\}\cong {\cal B}$, and so 
$2\Pi_{\cal B}-I = U_{\cal S} \left(2\Pi_{{\cal L}'}  -I\right) U_{\cal S}^\dagger.$
To implement $\left(2\Pi_{{\cal L}'}  -I\right)$, it is enough to check that the first register is 1, and $u$ and $v$ are in sorted order (we know $\{u,v\}\in E(G)$ by the structure of $H'$). This can be done in $O(n)$ basic operations.
\end{proof}

\paragraph{Negative Analysis:} For the negative analysis, it would be sufficient to upper bound the total weight of $G$ and appeal to \thm{full-framework}, but we will instead explicitly construct a negative witness (see \defin{neg-witness}) in order to appeal to \thm{lin-alg-fwk}. 
That is, we show explicitly how to express $\ket{\psi_0}=\ket{s,v_0}$ as the sum of something in ${\cal A}$ and something in ${\cal B}$, when $t$ is not marked. We let:
\begin{equation}
\begin{split}
\ket{w_{\cal A}} &:= \frac{1}{\sqrt{\w_0}}\sum_{u\in V(G)}\ket{\psi_\star^{G'}(u)}
\mbox{ and }\ket{w_{\cal B}} := -\frac{1}{\sqrt{\w_0}}\sum_{(u,v)\in \overrightarrow{E}(G)}\sqrt{\w_{u,v}}(-1)^{\Delta_{u,v}}\ket{\psi_{u,v}}.
\end{split}\label{eq:welded-neg-witness}
\end{equation}
Then we prove the following.
\begin{lemma}\label{lem:welded-neg-analysis}
Suppose $M=\emptyset$. Then $\ket{w_{\cal A}},\ket{w_{\cal B}}$ form a 0-negative witness with 
$
\norm{\ket{w_{\cal A}}}^2=O(n/\w_0).
$
\end{lemma}
\begin{proof}
When $M=\emptyset$ (that is, $t\not\in M$), the graph $G'$ is simply $G$ with an additional vertex $v_0$ connected to $s$ by an edge from $s$ to $v_0$ of weight $\w_0$. 
Let $\Gamma_G(u)$ denote the neighbourhood of $u$ in $G$, and $\Gamma_{G'}(u)$ the neighbourhood of $u$ in $G'$, so, assuming $M\neq \emptyset$, for all $u\in V(G)\setminus \{s\}$, $\Gamma_G(u)=\Gamma_{G'}(u)$, and $\Gamma_{G'}(s)=\Gamma_G(s)\cup\{v_0\}$.
Thus, referring to \eq{welded-neg-witness} and \eq{welded-stars}, we have: 
\begin{align*}
\sqrt{\w_0}\ket{w_{\cal A}} &= \sum_{u\in V(G)}\sum_{v\in\Gamma_{G'}(u)}\sqrt{\w_{u,v}}(-1)^{\Delta_{u,v}}\ket{u,v}\\
&=\sum_{u\in V(G)}\sum_{v\in\Gamma_G(u)}\sqrt{\w_{u,v}}(-1)^{\Delta_{u,v}}\ket{u,v}+\sqrt{\w_0}\ket{s,v_0},
\end{align*}
and referring to \eq{welded-transition-states}, we have:
\begin{align*}
\sqrt{\w_0}\ket{w_{\cal B}} &= -\sum_{(u,v)\in \overrightarrow{E}(G)}\sqrt{\w_{u,v}}(-1)^{\Delta_{u,v}}(\ket{u,v}-\ket{v,u})\\
&= -\sum_{u\in V(G)}\sum_{v\in\Gamma_G^+(u)}\sqrt{\w_{u,v}}(-1)^{\Delta_{u,v}}\ket{u,v}-\sum_{v\in V(G)}\sum_{u\in \Gamma_G^-(v)}\sqrt{\w_{v,u}}(-1)^{\Delta_{v,u}}\ket{v,u}\\
&=-\sum_{u\in V(G)}\sum_{v\in\Gamma_G(u)}\sqrt{\w_{u,v}}(-1)^{\Delta_{u,v}}\ket{u,v},
\end{align*}
where we have used the fact that $\w_{u,v}=\w_{v,u}$ and $(-1)^{\Delta_{u,v}}=-(-1)^{\Delta_{v,u}}$.
 Thus, we see that:
\begin{align*}
\sqrt{\w_0}\left(\ket{w_{\cal A}}+\ket{w_{\cal B}}\right) &= \sqrt{\w_0}\ket{s,v_0}=\sqrt{\w_0}\ket{\psi_0}. 
\end{align*}
It is simple to check that $\ket{w_{\cal A}}\in \mathrm{span}\{\Psi^{\cal A}\}$ and $\ket{w_{\cal B}}\in \mathrm{span}\{\Psi^{\cal B}\}$ (see \eq{welded-Psi}), so we see that these states form a 0-negative witness. 

We can analyse the complexity of this witness by computing an upper bound on $\norm{\ket{w_{\cal A}}}^2$:
\begin{align*}
\norm{\ket{w_{\cal A}}}^2 &= \frac{2}{\w_0}\sum_{e\in E(G)}\w_e = \frac{2}{\w_0}\sum_{k=0}^{2n+1}|E_k|\w_k\\
&=\frac{1}{\w_0}\sum_{k=0}^{n}2^{k}\frac{1}{2^{2\lceil k/2\rceil}}+\frac{2}{w_0}\sum_{k=n+1}^{2n+1}2^{2n+1-k+1}\frac{1}{2^{2n+4-2\lceil k/2\rceil}}=O(n/\w_0)\label{eq:welded-trees-neg}
\end{align*}
using the fact that edges in $E_k$ have weight $\w_k$ defined in \eq{welded-w_k}, and $|E_k|$ in \eq{welded-E_k}.
\end{proof}

\paragraph{Positive Analysis:} In the case when $t$ is marked, so $M=\{t\}\neq\emptyset$, we exhibit a positive witness (see \defin{pos-witness}) $\ket{w}$ that is orthogonal to all states in $\Psi^{\cal A}\cup \Psi^{\cal B}$, and that has non-zero overlap with $\ket{\psi_0}=\ket{s,v_0}$. If $\theta$ is any $st$-flow on $G$ (see \defin{flow}), as long as $M=\{t\}$, so there is an edge from $t$ to $v_0$, we can extend $\theta$ to a circulation on $G'$ by sending the unit flow coming into $t$ out to $v_0$, and then back into $s$. That is, define $\theta(t,v_0)=1$, and $\theta(s,v_0)=-1$. Then if we define 
\begin{equation}
\ket{w}=\frac{\theta(s,v_0)}{\sqrt{\w_0}}\ket{s,v_0}+\sum_{(u,v)\in\overrightarrow{E}(G)}\frac{\theta(u,v)}{\sqrt{\w_{u,v}}}(\ket{u,v}+\ket{v,u})+\frac{\theta(t,v_0)}{\sqrt{\w_{\sf M}}}\ket{t,v_0}
\label{eq:welded-pos-witness}
\end{equation}
it turns out that this will always be orthogonal to all star states $\ket{\psi_\star^{G'}(u)}$, as well as all transition states $\ket{\psi_{u,v}}$. However, there are additional states $\ket{\psi^j_\star(u)}\in\Psi^{\cal A}\cup \Psi^{\cal B}$, and in order to be orthogonal to all of these, the flow must satisfy additional constraints. We will show that all these constraints are satisfied by the natural choice of flow that, for each vertex, comes in from the parent and then sends half to each child. That is, letting $E_k$ for $k\in \{1,\dots,2n+1\}$ be as in \eq{welded-E_k}, and $E_0=\{(v_0,s)\}$, define:
\begin{equation}
\forall k\in \{1,\dots,2n+1\}, (u,v)\in E_k,\;\; \theta(u,v) := \frac{1}{|E_k|}=2^{-k}. \label{eq:welded-flow}
\end{equation}
Then we first prove the following:
\begin{claim}\label{clm:welded-pos}
Let $u\in V_{\mathrm{even}}\setminus\{s\}$, and let $\ket{w_u} = (\ket{u}\bra{u}\otimes I)\ket{w}.$
Then $\ket{w_u}\propto \ket{\widehat{\psi}^0(u)}$.
\end{claim}
\begin{proof}
Since $u\not\in \{s,t\}$, we have:
\begin{align*}
\ket{w_u} = \sum_{v\in\Gamma^+(u)}\frac{\theta(u,v)}{\sqrt{\w_{u,v}}}\ket{u,v}+\sum_{u'\in\Gamma^-(u)}\frac{\theta(u',u)}{\sqrt{\w_{u,u'}}}\ket{u,u'}
= \sum_{v\in\Gamma(u)}(-1)^{\Delta_{u,v}}\frac{\theta(u,v)}{\sqrt{\w_{u,v}}}\ket{u,v}.
\end{align*}
using $\theta(u,v)=-\theta(v,u)$, $\w_{u,v}=\w_{v,u}$, and $\Delta_{u,v}=0$ if $v\in\Gamma^+(u)$, and $1$ otherwise.
Recall that $u$ has three neighbours: a parent $p(u)$ and two children $c_1(u)$ and $c_2(u)$.
Since $u\in V_{2\ell}$ for some $\ell$, the edges adjacent to $u$ are (up to direction) in $E_{2\ell}$ and $E_{2\ell+1}$. If $\ell$ is even, $2\ell=0\;\mathrm{mod}\; 4$ and $2\ell+1=1\;\mathrm{mod}\; 4$, so by \eq{welded-E_k}, $\Delta_{p(u),u}=\Delta_{u,c_1(u)}=\Delta_{u,c_2(u)}=0$ if $\ell\in\{1,\dots,n/2\}$ (i.e.~$u$ is in the left tree, so its parent is to its left) and $=1$ otherwise. If $\ell$ is odd, $2\ell=2\;\mathrm{mod}\; 4$ and $2\ell+1=3\;\mathrm{mod}\; 4$, so $\Delta_{p(u),u}=\Delta_{u,c_1(u)}=\Delta_{u,c_2(u)}=1$ if $\ell\in\{1,\dots,n/2\}$ and $=0$ otherwise. Thus, since $(-1)^{\Delta_{u,p(u)}}=-(-1)^{\Delta_{p(u),u}}$, we always have:
\begin{align*}
\ket{w_u} &= \pm\left(-\frac{\theta(u,p(u))}{\sqrt{\w_{u,p(u)}}}\ket{u,p(u)}+\frac{\theta(u,c_1(u))}{\sqrt{\w_{u,c_1(u)}}}\ket{u,c_1(u)}+\frac{\theta(u,c_2(u))}{\sqrt{w_{u,c_2(u)}}}\ket{u,c_2(u)}\right).
\end{align*}
Suppose $\ell\in\{1,\dots,n/2\}$, so $u$ is in the left tree. Then $(p(u),u)\in E_{2\ell}$, so we have
$$\theta(u,p(u))=-\theta(p(u),u)=-\frac{1}{|E_{2\ell}|}=-2^{-2\ell}
\mbox{ and }
\sqrt{\w_{u,p(u)}}=\sqrt{\w_{2\ell}}=2^{-\lceil 2\ell/2\rceil }=2^{-\ell}$$
by \eq{welded-w_k}, and for $i\in\{1,2\}$, $(u,c_i(u))\in E_{2\ell+1}$, so we have
$$\theta(u,c_i(u))=\frac{1}{|E_{2\ell+1}|}=2^{-(2\ell+1)}
\mbox{ and }
\sqrt{\w_{u,c_i(u)}}=\sqrt{w_{2\ell+1}} = 2^{-\lceil (2\ell+1)/2\rceil }=2^{-(\ell+1)}
$$
also by \eq{welded-w_k}. Thus:
\begin{align*}
\ket{w_u} &= \pm\left( -\frac{-2^{-2\ell}}{2^{-\ell}}\ket{u,p(u)}+\frac{2^{-(2\ell+1)}}{2^{-(\ell+1)}}\ket{u,c_1(u)}+\frac{2^{-(2\ell+1)}}{2^{-(\ell+1)}}\ket{u,c_2(u)}\right)\\
&=\pm 2^{-\ell}\left(\ket{u,p(u)}+\ket{u,c_1(u)}+\ket{u,c_2(u)}\right).
\end{align*}

\noindent On the other hand, if $\ell\in \{n/2+1,\dots, n\}$, so that $u$ is in the right tree, we have $(u,p(u))\in E_{2\ell+1}$, so:
$$\theta(u,p(u)) = \frac{1}{|E_{2\ell+1}|} = 2^{-(2n+2-2\ell-1)}
\mbox{ and }
\sqrt{\w_{u,p(u)}}=\sqrt{\w_{2\ell+1}}=2^{-(n+2-\lceil (2\ell+1)/2\rceil)} = 2^{-(n+1-\ell)},
$$
and for $i\in\{1,2\}$, $(c_i(u),u)\in E_{2\ell}$, so:
$$\theta(u,c_i(u)) = -\theta(c_i(u),u)=-\frac{1}{|E_{2\ell}|} = -2^{-(2n+2-2\ell)}
\mbox{ and }
\sqrt{\w_{u,c_i(u)}}=\sqrt{\w_{2\ell}}=2^{-(n+2-\lceil \frac{2\ell}{2}\rceil)} = 2^{-(n+2-\ell)}.
$$
Thus
\begin{align*}
\ket{w_u} &= \pm\left(-\frac{2^{-(2n+1-2\ell)}}{2^{-(n+1-\ell)}}\ket{u,p(u)}+\frac{-2^{-(2n+2-2\ell)}}{2^{-(n+2-\ell)}}\ket{u,c_1(u)}+\frac{-2^{-(2n+2-2\ell)}}{2^{-(n+2-\ell)}}\ket{u,c_2(u)}\right)\\
&=\mp 2^{n-\ell}\left(\ket{u,p(u)}+\ket{u,c_1(u)}+\ket{u,c_2(u)}\right).
\end{align*}
Thus, letting $\{v_1,v_2,v_3\}=\{p(u),c_1(u),c_2(u)\}$ with $v_1<v_2<v_3$, for any $\ell\in \{1,\dots,n\}$, if $u\in V_{2\ell}$, we have:
$\ket{w_u}\propto \ket{u,v_1}+\ket{u,v_2}+\ket{u,v_3}.$
\end{proof}

\noindent Then we have the following.
\begin{lemma}\label{lem:welded-pos-analysis}
Let $\w_{\sf M}=\w_0$.
Suppose $M=\{t\}$, and let $\ket{w}$ be as defined in \eq{welded-pos-witness} with respect to the flow defined in \eq{welded-flow}. 
Then $\ket{w}$ is a 0-positive witness (see \defin{pos-witness}) with: 
\begin{equation*}
\frac{\norm{\ket{w}}^2}{|\braket{w}{\psi_0}|^2} = O(\w_0n).
\end{equation*}
\end{lemma}
\begin{proof}
To show that $\ket{w}$ is a 0-positive witness, we must show that it is orthogonal to all states in $\Psi^{\cal A}\cup\Psi^{\cal B}$. For $(u,v)\in\overrightarrow{E}(G)$, it is clear from the definition of $\ket{w}$, and the definition of $\ket{\psi_{u,v}}=\ket{u,v}-\ket{v,u}$ (see \eq{welded-transition-states}) that $\braket{w}{\psi_{u,v}}=0$. 

We next check that $\braket{w}{\psi_\star^{G'}(u)}=0$ for all $u\in V(G)$, which follows from the fact that $\theta$ is a circulation on $G'$. First, suppose $u\in V(G)\setminus\{s,t\}$:
\begin{align*}
\braket{w}{\psi_\star^{G'}(u)}&=\sum_{(u',v')\in \overrightarrow{E}(G)}\frac{\theta(u',v')}{\sqrt{\w_{u',v'}}}(\bra{u',v'}+\bra{v',u'})\sum_{v\in\Gamma(u)}\sqrt{\w_{u,v}}(-1)^{\Delta_{u,v}}\ket{u,v} & \mbox{see }\eq{welded-stars}\\
&= \sum_{v\in \Gamma^+(u)}\theta(u,v)(-1)^{\Delta_{u,v}}\braket{u,v}{u,v}+\sum_{v\in\Gamma^-(u)}\theta(v,u)(-1)^{\Delta_{u,v}}\braket{u,v}{u,v}\\
&=\sum_{v\in\Gamma^+(u)}\theta(u,v) +\sum_{v\in\Gamma^-(u)}(-\theta(u,v))(-1) = \sum_{v\in \Gamma(u)}\theta(u,v),
\end{align*}
where we used the fact that $(-1)^{\Delta_{u,v}}=1$ when $v\in\Gamma^+(u)$ and $(-1)$ if $v\in\Gamma^-(u)$, and the fact that $\theta(v,u)=-\theta(u,v)$. This is 0 whenever $\theta$ is a circulation (see \defin{flow}), so we now simply check that $\theta$, as defined, is a circulation (at least on vertices other than $s$ and $t$). Suppose $u\in V_k$ for some $k\in\{1,\dots,n\}$. Then $u$ has three neighbours: a parent $p(u)\in V_{k-1}$, and two children $c_1(u),c_2(u)\in V_{k+1}$. We have 
$$\theta(u,p(u))= -\theta(p(u),u) = -\frac{1}{|E_k|},
\mbox{ and }
\theta(u,c_1(u))=\theta(u,c_2(u)) = \frac{1}{|E_{k+1}|} = \frac{1}{2|E_k|},$$
and thus
$$\theta(u,p(u))+\theta(u,c_1(u))+\theta(u,c_2(u)) = 0.$$
The case for $k\in\{n+1,\dots,2n\}$ is nearly identical. We still need to check orthogonality with $\ket{\psi_\star^{G'}(s)}$ and $\ket{\psi_\star^{G'}(t)}$. Suppose $t$ has children $u_1$ and $u_2$. Then for $i\in\{1,2\}$, $(u_i,t)\in E_{2n+1}$, so since $2n+1=1\mod 4$ (we are assuming $n$ is even), we have $(u_i,t)\in \overrightarrow{E}(G)$ (see \eq{welded-EG}). Thus, referring to \eq{welded-st-stars}, 
\begin{align*}
\braket{w}{\psi_\star^{G'}(u)}&=\bra{w}\left(\sqrt{\w_{\sf M}}\ket{t,v_0} - \frac{1}{2}\ket{t,u_1} - \frac{1}{2}\ket{t,u_2} \right) \\
&=\braket{t,v_0}{t,v_0}-\sum_{(u,v)\in \overrightarrow{E}(G)}\frac{\theta(u,v)}{\sqrt{\w_{u,v}}}(\bra{u,v}+\bra{v,u})\left(\frac{1}{2}\ket{t,u_1}+\frac{1}{2}\ket{t,u_2} \right)\\
&= 1-\frac{\theta(u_1,t)}{\sqrt{\w_{u_1,t}}}\frac{1}{2}\braket{t,u_1}{t,u_1}-\frac{\theta(u_2,t)}{\sqrt{\w_{u_2,t}}}\frac{1}{2}\braket{t,u_2}{t,u_2}\\
&= 1-\frac{1/|E_{2n+1}|}{\sqrt{\w_{2n+1}}}\frac{1}{2} - \frac{1/|E_{2n+1}|}{\sqrt{\w_{2n+1}}}\frac{1}{2}
= 1 - \frac{1/2}{\sqrt{2^{-2(n+2-\lceil (2n+1)/2\rceil)}}} = 0
\end{align*}
by \eq{welded-w_k}.
This is also simply following from the fact that $\theta$ is a circulation. 
A nearly identical argument works for $\ket{\psi_\star^{G'}(s)}$. 

It Thus, only remains to show orthogonality of $\ket{w}$ with the states of $\Psi^{\cal A}$ that are not star states of $G'$. The only such states are those in \eq{welded-alt-neighbourhood} (some of which are also star states of $G'$). By \clm{welded-F-basis}, it is sufficient to show orthogonality with the states $\ket{\widehat{\psi}^j(u)}$, for $j\in\{1,2\}$ and $u\in V_{\mathrm{even}}\setminus\{s\}$. 
Then letting $v_1<v_2<v_3$ be the neighbours of $u$, and appealing to \clm{welded-pos}:
\begin{align*}
\sqrt{3}\braket{w}{\psi_\star^i(u)}
&= \bra{w_u}\left(\ket{u,v_1}+\omega_3^j\ket{u,v_2}+\omega_3^{2j}\ket{u,v_3}\right)\\
&\propto \left(\bra{u,v_1}+\bra{u,v_2}+\bra{u,v_3} \right)\left(\ket{u,v_1}+\omega_3^j\ket{u,v_2}+\omega_3^{2j}\ket{u,v_3}\right)
\propto 1+\omega_3^j+\omega_3^{2j} = 0.
\end{align*}

Since we can also immediately see that:
$|\braket{w}{\psi_0}|^2  = {1}/{\w_0}$,
$\ket{w}$ is a positive witness. To complete the analysis of its complexity, we have, using $\w_{\sf M}=\w_0$, and the fact that all edges in $E_k$ have the same weight, $\w_k$ (see \eq{welded-w_k}), and flow, $\frac{1}{|E_k|}$:
\begin{align*}
\norm{\ket{w}}^2 &= \frac{1}{\w_0}+ 2\sum_{(u,v)\in\overrightarrow{E}(G)}\frac{\theta(u,v)^2}{\w_{u,v}}+\frac{1}{\w_{\sf M}}
=\frac{2}{\w_0} + 2\sum_{k=1}^{2n+1}|E_k|\frac{1}{|E_k|^2}\frac{1}{\w_k}\\
&= \frac{2}{\w_0}+2\sum_{k=1}^n\frac{1}{2^k}\frac{1}{2^{-2\lceil k/2\rceil}} + 2\sum_{k=n+1}^{2n+1}\frac{1}{2^{2n+2-k}}\frac{1}{2^{-2(n+2-\lceil k/2\rceil)}}  
= \frac{2}{\w_0}+O(n).\qedhere
\end{align*}
\end{proof}

\begin{remark}\label{rmk:flow-duality}
The reader may wonder why the weights change by a factor of 4 every two layers, rather than by a factor of 2 every layer. If we set all the weights to 1, the positive witness size is constant, while the negative witness size is exponential. If we change weights by a factor of two at each layer, the negative witness size is constant, whereas the positive witness size is exponential. With the setting of weights that we have chosen, both witness sizes are linear in $n$ (up to scaling by $\w_0$). This setting of weights and edge directions creates a perfect duality between positive and negative witnesses. For vertices $u\in V_{\mathrm{odd}}$, we include the star state, which is proportional to $\ket{\widehat{\psi}^0(u)}$ (see \clm{welded-F-basis}) in $\Psi^{\cal A}$, so the flow through $u$ must be in $\mathrm{span}\{\ket{\widehat{\psi}^1(u)},\ket{\widehat{\psi}^2(u)}\}$. Conversely, for vertices $u\in V_{\mathrm{even}}$, we include $\mathrm{span}\{\ket{\widehat{\psi}^1(u)},\ket{\widehat{\psi}^2(u)}\}$ in $\Psi^{\cal A}$, so the flow through $u$ must be proportional to $\ket{\widehat{\psi}^0(u)}$.
\end{remark}

\paragraph{Conclusion of Proof:} We now apply \thm{lin-alg-fwk} to conclude the proof of \thm{welded}. By \lem{welded-pos-analysis}, there is some constant $c$ such that setting $\w_0={1}/(cn)$, whenever $M=\{t\}$, there exists a positive witness $\ket{w}$ with 
$$\frac{\norm{\ket{w}}^2}{|\braket{w}{\psi_0}|^2}\leq c_+:=50.$$
Then by \lem{welded-neg-analysis}, there is some 
$${\cal C}_-=O(n/\w_0)=O(n^2)$$
such that whenever $M=\emptyset$, there exists a negative witness with
$\norm{\ket{w_{\cal A}}}^2 \leq {\cal C}_-$.
Then since the initial state can be prepared in ${\sf S}_q=0$ queries and ${\sf S}=O(n)$ time, and by \lem{welded-U}, the unitary can be implemented in $O(1)$ query to $O_G$, and $O(n)$ time, the phase estimation algorithm distinguishes between the cases $M=\emptyset$ and $M=\{t\}$ in 
$$O\left(0+\sqrt{{\cal C}_-}\right)=O(n)
\;\mbox{ and }\;O\left(n+\sqrt{{\cal C}_-}n\right)=O(n^2)$$
queries and time respectively.

\paragraph{Removing the Assumption that $u\in V_{\mathrm{even}}$ can be Checked:} We do not actually require an extra assumption that the algorithm can efficiently check, for a vertex $u$, if it is in $V_{\mathrm{even}}$ or $V_{\mathrm{odd}}$. Intuitively, this is because if a walker starts at $u$, she can always keep track of the parity of the distance from $u$, by keeping track of a bit that is initially 0, and flips every time she takes a step. More precisely, we can define a graph $G_0$ as follows:
$$V(G_0) = V_{\mathrm{even}}\times\{0\} \cup V_{\mathrm{odd}}\times \{1\}$$
$$E(G_0) = \{\{(u,0),(v,1)\}: \{u,v\}\in E(G), u\in V_{\mathrm{even}}\},$$
so that a walk on $G_0$ is like a walk on $G$, except that there is a bit indicating which of the two independent sets we are in, which we flip at every step. To find the neighbours of any vertex $(u,b)$, simply query $O_G$ and append $b\oplus 1$ to each of the three returned strings. We let $(s,0)$ and $(t,1)$, which are both in $V(G_0)$, take the places of $s$ and $t$.

\section{\texorpdfstring{$k$}{k}-Distinctness}\label{sec:k-dist-full}

Fix any constant $k$. Formally, $k$-distinctness is defined as follows. Given an input $x\in [q]^n$, for some $q\in {\sf poly}(n)$, decide if there exist distinct $a_1,\dots,a_k\in [n]$ such that $x_{a_1}=\dots=x_{a_k}$, called a $k$-collision. A search version of this problem asks that the algorithm find a $k$-collision if one exists. The search and decision versions are equivalent up to log factors, so we focus on the decision version. The main result of this section is a quantum algorithm that solves $k$-distinctness in $\widetilde{O}(n^{\frac{3}{4}-\frac{1}{4}\frac{1}{2^k-1}})$ time complexity (see \thm{k-dist}) for any $k\geq 3$, which is a new result for $k>3$. As a warm-up, we describe the $k=3$ case of our algorithm in \sec{3-dist}, before giving the full algorithm in \sec{k-dist-alg}. First, we describe some assumptions on the structure of the input in \sec{k-dist-assumptions}. 

\subsection{Assumptions on the Input}\label{sec:k-dist-assumptions}
We assume that either there is no $k$-collision, or there is a unique $k$-collision, $a_1,\dots,a_k\in [n]$. This is justified by the following lemma, which follows from \cite[Section 5]{ambainis2004QWalkForElementDist}.
\begin{lemma}\label{lem:unique-to-multiple}
	Fix constants $k\geq 2$ and $\lambda\in[1/2,1)$. 
	Let ${\cal A}$ be an algorithm that decides $k$-distinctness in bounded error with complexity $\widetilde{O}(n^{\lambda})$ when there is at most one $k$-collision. Then there is an algorithm ${\cal A}'$ that decides $k$-distinctness (in the general case) in bounded error in complexity $\widetilde{O}(n^{\lambda})$.
\end{lemma}
This fact has been exploited in nearly every quantum algorithm for $k$-distinctness.
Another standard trick is to assume that $[n]$ is partitioned as:
$$[n]=A_1\cup \dots \cup A_k$$
such that the unique $k$-collision $(a_1,\dots,a_k)$, (if it exists) is in $A_1\times\dots\times A_k$. Towards fixing \textbf{Problem 1} from \sec{intro-3-dist}, we further partition each of $A_2,\dots,A_{k-1}$ as 
$$A_{\ell}=A_{\ell}^{(1)}\cup\dots\cup A_{\ell}^{(m_{\ell})}$$
for some $m_{\ell}$. We will choose these partitions as follows. Fix a $d$-wise independent permutation $\tau:[n]\rightarrow [n]$, for $d=\log^{2^{k-1}}(n)$ that is both efficiently computable, and efficiently invertible (see \defin{d-wise} and the discussion below).
For ${\ell}\in [k]$, define:
$$A_{\ell}=\{\tau(i):i\in \{(\ell-1)n/k+1,\dots,\ell n/k\}\}$$
and for $j\in [m_\ell]$, define:
\begin{equation}\label{eq:k-dist-partition}
A_{\ell}^{(j)}=\left\{\tau(i): i\in \left\{(\ell-1)n/k+(j-1)\frac{n}{km_{\ell}}+1,\dots,(\ell-1)n/k+ j\frac{n}{km_\ell}\right\}\right\}.
\end{equation}
\noindent Then we will make use of the following facts:
\begin{lemma}
	\begin{enumerate}
		\item For any $i\in [n]$, we can check to which $A_{\ell}^{(j)}$ it belongs in ${\sf polylog}(n)$ complexity.
		\item For any $\ell\in [k]$, we can generate a uniform superposition over $A_\ell$, and for any $j\in [m_\ell]$, we can generate a uniform superposition over $A_\ell^{(j)}$, in ${\sf polylog}(n)$ complexity.
		\item $\Pr[a_1\in A_1,\dots,a_k\in A_k]=\Omega(1)$. 
	\end{enumerate}
\end{lemma}
\begin{proof}
Since $d\in{\sf polylog}(n)$, we can assume (see discussion below \defin{d-wise}) that both $\tau$ and $\tau^{-1}$ can be computed in ${\sf polylog}(n)$ complexity. Then for Item 1, it is enough to compute $\tau^{-1}(i)$. 

For Item 2, we describe how to perform a superposition over $\{\tau(i):i\in \{\ell,\dots, r\}\}$ for any integers ${\ell}<r$. First generate the uniform superposition over the set $\{\ell,\dots,r\}$, and compute $\tau$ in a new register, to get (up to normalization)
$\sum_{i=\ell}^r\ket{i}\ket{\tau(i)}$.
Then uncompute the first register by computing $\tau^{-1}$ of the second register and adding it (bitwise, mod 2) into the first. 

Finally, Item 3 follows from the $d$-wise independence of $\tau$, since $d>k$. 
\end{proof}
\noindent For any disjoint subsets of $[n]$, $S_1,\dots,S_{\ell}$, define:
\begin{equation}
	{\cal K}(S_1,\dots,S_{\ell})=\{(i_1,\dots,i_{\ell})\in S_1\times\dots\times S_{\ell}: x_{i_1}=\dots=x_{i_{\ell}}\}.\label{eq:collision-set}
\end{equation}
Then without loss of generality, we can assume that for each $A_j^{(\ell)}$, ${\cal K}(A_1,\dots,A_{j-1},A_j^{(\ell)})=\Theta(|A_j^{(\ell)}|)$, because we can simply pad the input with $\Theta(n)$ extra $(k-1)$-collisions, evenly spread across the blocks.

\subsection{Warm-up: 3-Distinctness Algorithm}\label{sec:3-dist}

In this section, we prove the following upper bound on the time complexity of $3$-distinctness.
\begin{theorem}\label{thm:3-dist}
There is a quantum algorithm that decides 3-distinctness with bounded error in $\widetilde{O}(n^{5/7})$ complexity.
\end{theorem}
\noindent This upper bound is not new, having been proven in~\cite{belovs2013TimeEfficientQW3Distintness}, but its proof in our new framework is a useful warm-up for \sec{k-dist-alg}, where we generalise the algorithm to all constants $k>3$. Throughout this section, $\widetilde{O}$ will surpress polylogarithmic factors in $n$. 

Our algorithm will roughly follow the one described in \sec{intro-3-dist}, but with the modifications, also briefly mentioned in \sec{intro-3-dist}, needed to circumvent the problems with the approach, for which we need our new Mutltidimensional Quantum Walk Framework, \thm{full-framework}. We start by setting up these modifications, before formally defining the graph that will be the basis for our quantum walk algorithm, and then performing the necessary analysis to apply \thm{full-framework}.

\vskip10pt

Recall from \sec{intro-3-dist} that the basic idea of our quantum walk algorithm is to walk on sets $R = (R_1,R_2)$ where $R_1\subset A_1$ and $R_2\subset A_2$. 

\paragraph{Towards Fixing Problem 1:} The first problem identified in \sec{intro-3-dist} is that $|R_2|$ is larger than the total time we would like our algorithm to spend, meaning we do not want to spend $|R_2|$ steps sampling and writing down the set $R_2$. To this end, we have partitioned $A_2$ into equal sized blocks:
$$A_2=A_2^{(1)}\cup\dots\cup A_2^{(m_2)},$$
(see \sec{k-dist-assumptions} for details of how this partition is chosen). We redefine $R_2$ as follows: whenever we want to choose a subset of $A_2$, we do so by selecting $R_2\subset [m_2]$, which encodes the subset of $A_2$:
$$\overline{R}_2:=\bigcup_{j\in R_2} A_2^{(j)}.$$
We choose $m_2$ so that $|A_2^{(j)}|=\frac{n}{3m_2}$ is large enough so that for a random set $R_1$ of size $r_1$, the expected size of ${\cal K}({R}_1,A_2^{(j)})$ is constant, so we set $m_2=\Theta(r_1)$. Finally, we choose $t_2=|R_2|$ so that $|\overline{R}_2|=t_2\frac{n}{3m_2}$ is the desired size of $R_2$ (denoted $r_2$ in \cite{belovs2012kDist}) and for consistency also define $t_1 = r_1$. We will find that the optimal parameter settings are $t_1=n^{5/7}$ and $t_2=n^{4/7}$ (so $m_2=\Theta(n^{5/7})$). 

\paragraph{Towards Fixing Problem 2:} In order to solve the second problem discussed in \sec{intro-3-dist}, following a similar construction in \cite{belovs2012kDist}, each of $R_1$ and $R_2$ will be a \emph{tuple} of disjoint sets, as follows. We have $R_1=(R_1(\{1\}),R_1(\{2\}),R_1(\{1,2\}))$ where $R_1(\{1\})$, $R_1(\{2\})$, and $R_1(\{1,2\})$ are disjoint subsets of $A_1$ of size $t_1$; and $R_2=(R_2(1),R_2(2))$, where $R_2(1)$ and $R_2(2)$ are disjoint subsets of $[m_2]$ of size $t_2$ (note that this alters $\abs{R_1}$ and $\abs{R_2}$ by a constant factor), meaning for $s\in\{1,2\}$, 
$$\overline{R}_2(s):=\bigcup_{j\in R_2(s)}A_2^{(j)}$$
are disjoint subsets of $A_2$ of size $t_2\frac{n}{3m_2}$. We also use $R_1$ and $R_2$ to denote the union of sets in the tuple, so for example, $j\in \overline{R}_2$ means $j\in \overline{R}_2(1)\cup\overline{R}_2(2)$. 
For a vertex labelled by $R = (R_1,R_2)$, we maintain \emph{data} with the following components. We query everything in $R_1$, so for $S\in 2^{\{1,2\}}\setminus\emptyset$, we define:
\begin{equation*}\label{eq:3-dist-D1}
\begin{split}
D_1(R_1(S)) &:= \{(i_1,x_{i_1}): i_1\in R_1(S)\}\\
D_1(R) &:= (D_1(R_1(\{1\})),D_1(R_1(\{2\})),D_1(R_1(\{1,2\})))
\end{split}
\end{equation*}
and for $s\in \{1,2\}$ define
\begin{equation}\label{eq:3-dist-D2}
\begin{split}
D_2(R_2(s)|R_1) &:= \bigcup_{S\subseteq \{1,2\}:s\in S} \{(i_1,i_2,x_{i_1}): i_2\in \overline{R}_2(s), i_1\in R_1(S),x_{i_1}=x_{i_2}\}\\
D_2(R) &:= (D_2(R_2(1)|R_1),D_2(R_2(2)|R_1)).
\end{split}
\end{equation}
Finally we let
\begin{align*}
D(R) &:= (D_1(R),D_2(R)).
\end{align*}
So to summarise, we query everything in $R_1$, but we only query those things in $\overline{R}_2$ that have a collision in $R_1$, and even then, not in every case: if $i_2\in \overline{R}_2(s)$, we only query it if it has a collision with $R_1(\{s\})$ or $R_1(\{1,2\})$ (see \fig{set-partitions}). This partially solves \textbf{Problem 2}, because it ensures that if we choose to add a new index $i_1$ to $R_1$, we have three choices of where to add it, and either all of those choices are fine (they don't introduce a \emph{fault} in $D_2(R)$), or exactly one of them is fine. 

For a finite set ${\cal S}$, and positive integers $r$ and $\ell$, we will use the notation 
\begin{equation}
\binom{{\cal S}}{r^{(\ell)}}:=\binom{{\cal S}}{\underbrace{r,\dots,r}_{\ell\;\mathrm{times}}}\label{eq:weird-binom}
\end{equation}
to denote the set of all $\ell$-tuples of disjoint subsets of ${\cal S}$, each of size $r$. 
Finally, we define:
\begin{equation}
{\binom{\cal S}{r^{(\ell)}}^+} := \bigcup_{\ell'=1}^{\ell} \left(\binom{\cal S}{r^{(\ell'-1)},r+1,r^{(\ell-\ell')}}\right),\label{eq:weird-binom+}
\end{equation}
to be the set of all $\ell$-tuples of disjoint sets of ${\cal S}$ such that exactly one of the sets has size $r+1$, and all others have size $r$. We let $\mu(S)$ denote the smallest element of $S$.

\subsubsection{The Graph \texorpdfstring{$G$}{G}}\label{sec:3-dist-G}

We now define $G$, by defining disjoint vertex sets $V_0,V_0^+,V_1,V_2,V_3$ whose union will make up $V(G)$, as well as the edges between adjacent sets. 

\paragraph{$V_0$:} We first define 
\begin{equation}
V_0:=\left\{v^0_{R_1,R_2}=(0,R_1,R_2,D(R_1,R_2)):(R_1,R_2)\in\binom{A_1}{t_1^{(3)}}\times\binom{[m_2]}{t_2^{(2)}}\right\}\label{eq:3-dist-V1}
\end{equation}
on which the initial distribution will be uniform: $\sigma(v^0_{R_1,R_2})=\frac{1}{|V_0|}$. We implicitly store all sets including $R_1$, $R_2$ and $D(R_1,R_2)$ in a data structure with the properties described in \sec{data}. This will only be important when we analyse the time complexity of the setup and transition subroutines. 

\paragraph{$V_0^+$ and $E_0^+\subset V_0\times V_0^+$:} Next, each vertex in $V_0^+$ will be labeled by a vertex in $V_0$, along with an index $i_1\not\in R_1$ that we have decided to add to one of $R_1(\{1\})$, $R_1(\{1,2\})$ or $R_1(\{2\})$. We have not yet decided to which of the three sets it will be added, nor added it. 
\begin{align}
V_0^+ &:=\left\{v^0_{R_1,R_2,i_1}:=((0,+),R_1,R_2,D(R_1,R_2),i_1): v^0_{R_1,R_2}\in V_0,i_1\in A_1\setminus R_1\right\},\nonumber\\
\mbox{so }|V_0^+| &= |V_0| (n/3-3t_1).\label{eq:V2}
\end{align}
There is an edge between $v^0_{R}\in V_0$ and $v^0_{R,i_1}\in V_0^+$ for any $i_1\in A_1\setminus R_1$, and for any $v^0_{R,i_1}\in V_0^+$, $v^0_{R}\in V_0$ is its unique in-neighbour, so we define edge label sets (see \defin{QW-access})
$$L^+(v_{R}^0):=A_1\setminus R_1
\;\mbox{ and }\;
L^-(v_{R,i_1}^0):=\{\leftarrow\},$$
and let $f^+_{v_{R}^0}(i_1) = v_{R,i_1}^0$, and $f^-_{v_{R,i_1}^0}(\leftarrow) = v_{R}^0$. Here we have added the superscript $+$ (respectively $-$) to denote the restriction of $f_{v_{R}^0}$ to $L^+(v_{R}^0)$ (respectively $L^-(v_{R,i_1}^0)$). By writing $f^+_{v_{R}^0}(i_1)$ and $f^-_{v_{R,i_1}^0}(\leftarrow)$, we emphasise that the index $i_1$ is an element of $L^+(v_{R}^0)$ and the index $\leftarrow$ is an element of $L^-(v_{R,i_1}^0)$. We stick to this notation convention for the rest of the section. 

\noindent We let $E_0^+$ be the set of all edges,
$$E_0^+:=\left\{\left(v^0_{R},v^0_{R,i_1}\right):v^0_{R}\in V_0, i_1 \in A_1 \setminus R_1\right\},$$
and set $\w_e = \w_0^+=1$ for all $e \in E_0^+$. This together with \eq{V2} implies that
\begin{equation}
|E_0^+|=|V_0^+|=|V_0|(n/3-3t_1).\label{eq:E1}
\end{equation}

Note that we break the move from $V_0$ to $V_1$, where we add some $i_1$ to one of the sets $R_1(\{1\})$, $R_1(\{2\})$ or $R_1(\{1,2\})$, into two steps: First we select an index $i_1$ to add -- that's the step we have just described, from $V_0$ to $V_0^+$. Next, we choose one of the three sets and add $i_1$ there -- that's the step we are about to describe, from $V_0^+$ to $V_1$. The reason we do this in two steps is that we will use the alternative neighbourhoods trick to ensure we can efficiently implement the second step, only adding $i_1$ to a set where it won't cause a fault, despite not being able to efficiently decide which sets these are. It is useful to have this somewhat more complicated-to-implement part of the walk isolated, in vertices of constant degree, as the vertices of $V_0^+$ will be. Note that in defining the graph, as we are currently doing, this complication does not appear, except that the reader may notice that it looks difficult to implement a step of the walk from $V_0^+$ -- it is indeed more complicated, requiring the use of alternative neighbourhoods later. Let us continue with the description of the graph.

\paragraph{$V_1$ and $E_1\subset V_0^+\times V_1$:} Continuing, vertices in $V_1$ represent having added an additional index to $R_1$, so we define:
\begin{align}
V_1(S) &:=\Bigg\{v^1_{R_1,R_2}=(1,R_1,R_2,D(R_1,R_2)): (R_1,R_2)\in\binom{A_1}{t_1^{(3)}}^+\times \binom{[m_2]}{t_2^{(2)}},|R_1(S)|=t_1+1\Bigg\},\nonumber\\
V_1 &:=\bigcup_{S\in 2^{\{1,2\}}\setminus\{\emptyset\}}V_1(S)\nonumber\\
\mbox{so }|V_1|&=3\binom{n/3}{t_1+1,t_1,t_1}\binom{m_2}{t_2,t_2}
= 3\frac{n/3-3t_1}{t_1+1}\binom{n/3}{t_1,t_1,t_1}\binom{m_2}{t_2,t_2}
= \frac{n-9t_1}{t_1+1}|V_0|.
\label{eq:V3}
\end{align}

For a vertex $v_{R,i_1}^0\in V_0^+$ we have chosen an index $i_1$ to add to $R_1$, but we have not yet decided to which part of $R_1$ it should be added. A transition to a vertex in $V_1$ consists of choosing an $S \in 2^{\{1,2\}} \setminus \{\emptyset\}$ and adding $i_1$ to $R_1(S)$, so
$$L^+(v_{R,i_1}^0):=2^{\{1,2\}} \setminus \{\emptyset\},$$
and $f^+_{v_{R,i_1}^0}(S) = v^1_{R'}$, where $R'$ is obtained from $R$ by inserting $i_1$ into $R_1(S)$. Note that not all of these labels represent edges with non-zero weight, as we want to ensure that adding $i_1$ to $R_1(S)$ does not introduce a \emph{fault}, meaning that adding $i_1$ to $R_1(S)$ should not require that any collision involving $i_1$ be added to $D_2(R)$. 

Viewing transitions in $E_1$ from the other direction, a vertex $v^1_{R'} \in V_1(S)$ is connected to a vertex $v_{R,i_1}^0\in V_0^+$ if we can obtain $R$ from $R'$ by removing $i_1$ from $R_1'(S)$, and if doing so does not require an update to $D_2(R')$, meaning there do not exist any $s \in S$ and $i_2\in \overline{R}'_2(s)$ such that $x_{i_1}=x_{i_2}$. So for any $v^1_{R'}\in V_1(S)$, we let
\begin{equation}
\begin{split}
	L^-(v_{R'}^1):&=\{i_1 \in R_1'(S): \nexists s \in S, i_2\in \overline{R}'_2(s)\mbox{ s.t. }x_{i_1}=x_{i_2} \} \\
	&=\{i_1\in R_1'(S):\nexists i_2\mbox{ s.t. }(i_1,i_2,x_{i_1})\in D_2(R')\},
\end{split}\label{eq:3-dist-L-minus-1}
\end{equation}
and $f^-_{v_{R'}^1}(i_1) = v^0_{R'\setminus \{i_1\},i_1}$. It is currently not clear how to define $E_1$, the set of (non-zero weight) edges between $V_0^+$ and $V_1$, because $|V_0^+|\cdot |L^+(v_{R,i_1}^0)| > |V_1|\cdot |L^-(v_{R'}^1)|$, so in particular, we cannot assign nonzero weights $\w_{u,i}$ to all $u\in V_0^+$, $i\in L^+(u)$, because that would make $E_1$ larger than we have labels $L^-$ for. We will instead assign non-zero weights $\w_{v,j}$ to those edges where $v \in V_1$ and $j\in L^-(v)$. That is, define:
$$E_1 := \left\{\left(f^-_{v_{R'}^1}(i_1),v^1_{R'}\right)= \left(v^0_{R'\setminus \{i_1\},i_1},v^1_{R'}\right): v^1_{R'} \in V_1, i_1 \in L^-(v_{R'}^1)\right\}$$
and give weight $\w_e = \w_1 = 1$ to all $e \in E_1$. This means that for $u=v_{R,i_1}^0\in V_0^+$, there are some $S\in 2^{\{1,2\}}\setminus\{\emptyset\}$ with $\w_{u,S}=0$ -- namely those with $f_v^{-1}(u)\not\in L^-(v)$ for $v=f_u(S)$. To investigate which $S$ this applies to, we introduce for all $v_{R,i_1}^0\in V_0^+$:
\begin{equation}
{\cal I}(v^{0}_{R,i_1}):=\left\{s\in \{1,2\}: \exists i_{2}\in \overline{R}_{2}(s)\mbox{ s.t. }x_{i_2}=x_{i_1}\right\},\label{eq:3-dist-cal-I}
\end{equation}
so ${\cal I}(v^{0}_{R,i_1})$ consists of those $s \in \{1,2\}$ where a fault occurs if $i_1$ is added to $R_1(S)$ such that $s\in S$. Note that since we assume the unique 3-collision has a part in $A_3$, $i_1$ can have at most one colliding element in $\overline{R}_2$, and so it cannot be in both $\overline{R}_2(1)$ and $\overline{R}_2(2)$, which are disjoint. Thus, ${\cal I}(v^{0}_{R,i_1}) \subsetneq \{1,2\}$ -- so it is $\emptyset$, $\{1\}$, or $\{2\}$ (this heavy-handed notation is overkill here, but we are warming up for $k$-distinctness, where it is necessary). We now have the following:
\begin{lemma}\label{lem:3-dist-E1}
	Let $R^{S \leftarrow i_1}$ be obtained from $R$ by inserting $i_1$ into $R_1(S)$. Then
	\begin{align*}
		E_1 &= \left\{\left(v^{0}_{R,i_1},v^{1}_{R^{S \leftarrow i_1}}\right): v_{R,i_1}^0\in V_0^+, S \in 2^{\{1,2\}\setminus {\cal I}(v_{R,i_1}^0)}\setminus \{\emptyset\} \right\}.\label{eq:3-dist-E1}
	\end{align*}
	So for all $v^0_{R,i_1}\in V_0^+$, and $S\in L^+(v^0_{R,i_1})$,
	$\displaystyle\w_{v^0_{R,i_1},S} = \left\{\begin{array}{ll}
		\w_1 =1 & \mbox{if }S\cap {\cal I}(v^0_{R,i_1}) = \emptyset\\
		0 & \mbox{else.}
	\end{array}\right.$
\end{lemma}

\begin{proof}
	Let $E_1'$ be the right-hand side of the identity in the theorem statement, so we want to show $E_1=E_1'$. 
	Fix any $v_{R,i_1}^0\in V_0^+$ and $S \in 2^{\{1,2\}\setminus {\cal I}(v_{R,i_1}^0)}\setminus \{\emptyset\}$, and let $R'=R^{S\leftarrow i_1}$. Then since $S\cap {\cal I}(v_{R,i_1})=\emptyset$, by definition of ${\cal I}(v_{R,i_1})$ there does not exist any $s \in S$ and $i_2\in \overline{R}'_2(s)$ such that $x_{i_1}=x_{i_2}$. Hence, $L^-(v_{R'}^1)$, which implies $E_1'\subseteq E_1$.
	
	For the other direction, fix any $v_{R'}^1\in V_1(S)$ and $i_1\in L^-(v_{R'}^1)$. Since $i_1\in R'_1(S)$, we have $v^0_{R' \setminus \{i_1\},i_1}\in V_0^+$ and $(R'\setminus\{i_1\})^{S\leftarrow i_1}=R'$. Since by definition of $L^-(v_{R'}^1)$ there does not exist $s \in S$ and $i_2\in \overline{R}'_2(s)$ such that $x_{i_1}=x_{i_2}$, we immediately have $S\cap {\cal I}(v_{R'\setminus\{i_1\},i_1})=\emptyset$. This implies $E_1\subseteq E_1'$.
\end{proof}

\noindent From \eq{V2} we now have:
\begin{equation}
|E_1|\leq 3|V_0^+|=3|V_0|(n/3-3t_1).\label{eq:E2}
\end{equation}

\paragraph{$V_2$ and ${E}_2\subset V_1\times V_2$:} Vertices $v_{R}^1\in V_1(S)$ represent having added an additional index $i_1$ to $R_1(S)$, so $|R_1(S)|=t_1+1$. A vertex $v_{R_1,R_2'}^2\in V_2$ is adjacent to $v_R^1$ if $R_2'$ is obtained from $R_2$ by adding $j_2\not\in R_2$ to $R_2(s)$ for some choice of $s\in\{1,2\}$.
We will not let this choice of $s$ be arbitrary though and instead, in order to simplify things in the more complicated $k$-distinctness setting, we require that 
$j_2$ be added to $R_2(\mu(S))$, where $\mu(S)$ denotes the minimum element of $S$. 
\begin{align*}
V_2(S) &:= \Bigg\{v^2_{R}=(2,R,D(R)):R\in\binom{A_1}{t_1^{(3)}}^+\times \binom{[m_2]}{t_2^{(2)}}^+ ,
\abs{R_1(S)} = t_1+1, \abs{R_2(\mu(S))} = t_2+1 \Bigg\},\\
V_2 &:=\bigcup_{S\in 2^{\{1,2\}}\setminus\{\emptyset\}}V_2(S). \label{eq:V-2-Ss}
\end{align*}
\noindent This means that
\begin{equation}
\begin{split}
|V_2|&=3\binom{n/3}{t_1+1,t_1,t_1}\binom{m_2}{t_2+1,t_2} = 3\frac{n/3-3t_1}{t_1+1}\binom{n/3}{t_1,t_1,t_1}\frac{m_2-2t_2}{t_2+1}\binom{m_2}{t_2,t_2}
 = O\left(\frac{nm_2}{t_1t_2}|V_0|\right).
\end{split}\label{eq:V4}
\end{equation}
We move from $v^1_R\in V_1$ to $v^2_{R'}\in V_2$ by selecting some $j_2\in [m_2]\setminus R_2$ to add to $R_2$; and from $v^2_{R'}$ to $v^1_R$ by selecting some $j_2$ to remove from $R_2$, so for $v^1_R\in V_1(S)$ and $v^2_{R'}\in V_2(S)$, we let
$$L^+(v^1_{R}):= [m_2]\setminus R_2 
\mbox{ and }
L^-(v^2_{R'}):=R'_2(\mu(S)).$$
The sets $L^+(v^1_R)$ and $L^-(v^1_R)$ (defined in \eq{3-dist-L-minus-1}) should be disjoint, but this does not appear to be the case. To ensure this, we implicitly append a label $\leftarrow$ to every label in $L^-(u)$ for any $u$, and $\rightarrow$ to every label in $L^+(u)$. 
We let $f^+_{v^1_{R}}(j_2) = v_{R_1,R_2^{\mu(S) \leftarrow j_2}}^2$ when $v_R^1\in V_1(S)$, and $f^-_{v^2_{R_1,R_2'}}(j_2) = v_{R_1,R'_2 \setminus \{j_2\}}^1$. Accordingly we define $E_2(S)$ to be the set of all such edges:
\begin{align*}
	E_2:&=\bigcup_{S\in 2^{\{1,2\}}\setminus\{\emptyset\}}\{(v^1_{R},v_{R_1,R_2^{\mu(S) \leftarrow j_2}}^2):v^1_{R}\in V_1(S), j_2 \in [m_2]\setminus R_2\} \\
	&=\bigcup_{S\in 2^{\{1,2\}}\setminus\{\emptyset\}}\left\{\left(v^1_{R_1,R'_2\setminus \{j_2\}},v_{R_1,R'_2}^2\right):v^2_{R_1,R'_2}\in V_2(S), j_2 \in R_2(\mu(S))\right\}.
\end{align*}
We set $\w_e = \w_2=\sqrt{n/m_2}$ for all $e \in E_2$, and observe, using \eq{V3}, that:
\begin{equation}
|E_2|=(m_2-2t_2)|V_1|=\frac{(m_2-2t_2)(n-9t_1)}{t_1+1}|V_0|. \label{eq:E3}
\end{equation}

\paragraph{The Final Stage: $V_3$ and $E_3$:} The last stage is very simple, as every vertex in $V_3$ represents having added an additional index to each of $R_1, R_2$ and chosen some $i_3 \in A_3$:
$$V_3:=\{v^3_{R_1,R_2,i_3}=(3,R_1,R_2,D(R_1,R_2),i_3):v^2_{R_1,R_2}\in V_2, i_3\in A_3\}.$$
There is an edge between $v^2_{R}\in V_2$ and $v^3_{R,i_3}\in V_3$ for any $i_3\in A_3$, and for any $v^3_{R,i_3}\in V_3$, $v^2_{R}$ is its unique (in-)neighbour, so we define
$$L^+(v_{R}^2):=A_3
\mbox{ and }
L(v_{R,i_3}^3)=L^-(v_{R,i_3}^3):=\{\leftarrow\},$$
and let $f^+_{v_{R}^2}(i_3) = v_{R,i_3}^3$, and $f^-_{v_{R,i_3}^3}(\leftarrow) = v_{R}^2$. We let $E_3$ be the set of all such edges,
$$E_3=\left\{\left(v^2_{R},v^3_{R,i_3}\right):v^2_{R}\in V_2, i_3 \in A_3\right\},$$
and set $\w_e = \w_3=1$ for all $e \in E_3$. Then using \eq{V4} we observe
\begin{equation}
|E_3|=\frac{n}{3}|V_2|= O\left(\frac{n^2m_2}{t_1t_2}|V_0|\right).\label{eq:E4}
\end{equation}

\begin{table}
\renewcommand{\arraystretch}{1.5}
\centering
\begin{tabular}{ |c|c|c|c|c|}
\hline
 $u$ & $j\in L^-(u)$ & $f^-_u(j)$ & $i\in L^+(u)$ & $f^+_u(i)$\\
\hline
\hline
$v^0_{R}\in V_0$ & $\emptyset$ & & $ i_1\in A_1\setminus R_1$ & $v_{R,i_1}^0$\\
\hline
$v^0_{R,i_1}\in V_0^+$ & $\leftarrow$ & $v_R^0$ & $S\in 2^{\{1,2\}} \setminus \{\emptyset\}$ & $v_{R^{S\leftarrow i_1}}^1$\\
\hline
$v^1_{R}\in V_1(S)$ & $i_1\in R_1(S):d_R^{\rightarrow}(i_1)=0$ & $v_{R\setminus\{i_1\},i_1}^0$ & $j_2\in [m_2]\setminus R_2$ & $v_{R^{\mu(S)\leftarrow j_2}}^2$\\
\hline
$v^2_{R}\in V_2(S)$ & $j_2\in R_2(\mu(S))$ & $v_{R\setminus\{j_2\}}^1$ &  $i_3\in A_3$ & $v_{R,i_3}^3$\\
\hline
$v^3_{R,i_3}\in V_3$ & $\leftarrow$ & $v_{R}^{2}$ & $\emptyset$ &\\
\hline
\end{tabular}
\caption{A summary of the vertex sets and the labels of edges coming into and out of each vertex. 
Foreshadowing \sec{k-dist-alg}, we here define $d_R^{\rightarrow}(i_1)$ to be 0 if and only if there is no $i_2\in \bigcup_{s\in S}R_2(s)$ such that $x_{i_1}=x_{i_2}$. Here $R^{\mu(S)\leftarrow j_2}$ is obtained from $R$ by inserting $j_2$ into $R_2(\mu(S))$, where $\mu(S)$ is the minimum element of $S$. We remark that $L^-(u)$ and $L^+(u)$ should always be disjoint. To ensure that this holds, we implicitly append a $\leftarrow$ label to all of $L^-(u)$ and a $\rightarrow$ label to all of $L^+(u)$.
}\label{tab:3-dist-index-sets}
\end{table}

\paragraph{The Graph $G$:} The full graph $G$ is defined by:
\begin{align*}
V(G)&=V_0\cup V_0^+ \cup V_1\cup V_2\cup V_3\\
\mbox{and }\overrightarrow{E}(G) &= \{(u,v):u\in V(G),i\in L^+(u), \w_{u,i}\neq 0\} = E_0^+ \cup E_1 \cup E_2 \cup E_3,
\end{align*}
where the sets $L^+(u)$ are summarised in \tabl{3-dist-index-sets}, and the condition under which $\w_{u,i}=0$ can be found in \lem{3-dist-E1}. Non-zero edge weights are summarised in \tabl{3-dist-weights}.

\paragraph{The Marked Set and Checking Cost:}  In the notation of \thm{full-framework}, we let $V_{\sf M}=V_3$, and we will define a subset $M\subseteq V_3$ as follows. 
If $(a_1,a_2,a_3)\in A_1\times A_2\times A_3$ is the unique 3-collision (see \sec{k-dist-assumptions}), we let 
\begin{equation}
M = \left\{v^3_{R_1,R_2,i_3}\in V_3: \exists S\in 2^{\{1,2\}}\setminus\{\emptyset\}, \mbox{ s.t. }a_1\in R_1(S),a_2\in \overline{R}_2(\mu(S)), a_3=i_3\right\},\label{eq:3-dist-M}
\end{equation}
and otherwise $M=\emptyset$.
Recall that $v^3_{R_1,R_2,i_3}=(3,R_1,R_2,D(R_1,R_2),i_3)$, where $D(R_1,R_2)$ includes $D_2(R)$, defined in \eq{3-dist-D2}, storing all pairs $(i_1,i_2,x_{i_1})$ such that $x_{i_1}=x_{i_2}$ and $\exists S\in 2^{\{1,2\}}$ and $s\in S$ with $i_1\in R_1(S)$ and $i_2\in R_2(s)$.
Thus, we can decide if $v^3_{R_1,R_2,i_3}\in V_3$ is marked by querying $i_3$ to obtain $x_{i_3}$ and looking it up (see \sec{data}) in $D_2(R)$ to see if we find some $(i_1,i_2,x_{i_3})$, in which case, it must be that $a_1=i_1$, $a_2=i_2$ and $a_3=i_3$. Thus, the checking cost is at most
\begin{equation}
{\sf C}=O(\log n).\label{eq:3-dist-C}
\end{equation}

\begin{table}
\renewcommand{\arraystretch}{1.5}
\centering
\begin{tabular}{|c|c|c|}
\hline
Edge set & Weights & Complexity\\
\hline
\hline
$E_0^+\subset V_0\times V_0^+$ & $\w_0^+=1$ & ${\sf T}_0^+=\widetilde{O}(1)$\\
\hline
$E_1\subset V_0^+\times V_1$ & $\w_1=1$ & ${\sf T}_1=\widetilde{O}(1)$\\
\hline
$E_2\subset V_1\times V_2$ & $\w_2=\sqrt{n/m_2}$ & ${\sf T}_2=\widetilde{O}(\sqrt{n/m_2})$\\
\hline
$E_3\subset V_2\times V_3$ & $\w_3=1$ & ${\sf T}_3=\widetilde{O}(1)$\\
\hline
\end{tabular}
\caption{A summary of the weights and complexities (see \sec{3-dist-transitions}) of each edge set.}\label{tab:3-dist-weights}
\end{table}

\subsubsection{The Star States and their Generation}\label{sec:3-dist-star-states}

We define a set of alternative neighbourhoods for $G$ (see \defin{alternative}). 
For all $u\in V(G)\setminus V_0^+$, we define $\Psi_\star(u)=\{\ket{\psi_\star^G(u)}\}$, which by \tabl{3-dist-index-sets} is equal to the following: for $u=v_{R_1,R_2}^0 \in V_0$,
\begin{equation}
	\ket{\psi_\star^G(u)} = \sum_{i_1 \in A_1 \setminus R_1} \sqrt{\w_0^+}\ket{v_{R_1,R_2}^0,i_1};\label{eq:3-star-v0}
\end{equation}
for $u=v_{R_1,R_2}^1 \in V_1(S)$,
\begin{equation}
	\ket{\psi_\star^G(u)} = -\sum_{\substack{i_1 \in R_1(S):\\ \nexists i_2,\; (i_1,i_2,x_{i_1}) \in D_2(R)}} \sqrt{\w_1}\ket{v_{R_1,R_2}^1,\leftarrow,i_1} + \sum_{j_2 \in [m_2] \setminus R_2} \sqrt{\w_2}\ket{v_{R_1,R_2}^1,\rightarrow,j_2};\label{eq:3-star-v1}
\end{equation}
for $u=v_{R_1,R_2}^2 \in V_2(S)$,\footnote{Here we explicitly include the $\rightarrow$ and $\leftarrow$ parts of each element of $L^+(u)$ and $L^-(u)$, which are normally left implict, in order to stress that the first and second sum are orthogonal.}
\begin{equation}
	\ket{\psi_\star^G(u)} = -\sum_{\substack{j_2 \in R_2(\mu(S))}} \sqrt{\w_2}\ket{v_{R_1,R_2}^2,\leftarrow,j_2} + \sum_{i_3 \in A_3} \sqrt{\w_3}\ket{v_{R_1,R_2}^2,\rightarrow,i_3};\label{eq:3-star-v2}
\end{equation}
and finally for $u=v_{R_1,R_2,i_3}^3 \in V_3$,
\begin{equation}
	\ket{\psi_\star^G(u)} = -\sqrt{\w_3}\ket{v_{R_1,R_2,i_3}^3,\leftarrow}.\label{eq:3-star-v3}
\end{equation}

From \tabl{3-dist-index-sets}, as well as the description of $\w$ from \lem{3-dist-E1}, we can see that for $u=v_{R,i_1}^0\in V_0^+$, 
$$\ket{\psi^G_\star(u)}=-\sqrt{\w_0^+}\ket{u,\leftarrow}+\sum_{S_1\subseteq\{1,2\}\setminus{\cal I}(u):S_1\neq\emptyset}\sqrt{\w_1}\ket{u,S_1}.$$
To generate this state, one would have to compute ${\cal I}(u)$ (see \eq{3-dist-cal-I}), which would require finding any $i_2\in R_2$ such that $x_{i_1}=x_{i_2}$, which is too expensive. Hence, we simply add all three options, for possibilities ${\cal I}(u)\in\{\emptyset,\{1\},\{2\}\}$ (see also \fig{alt-neighbourhoods-S}),
to $\Psi_\star(u)$:
\begin{equation}
\begin{split}
\Psi_\star(u) := \{&\ket{\psi_\star^{\emptyset}(u)}:=\sqrt{\w_0^+}\ket{u,\leftarrow}+\sqrt{\w_1}\ket{u,\{1\}}+\sqrt{\w_1}\ket{u,\{1,2\}}+\sqrt{\w_1}\ket{u,\{2\}},\\
&\ket{\psi_\star^{\{1\}}(u)}:=\sqrt{\w_0^+}\ket{u,\leftarrow}+\sqrt{\w_1}\ket{u,\{2\}},\\
&\ket{\psi_\star^{\{2\}}(u)}:=\sqrt{\w_0^+}\ket{u,\leftarrow}+\sqrt{\w_1}\ket{u,\{1\}}
\} \ni \ket{\psi_\star^G(u)}.
\end{split}\label{eq:3-dist-stars}
\end{equation}
Note that it is important that each state in $\bigcup_{u\in V_0^+}\Psi_\star(u)$ (and therefore each $\ket{\psi_\star^{G}(u)}$) have at least one outgoing (i.e.~forward) edge. Otherwise, it would be impossible to satisfy \textbf{P2} of \thm{full-framework} (or equivalently, Item 2 of \lem{3-dist-positive}). This is satisfied because ${\cal I}(u)$ is always a proper subset of $\{1,2\}$.

We now describe how to generate the states in $\bigcup_{u\in V(G)}\Psi_\star(u)$ in $\widetilde{O}(1)={\sf polylog}(n)$ complexity (see \defin{alternative}). We will make use of the following lemma.

\begin{lemma}\label{lem:const-stars}
Let $V'\subseteq V(G)\setminus V_0\cup V_{\sf M}$ be such that there exists some constant $c$ such that for all $u\in V'$, $L(u)\subseteq \{0,1\}^c$. Suppose for all $u\in V'$, 
$$\Psi_\star(u)=\{\ket{u}\ket{\phi_\ell}:\ell\in [d']\}$$
for some constant $d'$, and states $\ket{\phi_{\ell}}\in \mathrm{span}\{\ket{j}:j\in \{0,1\}^c\}$. 
Then for some $d\leq d'$, there is an orthonormal basis $\overline{\Psi}(u)=\{\ket{\overline{\psi}_{u,1}},\dots,\ket{\overline{\psi}_{u,d}}\}$ for $\mathrm{span}\{\Psi_\star(u)\}$, for each $u\in V'$, and a map $U_\star'$ that can be implemented in cost $O(1)$ such that for all $u\in V'$ and $\ell\in[d]$, $U_\star'\ket{u,\ell}=\ket{\overline{\psi}_{u,\ell}}$. 
\end{lemma}
\begin{proof}
First note that by the assumptions we are making, $d:=\dim \mathrm{span}\{\Psi_\star(u)\}$ for all $u\in V'$, and $d$ is a constant. Fix any orthonormal basis 
$\{\ket{\overline{\phi}_1},\dots,\ket{\overline{\phi}_d}\}$ for 
$\mathrm{span}\{\ket{\phi_{\ell}}:\ell\in [d']\},$
which is independent of $u$. Since the basis lives in a constant-dimensional subspace, the map:
$C_\star:\ket{\ell}\mapsto \ket{\overline{\phi}_\ell}$
acts on a constant number of qubits, and so can be implemented in $O(1)$ elementary gates. We complete the proof by letting $U_\star'=I\otimes C_{\star}$, and observe that: $U_\star'\ket{u,\ell} = \ket{u}\ket{\overline{\phi}_\ell}=:\ket{\overline{\psi}_{u,\ell}}$.
\end{proof}

\begin{lemma}\label{lem:3-dist-star-states}
The states $\Psi_\star=\{\Psi_\star(u)\}_{u\in V(G)}$ can be generated in $\widetilde{O}(1)$ complexity. 
\end{lemma}
\begin{proof}
The description of a vertex $u\in V(G)$ begins with a label indicating to which of $V_0,V_0^+,V_1,V_2,V_3$ it belongs. 
Thus, we can define subroutines $U_0,U_{0,+},U_1,U_2,U_3$ that generate the star states in each vertex set respectively,
and then $U_\star=\sum_{\ell = 0}^3\ket{\ell}\bra{\ell}\otimes U_{\ell} + \ket{0,+}\bra{0,+}\otimes U_{0,+}$ will generate the star states in the sense of \defin{alternative}. 

We begin with $U_0$. For $v^0_{R}\in V_0$, we have $\Psi_\star(v^0_R)=\{\ket{\psi_\star^G(v^0_R)}\}$, where $\ket{\psi_\star^G(v^0_R)}$ is as in \eq{3-star-v0}. Thus, implementing the map $U_0: \ket{v^0_R}\ket{0}\mapsto \propto \ket{\psi_\star^G(v^0_R)}$ is as simple as generating a uniform superposition over $A_1$, and then using $O(\log n)$ rounds of amplitude amplification to get inverse polynomially close to the uniform superposition over $A_1\setminus R_1$. 

For $U_{0,+}$, since all $v_{R,i_1}^0\in V_0^{+}$ have the same star states, modulo $v_{R,i_1}^0$ itself, with constant-sized label set $L=\{\{1\},\{2\},\{1,2\},\leftarrow\}$, we can apply \lem{const-stars}, to get a $U_{0,+}$ that costs $O(1)$. 

We continue with $U_1$. For $v^1_{R}\in V_1$, we have $\Psi_\star(v^1_R)=\{\ket{\psi_\star^G(v^1_R)}\}$, where $\ket{\psi_\star^G(v^1_R)}$ is as in \eq{3-star-v1}. Thus, to implement the map $U_1: \ket{u}\ket{0}\mapsto \propto \ket{\psi_\star^G(u)}$, we first compute (referring to \tabl{3-dist-weights} for the weights):
$$\ket{u,0}\mapsto\propto \ket{u}\left(-\sqrt{\w_1}\ket{\leftarrow}+\sqrt{\w_2}\ket{\rightarrow}\right)\ket{0}=\ket{u}\left(-\ket{\leftarrow}+(n/m_2)^{1/4}\ket{\rightarrow}\right)\ket{0},$$
which can be implemented by a $O(1)$-qubit rotation. Then conditioned on $\leftarrow$, generate a uniform superposition over $i_1\in R_1$, and then use $O(\log n)$ rounds of amplitude amplification to get inverse polynomially close to a superposition over $i_1\in R_1$ such that there is no $(i_1,i_2,x_{i_1})\in D_2(R)$. We have used the fact that our data structure supports taking a uniform superposition (see \sec{data}). Finally, conditioned on $\rightarrow$, generate a uniform superposition over $j_2\in [m_2]\setminus R_2$. 

The implementation of $U_2$ is similar, but instead (see \eq{3-star-v2}) we perform a single qubit rotation to get $-\sqrt{\w_2}\ket{\leftarrow}+\sqrt{\w_3}\ket{\rightarrow}$ in the last register, and then conditioned on the value of this register, we either generate a uniform superposition over $R_2(\mu(S))$ or $A_3$.

Finally, referring to \eq{3-star-v3}, we see that the implementation of $U_3$ is trivial.
We Thus, conclude that $U_\star$ can be implemented in $\widetilde{O}(1)={\sf polylog}(n)$ complexity.
\end{proof}

\subsubsection{The Transition Subroutines}\label{sec:3-dist-transitions}

In this section we show how to implement the transition map  $\ket{u,i}{\mapsto}\ket{v,j}$ for $(u,v)\in\overrightarrow{E}(G)$ with $i=f_u^{-1}(v)$ and $j=f_v^{-1}(u)$ (see \defin{QW-access}). We do this by exhibiting uniform (in the sense of \lem{uniform-alg}) subroutines ${\cal S}_0^{+},{\cal S}_1,{\cal S}_2,{\cal S}_3$ that implement the transition map for $(u,v)$ in $E_0^{+},E_1,E_2,E_3$ respectively (defined in \sec{3-dist-G}) whose union makes up $\overrightarrow{E}(G)$. In \cor{3-dist-transitions}, we will combine these to get a quantum subroutine (\defin{variable-time}) for the full transition map. 

\begin{lemma}\label{lem:3-dist-T1}
There is a uniform subroutine ${\cal S}_0^{+}$ such that for all $(u,v)\in E_0^{+}$ with $i=f_u^{-1}(v)$ and $j=f_v^{-1}(u)$, ${\cal S}_0^+$ maps $\ket{u,i}$ to $\ket{v,j}$ with error 0 in complexity ${\sf T}_{u,v}={\sf T}_0^+=\widetilde{O}(1)$.
\end{lemma}
\begin{proof}
For $(v^0_{R},v^0_{R,i_1})\in E_0^+$, ${\cal S}_0^+$ should implement the map:
\begin{align*}
\ket{v^0_{R},i_1} &\mapsto \ket{v^0_{R,i_1},\leftarrow}\\
\equiv \; 
\ket{(0,R,D(R)),i_1} &\mapsto \ket{((0,+),R,D(R),i_1),\leftarrow}.
\end{align*}
It is easy to see that this can be done in ${\sf polylog}(n)$ complexity (and is therefore trivially uniform): we just need to do some accounting to move $i_1$ from the edge label register to the vertex register, and update the first register $\ket{0}\mapsto \ket{(0,+)}$. 
\end{proof}

\begin{lemma}\label{lem:3-dist-T2}
There is a uniform subroutine ${\cal S}_1$ such that for all $(u,v)\in E_1$ with $i=f_u^{-1}(v)$ and $j=f_v^{-1}(u)$, ${\cal S}_1$ maps $\ket{u,i}$ to $\ket{v,j}$ with error 0 in complexity ${\sf T}_{u,v}={\sf T}_1=\widetilde{O}(1)$.
\end{lemma}
\begin{proof}
For $(v^0_{R_1,R_2,i_1},v^1_{R_1',R_2})\in E_1$, where $v^1_{R_1',R_2}\in V_1(S)$, ${\cal S}_1$ should implement the map:
\begin{align*}
\ket{v^0_{R_1,R_2,i_1},S} &\mapsto \ket{v^1_{R_1',R_2},i_1}\\
\equiv \; 
\ket{((0,+),R_1,R_2,D(R_1,R_2),i_1),S} &\mapsto \ket{(1,R_1',R_2,D(R_1',R_2)),i_1}.
\end{align*}
To implement this transition, we need only insert $i_1$ into $R_1(S)$, query $i_1$ to obtain $x_{i_1}$ and update the data by inserting $(i_1,x_{i_1})$ into the $D_1(R)$ part of $D(R_1,R_2)=(D_1(R),D_2(R))$ (see \sec{data}). Note that we \emph{do not} attempt to update the $D_2(R)$ part of the data by searching $\overline{R}_2$ for collisions with $i_1$.
If there is some $s\in S$ and $i_2\in \overline{R}_2(s)$ such that $x_{i_1}=x_{i_2}$, then by definition of $E_1$, $(v^0_{R_1,R_2,i_1},v^1_{R_1',R_2})\not\in E_1$. 
To finish, we uncompute $S$ by checking which of the three parts of $R_1$ has size $t_1+1$, account for the moving of $i_1$ from the vertex register to the edge label register, and map $\ket{(0,+)}$ to $\ket{1}$ in the first register. The total cost is ${\sf polylog}(n)$. 
\end{proof}

We now move on to ${\cal S}_2$, which is somewhat more complicated. For $(v^1_{R_1,R_2},v^2_{R_1,R_2'})\in E_2$,  where $v^1_R\in V_1(S)$, and $R_2'(\mu(S))=R_2(\mu(S))\cup\{j_2\}$ for some $j_2\in [m_2]\setminus R_2$, ${\cal S}_2$ should act as:
\begin{equation}
\begin{split}
\ket{v^1_{R_1,R_2},j_2} &\mapsto \ket{v^2_{R_1,R_2'},j_2}\\
\equiv \; 
\ket{(1,R_1,R_2,D(R_1,R_2)),j_2} &\mapsto \ket{(2,R_1,R_2',D(R_1,R_2')),j_2}.
\end{split}\label{eq:T3}
\end{equation}
The complexity of this map, which we will implement with some error, depends on  $|{\cal K}({R}_1,A_2^{(j_2)})|$ (see \eq{collision-set}), the number of collisions to be found between $R_1$ and the block $A_2^{(j_2)}$, which is implicitly being added to $\overline{R}_2$ by adding $j_2$ to $R_2$. 
\lem{3-dist-T3} below describes how to implement this transition map as long as there are fewer than $c_{\max}\log n$ collisions to be found for some constant $c_{\max}$. 
For the case when $|{\cal K}({R}_1,A_2^{(j_2)})|\geq c_{\max}\log n$, we will let the algorithm fail (so there is no bound on the error for such transitions). That is, we let:
\begin{equation}
\tilde{E}:=\left\{\left(v^1_{R},v^2_{R'}\right)\in E_2: |{\cal K}(R_1,A^{(j_2)})|\geq c_{\max}\log n,
\mbox{ where }\{j_2\}=R_2'\setminus R_2\right\}.\label{eq:3-dist-tilde-E}
\end{equation}

\begin{lemma}\label{lem:3-dist-T3}
Fix any constant $\kappa$. There is a uniform subroutine ${\cal S}_2$ that implements the transition map that maps $\ket{u,i}$ to $\ket{v,j}$ for all $(u,v)\in E_2\setminus\tilde{E}$, with error $O(n^{-\kappa})$, in complexity ${\sf T}_{u,v}={\sf T}_2=\widetilde{O}(\sqrt{n/m_2})$. 
\end{lemma}
\begin{proof}
To implement the map in \eq{T3}, we need to insert $j_2$ into $R_2(\mu(S))$ to obtain $R_2'$, update $D_2(R)$ to reflect this insertion, and increment the first register. All of these take ${\sf polylog}(n)$ complexity, except for updating $D_2(R)$. To update $D_2(R)$, we need to search $A_2^{(j_2)}$ -- the new block we're adding to $\overline{R}_2$ -- to find anything that collides with $R_1$. Since the number of such collisions is less than $c_{\max}\log n$, we can do this using quantum search, which is uniform, with error $O(n^{-\kappa})$ for any desired constant $\kappa$ in complexity 
$O(\sqrt{n/m_2}\log^2 n)$, since $|A_2^{(j_2)}|=\sqrt{n/m_2}$.
\end{proof}

\begin{lemma}\label{lem:tilde-E}
For any constant $\kappa$, there exists a choice of constant $c_{\max}$ sufficiently large such that 
$|\tilde{E}|\leq n^{-\kappa}|E_2|$. 
\end{lemma}
\begin{proof}
By \lem{hypergeo} (or as a special case of \lem{setup-probability}), for every $j_2\in [m_2]$, if $R_1$ is uniformly random from $\binom{A_1}{t_1^{(3)}}$, there exists a constant $c_{\max}$ such that 
$\Pr[|{\cal K}({R}_1,A_2^{(j_2)})|\geq c_{\max}\log n]\leq n^{-\kappa}.$
It follows that the proportion of edges in $E_2$ that are in $\tilde{E}$ is at most $n^{-\kappa}$.
\end{proof}
\begin{lemma}\label{lem:3-dist-T4}
There is a subroutine ${\cal S}_3$ such that for all $(u,v)\in E_3$ with $i=f_u^{-1}(v)$ and $j=f_v^{-1}(u)$, ${\cal S}_3$ maps $\ket{u,i}$ to $\ket{v,j}$ with error 0 in complexity ${\sf T}_{u,v}=\widetilde{O}(1)$.
\end{lemma}
\begin{proof}
The proof is identical to that of \lem{3-dist-T1}.
\end{proof}

In order to apply \thm{full-framework}, we need to implement the full transition map as a quantum subroutine in the sense of \defin{variable-time}.
\begin{corollary}\label{cor:3-dist-transitions}
Let $\kappa$ be any constant. 
There is a quantum subroutine (in the sense of \defin{variable-time}) that implements the full transition map with errors $\epsilon_e\leq n^{-\kappa}$ for all $e\in \overrightarrow{E}(G)\setminus\tilde{E}$, and times ${\sf T}_e=\widetilde{O}(1)$ for all $e\in \overrightarrow{E}(G)\setminus E_2$, and ${\sf T}_e={\sf T}_2=\widetilde{O}(\sqrt{n/m_2})$ for all $e\in E_2$. 
\end{corollary}
\begin{proof}
We combine \lem{3-dist-T1}, \lem{3-dist-T2}, \lem{3-dist-T3} and \lem{3-dist-T4} using  \lem{uniform-combine}.
\end{proof}

\subsubsection{Initial State and Setup Cost}

The initial state is defined to be the uniform superposition over $V_1$: 
\begin{equation*}
\ket{\sigma} := \sum_{v^0_{R_1,R_2}\in V_0}\frac{1}{\sqrt{|V_0|}}\ket{v^0_{R_1,R_2}}.
\end{equation*}

\begin{lemma}\label{lem:setup}
The state $\ket{\sigma}$ can be generated with error $n^{-\kappa}$ for any constant $\kappa$ in complexity 
$$\textstyle{\sf S}=\widetilde{O}\left(t_1+t_2\sqrt{\frac{n}{m_2}}\right).$$
\end{lemma}
\begin{proof}
We start by taking a uniform superposition over all $R_1\in \binom{A_1}{t_1^{(3)}}$ and $R_2\in \binom{[m_2]}{t_2^{(2)}}$ stored in data structures as described in \sec{data}, and querying everything in $R_1$ to get $D_1(R)$, which altogether costs $\widetilde{O}(t_1+t_2)$. Next for each $s\in\{1,2\}$, we search for all elements of $\overline{R}_2(s)$ that collide with an element of $R_1(\{s\})$ or $R_1(\{1,2\})$. However, we do not want to spend too long on this step, so we stop if we find $ct_2$ collisions, for some constant $c$. If we do this before all collisions are found, that part of the state is not correct, but we argue that this only impacts a very small part of the state. The cost of this search is (up to log factors) $\sqrt{t_2|\overline{R}_2|}=t_2\sqrt{n/m_2}$. 

For a uniform $R_1$ and fixed $R_2$, the expected value of $Z=|{\cal K}({R}_1,\overline{R}_2)|$, the number of collisions, is 
$$\mu=O\left(\frac{|\overline{R}_2| t_1}{n}\right) = O\left(\frac{t_2\frac{n}{m_2}  t_1}{n}\right)=O(t_2),$$
since $m_2=\Theta(t_1)$. Let $c'$ be a constant such that $\mu\leq c't_2$, and choose $c=7c'$.
Since $Z$ is a hypergeometric random variable, we have, by \lem{hypergeo}, 
$\Pr[Z\geq ct_2] \leq e^{-ct_2} = o( n^{-\kappa})$
for any $\kappa$, since $t_2$ is polynomial in $n$. Thus, the state we generate is $n^{-\kappa}$-close to $\ket{\sigma}$. 
\end{proof}

\subsubsection{Positive Analysis}

For the positive analysis, we must exhibit a flow (see \defin{flow}) from $V_0$ to $M$ whenever $M\neq \emptyset$.

\begin{lemma}\label{lem:3-dist-positive}
There exists some ${\cal R}^{\sf T}=\widetilde{O}(|V_0|^{-1})$ such that the following holds. Whenever there is a unique $3$-collision $(a_1,a_2,a_3)\in A_1\times A_2\times A_3$, there exists a flow $\theta$ on $G$ that satisfies conditions \textbf{P1}-\textbf{P5} of \thm{full-framework}. Specifically:
\begin{enumerate}
	\item For all $e\in\tilde{E}$, $\theta(e)=0$.
	\item For all $u\in V(G)\setminus (V_0\cup V_3)$ and $\ket{\psi_\star(u)}\in \Psi_\star(u)$, 
$$\sum_{i\in L^+(u)}\frac{\theta(u,f^+_u(i))\braket{\psi_\star(u)}{u,i}}{\sqrt{\w_{u,i}}}-\sum_{i\in L^-(u)}\frac{\theta(u,f^-_u(i))\braket{\psi_\star(u)}{u,i}}{\sqrt{\w_{u,i}}}=0.$$
	\item $\sum_{u\in  V_0}\theta(u)=1$.
	\item $\sum_{u\in V_0}\frac{|\theta(u)-\sigma(u)|^2}{\sigma(u)}\leq 1$.
	\item ${\cal E}^{\sf T}(\theta)\leq {\cal R}^{\sf T}$.
\end{enumerate}
\end{lemma}
\begin{proof}
Recall that $M$ is the set of $v^3_{R,i_3}\in V_3$ such that for some $S\subseteq \{1,2\}$, 
$a_1\in R_1(S)$, $a_2\in \overline{R}_2(\mu(S))$ and $a_3=i_3$. Let $j^*\in [m_2]$ be the unique block label such that $a_2\in A_2^{(j^*)}$. Then $a_2\in \overline{R}_2(\mu(S))$ if and only if $j^*\in R_2(\mu(S))$. 
Assuming $M\neq\emptyset$, we define a flow $\theta$ on $G$ with all its sinks in $M$. It will have sources in both $V_1$ and $M$, but all other vertices will conserve flow. This will imply \textbf{Item 2} for all \emph{correct} star states of $G$, $\ket{\psi_\star^G(u)}$, but we will have to take extra care to ensure that \textbf{Item 2} is satisfied for the additional states in $\Psi_\star(u):u\in V_0^+$. 

To satisfy condition \textbf{P5} of \thm{full-framework}, we must upper bound ${\cal E}^{\sf T}(\theta)={\cal E}(\theta^{\sf T})$ (see \defin{nwk-length}), which is the energy of the flow $\theta$ extended to a graph $G^{\sf T}$, in which each edge of $G$ in $E_2$ has been replaced by a path of length ${\sf T}_2=\widetilde{O}(\sqrt{n/m_2})$, and all other edges have been replaced by paths of length $\widetilde{O}(1)$ (see \cor{3-dist-transitions}).
We define $\theta$ on $E_0^+$, $E_1$, $E_2$ and $E_3$ stage by stage, and upper bound the contribution to ${\cal E}^{\sf T}(\theta)$ for each stage.

\vskip7pt
\noindent\textbf{${\cal R}_0^+$, Item 3, and Item 4:} Let $M_0$ be the set of $v^0_{R_1,R_2}\in V_0$ such that $a_1\not\in R_1$, $j^*\not\in R_2$, and for $c_{\max}$ as in \lem{tilde-E}, $|{\cal K}({R}_1,A_1^{(j^*)})|<c_{\max}\log n$ (see~\eq{collision-set}). This latter condition is because we will later send flow down edges that add $j^*$ to $R_2$, and we don't want to have flow on edges in $\tilde{E}$. For all $v^0_{R}\in M_0$, let $\theta(v^0_{R},v^1_{R,a_1})={|M_0|}^{-1}$. For all other edges in $E_0^+$, let $\theta(e)=0$. Note that we can already see that $\theta(u)={|M_0|}^{-1}$ for all $u\in M_0$, so we satisfy \textbf{Item~3}. By \lem{tilde-E}, we know that the proportion of $R_1$ that are excluded because $|{\cal K}(\overline{R}_1,A_1^{(j^*)})|\geq c_{\max}\log n$ is $o(1)$, so we can conclude:
\begin{equation}
\frac{|V_0|}{|M_0|} = (1+o(1))\left(1+O\left(\frac{t_1}{n}\right)\right)\left(1+O\left(\frac{t_2}{m_2}\right)\right) = 1+o(1).\label{eq:V0-M0}
\end{equation}

Since $\sigma(u)=\frac{1}{|V_0|}$, we can conclude with \textbf{Item 4} of the lemma statement:
\begin{equation*}
\begin{split}
\sum_{u\in V_0}\frac{|\theta(u)-\sigma(u)|^2}{\sigma(u)} 
&= |V_0|^2\left(\frac{1}{|M_0|}-\frac{1}{|V_0|}\right)^2
= \left(\frac{|V_0|}{|M_0|}-1\right)^2=o(1).
\end{split}\label{eq:delta-bound}
\end{equation*}
Using $\w_0^+=1$ and ${\sf T}_e=\widetilde{O}(1)$ for all $e\in E_0^+$ (see \tabl{3-dist-weights}), the contribution of the edges in $E_0^+$ to the energy of the flow can be computed as:
\begin{equation}
{\cal R}_0^+ = \sum_{e\in E_0^+}{\sf T}_e\frac{\theta(e)^2}{\w_0^+} = \widetilde{O}\left(\sum_{u\in M_0}\frac{1}{|M_0|^2}\right)=\widetilde{O}\left(\frac{1}{|M_0|}\right),\label{eq:R1}
\end{equation}
since each vertex in $M_0$ has a unique outgoing edge with flow, and the flow is uniformly distributed.

\vskip7pt
\noindent\textbf{${\cal R}_1$ and Item 2:} Let $M_0^+$ be the set of $v^0_{R,i_1}\in V_0^+$ such that $v^0_{R}\in M_0$ and $i_1=a_1$, so $|M_0^+|=|M_0|$. These are the only vertices in $V_0^+$ that have flow coming in from $V_0$, and specifically, the incoming flow from $V_0$ to a vertex in $M_0^+$ is $\frac{1}{|M_0|}$.

The only way there could be a fault adding $a_1$ to $R_1$ would be if $a_2\in \overline{R}_2$, but we have ensured that that is not the case. Thus, for each $u\in M_0^+$, ${\cal I}(u)=\emptyset$, so there are three edges going into $V_1$ (labelled by $\{1\}$, $\{2\}$, and $\{1,2\}$, all disjoint from ${\cal I}(u)$) to which we can assign flow. 

\textbf{Item 2} is satisfied for all $\ket{\psi_\star^G(u)}:u\in V(G)\setminus (V_0\cup V_3)$ by virtue of $\theta$ conserving flow at all vertices in $V(G)\setminus (V_0\cup V_3)$ (we have not finished defining $\theta$, but it will be defined so that this holds). However, for $u\in V_0^+$, $\Psi_\star(u)=\{\ket{\psi_\star^{{\cal I}}(u)}\}_{{\cal I}\subsetneq\{1,2\}}$ (see \eq{3-dist-stars}) contains more than just $\ket{\psi_\star^G(u)}$. When $u\in V_0^+\setminus M_0^+$, there is no flow through $u$, so Item 2 is easily seen to be satisfied for all states in $\Psi_\star(u)$. For $u\in M_0^+$, $\ket{\psi_\star^G(u)}=\ket{\psi_\star^{\emptyset}(u)}$, so the additional constraints we need to take additional care to satisfy are those for $\ket{\psi_\star^{\{s\}}(u)}$ with $s\in \{1,2\}$:
\begin{align*}
&\sum_{i\in L^+(u)}\theta(u,f^+_{u}(i))\frac{\braket{\psi_\star^{\{s\}}(u)}{u,i}}{\sqrt{\w_1}}-\sum_{j\in L^-(u)}\theta(u,f^-_{u}(j))\frac{\braket{\psi_\star^{\{s\}}(u)}{u,i}}{\sqrt{\w_0^+}}\\
={}& \sum_{S \in 2^{\{1,2\}} \setminus \{\emptyset\}}\theta(u,f^+_{u}(S))\frac{\braket{\psi_\star^{\{s\}}(u)}{u,S}}{\sqrt{\w_1}}-\theta(u,f^-_{u}(\leftarrow))\frac{\braket{\psi_\star^{\{s\}}(u)}{u,\leftarrow}}{\sqrt{\w_0^+}} & \mbox{see \tabl{3-dist-index-sets}}\\
={}&\theta(u,f^+_{u}(\{3-s\}))\frac{\sqrt{\w_1}}{\sqrt{\w_1}}-\theta(u,f^-_u(\leftarrow))\frac{-\sqrt{\w_0^+}}{\sqrt{\w_0^+}} & \mbox{see \eq{3-dist-stars}}\\
={}&\theta(u,v_{\{3-s\}})+\theta(u,v^0),
\end{align*}
where $v^0 = f^-_u(\leftarrow)$ is the neighbour of $u$ in $V_0$, and $v_{\{3-s\}} = f^+_u(\{3-s\})$ is the neighbour of $u$ in $V_1$ with edge labelled by $\{3-s\}$ (see \fig{3-dist-flow-stars}). So for $s'\in\{1,2\}$, we must have:
$$0=\theta(u,v_{\{s'\}}) + \theta(u,v^0) = \theta(u,v_{\{s'\}}) - \frac{1}{|M_0|},$$
since $\theta(u,v^0) = -\theta(v^0,u) = -\frac{1}{|M_0|}$. To satisfy this, we set:
\begin{align*}
\theta(u,v_{\{1\}}) = \theta(u,v_{\{2\}}) = \frac{1}{|M_0|},
\end{align*}
meaning that all the flow that comes into $u$ along edge $(v^0,u)$ must leave $u$ along edge $(u,v_{\{1\}})$, but it must also all leave along edge $(u,v_{\{2\}})$. However, we have now assigned twice as much outgoing flow as incoming flow, so the only way for flow to be conserved at $u$ is to also have $\frac{1}{|M_0|}$ flow coming into $u$ along edge $(v_{\{1,2\}},u)$, so we set:
$$\theta(u,v_{\{1,2\}}) = -\frac{1}{|M_0|}.$$
This is shown visually in \fig{3-dist-flow-stars}.
Using $\w_1=1$ and ${\sf T}_1=\widetilde{O}(1)$, we can compute the contribution of edges in $E_1$ to the energy of the flow as:
\begin{equation}
{\cal R}_1 = \sum_{u\in M_0^+}{\sf T}_1\frac{3(1/|M_0|)^2}{\w_1} = \widetilde{O}\left(\frac{|M_0^+|}{|M_0|^2}\right)=\widetilde{O}\left(\frac{|M_0|}{|M_0|^2}\right)=\widetilde{O}\left(\frac{1}{|M_0|}\right).
\label{eq:R2}
\end{equation}

\begin{figure}
\centering
\begin{tikzpicture}[scale=1.2]
\node at (0,0) {
\begin{tikzpicture}[scale=1.2]
\draw[->] (0,0)--(.75,0); \draw (.75,0)--(1.5,0);		\draw[->] (1.5,0)--(2.25,.5); \draw (2.25,.5)--(3,1);
										\draw (1.5,0)--(2.25,0); \draw[<-] (2.25,0)--(3,0);
										\draw[->] (1.5,0)--(2.25,-.5); \draw (2.25,-.5)--(3,-1);

\filldraw (0,0) circle (.05);		\filldraw (1.5,0) circle (.08);	\filldraw (3,1) circle (.05);
												\filldraw (3,0) circle (.05);
												\filldraw (3,-1) circle (.05);

\node at (-.25,.2) {$v^0$};
\node at (1.5,.25) {$u$};
\node at (3.5,1) {$v_{\{1\}}$};
\node at (3.5,0) {$v_{\{1,2\}}$};
\node at (3.5,-1) {$v_{\{2\}}$};

\node at (1.1,-.15) {\color{blue}$\leftarrow$};
\node at (1.9,.5) {\color{blue}\small ${}_{\{1\}}$};
\node at (2.2,.15) {\color{blue}\small ${}_{\{1,2\}}$};
\node at (1.9,-.5) {\color{blue}\small ${}_{\{2\}}$};
\end{tikzpicture}
};

\node at (0,-1.5) {$\ket{\psi_\star^G(u)}=\ket{\psi_\star^{\emptyset}(u)}$};

\node at (6,0) {
\begin{tikzpicture}[scale=1.2]
\draw[->] (0,0)--(.75,0); \draw (.75,0)--(1.5,0);		\draw[->] (1.5,0)--(2.25,.5); \draw (2.25,.5)--(3,1);

\filldraw (0,0) circle (.05);		\filldraw (1.5,0) circle (.08);	\filldraw (3,1) circle (.05);

\node at (-.25,.2) {$v^0$};
\node at (1.5,.25) {$u$};
\node at (3.5,1) {$v_{\{1\}}$};
\node at (3.5,-1) {\color{white} $v_{\{2\}}$};

\node at (1.1,-.15) {\color{blue}$\leftarrow$};
\node at (1.9,.5) {\color{blue}\small ${}_{\{1\}}$};
\end{tikzpicture}
};

\node at (6,-1.5) {$\ket{\psi_\star^{\{2\}}(u)}$};

\node at (12,0) {
\begin{tikzpicture}[scale=1.2]
\draw[->] (0,0)--(.75,0); \draw (.75,0)--(1.5,0);		
										\draw[->] (1.5,0)--(2.25,-.5); \draw (2.25,-.5)--(3,-1);

\filldraw (0,0) circle (.05);		\filldraw (1.5,0) circle (.08);	
												\filldraw (3,-1) circle (.05);

\node at (-.25,.2) {$v^0$};
\node at (1.5,.25) {$u$};
\node at (3.5,1) {\color{white}$v_{\{1\}}$};
\node at (3.5,-1) {$v_{\{2\}}$};

\node at (1.1,-.15) {\color{blue}$\leftarrow$};
\node at (1.9,-.5) {\color{blue}\small ${}_{\{2\}}$};
\end{tikzpicture}
};

\node at (12,-1.5) {$\ket{\psi_\star^{\{1\}}(u)}$};

\end{tikzpicture}
\caption{The three star states in $\Psi_\star(u)$, for $u\in V_0^+$. Edge labels from $L(u)$ are shown in blue. Arrows in edges indicate the direction of flow. We have chosen the flow so that flow is conserved at $u$ in $G$, which can be seen by the fact that flow comes in on two edges, and leaves by two edges in the figure for $\ket{\psi_\star^G(u)}$; but flow is still conserved if we restrict to either of the other two neighbourhoods, which is necessary to satisfy Item 2 of \lem{3-dist-positive}.}\label{fig:3-dist-flow-stars}
\end{figure}
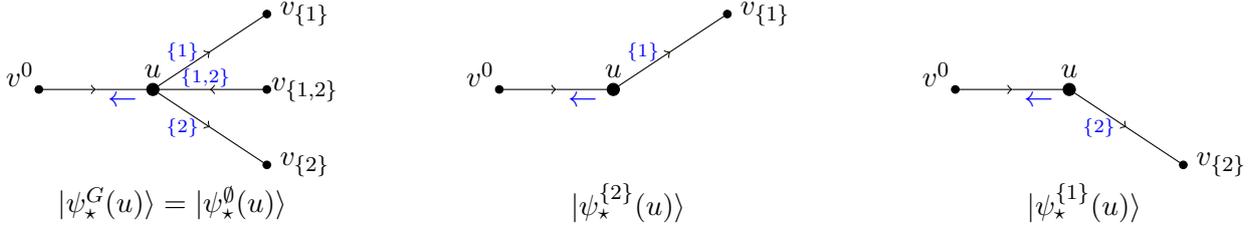

\vskip7pt
\noindent\textbf{${\cal R}_2$ and Item 1:} Let $M_1(S)$ be the set of $v^1_{R}\in V_1(S)$ such that $a_1\in R_1(S)$ and $j^*\not\in R_2$, and let $M_1=M_1(\{1\})\cup M_1(\{2\})\cup M_1(\{1,2\})$, so $|M_1|=3|M_0|$. These are exactly the vertices of $V_1$ that have non-zero flow coming in from $V_0^+$, and in particular, for $v^1_{R}\in M_1(S)$, the amount of flow coming in from $V_0^+$ is $(-1)^{|S|}\frac{3}{|M_1|}$, and we will send it along the edge $(v^1_{R},v^2_{R'})\in E_2$ that adds $j^*$ to the set $R_2(\mu(S))$ to obtain $R'$:
\begin{equation*}
\theta\left(v^1_{R},v^2_{R'}\right) = (-1)^{|S|+1}\frac{3}{|M_1|}=(-1)^{|S|+1}\frac{1}{|M_1(S)|}.
\end{equation*}
All other edges of $E_2$ will have $\theta(e)=0$. 
Using $\w_2=\sqrt{n/m_2}$ and ${\sf T}_2=\widetilde{O}(\sqrt{n/m_2})$, we can compute the contribution of edges in $E_2$ to the energy ${\cal E}^{\sf T}$ of the flow:
\begin{equation}
\begin{split}
{\cal R}_2 &= \frac{{\sf T}_2}{\w_2}|M_1|\frac{9}{|M_1|^2}
=\widetilde{O}\left(\frac{1}{|M_0|}\right).
\end{split}\label{eq:R3}
\end{equation}

We also note that by ensuring that there is only flow on $v^1_{R}\in V_1$ when ${\cal K}({R}_1,A_2^{(j^*)})$ is not too big, we have ensured that the flow on the edges in $\tilde{E}$ is 0, satisfying \textbf{Item 1}. 

\vskip7pt
\noindent\textbf{${\cal R}_3$:} Finally, let $M_2(S)$ be the set of $v^2_{R}\in V_2(S)$ such that $a_1\in R_1(S)$ and $j^*\in R_2(\mu(S))$, and let $M_2=M_2(\{1\})\cup M_2(\{2\})\cup M_2(\{1,2\})$. These are exactly the vertices of $V_2$ that have non-zero flow coming in from $V_1$, in the amount of
$(-1)^{|S|+1}{|M_2(S)|}^{-1}.$
We send this flow along the unique edge from $v^2_{R}$ into $V_3$ that adds $i_3=a_3$:
\begin{align*}
\theta(v^2_{R},v^3_{R,a_3}) &=(-1)^{|S|+1}\frac{1}{|M_2(S)|}=(-1)^{|S|+1}O\left(\frac{1}{|M_0|}\right).
\end{align*}
Using $\w_3=1$ and ${\sf T}_3=\widetilde{O}(1)$, the total contribution of edges in $E_3$ to the energy of the flow is:
\begin{equation}
\begin{split}
{\cal R}_3 &= \frac{{\sf T}_3}{\w_3}|M_2|O\left(\frac{1}{|M_2|^2}\right)
=\widetilde{O}\left(\frac{1}{|M_0|}\right).
\end{split}\label{eq:R4}
\end{equation}

\vskip5pt
\noindent\textbf{Item 5:} It remains only to upper bound the energy of the flow by adding up the 4 contributions in \eq{R1} to \eq{R4}, and applying $|V_0|=(1+o(1))|M_0|$ from \eq{V0-M0}:
\begin{align*}
{\cal E}^{\sf T}(\theta) &\leq {\cal R}_0^++{\cal R}_1+{\cal R}_2+{\cal R}_3
=\widetilde{O}\left(\frac{1}{|M_0|} \right)=\widetilde{O}\left(\frac{1}{|V_0|} \right). \qedhere
\end{align*}
\end{proof}

\subsubsection{Negative Analysis}

For the negative analysis, we need to upper bound the total weight of the graph, taking into account the subroutine complexities, ${\cal W}^{\sf T}(G)$ (see \defin{nwk-length}).
\begin{lemma}\label{lem:3-dist-negative}
There exists  ${\cal W}^{\sf T}$ such that:
$${\cal W}^{\sf T}(G)\leq {\cal W}^{\sf T}\leq \widetilde{O}\left(\left(n + \frac{n^2}{t_1} +\frac{n^2}{t_2}\right)|V_1|\right).$$
\end{lemma}
\begin{proof}
Recall that ${\cal W}^{\sf T}(G)={\cal W}(G^{\sf T})$ is the total weight of the graph $G^{\sf T}$, where we replace each edge $e$ of $G$, with weight $\w_e$, by a path of ${\sf T}_e$ edges of weight $\w_e$, where ${\sf T}_e$ is the complexity of the edge transition $e$. Thus, ${\cal W}^{\sf T}(G)=\sum_{e\in E(G)}{\sf T}_e\w_e$. By \cor{3-dist-transitions}, for all $e\in \overrightarrow{E}(G)\setminus E_2$, ${\sf T}_e=\widetilde{O}(1)$, and $\w_e=1$ (see \tabl{3-dist-weights}). Thus, using \eq{E1}, the total contribution to the weight from the edges in $E_0^+$ is:
\begin{equation}
{\cal W}_0^+ := \w_0^+|E_0^+|{\sf T}_0^+=\widetilde{O}\left({n}|V_0|\right).\label{eq:W1}
\end{equation}
Using \eq{E2}, the total contribution from the edges in $E_1$ is:
\begin{equation}
{\cal W}_1:= \w_1|E_1|{\sf T}_1 = \widetilde{O}\left(n|V_0|\right).\label{eq:W2}
\end{equation}
The edges $e\in E_2$ have ${\sf T}_e={\sf T}_2=\widetilde{O}(\sqrt{n/m_2})$, by \cor{3-dist-transitions}, so using $\w_2=\sqrt{n/m_2}$ and \eq{E3}, the total contribution from the edges in $E_2$ is:
\begin{equation}
{\cal W}_2:= \w_2|E_2|{\sf T}_2 =\widetilde{O}\left(\sqrt{\frac{n}{m_2}}\frac{2(m_2-t_2)(n-9t_1)}{t_1+1}|V_0|\sqrt{\frac{n}{m_2}}\right)
=\widetilde{O}\left(\frac{n^2}{t_1}|V_0|\right).\label{eq:W3}
\end{equation}
Finally, using \eq{E4} and the fact that $m_2=\Theta(t_1)$, the total contribution from the edges in $E_3$ is:
\begin{equation}
{\cal W}_3:= \w_3|E_3|{\sf T}_3 = \widetilde{O}\left(\frac{n^2}{t_2}|V_0|\right).\label{eq:W4}
\end{equation}
Combining \eq{W1} to \eq{W4}, we get total weight:
\begin{align*}
{\cal W}^{\sf T}(G) &= {\cal W}_0^++{\cal W}_1+{\cal W}_2+{\cal W}_3
 = \widetilde{O}\left(\left(n + \frac{n^2}{t_1}+\frac{n^2}{t_2}\right)|V_0| \right).\qedhere
\end{align*}
\end{proof}

\subsubsection{Conclusion of Proof of Theorem~\ref{thm:3-dist}}

We can now conclude with the proof of \thm{3-dist}, showing an upper bound of $\widetilde{O}(n^{5/7})$ on the bounded error quantum time complexity of $3$-distinctness. 

\begin{proof}[Proof of Theorem~\ref{thm:3-dist}]
We apply \thm{full-framework} to $G$ (\sec{3-dist-G}), $M$ (\eq{3-dist-M}),  $\sigma$ the uniform distribution on $V_0$ (\eq{3-dist-V1}), and $\Psi_\star$ (\sec{3-dist-star-states}), with 
$${\cal W}^{\sf T}=\widetilde{O}\left(\left(n + \frac{n^2}{t_1} +\frac{n^2}{t_2}\right)|V_0|\right)
\mbox{ and }
{\cal R}^{\sf T}=\widetilde{O}\left(|V_0|^{-1} \right).$$
Then we have
$${\cal W}^{\sf T}{\cal R}^{\sf T} = \widetilde{O}\left(n + \frac{n^2}{t_1} +\frac{n^2}{t_2}\right)= o(n^2).$$

\noindent We have shown the following:
\begin{description}
\item[Setup Subroutine:] By \lem{setup}, the state $\ket{\sigma}$ can be generated in cost ${\sf S}=\widetilde{O}\left(t_1+t_2\sqrt{\frac{n}{m_2}}\right)$.
\item[Star State Generation Subroutine:] By \lem{3-dist-star-states}, the star states ${\Psi_\star}$ can be generated in $\widetilde{O}(1)$ complexity.
\item[Transition Subroutine:] By \cor{3-dist-transitions}, there is a quantum subroutine that implements the transition map with errors $\epsilon_{u,v}$ and costs ${\sf T}_{u,v}$ such that
\begin{description}
\item[TS1] For all $(u,v)\in \overrightarrow{E}(G)\setminus E_2$, $\epsilon_{u,v}=0$. For all $(u,v)\in E_2\setminus\tilde{E}$ (see \eq{3-dist-tilde-E}), taking $\kappa>2$ in \lem{3-dist-T3}, we have $\epsilon_{u,v}
=O(n^{-\kappa})=o(1/({\cal R}^{\sf T}{\cal W}^{\sf T}))$.
\item[TS2] By \lem{tilde-E}, using $\w_2=\sqrt{n/m_2}$ and $\kappa>2$:
\begin{align*}
\sum_{e\in \tilde{E}}\w_e &= \w_2|\tilde{E}| \leq \sqrt{\frac{n}{m_2}}n^{-\kappa}|E_2| = \sqrt{\frac{n}{m_2}}n^{-\kappa}\frac{2(m_2-t_2)(n-9t_1)}{t_1+1}|V_0| &\mbox{by }\eq{E3}\\
&= {O}\left(\sqrt{n}n^{-\kappa}n\frac{1}{{\cal R}^{\sf T}}\right) = o(1/{\cal R}^{\sf T}). 
\end{align*}
since $m_2=\Theta(t_1)$.

\end{description}

\item[Checking Subroutine:] By \eq{3-dist-C}, for any $u\in V_{\sf M}=V_3$, we can check if $u\in M$ in cost $\widetilde{O}(1)$. 
\item[Positive Condition:] By \lem{3-dist-positive}, there exists a flow satisfying conditions \textbf{P1}-\textbf{P5} of \thm{full-framework}, with ${\cal E}^{\sf T}(\theta)\leq {\cal R}^{\sf T}=\widetilde{O}\left(|V_0|^{-1} \right)$.
\item[Negative Condition:] By \lem{3-dist-negative}, ${\cal W}^{\sf T}(G)\leq {\cal W}^{\sf T}=\widetilde{O}\left(\left(n + \frac{n^2}{t_1} +\frac{n^2}{t_2}\right)|V_0|\right)$. 
\end{description}
Thus, by \thm{full-framework}, there is a quantum algorithm that decides if $M=\emptyset$ in bounded error in complexity:
\begin{align*}
\widetilde{O}\left({\sf S}+\sqrt{{\cal R}^{\sf T}{\cal W}^{\sf T}}\right) &= \widetilde{O}\left(t_1+t_2\sqrt{\frac{n}{m_2}}+\sqrt{n+\frac{n^2}{t_1}+\frac{n^2}{t_2}}\right)
=\widetilde{O}\left( t_1+t_2\sqrt{\frac{n}{t_1}}+\sqrt{n}+\frac{n}{\sqrt{t_1}}+\frac{n}{\sqrt{t_2}}\right).
\end{align*}
Choosing the optimal values of $t_1=n^{5/7}$ and $t_2=n^{4/7}$, we get an upper bound of $\widetilde{O}(n^{5/7})$.
Since $M\neq\emptyset$ if $x$ has a unique 3-collision, and $M=\emptyset$ if $x$ has no 3-collision, the algorithm distinguishes these two cases. By \lem{unique-to-multiple}, this is enough to solve 3-distinctness in general.
\end{proof}

\subsection{\texorpdfstring{$k$}{k}-Distinctness Algorithm}\label{sec:k-dist-alg}

In this section, we generalise the $3$-distinctness algorithm from \sec{3-dist} to prove the following.
\begin{theorem}\label{thm:k-dist}
	Let $k$ be any constant.
	There is a quantum algorithm that decides $k$-distinctness with bounded error in $\widetilde{O}\left(n^{\frac{3}{4}-\frac{1}{4}\frac{1}{2^k-1}}\right)$ complexity.
\end{theorem}
\noindent Throughout this section, $\widetilde{O}$ will surpress polylogarithmic factors in $n$.
We use the assumptions on the input defined in \sec{k-dist-assumptions}, including partitioning $[n]$ into $A_1\cup\dots\cup A_k$, and each $A_{\ell}$, for $\ell\in\{2,\dots, k-1\}$ into blocks $A_{\ell}^{(1)}\cup\dots\cup A_{\ell}^{(m_{\ell})}$ of size $\frac{n}{km_{\ell}}$. Additionally, for the uniformity of our notation in this section, we choose to also partition $A_1$ into blocks $A_{1}^{(1)}\cup\dots\cup A_{\ell}^{(m_{1})}$ of size $\frac{n}{km_{1}}$. By choosing $m_1 = \frac{n}{k} = \abs{A_1}$, this becomes the trivial partition, where each block is of size $1$. A summary of the parameters of the algorithm appears in \tabl{k-dist-setsizes}.

\paragraph{Tuples of Sets:} Fix constants $c_1,\dots,c_{k-1}$ and parameters $t_1,\dots,t_{k-1}$ as in \tabl{k-dist-setsizes}. The vertices of our graph are labelled by sets $R=(R_1,\dots,R_{k-1})$, where each $R_{\ell}$ is a tuple of $c_1\dots c_{\ell-1}(2^{c_{\ell}}-1)$ disjoint subsets of $[m_{\ell}]$ of size $t_\ell$:
$$R_{\ell}=(R_{\ell}(s_1,\dots,s_{\ell-1},S_{\ell}))_{s_1\in [c_1],\dots,s_{\ell-1}\in [c_{\ell-1}],S_{\ell}\in 2^{[c_{\ell}]}\setminus\{\emptyset\}}.$$
We define:
$$\overline{R}_{\ell}(s_1,\dots,s_{\ell-1},S_{\ell}):= 
\bigcup_{j_{\ell}\in R_{\ell}(s_1,\dots,s_{\ell-1},S_{\ell})}A_{\ell}^{(j_{\ell})}.$$
and
$$\overline{R}_{\ell}:=
(\overline{R}_{\ell}(s_1,\dots,s_{\ell-1},S_{\ell}))_{s_1\in [c_1],\dots,s_{\ell-1}\in [c_{\ell-1}],S_{\ell}\in 2^{[c_{\ell}]}\setminus\{\emptyset\}}.$$

If we let $r_{\ell}=|\overline{R}_{\ell}|\approx t_{\ell}\frac{n}{m_{\ell}}$ for $\ell\in [k-1]$, we get the set sizes $r_{\ell}$ from~\cite{belovs2012kDist}. We will not use these variables, but we note that the values we get for $\{r_{\ell}\}_{\ell=1}^{k-1}$ (from the values of $\{t_{\ell}\}_{\ell=1}^{k-1}$) are the same as those obtained in~\cite{belovs2012kDist}, as our algorithm can be seen as an algorithmic version of the combinatorial construction used in~\cite{belovs2012kDist}. Finally, we choose the number of blocks in each $A_{\ell}$, $m_{\ell}$, so that $m_{\ell}=\Theta(t_{\ell-1})$ for each $\ell\in \{2,\dots,k-1\}$. This ensures that the expected size of ${\cal K}(\overline{R}_1,\dots,\overline{R}_{\ell-1},A_{\ell}^{(j_{\ell})})$ is constant. These values are summarised in \tabl{k-dist-setsizes}.

\begin{table}
	\renewcommand{\arraystretch}{2}
	\centering
	\begin{tabular}{ |r|l|}
		\hline
		$\ell\in\{1,\dots,k-1\}$, $t_{\ell}$ & $=n^{\frac{3}{4}-\frac{1}{4}\frac{1}{2^k-1} - \sum_{\ell'=2}^{\ell}\frac{2^{k-1-\ell'}}{2^k-1}}$\\
		\hline
		$m_1$ & $=\frac{n}{k}$\\
		\hline
		$\ell\in\{2,\dots,k-1\}$ $m_{\ell}$ & $=\Theta(t_{\ell-1})$\\
		\hline
		$c_1$ & $=k-1$\\
		\hline
		$\ell\in\{2,\dots,k-2\}$, $c_{\ell}$ & $=O(1)$ large enough for \cor{ktilde-E}\\
		\hline
		$c_{k-1}$ & $=1$\\
		\hline
		$\ell\in\{2,\dots,k-1\}$, $p_{\ell}$ & $={\sf polylog}(n)$ large enough for \cor{ktilde-E}.\\
		\hline
	\end{tabular}
	\caption{A summary of the (asymptotic) values of variables used in this section.}\label{tab:k-dist-setsizes}
\end{table}

\paragraph{Data:} With any $R$ defined as above, we keep track of some input-dependent data as follows. First, we query everything in $\overline{R}_1$, so we define:
\begin{equation}
\begin{split}
\forall S_1\in 2^{[c_1]}\setminus\{\emptyset\}, D_1(R_1(S_1)) &:= \{(i_1,x_{i_1}): i_1\in \overline{R}_1(S_1)\}\\
D_1(R) &:= \left( D_1(R_1(S_1))\right)_{S_1\in 2^{[c_1]}\setminus\{\emptyset\}}.
\end{split}\label{eq:k-dist-D-1}
\end{equation}

Next, for $\ell\in \{2,\dots,k-1\}$, and $(s_1,\dots,s_{\ell-1},S_{\ell})\in [c_1]\times\dots\times [c_{\ell-1}]\times (2^{[c_{\ell}]}\setminus\{\emptyset\})$, we only query some of the indices in $\overline{R}_{\ell}$, and which ones we query depends on $R$, specifically on $R_1,\dots,R_{\ell-1}$:
\begin{multline}
D_{\ell}(R_{\ell}(s_1,\dots,s_{\ell-1},S_{\ell})|R) := \bigcup_{\substack{S_{\ell-1}\subseteq [c_{\ell-1}]:\\ s_{\ell-1}\in S_{\ell-1}}}
\big\{  
(i_1,\dots,i_{\ell},x_{i_1}):x_{i_{\ell}}=x_{i_1},i_{\ell}\in \overline{R}_{\ell}(s_1,\dots,s_{\ell-1},S_{\ell}),\\
 (i_1,\dots,i_{\ell-1},x_{i_1})\in D_{\ell-1}(R_{\ell-1}(s_1,\dots,s_{\ell-2},S_{\ell-1})|R)  
\big\}.
\label{eq:D-ell}
\end{multline}
We will sometimes omit ``$|R$'' when the context is clear. We can group these together to get:
\begin{equation}
D_{\ell}(R) := \left(D_{\ell}(R_{\ell}(s_1,\dots,s_{\ell-1},S_{\ell}))\right)_{(s_1,\dots,s_{\ell-1},S_{\ell})\in [c_1]\times\dots[c_{\ell-1}]\times (2^{[c_{\ell}]}\setminus\{\emptyset\})}.\label{eq:D-ell-R}
\end{equation}

In addition to this data, we want to keep track of a number for each $j_{\ell}\in R_{\ell}$ that we call the \emph{forward collision degree}. Loosely speaking, for some $i_{\ell}\in \overline{R}_{\ell}$, a forward collision is an element $(i_1,\dots,i_{\ell},\dots,i_{\ell'},x_{i_1})\in D_{\ell'}(R)$, for some $\ell'>\ell$, and some $i_1,\dots,i_{\ell-1},i_{\ell+1},\dots,i_{\ell'}$. This can only exist if $(i_1,\dots,i_{\ell},i_{\ell+1},x_{i_1})\in D_{\ell+1}(R)$, so the \emph{forward collision degree of $i_{\ell}$}, $\bar{d}_{\ell}^{\rightarrow}(i_{\ell})$, counts these. Concretely, for $\ell \in \{1,\dots,k-2\}$ it is defined as:
\begin{equation}
\bar{d}_{R}^{\rightarrow}(i_{\ell}):=\abs{\left\{(i_1,\dots,i_{\ell-1},i_{\ell+1})\in \overline{R}_1\times\dots\times\overline{R}_{\ell-1}\times\overline{R}_{\ell+1} :(i_1,\dots,i_{\ell},i_{\ell+1},x_{i_1})\in D_{\ell+1}(R)\right\}}.\label{eq:k-dist-bar-d}
\end{equation}
For consistency, we also define $\bar{d}_{R}^{\rightarrow}(i_{k-1}):=0$ for $i_{k-1}\in \overline{R}_{k-1}$. 
Then we can define the forward collision degree of $j_{\ell}\in R_{\ell}$ for any $\ell \in \{1,\dots,k-1\}$ as:
\begin{equation}
d_{R}^{\rightarrow}(j_{\ell}):=\sum_{i_{\ell}\in A_{\ell}^{(j_{\ell})}}\bar{d}_{\ell}^{\rightarrow}(i_{\ell}).
\label{eq:forward-collision-degree}
\end{equation}
When our quantum walk removes some $j_{\ell}$ from $R_{\ell}$, we will want to make sure that $d^{\rightarrow}_R(j_{\ell})=0$, because otherwise we will have to uncompute all forward collisions from the data, which could be expensive. Thus, we also keep a database of forward collision degrees:
\begin{equation}
\forall \ell\in\{1,\dots,k-2\},\; C_{\ell}^{\rightarrow}(R) := \{(j_{\ell},d_R^{\rightarrow}(j_{\ell})):j_{\ell}\in R_{\ell},d_{R}^{\rightarrow}(j_{\ell})>0\}.
\label{eq:k-dist-forward-col-db}
\end{equation}

\noindent To summarise, the data we keep track of at a vertex $v_R$ includes:
\begin{equation}
D(R):=(D_1(R),\dots,D_{k-1}(R),C^{\rightarrow}_1(R),\dots,C^{\rightarrow}_{k-2}(R)).
\label{eq:k-dist-data}
\end{equation}

\subsubsection{Intuition about the Combinatorial Structure}

The way we partition each $\overline{R}_{\ell}$ (by partitioning $R_{\ell}$) precisely follows the combinatorial structure of~\cite{belovs2012kDist}, but in this section, we try to give some intuition about why this is done. This section is not technically necessary, and may be skipped without impacting correctness. We will be imprecise for the sake of intuition; for precision see the remainder of this paper. 

\noindent A vertex $v_R$ is labelled by a tuple of tuples of sets: 
\begin{align*}
{R}&=({R}_1,\dots,{R}_{k-1}) \\
&= \left(({R}_1(S_1))_{S_1\in 2^{[c_1]}\setminus\{\emptyset\}},\dots,({R}_{k-1}(s_1,\dots,s_{k-2},S_{k-1}))_{s_1\in [c_1],\dots,s_{k-2}\in [c_{k-2}],S_{k-1}\in 2^{[c_{k-1}]}\setminus\{\emptyset\}}\right).
\end{align*}
Let $\bar{t}_{\ell}:= t_{\ell} c_1c_2\dots c_{\ell-1}(2^{c_\ell}-1)$ be the size of $R_{\ell}$ if each part $R_{\ell}(s_1,\dots,s_{\ell-1},S_{\ell})$ has size $t_{\ell}$.
The set of such vertices where $|R_{\ell}|=\bar{t}_{\ell}$ for all $\ell\in\{1,\dots,k-1\}$ is called $V_0$. 
Starting from such a vertex in $V_0$, we may add an index to $R_1$ to get a vertex where now $|R_1|=\bar{t}_1+1$ -- call the set of such vertices $V_1$. Then we may add something to $R_2$ to get a vertex with $|R_2|=\bar{t}_2+1$ -- call vertices of that form $V_2$. We can continue until we get a vertex where $|R_{\ell}|=\bar{t}_{\ell}+1$ for all $\ell\in \{1,\dots,k-1\}$, the set of which is called $V_{k-1}$. 

Let us give some more detail on the process of moving from $V_0$ to $V_{k-1}$. For any vertex $v_R=v_R^0\in V_0$, we first choose an index $j_1$ to add to $R_1$ (but don't yet add it), to get a vertex $v_{R,j_{1}}^0$ (we call the set of such vertices $V_0^+$). Next, we want to actually add $j_1$ to $R_1$ to get a vertex in $V_1$, but before we can do that, we need to choose \emph{where} in $R_1$ to add $j_1$, so we first choose some non-empty $S_1\subseteq [c_1]$ and then add $j_1$ to $R_1(S_1)$, completing the transition to $V_1$. The new vertex does not remember which $j_1$ was most recently added to $R_1(S_1)$, but it remembers where it was added (i.e.~$S_1$), because it has $|R_1(S_1)|=t_1+1$.

Next, to move from $v_R^1\in V_1$ to some vertex in $V_2$, we again start by choosing the $j_2$ that we will eventually add to $R_2$, to get an intermediate vertex $v_{R,j_2}^1\in V_1^+$. Then we need to choose some $(s_1,S_2)$ and add $j_2$ to $R_2(s_1,S_2)$. We will ensure that we choose $s_1\in S_1$ (we remember $S_1$, because $|R_1(S_1)|=t_1+1$), but we do this deterministically by taking the minimum element of $S_1$. The choice of $S_2$ however is random, which we discuss more below. The reason we choose $s_1\in S_1$ is that in the analysis, we will construct a flow that goes along edges from $V_0$ to $V_1$ that add the unique $j_1$ of the block containing $a_1$; and then along edges that add the unique $j_2$ of the block containing $a_2$, etc. We need $s_1\in S_1$ to ensure that $(a_1,a_2,x_{a_1})\in D_2(R)$, so that this flow eventually goes into the set of vertices in $V_{k-1}$ that not only contain $a_1,\dots,a_{k-1}$, but have noticed that they form a $(k-1)$-collision. 

We continue moving from $V_{\ell}$ to $V_{\ell+1}$, making some choice $S_{\ell+1}$, and adding a new index to $R_{\ell+1}(s_1,\dots,s_{\ell},S_{\ell+1})$ (the $s_1,\dots,s_{\ell}$ are chosen deterministically), so that a vertex in $V_{\ell+1}$ has associated sets $S_1,\dots,S_{\ell+1}$. 
In this way, the choices of sets made in moving from $V_0$ to $V_{k-1}$ give rise to a kind of tree, of depth $k$, with the degree at level $\ell$ being $s^{c_{\ell+1}}-1$. (The nodes of this tree correspond to \emph{sets} of vertices of the graph we're walking on, and the edges of the tree correspond to sets of edges in our graph).

But the choice of $S_{\ell+1}$ requires care, because if, in our quantum walk, we add $j_{\ell+1}$ to a bad choice of $R_{\ell+1}(s_1,\dots,s_{\ell},S_{\ell+1})$, we might find that we have introduced a \emph{fault}, rendering our data incorrect. It turns out that we can avoid faults precisely by making a choice of ${S}_{\ell+1}$ that avoids a certain set ${\cal I}={\cal I}(v_{R,j_{\ell+1}}^{\ell})\subset [c_{\ell+1}]$ -- this ensures that if we add an index to $\overline{R}_{\ell+1}(s_1,\dots,s_{\ell},S_{\ell+1})$ that has a match in some part $\overline{R}_{\ell+2}(s_1,\dots,s_{\ell+1},S_{\ell+2})$, then $\overline{R}_{\ell+2}(s_1,\dots,s_{\ell+1},S_{\ell+2})$ is exempt from requiring matches with $\overline{R}_{\ell+1}(s_1,\dots,s_{\ell},S_{\ell+1})$ to be queried (because $s_{\ell+1}\not\in S_{\ell+1}$). 

Since we don't know ${\cal I}$, we use alternative neighbourhoods for all the different possibilities ${\cal I}\subsetneq [c_{\ell+1}]$. 
The resulting alternative neighbourhoods for $v_{R,j_{\ell+1}}^{\ell}$ are as follows: They all have a backwards edge to $v_{R}^{\ell}$. The possible sets of forward edges are those labelled by $S_{\ell+1}\in L^+(v_{R,j_{\ell+1}}^\ell)=2^{[c_{\ell+1}]}\setminus\emptyset$ such that $S_{\ell+1}\cap {\cal I} = \emptyset$ -- that is non-empty $S_{\ell+1}\subset [c_{\ell+1}]\setminus {\cal I}$.\footnote{It's possible that ${\cal I}=[c_{\ell+1}]$, but this is sufficiently unlikely that we can treat this case separately.} Below are just some of the alternative neighbourhoods for $c_{\ell+1}=3$, for the choices ${\cal I}=\emptyset$, ${\cal I}=\{1\}$ and ${\cal I}=\{3\}$. Dotted lines indicate missing edges (or edges of weight 0). 
\begin{center}
    
\begin{tikzpicture}[scale=1.5]

\node at (0,0) {
\begin{tikzpicture}[scale=1.5]
\draw[{Latex[length=2mm, width=2mm]}-] (0,0)--(1.5,0);		\draw[-{Latex[length=2mm, width=2mm]}] (1.5,0)-- +(90:1.5);
										\draw[-{Latex[length=2mm, width=2mm]}] (1.5,0)-- +(60:1.5);
										\draw[-{Latex[length=2mm, width=2mm]}] (1.5,0)-- +(30:1.5); 
										\draw[-{Latex[length=2mm, width=2mm]}] (1.5,0)-- +(0:1.5);
										\draw[-{Latex[length=2mm, width=2mm]}] (1.5,0)-- +(-30:1.5);
										\draw[-{Latex[length=2mm, width=2mm]}] (1.5,0)-- +(-60:1.5);
										\draw[-{Latex[length=2mm, width=2mm]}] (1.5,0)-- +(-90:1.5);
										
\filldraw (-.05,0) circle (.05);		\filldraw (1.5,0) circle (.08);	\filldraw (1.5,0)+(90:1.55) circle (.05);
										\filldraw (1.5,0)+(60:1.55) circle (.05);
										\filldraw (1.5,0)+(30:1.55) circle (.05);
										\filldraw (1.5,0)+(0:1.55) circle (.05);														\filldraw (1.5,0)+(-30:1.55) circle (.05);													\filldraw (1.5,0)+(-60:1.55) circle (.05);													\filldraw (1.5,0)+(-90:1.55) circle (.05);
\node at (-.3,.2) {$v_{R}^{\ell}$};
\node at (.9,.25) {$v_{R,j_{\ell+1}}^{\ell}$};

\node[rotate=90] at (1.35,.6) {\color{blue}\tiny ${}_{\{1\}}$};
\node[rotate=60] at (1.65,.5) {\color{blue}\tiny ${}_{\{2\}}$};
\node[rotate=30] at (2,.4) {\color{blue}\tiny ${}_{\{1,2\}}$};
\node at (2.05,.1) {\color{blue}\tiny ${}_{\{3\}}$};
\node[rotate=-30] at (2.05,-.2) {\color{blue}\tiny ${}_{\{1,3\}}$};
\node[rotate=-60] at (1.9,-.5) {\color{blue}\tiny ${}_{\{2,3\}}$};
\node[rotate=-90] at (1.6,-.8) {\color{blue}\tiny ${}_{\{1,2,3\}}$};

\end{tikzpicture}
};

\node at (4.5,0) {
\begin{tikzpicture}[scale=1.5]
\draw[{Latex[length=2mm, width=2mm]}-] (0,0)--(1.5,0);		\draw[dashed] (1.5,0)-- +(90:1.5);
										\draw[-{Latex[length=2mm, width=2mm]}] (1.5,0)-- +(60:1.5);
										\draw[dashed] (1.5,0)-- +(30:1.5); 
										\draw[-{Latex[length=2mm, width=2mm]}] (1.5,0)-- +(0:1.5);
										\draw[dashed] (1.5,0)-- +(-30:1.5);
										\draw[-{Latex[length=2mm, width=2mm]}] (1.5,0)-- +(-60:1.5);
										\draw[dashed] (1.5,0)-- +(-90:1.5);
										
\filldraw (-.05,0) circle (.05);		\filldraw (1.5,0) circle (.08);	
										\filldraw (1.5,0)+(60:1.55) circle (.05);
										\filldraw (1.5,0)+(0:1.55) circle (.05);		
										\filldraw (1.5,0)+(-60:1.55) circle (.05);	
\node at (-.3,.2) {$v_{R}^{\ell}$};
\node at (.9,.25) {$v_{R,j_{\ell+1}}^{\ell}$};

\node[rotate=90] at (1.35,.6) {\color{blue}\tiny ${}_{\{1\}}$};
\node[rotate=60] at (1.65,.5) {\color{blue}\tiny ${}_{\{2\}}$};
\node[rotate=30] at (2,.4) {\color{blue}\tiny ${}_{\{1,2\}}$};
\node at (2.05,.1) {\color{blue}\tiny ${}_{\{3\}}$};
\node[rotate=-30] at (2.05,-.2) {\color{blue}\tiny ${}_{\{1,3\}}$};
\node[rotate=-60] at (1.9,-.5) {\color{blue}\tiny ${}_{\{2,3\}}$};
\node[rotate=-90] at (1.6,-.8) {\color{blue}\tiny ${}_{\{1,2,3\}}$};

\end{tikzpicture}
};

\node at (9,0) {
\begin{tikzpicture}[scale=1.5]
\draw[{Latex[length=2mm, width=2mm]}-] (0,0)--(1.5,0);		\draw[-{Latex[length=2mm, width=2mm]}] (1.5,0)-- +(90:1.5);
										\draw[-{Latex[length=2mm, width=2mm]}] (1.5,0)-- +(60:1.5);
										\draw[-{Latex[length=2mm, width=2mm]}] (1.5,0)-- +(30:1.5); 
										\draw[dashed] (1.5,0)-- +(0:1.5);
										\draw[dashed] (1.5,0)-- +(-30:1.5);
										\draw[dashed] (1.5,0)-- +(-60:1.5);
										\draw[dashed] (1.5,0)-- +(-90:1.5);
										
\filldraw (-.05,0) circle (.05);		\filldraw (1.5,0) circle (.08);	\filldraw (1.5,0)+(90:1.55) circle (.05);
										\filldraw (1.5,0)+(60:1.55) circle (.05);
										\filldraw (1.5,0)+(30:1.55) circle (.05);
\node at (-.3,.2) {$v_{R}^{\ell}$};
\node at (.9,.25) {$v_{R,j_{\ell+1}}^{\ell}$};

\node[rotate=90] at (1.35,.6) {\color{blue}\tiny ${}_{\{1\}}$};
\node[rotate=60] at (1.65,.5) {\color{blue}\tiny ${}_{\{2\}}$};
\node[rotate=30] at (2,.4) {\color{blue}\tiny ${}_{\{1,2\}}$};
\node at (2.05,.1) {\color{blue}\tiny ${}_{\{3\}}$};
\node[rotate=-30] at (2.05,-.2) {\color{blue}\tiny ${}_{\{1,3\}}$};
\node[rotate=-60] at (1.9,-.5) {\color{blue}\tiny ${}_{\{2,3\}}$};
\node[rotate=-90] at (1.6,-.8) {\color{blue}\tiny ${}_{\{1,2,3\}}$};

\end{tikzpicture}
};
\end{tikzpicture}

\end{center}

The reason we have done this in such an involved way is that we need to be able to design a flow that is orthogonal to \emph{all} of these stars. Later, we will see that such a flow exists. Up to scaling by some positive real number, the flow is a sign that depends on the sizes of the sets $S_1,\dots,S_{\ell+1}$ that have been chosen so far, as shown in the following figure:
\begin{center}
\begin{tikzpicture}[scale=1.4]
\node at (0,0) {
\begin{tikzpicture}[scale=1.4]
\draw[->] (-2,0)--(-.25,0); \draw (-2,0)--(1.5,0);
		\draw[->] (1.5,0)-- +(90:1); \draw (1.5,0)-- +(90:1.5);
										\draw[->] (1.5,0)-- +(60:1); \draw (1.5,0)-- +(60:1.5);
										\draw[-<] (1.5,0)-- +(30:1); \draw[<-] (1.5,0)-- +(30:1.5); 
										\draw[->] (1.5,0)-- +(0:1); \draw (1.5,0)-- +(0:1.5);
										\draw[->] (1.5,0)-- +(-30:1); \draw (1.5,0)-- +(-30:1.5);
										\draw[->] (1.5,0)-- +(-60:1); \draw (1.5,0)-- +(-60:1.5);
										\draw[-<] (1.5,0)-- +(-90:1.25); \draw (1.5,0)-- +(-90:1.5);
										
\filldraw (-2.05,0) circle (.05);		\filldraw (1.5,0) circle (.08);	\filldraw (1.5,0)+(90:1.55) circle (.05);
										\filldraw (1.5,0)+(60:1.55) circle (.05);
										\filldraw (1.5,0)+(30:1.55) circle (.05);
										\filldraw (1.5,0)+(0:1.55) circle (.05);														\filldraw (1.5,0)+(-30:1.55) circle (.05);													\filldraw (1.5,0)+(-60:1.55) circle (.05);													\filldraw (1.5,0)+(-90:1.55) circle (.05);
\node at (-.25,-.3) {$(-1)^{|S_1|+\dots+|S_{\ell}|+\ell}$};

\node at (4.5,-.3) {$(-1)^{|S_1|+\dots+|S_{\ell}|+|{\color{blue}S_{\ell+1}}|+\ell+1}$};

\node at (-2.3,.2) {$v_{R}^{\ell}$};
\node at (.9,.3) {$v_{R,j_{\ell+1}}^{\ell}$};

\node[rotate=90] at (1.35,.6) {\color{blue}\tiny ${}_{\{1\}}$};
\node[rotate=60] at (1.65,.5) {\color{blue}\tiny ${}_{\{2\}}$};
\node[rotate=30] at (2,.4) {\color{blue}\tiny ${}_{\{1,2\}}$};
\node at (2.05,.1) {\color{blue}\tiny ${}_{\{3\}}$};
\node[rotate=-30] at (2.05,-.2) {\color{blue}\tiny ${}_{\{1,3\}}$};
\node[rotate=-60] at (1.9,-.5) {\color{blue}\tiny ${}_{\{2,3\}}$};
\node[rotate=-90] at (1.6,-.8) {\color{blue}\tiny ${}_{\{1,2,3\}}$};

\end{tikzpicture}
};

\node at (7,0) {
\begin{tikzpicture}[scale=1.4]
\draw[->] (0,0)--(.75,0); \draw (0,0)--(1.5,0);
		\draw[->] (1.5,0)-- +(90:1); \draw (1.5,0)-- +(90:1.5);
										\draw[->] (1.5,0)-- +(60:1); \draw (1.5,0)-- +(60:1.5);
										\draw[-<] (1.5,0)-- +(30:1); \draw[<-] (1.5,0)-- +(30:1.5); 
										\draw[->] (1.5,0)-- +(0:1); \draw (1.5,0)-- +(0:1.5);
										\draw[->] (1.5,0)-- +(-30:1); \draw (1.5,0)-- +(-30:1.5);
										\draw[->] (1.5,0)-- +(-60:1); \draw (1.5,0)-- +(-60:1.5);
										\draw[-<] (1.5,0)-- +(-90:1.25); \draw (1.5,0)-- +(-90:1.5);
										
\filldraw (-.05,0) circle (.05);		\filldraw (1.5,0) circle (.08);	\filldraw (1.5,0)+(90:1.55) circle (.05);
										\filldraw (1.5,0)+(60:1.55) circle (.05);
										\filldraw (1.5,0)+(30:1.55) circle (.05);
										\filldraw (1.5,0)+(0:1.55) circle (.05);														\filldraw (1.5,0)+(-30:1.55) circle (.05);													\filldraw (1.5,0)+(-60:1.55) circle (.05);													\filldraw (1.5,0)+(-90:1.55) circle (.05);
\node at (.75,-.25) {$1$};

\node at (4,-.3) {$(-1)^{|{\color{blue}S_{\ell+1}}|+1}$};

\node at (-.3,.2) {$v_{R}^{\ell}$};
\node at (.9,.3) {$v_{R,j_{\ell+1}}^{\ell}$};

\node[rotate=90] at (1.35,.6) {\color{blue}\tiny ${}_{\{1\}}$};
\node[rotate=60] at (1.65,.5) {\color{blue}\tiny ${}_{\{2\}}$};
\node[rotate=30] at (2,.4) {\color{blue}\tiny ${}_{\{1,2\}}$};
\node at (2.05,.1) {\color{blue}\tiny ${}_{\{3\}}$};
\node[rotate=-30] at (2.05,-.2) {\color{blue}\tiny ${}_{\{1,3\}}$};
\node[rotate=-60] at (1.9,-.5) {\color{blue}\tiny ${}_{\{2,3\}}$};
\node[rotate=-90] at (1.6,-.8) {\color{blue}\tiny ${}_{\{1,2,3\}}$};

\end{tikzpicture}
};
\end{tikzpicture}
\end{center}
An edge labelled by the set ${\color{blue}S_{\ell+1}}$ has flow $(-1)^{|S_1|+\dots+|S_{\ell}|+|{\color{blue}S_{\ell+1}}|+\ell+1}$, also indicated by the direction of the edge. 
Up to scaling by $(-1)^{|S_1|+\dots+|S_{\ell}|+\ell}$, the left-hand side is the same as the simplified picture on the right-hand-side above. 

Consider the inner product of this flow state and a star state for a set ${\cal I}$. The incoming flow from $v_R^\ell$ always contributes $(-1)$ (because in the star it's pointing out, and recall that switching edge direction switches the sign). For every non-empty $S_{\ell+1}\subset [c_{\ell+1}]\setminus {\cal I}$, the corresponding edge
contributes 1 if it has outgoing flow (direction of flow matches the star state), and $(-1)$ if it has incoming flow (direction of edge is opposite to star state). 
Then the inner product is:
$$-1 + \sum_{S_{\ell+1}\in 2^{[c_{\ell+1}]\setminus{\cal I}}\setminus\{\emptyset\}}(-1)^{|S_{\ell+1}|+1} = 0,$$
so flow is conserved with respect to all possible alternative neighbourhoods. This is precisely what we need, and why we use this complex combinatorial structure.

\subsubsection{The Graph: Vertex Sets}\label{sec:k-dist-G-vertices}

To define $G$, we begin by defining disjoint vertex sets $V_0$, $V_0^+$, $(V_{\ell})_{\ell=1}^{k-1}$, $(V_{\ell}^+)_{\ell=1}^{k-2}$, and $V_k$, whose union makes up $V(G)$. We will use the notation in \eq{weird-binom} and \eq{weird-binom+} for tuples of disjoint sets throughout this section. \tabl{k-dist-index-sets} summarises $G$.

\paragraph{$V_0$:} We define
\begin{multline}
	V_0=\Bigg\{v^0_{R_1,\dots,R_{k-1}}:=(0,R_1,\dots,R_{k-1},D(R_1,\dots,R_{k-1})): R_{\ell} \in \binom{[m_{\ell}]}{t_{\ell}^{(c_1\cdots c_{\ell-1}(2^{c_{\ell}}-1))}}\Bigg\}.\label{eq:k-dist-V1}
\end{multline}

Our initial distribution is uniform on $V_0$: $\sigma(u)=\frac{1}{|V_0|}$ for all $u\in V_0$. We implicitly store all sets including those making up $R_1,\dots,R_{k-1}$ and $D(R_1,\dots,R_{k-1})$ in a data structure with the properties described in \sec{data}. This will only be important when we analyse the time complexity of the setup and transition subroutines.

\paragraph{$V_0^+$:} At a vertex in $V_0^+$, we suppose we have chosen a new element $j_1$ to add to $R_1$, but not yet added it. Thus, we label such a vertex by a tuple of sets $R$, and an index $j_1\not\in R_1$: 
\begin{align*}
	V_0^+ & :=\Big\{v^{0}_{R_1,\dots,R_{k-1},j_1}:=((0,+),R_1,\dots,R_{k-1},D(R_1,\dots,R_{k-1}),j_1): 
	 v^0_{R_1,\dots,R_{k-1}} \in V_0, j_1\in [m_1]\setminus R_1\Big\},
\end{align*}
\begin{equation}
\begin{split}
\mbox{so } 
	\abs{V_0^+} &= \abs{V_0}\abs{[m_1] \setminus R_1} = O\left(n\abs{V_0}\right).
\end{split}\label{eq:kV1}
\end{equation}

\paragraph{$V_{\ell}$ for $\ell\in\{1,\dots,k-1\}$:} At a vertex in $V_{\ell}$, we suppose we have added a new element to each of $R_1,\dots,R_{\ell}$, meaning that for each $\ell'\in[\ell]$, there is some $(s_1,\dots,s_{\ell'-1},S_{\ell'})\in [c_1]\times\dots\times [c_{\ell'-1}]\times (2^{[c_{\ell'}]}\setminus\{\emptyset\})$ such that $|R_{\ell'}(s_1,\dots,s_{\ell'-1},S_{\ell'})|=t_{\ell'}+1$. However, we will not let the choices of $s_1,\dots,s_{\ell'-1}$ for different $\ell'$ be arbitrary. Instead, we define the following sets of vertices, for $(S_1,\dots,S_{\ell})\in (2^{[c_1]}\setminus\{\emptyset\})\times\dots\times(2^{[c_\ell]}\setminus\{\emptyset\})$, where $\mu(S)$ denotes the minimum element of a set $S$:
\begin{multline}
\!\!V_{\ell}(S_1,\dots,S_{\ell}) := \Big\{ v^{\ell}_{R}=(\ell,R,D(R)):
\forall \ell'\in [\ell], R_{\ell'}\in\binom{[m_{\ell}]}{t_\ell^{(c_1\dots c_{\ell-1}(2^{c_{\ell}}-1))}}^+\!;\\
\qquad\qquad\quad\forall \ell'\in\{\ell+1,\dots,k-1\}, R_{\ell'}\in\binom{[m_{\ell}]}{t_\ell^{(c_1\dots c_{\ell-1}(2^{c_{\ell}}-1))}};\\
\forall \ell'\in\{1,\dots,\ell\},\; |R_{\ell'}(\mu(S_1),\dots,\mu(S_{\ell'-1}),S_{\ell'})|=t_{\ell'}+1
\Big\}.\label{eq:V-ell-Ss}
\end{multline}
This is the set of vertices labelled by sets $R$ where we have added elements to each of $R_1,\dots,R_{\ell}$, not yet added elements to $R_{\ell+1},\dots,R_{k-1}$, and for $\ell'\in\{1,\dots,\ell\}$, the choice of \emph{where} the new element was added to $R_{\ell'}$ is determined by $S_1,\dots,S_{\ell}$. 
Then we can define:
\begin{equation}
V_{\ell} := \bigcup_{(S_1,\dots,S_{\ell})\in (2^{[c_1]}\setminus\{\emptyset\})\times\dots\times(2^{[c_\ell]}\setminus\{\emptyset\})} V_{\ell}(S_1,\dots,S_{\ell}).\label{eq:V-ell-union}
\end{equation}
Using the fact that for all $\ell'\in\{2,\dots,k-2\}$, $m_{\ell'}=\Theta(t_{\ell'-1})$, we have:
\begin{equation}
\abs{V_{\ell}} = O\left(\abs{V_0} \prod_{\ell'=1}^{\ell}\frac{m_{\ell'}}{t_{\ell'}} \right)=O\left(\abs{V_0} \prod_{\ell'=1}^{\ell}\frac{t_{\ell'-1}}{t_{\ell'}} \right)=O\left(\frac{n}{t_{\ell}}\abs{V_0} \right).\label{eq:kVk0}
\end{equation}

\paragraph{$V_{\ell}^+:$ for $\ell\in\{1,\dots,k-2\}$:} At a vertex in $V_{\ell}^+$, we suppose, as in $V_{\ell}$, that we have already added an element to each of the sets $R_1,\dots,R_{\ell}$, but now have also selected an element $j_{\ell+1}\in [m_{\ell+1}]$ to add to $R_{\ell+1}$:

\begin{align*}
	V_{\ell}^+(S_1,\dots,S_{\ell})&:=\Big\{v^{\ell}_{R,j_{\ell+1}}:=((\ell,+),R,D(R),j_{\ell+1}): 
	 v^{\ell}_{R} \in V_{\ell}(S_1,\dots,S_{\ell}), j_{\ell+1}\in [m_{\ell+1}]\setminus R_{\ell+1}\Big\}\\
V_{\ell}^+ &:= \bigcup_{(S_1,\dots,S_{\ell})\in (2^{[c_1]}\setminus\{\emptyset\})\times\dots\times(2^{[c_\ell]}\setminus\{\emptyset\})} V_{\ell}^+(S_1,\dots,S_{\ell}),
\end{align*}
\noindent so together with \eq{kVk0} and $m_{\ell+1}=\Theta(t_{\ell})$, this implies
\begin{equation}
	\abs{V_{\ell}^+} = \abs{V_{\ell}}\abs{[m_{\ell+1}]\setminus T_{\ell+1}} = O\left(n\abs{V_0}\right).\label{eq:kVk1}
\end{equation}

\paragraph{The Final Stage, $V_k$:} At a vertex in $V_k$, we have added a new element to each of $R_1,\dots,R_{k-1}$, as in $V_{k-1}$, and also selected some $i_k\in A_k$, which we can view as a candidate for completing one of the $(k-1)$-collisions in $D_{k-1}(R)$ to a $k$-collision:
\begin{equation}
\begin{split}
	V_k&:=\big\{v^{k}_{R,i_k}:=(k,R,D(R),i_k): 
	v^{k-1}_{R} \in V_{k-1}, i_k\in A_k.\big\},\\
\mbox{so }
	\abs{V_k} &= \abs{V_{k-1}}\abs{A_k} = O\left(\frac{n^2}{t_{k-1}}\abs{V_0}\right).
\end{split}\label{eq:kVk}
\end{equation}

\subsubsection{The Graph: Edge Sets}\label{sec:k-dist-G-edges}

We now define the sets of edges that make up $\overrightarrow{E}(G)$, as well as the edge label sets $L(u)$ (see \defin{QW-access}) for each $u\in V(G)$. These are also summarised in \tabl{k-dist-index-sets}.

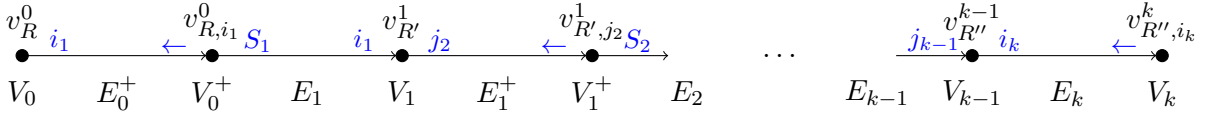
\begin{figure}
\centering
\begin{tikzpicture}[scale=1.25]
\filldraw (0,0) circle (.08);	\draw[-{Latex[length=2mm, width=2mm]}] (0,0) -- (2.42,0);
\filldraw (2.5,0) circle (.08);	\draw[-{Latex[length=2mm, width=2mm]}] (2.5,0) -- (4.92,0);
\filldraw (5,0) circle (.08);	\draw[-{Latex[length=2mm, width=2mm]}] (5,0) -- (7.42,0);
\filldraw (7.5,0) circle (.08);	\draw[-{Latex[length=2mm, width=2mm]}] (7.5,0) -- (8.5,0);
					\draw[-{Latex[length=2mm, width=2mm]}] (11.5,0) -- (12.42,0);
\filldraw (12.5,0) circle (.08);	\draw[-{Latex[length=2mm, width=2mm]}] (12.5,0) -- (14.92,0);
\filldraw (15,0) circle (.08);	

\node at (0,.45) {$v^0_R$};
\node at (0,-.5) {$V_0$};

	\node at (.5,.2) {\small\color{blue}$j_1$};
	\node at (1.25,-.5) {$E_0^+$};
	\node at (2,.1) {\small\color{blue}$\leftarrow$};

\node at (2.5,.45) {$v^0_{R,j_1}$};
\node at (2.5,-.5) {$V_0^+$};

	\node at (3.1,.2) {\small\color{blue}$S_1$};
	\node at (3.75,-.5) {$E_1$};
	\node at (4.5,.2) {\small\color{blue}$j_1$};

\node at (5,.45) {$v^1_{R'}$};
\node at (5,-.5) {$V_1$};

	\node at (5.5,.2) {\small\color{blue}$j_2$};
	\node at (6.25,-.5) {$E_1^+$};
	\node at (7,.1) {\small\color{blue}$\leftarrow$};

\node at (7.5,.45) {$v^1_{R',j_2}$};
\node at (7.5,-.5) {$V_1^+$};

	\node at (8.1,.2) {\small\color{blue}$S_2$};
	\node at (8.75,-.5) {$E_2$};

\node at (10,0) {$\dots$};

	\node at (11.25,-.5) {$E_{k-1}$};
	\node at (12,.2) {\small\color{blue}$j_{k-1}$};

\node at (12.5,.45) {$v^{k-1}_{R''}$};
\node at (12.5,-.5) {$V_{k-1}$};

	\node at (13,.2) {\small\color{blue}$i_k$};
	\node at (13.75,-.5) {$E_k$};
	\node at (14.5,.1) {\small\color{blue}$\leftarrow$};

\node at (15,.45) {$v^k_{R'',i_k}$};
\node at (15,-.5) {$V_k$};
\end{tikzpicture}
\caption{A path from $V_0$ to $V_k$, with edge labels shown in blue. $R'$ is obtained from $R$ by inserting $j_1$ into $R_1(S_1)$. $R''$ is obtained from $R'$ by inserting $j_2$ into $R_2(\mu(S_1),S_2)$, and for some choice of $S_3,\dots,S_{k-1}$, inserting, for each $\ell\in \{3,\dots,k-1\}$, some $j_{\ell}$ into $R_{\ell}(\mu(S_1),\dots,\mu(S_{\ell-1}),S_{\ell})$.}

\end{figure}

\begin{table}
\renewcommand{\arraystretch}{1.3}
\begin{tabular}{|c|c|c|c|c|}
\hline
$u$ & $j\in L^-(u)$ & $f^-_u(j)$ & $i\in L^+(u)$ & $f^+_u(i)$\\
\hline
\hline
$v_R^0\in V_0$ & $\emptyset$ & & $j_1\in [m_1]\setminus R_1$ & $v_{R,j_1}^0$\\ 
\hline
$v_{R,j_1}^0\in V_0^+$ & $\leftarrow$ & $v_R^0$ & $S_1\in 2^{[c_1]}\setminus\{\emptyset\}$ & $v_{R^{S_1 \leftarrow j_1}}^1$\\
\hline
$v_{R}^{\ell}\in V_{\ell}(S)$ & $j_{\ell}\in R_{\ell}(\hat\mu(S)):d_R^{\rightarrow}(j_{\ell})=0$ & $v_{R\setminus\{j_{\ell}\},j_{\ell}}^{\ell-1}$ & $j_{\ell+1}\in [m_{\ell+1}]\setminus R_{\ell+1}$ & $v_{R,j_{\ell+1}}^{\ell}$\\ 
\hline
$v_{R,j_{\ell+1}}^{\ell}\in V_{\ell}^+$ & $\leftarrow$ & $v_{R}^{\ell}$ & $S_{\ell+1}\in 2^{[c_{\ell+1}]}\setminus\{\emptyset\}$ & $v_{R^{S_{\ell+1}\leftarrow j_{\ell+1}}}^{\ell+1}$\\
\hline
$v_{R}^{k-1}\in V_{k-1}(S)$ & $j_{k-1}\in R_{k-1}(\hat\mu(S))$ & $v_{R\setminus\{j_{k-1}\},j_{k-1}}^{k-2}$ & $i_k\in A_k$ & $v_{R,i_k}^k$\\
\hline
$v_{R,i_k}^k\in V_k$ & $\leftarrow$ & $v_R^{k-1}$ & $\emptyset$ & \\
\hline
\end{tabular}
\caption{The sets labeling incoming ($L^-$) and outgoing ($L^+$) edges of each vertex $u\in V(G)$, and the neighbouring vertices at the end of every such edge. $\ell\in\{1,\dots,k-2\}$, $S=(S_1,\dots,S_{\ell})$, and for brevity we use $\hat\mu(S):=(\mu(S_1),\dots,\mu(S_{\ell-1}),S_{\ell})$, where $\mu$ is the minimum. 
$R^{S_1\leftarrow j_1}$ is obtained from $R$ by inserting $j_1$ into $R_1(S_1)$, and for $v^{\ell}_{R,j_{\ell+1}}\in V_{\ell+1}^+(S)$, $R^{S_{\ell+1}\leftarrow j_{\ell+1}}$ is obtained from $R$ by inserting $j_{\ell+1}$ into $R_{\ell+1}(\hat\mu(S))$. To ensure that $L^-(u)$ and $L^+(u)$ are always disjoint, we implicitly append a $\leftarrow$ label to all of $L^-(u)$ and a $\rightarrow$ label to all of $L^+(u)$. }\label{tab:k-dist-index-sets}
\end{table}

\paragraph{$E_{\ell}^+\subset V_{\ell}\times V_{\ell}^+$ for $\ell \in \{0,\dots,k-2\}$:} There is an edge between $v^{\ell}_{R}\in V_{\ell}$ and $v^{\ell}_{R,j_{\ell+1}} \in V_{\ell}^+$ for any $j_{\ell+1}\in [m_{\ell+1}]\setminus R_{\ell+1}$, so we define
$$L^+(v^{\ell}_R):=[m_{\ell+1}]\setminus R_{\ell+1}
\mbox{ and }
L^-(v^{\ell}_{R,j_{\ell+1}}):=\{\leftarrow\},$$
and let $f^+_{v_R^{\ell}}(j_{\ell+1})=v_{R,j_{\ell+1}}^{\ell}$ and $f^-_{v_{R,j_{\ell+1}}^{\ell}}(\leftarrow)=v^{\ell}_R$. We let $E_{\ell}^+$ be the set of all such edges
$$E_{\ell}^{+}:=\left\{(v^{\ell}_{R},v^{\ell}_{R,j_{\ell+1}}):v^{\ell}_{R}\in V_{\ell}, j_{\ell+1}\in [m_{\ell+1}]\setminus R_{\ell+1}\right\}$$
and set $\w_e=\w_{\ell}^+=1$ for all $e\in E_{\ell}^+$.
This together with \eq{kVk1} implies that 
\begin{equation}
	\abs{E_{\ell}^+} = \abs{V_{\ell}^+} = O\left(n\abs{V_0}\right).\label{eq:kEk0}
\end{equation}

\paragraph{Faults:} Fix $\ell\in\{1,\dots,k-1\}$. As in the case of 3-distinctness, if we add a new block index $j_{\ell}$ to certain parts of $R_{\ell}$, to get $R'$, such that $d_{R'}^{\rightarrow}(j_{\ell})>0$, this introduces a \emph{fault} in the data, which our quantum walk will want to avoid. The case for $k>3$ is slightly more complicated, so we examine exactly when a fault is introduced before describing the remaining edge sets. 

Suppose $v_R^{\ell-1}\in V_{\ell-1}(S_1^*,\dots,S_{\ell-1}^*)$ (see \eq{V-ell-Ss}) and we add some $j_{\ell}$ to $R_{\ell}(\mu(S_1^*),\dots,\mu(S_{\ell-1}^*),S_{\ell})$, for some $S_{\ell}\subseteq [c_{\ell}]$. 
For $\ell\in\{2,\dots,k-2\}$, this introduces a fault if the following conditions are satisfied, by some $i_{\ell}\in A_{\ell}^{(j_{\ell})}$, which is added to $\overline{R}_{\ell}(\mu(S_1^*),\dots,\mu(S_{\ell-1}^*),S_{\ell})$, (we use $[{\cal E}]$ to denote the logical value of an event ${\cal E}$):
\begin{equation}
\begin{split}
\mathbf{C}^{\leftarrow}(i_{\ell},R,S_{\ell})&:= \big[\exists  (i_1,\dots,i_{\ell-1})\in R_1\times \overline{R}_2\times \dots \times \overline{R}_{\ell-1}\mbox{ s.t. }\\
& \qquad\qquad\qquad\qquad\qquad\quad\; (i_1,\dots,i_{\ell-1},i_{\ell},x_{i_1})\in D_{\ell}(R_{\ell}(\mu(S_1^*),\dots,\mu(S_{\ell-1}^*),S_{\ell}))\big]\\
\mathbf{C}^{\rightarrow}(i_{\ell},R,S_{\ell})&:= \Bigg[\exists s_{\ell}\in S_{\ell} \mbox{ s.t. }\\
&\qquad\quad\exists i_{\ell+1}\in \bigcup_{S_{\ell+1}\in 2^{[c_{\ell+1}]}\setminus\{\emptyset\}}\overline{R}_{\ell+1}(\mu(S_1^*),\dots,\mu(S_{\ell-1}^*),s_{\ell},S_{\ell+1}) \mbox{ s.t. }x_{i_{\ell+1}}=x_{i_{\ell}}\Bigg].
\end{split}\label{eq:k-dist-Carrows}
\end{equation}
In words $\mathbf{C}^{\leftarrow}$ is the condition that $i_{\ell}$ forms a collision $(i_1,\dots,i_{\ell},x_{i_1})$ that would be stored in $D_{\ell}(R)$, and $\mathbf{C}^{\rightarrow}$ is the condition that $i_{\ell}$ collides with something in $\overline{R}_{\ell+1}$ such that if $\mathbf{C}^{\rightarrow}$ holds, $(i_1,\dots,i_{\ell+1},x_{i_1})$ would be stored in $D_{\ell+1}(R)$. For $\ell=1$, $\mathbf{C}^{\rightarrow}$ is also defined, and $i_1$ introduces a fault whenever $\mathbf{C}^{\rightarrow}$ is true. For $\ell=k-1$, $\mathbf{C}^{\rightarrow}$ can never be true, so there is never a fault. We set $c_{k-1}=1$ (see \tabl{k-dist-setsizes}).

Then for any $\ell\in \{1,\dots,k-2\}$, $v_R^{\ell-1}\in V_{\ell-1}(S_1^*,\dots,S_{\ell-1}^*)$, $i_{\ell}\in A_{\ell}\setminus\overline{R}_{\ell}$, and $S_{\ell}\in 2^{[c_{\ell}]}\setminus\{\emptyset\}$, condition $\mathbf{C}^{\rightarrow}$ is false if and only if $S_{\ell}$ is disjoint from the following set:
\begin{equation}
{\cal I}(v^{\ell-1}_R,i_{\ell}):=\left\{s_{\ell}\in [c_{\ell}]: \exists i_{\ell+1}\in\!\!\!\!\!\! \bigcup_{S_{\ell+1}\in 2^{[c_{\ell+1}]}\setminus\{\emptyset\}}\!\!\!\!\!\!\!\overline{R}_{\ell+1}(\mu(S_1^*),\dots,\mu(S_{\ell-1}^*),s_{\ell},S_{\ell+1})\mbox{ s.t. }x_{i_{\ell+1}}=x_{i_{\ell}}\right\}.\label{eq:k-dist-I-i}
\end{equation}
For $\ell=k-1$, we define ${\cal I}(v^{k-2}_R,i_{k-1}):=\emptyset$. When $\ell=1$ we can define, for $v_{R,j_1}^0\in V_0^+$:
\begin{equation}
	{\cal I}(v^{0}_{R,j_{1}}):=\bigcup_{i_{1}\in A_{1}^{(j_{1})}}{\cal I}(v^{0}_R,i_{1})\label{eq:k-dist-cal-I-1}
\end{equation}
As long as we choose some $S_1$ that avoids this set, we will not introduce a fault. For $\ell>1$,
examining condition \textbf{C}$^{\leftarrow}$ above, although it appears to depend on $S_{\ell}$, it does not. Referring to \eq{D-ell}, we can rewrite \textbf{C}$^{\leftarrow}$ as:
\begin{multline*}
\mathbf{C}^{\leftarrow}(i_{\ell},R,S_{\ell})\Leftrightarrow \mathbf{C}^{\leftarrow}(i_{\ell},R):= \bigg[\exists S_{\ell-1}\subseteq [c_{\ell-1}] \mbox{ s.t. }\mu(S_{\ell-1}^*)\in S_{\ell-1},\\
\exists (i_1,\dots,i_{\ell-1},x_{i_1})\in D_{\ell-1}(R_{\ell-1}(\mu(S_1^*),\dots,\mu(S_{\ell-2}^*),S_{\ell-1}))
\mbox{ s.t. }x_{i_{\ell}}=x_{i_1}\bigg].\label{eq:k-dist-Carrow}
\end{multline*}
Thus, for $\ell\in\{2,\dots,k-2\}$, for any $v^{\ell-1}_{R,j_{\ell}}\in V_{\ell-1}^+$, we can define:
\begin{equation}
{\cal I}(v^{\ell-1}_{R,j_{\ell}}):=\bigcup_{i_{\ell}\in A_{\ell}^{(j_{\ell})}:\mathbf{C}^{\leftarrow}(i_{\ell},R)}{\cal I}(v^{\ell-1}_R,i_{\ell}).\label{eq:k-dist-cal-I}
\end{equation}

\begin{lemma}\label{lem:faults-equiv}
For any $\ell\in\{1,\dots,k-1\}$, fix $v^{\ell-1}_{R,j_{\ell}}\in V_{\ell-1}^+(S_1^*,\dots,S_{\ell-1}^*)$, and non-empty $S_{\ell}\subseteq [c_{\ell}]$, and let $R'$ be obtained from $R$ by inserting $j_{\ell}$ into $R_{\ell}(\mu(S_1^*),\dots,\mu(S_{\ell-1}^*),S_{\ell})$. Then $d^{\rightarrow}_{R'}(j_{\ell})=0$ if and only if $S_{\ell}\cap {\cal I}(v^{\ell-1}_{R,j_{\ell}})=\emptyset$. 
\end{lemma}
\begin{proof}
For $\ell=k-1$, $d_{R'}^{\rightarrow}(j_{k-1})=0$ and ${\cal I}(v^{k-2}_{R,j_{k-1}})=\emptyset$ always hold, by definition. 
For $\ell\in\{1,\dots,k-2\}$,
\begin{align*}
\bar{d}_{R'}^{\rightarrow}(i_{\ell}) &= |\{(i_1,\dots,i_{\ell},i_{\ell+1},x_{i_1})\in D_{\ell+1}(R')\}| & \mbox{see \eq{k-dist-bar-d}}\\
&= \!\!\!\!\!\!\!\!\!\sum_{\substack{(s_1,\dots,s_{\ell},S_{\ell+1})\in\\ [c_1]\times\dots\times[c_{\ell}]\times (2^{[c_{\ell+1}]}\setminus\{\emptyset\})}}
|\{(i_1,\dots,i_{\ell+1},x_{i_1})\in D_{\ell+1}(R(s_1,\dots,s_{\ell},S_{\ell+1}))\}| & \mbox{see \eq{D-ell-R}}\\
&= \!\!\!\!\!\!\!\!\!\sum_{\substack{(s_1,\dots,s_{\ell},S_{\ell+1})\in\\ [c_1]\times\dots\times[c_{\ell}]\times (2^{[c_{\ell+1}]}\setminus\{\emptyset\})}}
\sum_{\substack{S_{\ell}'\subseteq[c_{\ell}]:\\ s_{\ell}\in S_{\ell}}}\!\!
|\{(i_1,\dots,i_{\ell+1},x_{i_1}):  i_{\ell+1}\in\overline{R}_{\ell+1}(s_1,\dots,s_{\ell},S_{\ell+1}),\\[-20pt]
&\qquad\qquad\qquad\qquad\qquad\qquad\quad\, x_{i_{\ell+1}}=x_{i_{\ell}}, (i_1,\dots,i_{\ell},x_{i_1})\in D_{\ell}(R_{\ell}'(s_1,\dots,s_{\ell-1},S_{\ell}'))\}| & \mbox{see \eq{D-ell}}.
\end{align*}
Suppose $i_{\ell}\in A_{\ell}^{(j_{\ell})}$, meaning we have $i_{\ell}\in \overline{R}_{\ell}(\mu(S_1^*),\dots,\mu(S_{\ell-1}^*),S_{\ell})$. By \eq{D-ell},
$\ell$-collisions of the form $(i_1,\dots,i_{\ell},x_{i_1})$ can only occur in $D_{\ell}({R}_{\ell}(\mu(S_1^*),\dots,\mu(S_{\ell-1}^*),S_{\ell}))$, so we continue:
\begin{align*}
\bar{d}_{R'}^{\rightarrow}(i_{\ell}) &= \sum_{\substack{s_{\ell}\in S_{\ell},S_{\ell+1}\in2^{[c_{\ell+1}]}\setminus\{\emptyset\}}}
|\{(i_1,\dots,i_{\ell+1},x_{i_1}):  i_{\ell+1}\in\overline{R}_{\ell+1}(\mu(S_1^*),\dots,\mu(S_{\ell-1}^*),s_{\ell},S_{\ell+1}),\\[-17pt]
&\qquad\qquad\qquad\qquad\qquad\qquad\quad\, x_{i_{\ell+1}}=x_{i_{\ell}}, (i_1,\dots,i_{\ell},x_{i_1})\in D_{\ell}(R_{\ell}'(\mu(S_1^*),\dots,\mu(S_{\ell-1}^*),S_{\ell}))\}|,
\end{align*}
and thus $\bar{d}_{R'}^{\rightarrow}(i_{\ell})>0$ if and only if:
\begin{align}
& \exists s_{\ell}\in S_{\ell},S_{\ell+1}\in 2^{[c_{\ell+1}]}\setminus\{\emptyset\}\mbox{ s.t. }\exists i_{\ell+1}\in \overline{R}_{\ell+1}(\mu(S_1^*),\dots,\mu(S_{\ell-1}^*),s_{\ell},S_{\ell+1})\mbox{ s.t. }x_{i_{\ell+1}}=x_{i_{\ell}}\label{eq:k-dist-fault-1}\\
\mbox{and }&\exists (i_1,\dots,i_{\ell-1},i_{\ell},x_{i_1})\in D_{\ell}(R_{\ell}'(\mu(S_1^*),\dots,\mu(S_{\ell-1}^*),S_{\ell})).\label{eq:k-dist-fault-2}
\end{align}
Suppose ${d}_{R'}^{\rightarrow}(j_{\ell})>0$. By \eq{forward-collision-degree}, this happens if and only if there exists $i_{\ell}\in A_{\ell}^{(j_{\ell})}$ such that $\bar{d}_{R'}^{\rightarrow}(i_{\ell})>0$, which holds if and only if \eq{k-dist-fault-1} and \eq{k-dist-fault-2} are true.  We know \eq{k-dist-fault-1} if and only if $\mathbf{C}^{\rightarrow}(i_{\ell},R,S_{\ell})$ holds, if and only if $S_{\ell}\cap {\cal I}(v_R^{\ell-1},i_{\ell})\neq \emptyset$. For \eq{k-dist-fault-2}, we make a distinction based on the value of $\ell$. 

In the case $\ell=1$, \eq{k-dist-fault-2} is just $(i_1,x_{i_1})\in D_1(R_1'(S_1))$, which is true by \eq{k-dist-D-1}, since we just added $j_1$ to $R_1(S_1)$ to get $R'(S_1)$. For the other direction, if $S_{1}\cap {\cal I}(v_{R,j_{1}}^{0})\neq\emptyset$, then by \eq{k-dist-cal-I-1}, $\exists i_{1}\in A_{1}^{(j_{1})}$ satisfying \eq{k-dist-fault-2}. This completes the $\ell=1$ case. 

Continuing with the case $\ell\in\{2,\dots,k-2\}$, by \eq{D-ell}, using the fact that $R_{\ell-1}'=R_{\ell-1}$, we have \eq{k-dist-fault-2} if and only if $\mathbf{C}^{\leftarrow}(i_{\ell},R)$. Thus, we have:
\begin{equation*}
	\left[{d}_{R'}^{\rightarrow}(j_{\ell})>0\right] \Leftrightarrow \underbrace{\exists i_{\ell}\in A_{\ell}^{(j_{\ell})}\mbox{ s.t }\left[\left[ S_{\ell}\cap {\cal I}(v_R^{\ell-1},i_{\ell})\neq \emptyset \right]\wedge \mathbf{C}^{\leftarrow}(i_{\ell},R)\right]}_{=:\mathbf{C}}.\label{eq:k-dist-fault-final}
\end{equation*}
If $\mathbf{C}$ holds, then by \eq{k-dist-cal-I}, ${\cal I}(v_R^{\ell-1},i_{\ell})\subseteq {\cal I}(v_{R,j_{\ell}}^{\ell-1})$, and so, also by $\mathbf{C}$, $S_{\ell}\cap {\cal I}(v_{R,j_{\ell}}^{\ell-1})\neq\emptyset$. For the other direction, if $S_{\ell}\cap {\cal I}(v_{R,j_{\ell}}^{\ell-1})\neq\emptyset$, then by \eq{k-dist-cal-I}, $\exists i_{\ell}\in A_{\ell}^{(j_{\ell})}$ satisfying both conditions of $\mathbf{C}$.
\end{proof}

\paragraph{$E_1\subset V_0^+\times V_1$:} Recall that $V_0^+$ is the set of vertices $v_{R,j_1}^0$ in which we have chosen an index $j_1$ to add to $R_1$, but not yet decided to which part of $R_1$ it should be added. A transition in $E_1$ represents selecting some $S_1\in 2^{[c_1]}\setminus\{\emptyset\}$ and then adding $j_1$ to $R_1(S_1)$, so we have 
$$L^+(v_{R,j_1}^0):= 2^{[c_1]}\setminus\{\emptyset\},$$
and $f^+_{v^0_{R,j_1}}(S_1)=v^1_{R'}$, where $R'$ is obtained from $R$ by inserting $j_1$ into $R_1(S_1)$. As in the case of 3-distinctness, not all of these labels represent edges with non-zero weight.
To go from a vertex $v_{R'}^1\in V_1(S_1^*)$, we choose some $j_1$ to remove from $R'_1(S_1^*)$, the part of $R_1$ that has had an index added, to get some $R$ such that $v_R^0\in V_0$ (and $v_{R,j_1}^0\in V_0^+$). However, we make sure to choose an $j_1$ with no forward collisions -- i.e.~$d^{\rightarrow}_{R'}(j_1)=0$ -- so we let
$$L^-(v_{R'}^1):=\{j_1\in R_1'(S_1^*):d_{R'}^{\rightarrow}(j_1)=0\},$$
and then set $f^-_{v_{R'}^1}(j_1)=v_{R,j_1}^0$ where $R=R'\setminus j_1$ is obtained from $R'$ by removing $j_1$. Importantly, given $v^1_{R'}$, we can take a superposition over this set, because we store the set $C^{\rightarrow}_1(R)$ defined in \eq{k-dist-forward-col-db} (this is necessary in \sec{k-dist-star-states}).

As in the case of 3-distinctness, it is not yet clear how to define $E_1$, the set of (non-zero weight) edges between $V_0^+$ and $V_1$, because $|V_0^+|\cdot |L^+(v_{r,j_1}^0)| > |V_1|\cdot |L^-(v_{R'}^1)|$.  We define it as follows.
$$E_1:=\left\{\left(f^-_{v_{R'}^1}(j_1),v_{R'}^1\right)=\left(v_{R'\setminus\{j_1\},j_1}^0,v_{R'}^1\right):v_{R'}^1\in V_1,j_1\in L^-(v_{R'}^1)\right\}$$
and give weight $\w_1=1$ to all edges in $E_1$. Then we have the following.
\begin{lemma}\label{lem:k-dist-E1}
Let $R^{S_1 \leftarrow j_1}$ be obtained from $R$ by inserting $j_1$ into $R_1(S_1)$. Then
\begin{align*}
E_1 &= \left\{\left(v^{0}_{R,j_1},v^{1}_{R^{S_1 \leftarrow j_1}}\right): v_{R,j_1}^0\in V_0^+, S_1 \in 2^{[c_{1}]\setminus {\cal I}(v_{R,j_1}^0)}\setminus \{\emptyset\} \right\}.\label{eq:k-dist-E1}
\end{align*}
So for all $v^0_{R,j_1}\in V_0^+$, and $S_1\in L^+(v^0_{R,j_1})$,
$\displaystyle\w_{v^0_{R,j_1},S_1} = \left\{\begin{array}{ll}
\w_1 =1 & \mbox{if }S_1\cap {\cal I}(v^0_{R,j_1}) = \emptyset\\
0 & \mbox{else.}
\end{array}\right.$
\end{lemma}
\begin{proof}
Let $E_1'$ be the right-hand side of the identity in the theorem statement, so we want to show $E_1=E_1'$. 
Fix any $v_{R,j_1}^0\in V_0^+$ and non-empty $S_1\subseteq [c_1]\setminus {\cal I}(v_{R,j_1})$, and let $R'=R^{S_1\leftarrow j_1}$. Then since $S_1\cap {\cal I}(v_{R,j_1})=\emptyset$, by \lem{faults-equiv}, $d^{\rightarrow}_{R'}(j_1)=0$.
This implies $E_1'\subseteq E_1$.

For the other direction, fix any $v^1_{R'}\in V_1(S_1^*)$ and $j_1\in L^-(v^1_{R'})$. Since $j_1\in R_1(S_1^*)$, we have $v^1_{R'\setminus\{j_1\}}\in V_0^+$ (that is, we have removed an index from the set that had size $t_1+1$) and $(R'\setminus\{j_1\})^{S_1^*\leftarrow j_1}=R'$. Since $d_{R'}^{\rightarrow}(j_1)=0$, by \lem{faults-equiv}, $S_1^*\cap {\cal I}(v_{R'\setminus\{j_1\},j_1}^0)=\emptyset$. This implies $E_1\subseteq E_1'$.
\end{proof}

We remark that for any $j_1\in [m_1]$, $d_R^{\rightarrow}(i_1)$ is always at most $k-2$. Otherwise, there are at least $k-1$ elements $i_2\in \overline{R}_2\subset A_2$ such that $x_{i_1}=x_{i_2}$ (where $i_1$ is the unique element in $A_1^{(j_1)}$), and together with $i_1$ these form a $k$-collision, which contradicts our assumption that the unique $k$-collision is in $A_1\times\dots\times A_k$. Thus, if we set $c_1=k-1$, we have for any $v_{R,j_1}^0\in V_0^+$, ${\cal I}(v_{R,j_1}^0)\subsetneq [c_1]$, which will be important in \sec{k-dist-star-states}.

Finally, it follows from \eq{kV1}, that 
\begin{equation}
	\abs{E_1} \leq \abs{L^+(v^0_{R,j_1})}\abs{V_0^+} = O\left(n\abs{V_0}\right).\label{eq:kE2}
\end{equation}

\paragraph{$E_{\ell}\subset V_{\ell-1}^+\times V_{\ell}$ for $\ell\in\{2,\dots,k-1\}$:} For $E_{\ell}$, we generalise the construction of $E_1$. Similar to the definition $E_1$, we define, for any $v^{\ell-1}_{R,j_{\ell}}\in V_{\ell-1}^+$, and $v_{R'}^{\ell}\in V_{\ell}(S_1^*,\dots,S_{\ell}^*)$:
$$L^+(v^{\ell-1}_{R,j_{\ell}}):=2^{[c_{\ell}]}\setminus\{\emptyset\}
\mbox{ and }
L^-(v^{\ell}_{R'}):=\{j_{\ell}\in R_{\ell}(\mu(S_1^*),\dots,\mu(S_{\ell-1}^*),S_{\ell}): d_R^{\rightarrow}(j_{\ell})=0\}.$$
We set $f^+_{v^{\ell-1}_{R,j_{\ell}}}(S_{\ell})=v^{\ell}_{R'}$ where if $v^{\ell-1}_R\in V_{\ell-1}(S_1^*,\dots,S_{\ell-1}^*)$, $R'$ is obtained from $R$ by inserting $j_{\ell}$ into $R_{\ell}(\mu(S_1^*),\dots,\mu(S_{\ell-1}^*),S_{\ell})$. We set $f^-_{v^{\ell}_{R'}}(j_{\ell})=v_{R'\setminus\{j_{\ell}\},j_{\ell}}^{\ell-1}$. Similar to $E_1$, we define:
\begin{equation}
{E}_{\ell}:=\left\{\left(f^-_{v^{\ell}_{R'}}(j_{\ell}), v^{\ell}_{R'}\right)=\left(v^{\ell-1}_{R'\setminus\{j_{\ell}\},j_{\ell}},v^{\ell}_{R'}\right): v^{\ell}_{R'}\in V_{\ell}, j_{\ell}\in L^-(v^{\ell}_{R'})\right\},\label{eq:k-dist-E_ell}
\end{equation}
and give weight $\w_{\ell}:=\sqrt{n/m_{\ell-1}}$ to all edges in ${E}_{\ell}$. Then we have the following.
\begin{lemma}\label{lem:k-dist-E-ell}
For any $(S_1,\dots,S_{\ell-1})\in (2^{[c_1]}\setminus\{\emptyset\})\times\dots\times(2^{[c_{\ell-1}]}\setminus\{\emptyset\})$, define:
\begin{multline*}
	{E}_{\ell}(S_1,\dots,S_{\ell-1})=\bigg\{\left(v^{\ell-1}_{R,j_{\ell}},v^{\ell}_{R'}\right): v^{\ell-1}_{R,j_{\ell}}\in V_{\ell-1}^+(S_1,\dots,S_{\ell-1}),\\
	 \exists S_{\ell} \in 2^{[c_{\ell}]\setminus {\cal I}(v^{\ell-1}_{R,j_{\ell}})} \setminus \{\emptyset\},\;
R'=R^{(\mu(S_1),\dots,\mu(S_{\ell-1}),S_{\ell})\leftarrow j_{\ell}} \bigg\},
\end{multline*}
where $R^{(\mu(S_1),\dots,\mu(S_{\ell-1}),S_{\ell})\leftarrow j_{\ell}}$ is obtained from $R$ by inserting $j_{\ell}$ into $R_{\ell}(\mu(S_1),\dots,\mu(S_{\ell-1}),S_{\ell})$. Then
$${E}_{\ell} = \bigcup_{(S_1,\dots,S_{\ell-1})\in (2^{[c_1]}\setminus\{\emptyset\})\times\dots\times(2^{[c_{\ell-1}]}\setminus\{\emptyset\})} {E}_{\ell}(S_1,\dots,S_{\ell-1}).$$
\end{lemma}
\begin{proof}
Fix $S_1,\dots,S_{\ell-1}$ and suppose $(v_{R,j_{\ell}}^{\ell-1},v_{R'}^{\ell})\in {E}_{\ell}(S_1,\dots,S_{\ell-1})$. Then by \lem{faults-equiv}, since $S_{\ell}$ is chosen so that $S_{\ell}\cap {\cal I}(v_{R,j_{\ell}}^{\ell-1})=\emptyset$,
$d_{R'}^{\rightarrow}(j_{\ell})=0$, and Thus, $j_{\ell}\in L^-(v_{R'}^{\ell})$, so $(v_{R,j_{\ell}}^{\ell-1},v_{R'}^{\ell})\in E_{\ell}$.

For the other direction, suppose $\left(v^{\ell-1}_{R'\setminus\{j_{\ell}\},j_{\ell}},v^{\ell}_{R'}\right)\in E_{\ell}$, and let $S_1^*,\dots,S_{\ell}^*$ be such that $v^{\ell}_{R'}\in V_{\ell}(S_1^*,\dots,S_{\ell}^*)$. Then $R'$ is obtained from $R'\setminus\{j_{\ell}\}$ by adding $j_{\ell}$ to $R_{\ell}(\mu(S_1^*),\dots,\mu(S_{\ell-1}^*),S_{\ell}^*)$. Then since $j_{\ell}\in L^-(v_{R'}^{\ell})$, $d_{R'}^{\rightarrow}(j_{\ell})=0$, so by \lem{faults-equiv}, $S_{\ell}^*\in 2^{[c_{\ell}]\setminus{\cal I}(v_{R,j_{\ell}}^{\ell-1})}\setminus\{\emptyset\}$, and so $\left(v^{\ell-1}_{R'\setminus\{j_{\ell}\},j_{\ell}},v^{\ell}_{R'}\right)\in E_{\ell}(S_1^*,\dots,S_{\ell-1}^*)$.
\end{proof}
While $E_{\ell}$ represents all edges between $V_{\ell-1}^{+}$ and $V_{\ell}$, we now define two sets of edges $\tilde{E}_{\ell}\subset E_{\ell}$, and $\tilde{E}_{\ell}'$ disjoint from $V_{\ell-1}^+\times V_{\ell}$, that each solve a different technical issue. First, in \sec{k-dist-transitions}, we will see that the complexity of transitions in $E_{\ell}$ depends on the number of collisions between the new block $A_{\ell}^{(j_{\ell})}$ being added, and $(\ell-1)$-collisions already stored in $D_{\ell-1}(R)$, so we will only attempt to implement the transition subroutine correctly when this set is not too large. In anticipation of this, we define:
\begin{equation}
	\tilde{E}_{\ell}:=\left\{(v^{\ell-1}_{R,j_{\ell}},v^{\ell}_{R'})\in E_{\ell}: |{\cal K}(\overline{R}_1,\dots,\overline{R}_{\ell-1},A_{\ell}^{(j_{\ell})})| \geq p_{\ell}
	\right\}\subset E_{\ell},\label{eq:k-dist-tilde-E-ell}
\end{equation}
where $p_{\ell}\in {\sf polylog}(n)$,
which will be part of $\tilde{E}$, the set of edges whose transitions we fail to implement. 
Second, if any $v_{R,j_{\ell}}^{\ell-1}\in V_{\ell-1}^+$ has no neighbour in $V_{\ell}$, which happens exactly when ${\cal I}(v_{R,j_{\ell}}^{\ell-1})=[c_{\ell}]$, then its correct star state would simply have one incoming edge from $V_0$, which, as discussed in \sec{3-dist-star-states}, would make it impossible to define a flow satisfying all star state constraints (\textbf{P2} of \thm{full-framework}).
Unlike in the case of $E_1$, there is no constant $c_{\ell}$ such that we can assume ${\cal I}(v_{R,j_{\ell}}^{\ell-1})\subsetneq [c_{\ell}]$ for all $v_{R,j_{\ell}}^{\ell-1}\in V_{\ell-1}^+$. That is because while each $i_{\ell}\in A_{\ell}$ can have at most $k-2$ collisions in $R_{\ell+1}$, the total number of such collisions for all $i_{\ell}\in A_{\ell}^{(j_{\ell})}$ may be linear in $|A_{\ell}^{(j_{\ell})}|$. Fortunately this happens for only a very small fraction of $v_{R,j_{\ell}}^{\ell-1}\in V_{\ell-1}^+$. Thus, we define (choosing $\{1\}$ arbitrarily):
\begin{equation}
	\tilde{E}_{\ell}':=\left\{\left(v^{\ell-1}_{R,j_{\ell}},f^+_{v^{\ell-1}_{R,j_{\ell}}}(\{1\})\right): {\cal I}(v^{\ell-1}_{R,j_{\ell}})=[c_{\ell}]\right\},\label{eq:k-dist-tilde-E-prime}
\end{equation}
which is disjoint from $E_{\ell}$. Note that for $\ell=k-1$, we always have ${\cal I}(v^{k-2}_{R,j_{k-1}})=\emptyset$ and $[c_{k-1}]=[1]$, so $\tilde{E}_{k-1}'=\emptyset$.
Since $\tilde{E}_{\ell}'$ will be part of $\tilde{E}$, we assume its endpoints $f^+_{v^{\ell-1}_{R,j_{\ell}}}(\{1\})$ are just otherwise isolated vertices that we do not consider a part of $V(G)$ (see \rem{orphan-vertices}). As with $E_{\ell}$, we set all edges in $\tilde{E}_{\ell}'$ to have weight $\w_{\ell}$. Thus, from this discussion as well as \lem{k-dist-E-ell} we have,
for all $v^{\ell-1}_{R,j_{\ell}}\in V_{\ell-1}^+$ and $S_{\ell}\in L^+(v^{\ell-1}_{R,j_{\ell}})$,
\begin{equation}
\w_{v^{\ell-1}_{R,j_{\ell}},S_{\ell}} = \left\{\begin{array}{ll}
\w_{\ell} =\sqrt{n/m_{\ell}} & \mbox{if }S_{\ell}\cap {\cal I}(v^{\ell-1}_{R,j_{\ell}}) = \emptyset\\
\w_{\ell} =\sqrt{n/m_{\ell}} & \mbox{if }{\cal I}(v^{\ell-1}_{R,j_{\ell}})=[c_{\ell}]\mbox{ and }S_{\ell}=\{1\}\\
0 & \mbox{else.}
\end{array}\right.\label{eq:k-dist-w-ell-cond}
\end{equation}

\noindent We can see from \eq{kV1}, that 
\begin{equation}
	\abs{E_{\ell}}+\abs{\tilde{E}_{\ell}'} \leq \abs{L^+(v_{R,j_{\ell}}^{\ell-1})}\abs{V_{\ell-1}^+} = O\left(n\abs{V_0}\right).\label{eq:kEk1}
\end{equation}

\paragraph{$E_k\subset V_{k-1}\times V_k$:} Finally, there is an edge between $v^{k-1}_{R}\in V_{k-1}$ and $v^k_{R,i_k}\in V_k$ for any $i_k\in A_k$, so we define 
$$L^+(v_R^{k-1}):= A_k
\mbox{ and }
L^-(v_{R,i_k}^k):=\{\leftarrow\},$$
and let $f^+_{v_R^{k-1}}(i_k)=v_{R,i_k}^k$, and $f^-_{v_{R,i_k}^k}(\leftarrow)=v_{R}^{k-1}$. We let $E_k$ be the set of such edges:
$$E_k:=\left\{(v_{R}^{k-1},v_{R,i_k}^k):v_{R}^{k-1}\in V_{k-1},i_k\in A_k\right\},$$
and we set $\w_e=\w_k=1$ for all $e\in E_k$. This, together with \eq{kVk}, implies that 
\begin{equation}
|E_k|=O\left(\frac{n^2}{t_{k-1}}|V_0|\right).
\label{eq:kEk}
\end{equation}

\begin{table}
	\renewcommand{\arraystretch}{1.5}
	\centering
	\begin{tabular}{ |c|c|c|c|c|c| }
		\hline
		Edge Set & $(u,v)$ & $(u,i)$ & $(v,j)$ & $\w_{u,v}$ & ${\sf T}_{u,v}$\\
		\hline
		\hline
		$E_0^+\subset V_0\times V_0^+$ & $(v^0_R,v^0_{R,j_1})$ & $(v^0_R,j_1)$ & $(v^0_{R,j_1},\leftarrow) $  & $\w_0^+=1$ & ${\sf T}_0^+=\widetilde{O}(1)$\\
		\hline
		$E_1\subset V_0^+\times V_1$ & $(v^0_{R,j_1},v^1_{R'})$ & $(v^0_{R,j_1},S_1)$ & $(v^1_{R'},j_1) $  & $\w_1=1$ & ${\sf T}_1=\widetilde{O}(1)$\\
		\hline
		$\{E_\ell^+\subset V_{\ell}\times V_{\ell}^+\}_{\ell=1}^{k-2}$ & $(v^1_{R},v^1_{R,j_{\ell+1}})$ & $(v^1_{R},j_{\ell+1})$ & $(v^1_{R,j_{\ell+1}},\leftarrow)$  & $\w_{\ell}^+=1$ & ${\sf T}_{\ell}^+=\widetilde{O}(1)$\\
		\hline
		$\{E_\ell\subset V_{\ell-1}^+\times V_{\ell}\}_{\ell=1}^{k-1}$ & $(v^{\ell-1}_{R,j_{\ell}},v^{\ell}_{R'})$ & $(v^{\ell-1}_{R,j_{\ell}},S_{\ell})$ & $(v^{\ell}_{R'},j_{\ell})$  & $\w_{\ell}=\sqrt{\frac{n}{m_{\ell}}}$ & ${\sf T}_{\ell}=\widetilde{O}\left(\sqrt{\frac{n}{m_{\ell}}}\right)$\\
		\hline
		$E_k\subset V_{k-1}^+\times V_k$ & $(v^{k-1}_{R},v^k_{R,i_k})$ & $(v^{k-1}_{R},i_k)$ & $(v^k_{R,i_k},\leftarrow)$  & $\w_k=1$ & ${\sf T}_k=\widetilde{O}(1)$\\
		\hline
	\end{tabular}
	\caption{
For each edge in $\protect\overrightarrow{E}(G)$, we can describe it in three ways: as a pair of vertices $(u,v)$; as a vertex $u$ and forward label $i=f_u^{-1}(v)$; and as a vertex $v$ and backward label $j=f_v^{-1}(u)$ (see \defin{QW-access}). We summarise these three descriptions for the edge sets that make up $\protect\overrightarrow{E}(G)\setminus\tilde{E}$, along with the edge weights, and transitions costs (see \cor{k-dist-transitions}). The edge labels $i$ and $j$ range across (sometimes strict) subsets of $L^+(u)$ and $L^-(u)$ (see \tabl{k-dist-index-sets}). For example, for $u\in V_0^+$, $i=S_1\in L^+(u)=2^{[c_1]}\setminus\{\emptyset\}$, $(u,i)$ only represents an edge of $\protect\overrightarrow{E}(G)$ when $S_1\cap {\cal I}(u)=\emptyset$ (see \lem{k-dist-E1} and \eq{k-dist-w-ell-cond}).}\label{tab:k-dist-edges}
\end{table}

\paragraph{The Graph $G$:} The full graph $G$ is defined by:
\begin{align*}
V(G)&=\bigcup_{\ell=0}^{k} V_{\ell}\cup \bigcup_{\ell=0}^{k-2} V_{\ell}^+\\
\overrightarrow{E}(G) &= \{(u,v):u\in V(G), i\in L^+(u):\w_{u,i}\neq 0\}=\bigcup_{\ell=0}^{k-2} E_{\ell}^+\cup E_1 \cup \bigcup_{\ell=2}^{k-1} (E_{\ell}\cup \tilde{E}_{\ell}')\cup E_k,
\end{align*}
where the edge label sets $L^+(u)$ are summarised in \tabl{k-dist-index-sets}, and weights are summarised in \tabl{k-dist-edges}. We define (recall that $\tilde{E}_{\ell}\subset E_{\ell}$): 
\begin{equation}
\tilde{E}:=\bigcup_{\ell=2}^{k-1}(\tilde{E}_{\ell}\cup\tilde{E}_{\ell}').\label{eq:k-dist-tilde-E}
\end{equation}

\paragraph{The Marked Set and Checking Cost:}  In the notation of \thm{full-framework}, we let $V_{\sf M}=V_k$, and we define a subset $M\subseteq V_{\sf M}$ as follows. 
If $(a_1,\dots,a_k)\in A_1\times \cdots \times A_k$ is the unique $k$-collision, we let 
\begin{equation}
\begin{split}
	M &= \big\{v^{k}_{R_1,\dots,R_{k-1},i_k} \in V_{k}: \exists (i_1,\dots,i_{k-1},x_{i_1})\in D_{k-1}(R)\mbox{ s.t. }x_{i_1}=x_{i_k}\}\\
&= \big\{v^k_{R_1,\dots,R_{k-1},i_k}: \exists 
S_{1} \subseteq [c_{1}],\dots,S_{k-1} \subseteq [c_{k-1}], s_1 \in S_1,\dots,s_{k-1} \in S_{k-1} \text{ s.t. }\\
	&\qquad\qquad\qquad\qquad\qquad\qquad \forall \ell \in \{1,\dots,k-1\},\; a_{\ell} \in \overline{R}_{\ell}(s_1,\dots,s_{\ell-1},S_{\ell})\mbox{ and }i_k = a_k\big\},
\end{split}\label{eq:k-dist-M}
\end{equation}
and otherwise, if there is no $k$-collision, $M=\emptyset$. We can decide whether $v^{k}_{R,i_k} \in V_{k}$ is marked by querying $i_k$ to obtain the value $x_{i_k}$ and looking it up in $D_{k-1}(R)$ to see if we find some $(i_1,\dots,i_{k-1},x_{i_k})$, in which case, it must be that $a_1=i_1$,\dots,$a_k=i_k$. Thus, the checking cost is at most
\begin{equation}
	{\sf C}=O(\log n).\label{eq:k-dist-C}
\end{equation}

\subsubsection{The Star States and their Generation}\label{sec:k-dist-star-states}

We define the set of alternative neighbourhoods (\defin{alternative}) with which we will apply \thm{full-framework}. 
For $\ell\in\{0,\dots,k\}$, for all  $v^{\ell}_R \in V_\ell$, we add a single star state to $\Psi_\star(u)$, which has one of three forms, depending on $\ell$ (refer to \tabl{k-dist-index-sets}): for $v^0_R\in V_0$,
\begin{equation}
\Psi_\star(v^{0}_R):=\left\{\ket{\psi_\star^G(v^{0}_R)}=
\sum_{j_1\in [m_1]\setminus R_1}\sqrt{\w_{0}^+}\ket{v^0_R,j_1}
\right\};\label{eq:k-star-v0}
\end{equation} 
for $\ell\in\{1,\dots,k-1\}$, and 
$v^{\ell}_R\in V_{\ell}(S_1,\dots,S_{\ell})$,\footnote{Here we explicitly include the $\rightarrow$ and $\leftarrow$ parts of each element of $L^+$ and $L^-$, which are normally left implicit, in order to stress that the two sum are orthogonal.}
\begin{equation}
\Psi_\star(v^{\ell}_R):=\left\{\ket{\psi_\star^G(v^{\ell}_R)}=-\!\!\!\!
\sum_{\substack{ j_{\ell}\in R_{\ell}(\mu(S_1),\dots,\mu(S_{\ell-1}),S_{\ell}):\\ d_R^{\rightarrow}(j_{\ell})=0}}\!\!\!\!\!\sqrt{\w_{\ell}}\ket{v^{\ell}_R,\leftarrow,j_{\ell}} 
+\!\!\!\sum_{j_{\ell+1}\in [m_{\ell+1}]\setminus R_{\ell+1}}\!\!\!
\sqrt{\w_{\ell}^+}\ket{v^{\ell}_R,\rightarrow,j_{\ell+1}}
\right\};\label{eq:k-star-v-ell}
\end{equation} 
and finally, for $v^k_{R,i_k}\in V_k$,
\begin{equation}
\Psi_\star(v^{k}_{R,i_k}):=\left\{\ket{\psi_\star^G(v^{k}_{R,i_k})}=
-\sqrt{\w_k}\ket{v^k_{R,i_k},\leftarrow}
\right\}.\label{eq:k-star-vk}
\end{equation}

\noindent From \tabl{k-dist-index-sets}, along with the description of $\w$ in \lem{k-dist-E1}, we can see that for $v^0_{R,j_1}\in V_0^+$, 
\begin{equation*}
\ket{\psi_\star^{G}(v^0_{R,j_1})} = -\sqrt{\w_0^+}\ket{v_{R,j_1}^0,\leftarrow}+\sum_{S_1\in 2^{[c_1]\setminus{\cal I}(v_{R,j_1}^0)}\setminus\{\emptyset\}}\sqrt{\w_1}\ket{v_{R,j_1}^0,S_1}.
\end{equation*}
To generate this state, one would have to compute ${\cal I}(v^{0}_{R,j_1})$ (see \eq{k-dist-I-i}), which would require determining the locations of all forward collisions of $j_1$, which is far too expensive. Hence, we simply add all options to $\Psi_\star(v^{0}_{R,j_1})$ (we see in \lem{k-dist-star-states} that generating this \emph{set} is not difficult):
\begin{equation}
	\Psi_\star(v^0_{R,j_1}) := \left\{\ket{\psi_\star^{{\cal I}_1}(v^0_{R,j_1})}:=-\sqrt{\w_0^+}\ket{v^0_{R,j_1},\leftarrow}+ \sum_{S_{1}\in 2^{[c_{1}] \setminus {\cal I}_{1}}\setminus\{\emptyset\}} \sqrt{\w_{1}}\ket{v^0_{R,j_1},S_1} : {\cal I}_{1}\subsetneq[c_{1}]\right\}.\label{eq:k-Psi-v-0-plus}
\end{equation}
Thus, since we always have ${\cal I}(v^0_{R,j_1})\subsetneq[c_1]$, $\ket{\psi_\star^G(v^0_{R,j_1})}=\ket{\psi_\star^{{\cal I}(v^0_{R,j_1})}(v^0_{R,j_1})}\in \Psi_{\star}(v^0_{R,j_1})$. 

Similarly, for $\ell\in \{2,\dots,k-1\}$ and $v^{\ell-1}_{R,j_{\ell}}\in V_{\ell-1}^+$ define:
\begin{equation}
	\Psi_\star(v^{\ell-1}_{R,j_{\ell}}) := \left\{\ket{\psi_\star^{{\cal I}_{\ell}}(v^{\ell-1}_{R,j_{\ell}})}:=-\sqrt{\w_{\ell-1}^+}\ket{v^{\ell-1}_{R,j_{\ell}},\leftarrow}+ \sum_{S_{\ell}\in 2^{[c_{\ell}] \setminus {\cal I}_{\ell}}\setminus\{\emptyset\}} \sqrt{\w_{\ell}}\ket{v^{\ell-1}_{R,j_{\ell}},S_{\ell}} : {\cal I}_{\ell}\subsetneq[c_{\ell}]\right\}.\label{eq:k-Psi-v-ell-plus}
\end{equation}
Then from \eq{k-dist-w-ell-cond},  we have:
\begin{equation}
\ket{\psi_\star^G(v^{\ell}_{R,j_{\ell}})}=\left\{\begin{array}{ll}
\ket{\psi_\star^{{\cal I}(v^{\ell}_{R,j_{\ell}})}(v^{\ell}_{R,j_{\ell}})} & \mbox{if }{\cal I}(v^{\ell}_{R,j_{\ell}})\subsetneq[c_{\ell}]\\
\ket{\psi_\star^{[c_{\ell}]\setminus\{1\}}(v^{\ell}_{R,j_{\ell}})} & \mbox{if }{\cal I}(v^{\ell}_{R,j_{\ell}})=[c_{\ell}], 
\end{array}\right.
\label{k-star-v-ell-plus}
\end{equation}
where ${\cal I}(v^{\ell}_{R,j_{\ell}})$ is defined in \eq{k-dist-cal-I}.

\noindent We now describe how to generate the states in $\bigcup_{u\in V(G)}\Psi_\star(u)$ in $\widetilde{O}(1)$ complexity (see \defin{alternative}):
\begin{lemma}\label{lem:k-dist-star-states}
	The states $\Psi_\star=\{\Psi_\star(u)\}_{u\in V(G)}$ can be generated in $\widetilde{O}(1)$ complexity. 
\end{lemma}
\begin{proof}
The description of a vertex $u\in V(G)$ begins with a label indicating to which of $V_0,\dots,V_k$ or $V_0^+,\dots,V_{k-2}^+$ it belongs, so we can define subroutines $U_0,\dots,U_k,U_{0,+},\dots,U_{k-2,+}$ that generate the star states in each vertex set respectively, and then 
$$U_\star = \sum_{\ell=0}^k\ket{\ell}\bra{\ell}\otimes U_{\ell}+\sum_{\ell=0}^{k-2}\ket{\ell,+}\bra{\ell,+}\otimes U_{\ell,+}$$
will generate the star states in the sense of \defin{alternative}.

	We begin with $U_0$. For $v^{0}_{R} \in V_0$, we have $\Psi_\star(v^0_R)=\{\ket{\psi_\star^G(v_R^0)}\}$, where $\ket{\psi_\star^G(v_R^0)}$ is as in \eq{k-star-v0}. Thus, implementing the map $U_0:\ket{u}\ket{0}\mapsto \propto \ket{\psi_\star^G(u)}$ is as simple as generating a uniform superposition over $[m_1]$, and then using $O(\log n)$ rounds of amplitude amplification to get inverse polynomially close to the uniform superposition over $[m_1]\setminus R_1$.

For $\ell\in\{1,\dots,k-1\}$, and $v^{\ell}_R\in V_{\ell}$, we again have $\Psi_\star(v^\ell_R)=\{\ket{\psi_\star^G(v_R^\ell)}\}$, where $\ket{\psi_\star^G(v_R^\ell)}$ is as in \eq{k-star-v-ell}. To implement $U_{\ell}:\ket{u}\ket{0}\mapsto \propto\ket{\psi_\star^G(u)}$, we first compute (referring to \tabl{k-dist-edges} for the weights):
$$\ket{u,0}\mapsto \propto\ket{u}\left(-\sqrt{\w_{\ell}}\ket{\leftarrow}+\sqrt{\w_{\ell}^+}\ket{\rightarrow}\right)\ket{0}=\ket{u}\left(-(n/m_{\ell})^{1/4}\ket{\leftarrow}+\ket{\rightarrow}\right)\ket{0},$$
which can be implemented by a $O(1)$-qubit rotation. Then conditioned on $\leftarrow$, generate a uniform superposition over $j_{\ell}\in R_{\ell}(\mu(S_1^*),\dots,\mu(S_{\ell-1}^*),S_{\ell}^*)$ (we can learn the sets $S_1^*,\dots,S_{\ell}^*$ by seeing which sets are bigger, or assume we simply keep track of these values in some convenient way), and then using $O(\log n)$ rounds of amplitude amplification to get inverse polynomially close to the superposition over such $j_{\ell}$ such that $d_R^{\rightarrow}(j_{\ell})=0$, which we can check by looking up $j_{\ell}$ in $C_{\ell}^{\rightarrow}(R)$. We have used the fact that our data structure supports taking a uniform superposition (see \sec{data}). Finally, conditioned on $\rightarrow$, generate a uniform superposition over $j_{\ell+1}\in [m_{\ell+1}]$, and use $O(\log n)$ rounds of  amplitude amplification to get inverse polynomially close to a superposition over $j_{\ell+1}\in [m_{\ell+1}]\setminus R_{\ell+1}$.  

For $\ell\in\{2,\dots,k-1\}$ (the case for $\ell=0$ is nearly identical), and $v^{\ell-1}_{R,j_{\ell}}\in V_{\ell-1}^+$, $\Psi_\star(v^{\ell-1}_{R,j_{\ell}})$ is a set of multiple states, as in \eq{k-Psi-v-ell-plus}. To implement $U_{\ell-1,+}:\ket{v^{\ell-1}_{R,j_{\ell}}}\ket{{\cal I}_{\ell}}\mapsto \ket{\psi_\star^{{\cal I}_{\ell}}(v^{\ell-1}_{R,j_{\ell}})}$ for all ${\cal I}_{\ell}\subsetneq [c_{\ell}]$, we note that each of these states is just $\ket{v^{\ell-1}_{R,j_{\ell}}}$ tensored with a constant-sized state depending only on ${\cal I}_{\ell}$, so we can implement $U_{\ell-1,+}$ in $O(1)$ time, by \lem{const-stars}.
	
	Finally, we implement $U_k$. For $v^{k}_{R,i_k} \in V_{k}$, we have $\Psi_\star(v^{k}_{R,i_k})=\{\ket{\psi_\star^G(v^{k}_{R,i_k})}\}$ as in \eq{k-star-vk}, so implementing the map $U_k:\ket{v^{k}_{R,i_k}}\ket{0}\mapsto\propto \ket{\psi_\star^G(v^{k}_{R,i_k})}\propto \ket{v^{k}_{R,i_k}}\ket{\leftarrow}$ is trivial.
	
	\noindent We Thus, conclude that $U_\star$ can be implemented in ${\sf polylog}(n)=\widetilde{O}(1)$ complexity.
\end{proof}

\subsubsection{Tail Bounds on Number of Collisions}\label{sec:k-dist-tail}

If $\overline{R}_1,\dots,\overline{R}_{k-1}$ were uniform random subsets of $A_1,\dots,A_{k-1}$ respectively, it would be simple to argue that, for example, the number of collisions stored in $D_{\ell}(R)$ for any $\ell$, which is a subset of ${\cal K}(\overline{R}_1,\dots,\overline{R}_{\ell})$, is within a constant of the average, with high probability. Since $\overline{R}_{\ell}$ is instead chosen from $A_{\ell}$ by taking $t_{\ell}$ blocks of $A_{\ell}$, and these blocks themselves are not uniform random, but rather chosen by a $d$-wise independent permutation for some $d={\sf polylog}(n)$, proving the necessary bounds, which are needed to upper bound the setup and transitions costs, is somewhat more subtle. 

\begin{lemma}\label{lem:setup-probability}
For any $\ell',\ell \in \{1,\dots,k-1\}$ where $\ell > \ell'$ and for any constant $\kappa$ there exists a constant $c$ such that the following holds. 
If $v_{R}^{\ell-1}$ is chosen uniformly at random from $V_{\ell-1}$, then for any fixed (non-random) $j\in [m_{\ell}]$ we have 
	$$\Pr\left[\abs{{\cal K}(\overline{R}_1,\dots,\overline{R}_{\ell'},A_{\ell}^{(j)})}\geq c\frac{t_{\ell'}}{m_{\ell}}\log^{2^{\ell'-1}} (n)\right]\leq n^{-\kappa}.$$
\end{lemma}
\begin{proof}
	The proof proceeds by induction on $\ell'$.

\noindent\textbf{Base case:}
For $\ell'=1$, we have $R_1$ is a uniform random subset of $[m_1]$ of size $\Theta(t_1)$ (for this proof, we ignore the partition of $R_1$ into $(R_1(S_1))_{S_1}$). Since $A_1$ has been partitioned the trivial way, where each block contains a single element, this means that $\overline{R}_1$ is a uniform random subset of $A_1$ of size $\Theta(t_1)$ as well. Hence, for any $\ell>1$ and fixed $A_\ell^{(j)}$, $Z=|{\cal K}(\overline{R}_1, A_{\ell}^{(j)})|$ is a hypergeometric random variable, where we draw $R_1$ from $[m_1]$ and where we consider any $j_1 \in R_1$ marked whenever the unique element in $A_1^{(j_1)}$ is part of a collision in ${\cal K}(\overline{R}_1, A_{\ell}^{(j)})$. This means $Z$ has parameters $N=m_1= \Theta(n)$, $K=|{\cal K}(A_1, A_{\ell}^{(j)})|$ and $d=|R_1|= \Theta(t_1)$. As mentioned in \sec{k-dist-assumptions}, we may assume for any $A_{\ell}^{(j)}$ that $K = \Theta\left(\abs{A_{\ell}^{(j)}}\right) = \Theta(n/m_{\ell})$. Then for any constant $c$, we have $c(t_1/m_{\ell})\log n\geq 7Kd/N$ for sufficiently large $n$, so by \lem{hypergeo}:
$$ \Pr\left[\abs{{\cal K}(\overline{R}_1, A_{\ell}^{(j)})}\geq c\frac{t_1}{m_{\ell}}\log n\right]\leq e^{-c\frac{t_1}{m_{\ell}}\log n}=n^{-ct_1/m_{\ell}}.$$

\noindent Referring to \tabl{k-dist-setsizes}, we have
$$\frac{t_1}{m_{\ell}} = \Theta\left(\frac{t_1}{t_{\ell-1}}\right) \geq \Omega \left(\frac{t_1}{t_{1}}\right)= \Omega(1).$$
Hence, we can choose $c$ sufficiently large so that $n^{-ct_1/m_{\ell}}\leq n^{-\kappa}$, completing the base case.
	
\noindent\textbf{Induction step:}
Suppose that for some $\ell'-1 \in \{1,\dots,k-2\}$ and all $\ell \in \{\ell',\dots,k-1\}$ the lemma holds; i.e.~for any fixed $j \in [m_{\ell}]$ and for any $\kappa'$ there exists a constant $c'$ such that
$$\Pr\left[\abs{{\cal K}(\overline{R}_1,\dots,\overline{R}_{\ell'-1},A_{\ell}^{(j)})}\geq c'\frac{t_{\ell'-1}}{m_{\ell}}\log^{2^{\ell'-2}} (n)\right]\leq n^{-\kappa'},$$
where $R_1$ is a uniform random subset of $[m_1]$ of size $\Theta(t_1)$, and for all $\ell''\in\{2,\dots,\ell'-1\}$, $R_{\ell''}$ is a uniform random subset of $[m_{\ell''}]$ of size $\Theta(t_{\ell''})$.

Now consider $Z = |{\cal K}(\overline{R}_1,\dots,\overline{R}_{\ell'},A_{\ell}^{(j)})|$, where $\ell > \ell'$ and where $R_{\ell'}$ is a uniform random subset of $[m_{\ell'}]$ of size $\Theta(r_{\ell'})$, and $A_{\ell}^{(j)}$ is still fixed. It is important to remark that this does not imply that $\overline{R}_{\ell'}$ is uniformly random, so instead we look at the blocks of $A_{\ell'}$. We say that any block $A_{\ell'}^{(j')}$ of $A_{\ell'}$ is \emph{marked} if it collides with an $\ell'$-collision in ${\cal K}(\overline{R}_1,\dots,\overline{R}_{\ell'-1},A_{\ell}^{(j)})$, and we let $B_1$ be the random variable that counts the number of such marked blocks in $A_{\ell'}$. Let ${\cal E}_1$ be the event that $|{\cal K}(\overline{R}_1,\dots,\overline{R}_{\ell'-1},A_{\ell}^{(j)})| <  c'\frac{t_{\ell'-1}}{m_{\ell}}\log^{2^{\ell'-2}}$, which happens with probability at least $1-n^{-\kappa'}$, by the induction hypothesis. Assuming ${\cal E}_1$ directly implies an upper bound on $B_1$ of the form 
\begin{equation}\label{eq:upperbound-b}
	B_1\leq c'\frac{t_{\ell'-1}}{m_{\ell}}\log^{2^{\ell'-2}} (n).
\end{equation}
Next we introduce a hypergeometric random variable $B_2$ that counts the number of marked blocks of $R_{\ell'}$, which we draw uniformly from $\{A_{\ell'}^{(j')}\}_{j'\in [m_{\ell'}]}$, which has $B_1$ marked blocks. Conditioned on event ${\cal E}_1$, this means that $B_2$ has parameters $N = m_{\ell'}$, $K = B_1 \leq c'\frac{t_{\ell'-1}}{m_{\ell}}\log^{2^{\ell'-2}} (n)$ (see \eq{upperbound-b}) and $d=|R_{\ell'}|=\Theta(t_{\ell'})$.

To relate $Z = |{\cal K}(\overline{R}_1,\dots,\overline{R}_{\ell'},A_{\ell}^{(j)})|$ to $B_2$, we need to analyse what the effect is of any marked block in $R_{\ell'}$ on the number of collisions in $|{\cal K}(\overline{R}_1,\dots,\overline{R}_{\ell'},A_{\ell}^{(j)})|$. By the induction hypothesis we know that for any block $A_{\ell'}^{(j')}$ of $A_{\ell'}$ and each $\kappa''$, there exists a constant $c''$ such that:
$$\Pr\left[\abs{{\cal K}(\overline{R}_1,\dots,\overline{R}_{\ell'-1},A_{\ell'}^{(j')})}\geq c''\log^{2^{\ell'-2}} (n)\right]\leq n^{-\kappa''}.$$  
If $|{\cal K}(\overline{R}_1,\dots,\overline{R}_{\ell'-1},A_{\ell'}^{(j')})| <  c''\log^{2^{\ell'-2}} (n)$ for every block $A_{\ell'}^{(j')}$ of $A_{\ell'}$, which we denote by event ${\cal E}_2$, then any marked block that gets added to $R_{\ell'}$ results in at most $c''\log^{2^{\ell'-2}} (n)$ collisions in ${\cal K}(\overline{R}_1,\dots,\overline{R}_{\ell'},A_{\ell}^{(j)})$, so $Z \leq B_2c''\log^{2^{\ell'-2}}(n) $. This implies that
\begin{equation}\label{eq:Z-B}
	\Pr\left[Z \geq cc''\frac{t_{\ell'}}{m_{\ell}}\log^{2^{\ell'-1}} (n)\middle| {\cal E}_1 \wedge {\cal E}_2 \right] \leq \Pr\left[B_2 \geq c\frac{t_{\ell'}}{m_{\ell}}\log^{2^{\ell'-1}-2^{\ell'-2}} n\middle| {\cal E}_1\right].
\end{equation}

Since $B_2$ is a hypergeometric random variable, and $c\frac{t_{\ell'}}{m_{\ell}}\log^{2^{\ell'-1}-2^{\ell'-2}} n \geq 7dK/N$ for sufficiently large $c$,
we can use \lem{hypergeo} and $m_{\ell'} = \Theta\left(t_{\ell'-1}\right)$ to derive:
\begin{equation}\label{eq:prob-B-1}
\begin{aligned}
	\Pr\left[B_2 \geq c\frac{t_{\ell'}}{m_{\ell}}\log^{2^{\ell'-1}-2^{\ell'-2}} n\middle| {\cal E}_1\right] &\leq 
 e^{-c\frac{t_{\ell'}}{m_{\ell}}\log^{2^{\ell'-1}-2^{\ell'-2}} n} 
\leq n^{-c t_{\ell'}/m_{\ell}},
\end{aligned}
\end{equation}
since $2^{\ell'-1}-2^{\ell'-2}\geq 1$ whenever $\ell'\geq 2$. By \tabl{k-dist-setsizes}, we have
$$\frac{t_{\ell'}}{m_{\ell}} = \Theta\left(\frac{t_{\ell'}}{t_{\ell-1}}\right) \geq \Omega \left(\frac{t_{\ell'}}{t_{\ell'}}\right)= \Omega(1).$$
Hence, by \eq{Z-B} and \eq{prob-B-1}, we can choose $c$ sufficiently large such that:
\begin{equation}\label{eq:prob-B-2}
	\Pr\left[Z \geq cc''\frac{t_{\ell'}}{m_{\ell}}\log^{2^{\ell'-1}} (n)\middle| {\cal E}_1 \wedge {\cal E}_2\right] \leq n^{-ct_{\ell'}/m_{\ell}} \leq n^{-\kappa'}.
\end{equation}

\noindent The only thing left to do is to use the union bound to upper bound
\begin{equation}\label{eq:prob-events}
	\Pr\left[\neg \left({\cal E}_1 \wedge {\cal E}_2\right) \right] \leq  n^{-\kappa'} +  m_{\ell'}n^{-\kappa''} \leq n^{-\kappa'} + n^{-\kappa''+1}, 
\end{equation}
where in the final inequality we have used the assumption from the lemma that $m_{\ell'}\leq n$.
We can now combine \eq{prob-B-2} and \eq{prob-events}, to conclude that for any $\kappa$ we can choose $\kappa'>\kappa$ and $\kappa'' > \kappa+1$ and a constant $c_{\ell'}$ large enough such that:
\begin{align*}
	\Pr\left[Z \geq c_{\ell}\frac{t_{\ell'}}{m_{\ell}}\log^{2^{\ell'-1}} (n)\right] &\leq \Pr\left[Z \geq c_{\ell}\frac{t_{\ell'}}{m_{\ell}}\log^{2^{\ell'-1}} (n)\middle| {\cal E}_1 \wedge {\cal E}_2 \right] + \Pr\left[\neg \left({\cal E}_1 \wedge {\cal E}_2\right) \right] \\
	&\leq n^{-\kappa'} +  n^{-\kappa'} + n^{-\kappa''+1} \leq n^{-\kappa}.\qedhere
\end{align*}
\end{proof}

\begin{corollary}\label{cor:fault-prob}
For $\ell\in\{2,\dots,k-2\}$, let $v_R^{\ell-1}$ be chosen uniformly at random from $V_{\ell-1}$. If the partition of $A_{\ell+1}$ into $\bigcup_{j\in [m_{\ell+1}]}A_{\ell+1}^{(j)}$ is chosen by a $d$-wise independent permutation for $d=\log^{2^{k-1}}(n)$ (see \sec{k-dist-assumptions}), then for all $j\in [m_{\ell}]$,
	$$\Pr\left[ |{\cal K}(\overline{R}_1,\dots,\overline{R}_{\ell-1},A_{\ell}^{(j)},\overline{R}_{\ell+1})|\geq 1\right] \leq o(1),$$
and for any constant $\kappa$ there exists a constant $c$ such that:
	$$\Pr\left[ |{\cal K}(\overline{R}_1,\dots,\overline{R}_{\ell-1},A_{\ell}^{(j)},\overline{R}_{\ell+1})|\geq c\right] \leq n^{-\kappa},$$
	where we note that $|{\cal K}(\overline{R}_1,\dots,\overline{R}_{\ell-1},A_{\ell}^{(j)},\overline{R}_{\ell+1})|$ is an upper bound on the number of potential faults if we add the block $A_{\ell}^{(j)}$ to $\overline{R}_{\ell}$. 
\end{corollary}
\begin{proof}
	By \lem{setup-probability}, for any $\kappa'$ there exists a $c'$ large enough such that for each fixed $j \in [m_{\ell}]$ we have
	$$\Pr\left[ |{\cal K}(\overline{R}_1,\dots,\overline{R}_{\ell-1},A_{\ell}^{(j)})|\geq c'\log^{2^{\ell-1}}(n)\right]\leq {n^{-\kappa'}}.$$
Let $i^{1},\dots,i^K\in A_{\ell+1}$ be all the points that collide with ${\cal K}(\overline{R}_1,\dots,\overline{R}_{\ell-1},A_{\ell}^{(j)})$, so w.h.p.~there are 
$$K < c'k\log^{2^{\ell-1}}(n) \leq \log^{2^{k-1}}(n)=d$$ 
of them, since each tuple in ${\cal K}(\overline{R}_1,\dots,\overline{R}_{\ell-1},A_{\ell}^{(j)})$ can collide with less than $k$ other indices. Let $\tau:[n]\rightarrow[n]$ be the $d$-wise independent permutation used to choose the partitions of $[n]$, as in \eq{k-dist-partition}. Then $\{\tau(i^1),\dots,\tau(i^K)\}$ is distributed as uniform set of size $K$, as long as $K\leq d$. In that case, the number of elements of $\{i^1,\dots,i^K\}$ that are included in $\overline{R}_{\ell+1}$, $Z$, is a hypergeometric random variable with $N=|A_{\ell+1}|=\Theta(n)$, $K< c'k\log^{2^{\ell-1}}(n)$ as above, and $d=|\overline{R}_{\ell+1}|=\Theta\left(\frac{nt_{\ell+1}}{m_{\ell+1}}\right)$ draws. 
From \tabl{k-dist-setsizes}, we have
	$$\frac{dK}{N}=\Theta\left(\frac{t_{\ell+1}}{m_{\ell+1}}\log^{2^{\ell-1}}(n)\right) = \Theta\left(n^{-\frac{2^{k-\ell-2}}{2^k-1}}\log^{2^{\ell}-1}(n)\right)\leq n^{-\Omega(1)}.$$
Thus, there is some constant $\epsilon$ such that $dK/N\leq n^{-\epsilon}$, so for any constant $c$, we can use the union bound and \cor{hypergeo} to get:
$$\Pr[Z\geq c]\leq 2e^c\left(cn^{\epsilon}\right)^{-c}=2\left(\frac{e}{c}\right)^c n^{-\epsilon c}+n^{-\kappa'}.$$
Hence, by choosing $\kappa' > \kappa$ and $c$ large enough we obtain $\Pr[Z\geq c] \leq n^{-\kappa}.$
Moreover, we also see that for $c=1$:
\begin{equation*}
\Pr[Z\geq 1]\leq O(n^{-\epsilon}+ n^{-\kappa'}) \leq o(1).\qedhere
\end{equation*}
\end{proof}

\begin{corollary}\label{cor:col-prob}
Let $v_{R}^{\ell}$ be chosen uniformly at random from $V_{\ell}$. 
	For any constant $\kappa$, there exist a constant $c$ such that:
	$$\Pr\left[|{\cal K}(\overline{R}_1,\dots,\overline{R}_{\ell})|\geq ct_{\ell}\log^{2^{\ell-2}}(n)\right]\leq n^{-\kappa}.$$
\end{corollary}
\begin{proof}
	By \lem{setup-probability}, for each fixed $j \in [m_{\ell}]$ and constant $\kappa'>\kappa+1$ there exists a $c$ large enough such that
	$$\Pr\left[ |{\cal K}(\overline{R}_1,\dots,\overline{R}_{\ell-1},A_{\ell}^{(j)})|\geq c\log^{2^{\ell-2}}(n)\right]\leq {n^{-\kappa'}}.$$
By the union bound, the probability of this bad event happening for \emph{any} $j\in R_{\ell}$ is at most $t_{\ell} n^{-\kappa'}\leq n^{\kappa}$, since $t_{\ell}<n$, from which the statement follows.
\end{proof}

\subsubsection{The Transition Subroutines}\label{sec:k-dist-transitions}

In this section we show how to implement the transition map  $\ket{u,i}{\mapsto}\ket{v,j}$ for $(u,v)\in\overrightarrow{E}(G)$ with $i=f_u^{-1}(v)$ and $j=f_v^{-1}(u)$ (see \defin{QW-access}). We do this by exhibiting uniform (see \lem{uniform-alg}) subroutines ${\cal S}_1,\dots,{\cal S}_k,{\cal S}_{0,+},\dots,{\cal S}_{k-2,+}$ that implement the transitions in each of the edge sets $E_1$, \dots ,$E_k$, $E_0^+$, \dots, $E_{k-2}^+$ defined in \sec{k-dist-G-edges}, whose union is $\overrightarrow{E}(G)\setminus \tilde{E}$. In \cor{k-dist-transitions}, we will combine these to get a quantum subroutine (\defin{variable-time}) for the full transition map. 

\begin{lemma}\label{lem:k-dist-T1}
For $\ell \in \{0,\dots,k-2\}$, there is a subroutine ${\cal S}_{\ell,+}$ such that for all $(u,v)\in E_{\ell}^+$ with $i=f_u^{-1}(v)$ and $j=f_v^{-1}(u)$, ${\cal S}_{\ell,+}$ maps $\ket{u,i}$ to $\ket{v,j}$ with error 0 in complexity ${\sf T}_{u,v}={\sf T}_{\ell}^+=\widetilde{O}(1)$.
\end{lemma}
\begin{proof} The proof is identical to that of \lem{3-dist-T1}.\end{proof}

\begin{lemma}\label{lem:k-dist-T2}
	There is a uniform subroutine ${\cal S}_1$ such that for all $(u,v)\in E_1$ with $i=f_u^{-1}(v)$ and $j=f_v^{-1}(u)$, ${\cal S}_1$ maps $\ket{u,i}$ to $\ket{v,j}$ with error 0 in complexity ${\sf T}_{u,v}={\sf T}_1=\widetilde{O}(1)$.
\end{lemma}
\begin{proof}
The proof is identical to that of \lem{3-dist-T2}.
\end{proof}

\noindent We now move on to ${\cal S}_{\ell}$, for $\ell \in \{2,\dots,k-1\}$, which is somewhat more complicated.
For $(v^{\ell-1}_{R,j_{\ell}},v^{\ell}_{R'})\in E_{\ell}$, where $v^{\ell-1}_{R,j_{\ell}}\in V_{\ell-1}^+(S_1^*,\dots,S_{\ell-1}^*)$, meaning that $R'$ is obtained from $R$ by inserting $j_{\ell}$ into $R_{\ell}(\mu(S_1^*),\dots,\mu(S_{\ell-1}^*),S_{\ell})$ for some $S_{\ell}$, ${\cal S}_{\ell}$ should act as:
\begin{equation}
\begin{split}
	\ket{v^{\ell-1}_{R,j_{\ell}},S_{\ell}} &\mapsto \ket{v^{\ell}_{R'},j_{\ell}}\\
\equiv \ket{(({\ell-1},+),R,D(R),j_{\ell}),S_{\ell}} &\mapsto \ket{(\ell,R',D(R')),j_{\ell}}.
	\end{split}\label{eq:kT3}
\end{equation}
The complexity of this map depends on $|{\cal K}(\overline{R}_1,\dots,\overline{R}_{\ell-1},A_{\ell}^{(j_{\ell})})|$, which is less than $p_{\ell}\in{\sf polylog}(n)$ whenever $(v^{\ell-1}_{R,j_{\ell}},v^{\ell}_{R'})\in E_{\ell}\setminus \tilde{E}_{\ell}=E_{\ell}\setminus\tilde{E}$ (see \eq{k-dist-tilde-E-ell} and \eq{k-dist-tilde-E}). \lem{k-dist-k1} below describes how to implement this transition map, up to some error, in that case. For the case when $(v^{\ell-1}_{R,j_{\ell}},v^{\ell}_{R'})\in \tilde{E}_{\ell}\subset \tilde{E}$ we let the algorithm fail.

\begin{lemma}\label{lem:k-dist-k1}
	Fix any constant $\kappa$. For each $\ell\in\{2,\dots,k-1\}$, there is a uniform subroutine ${\cal S}_{\ell}$ that implements the transition map that maps $\ket{u,i}$ to $\ket{v,j}$ for all $(u,v)\in E_{\ell}\setminus\tilde{E}$ with error $O(n^{-\kappa})$, in complexity ${\sf T}_{u,v}={\sf T}_{\ell}=\widetilde{O}(\sqrt{n/m_{\ell}})$. 
\end{lemma}
\begin{proof}
Suppose $u=v^{\ell}_{R,j_{\ell}}\in V_{\ell-1}^+(S_1^*,\dots,S_{\ell-1}^*)$. We can compute the values $S_1^*,\dots,S_{\ell-1}^*$ by checking which sets are larger, or just keeping track of these values in some convenient way, as they are chosen. Then to implement the map in \eq{kT3}, we need to insert $j_{\ell}$ into $R_{\ell}(\mu(S_1^*),\dots,\mu(S^*_{\ell-1}),S_{\ell})$ to obtain $R_{\ell}'$, update $D(R)$ to obtain $D(R')$, uncompute $S_{\ell}$ by checking which part of $R_{\ell}$ has size $t_{\ell}+1$, and increment the first register by mapping $\ket{\ell-1,+}\mapsto \ket{\ell}$. All of these take ${\sf polylog}(n)$ complexity, except for updating $D(R)$, which we now describe.

By \eq{k-dist-data}, we know that $D(R)$ consists of sets $\{D_{\ell'}(R)\}_{\ell'=1}^{k-1}$ where each $D_{\ell'}(R)$ contains a subset of ${\cal K}(\overline{R}_1,\dots,\overline{R}_{\ell'})$ (see \eq{D-ell}). When we go from $R$ to $R'$, we need to update each of these to account for any collisions involving indices $i_{\ell}\in A_{\ell}^{(j_{\ell})}$ that should be recorded in $D_{\ell'}(R')$. For $\ell'<\ell$, we can see that $D_{\ell'}(R)=D_{\ell'}(R')$, since $D_{\ell'}$ only depends on $R_1,\dots,R_{\ell'}$, which are unchanged. For $\ell'>\ell$, the existence of any $(i_1,\dots,i_{\ell-1},i_{\ell},i_{\ell+1},\dots,i_{\ell'},x_{i_1})\in D_{\ell'}(R')$ (which we would now need to find and add) implies $d_R^{\rightarrow}(j_{\ell})>0$ (see \eq{forward-collision-degree}), and this is not true for any $(v_{R,j_{\ell}}^{\ell-1},v_{R'}^{\ell})\in E_{\ell}$ (see \eq{k-dist-E_ell}). Thus, we only need to find any tuples $(i_1,\dots,i_{\ell},x_{i_1})$ such that $i_{\ell}\in A_{\ell}^{(j_{\ell})}$ that belong in $D_{\ell}(R')$. By \eq{D-ell}, such a tuple should be added to $D_{\ell}(R')$ if and only if $(i_1,\dots,i_{\ell-1},x_{i_1})\in D_{\ell-1}(R_{\ell-1}(\mu(S_1^*),\dots,\mu(S_{\ell-2}^*),S_{\ell-1}))$ such that $x_{i_{\ell}}=x_{i_1}$ for some $S_{\ell-1}$ such that $\mu(S_{\ell-1}^*)\in S_{\ell-1}$.

We search for $i_{\ell}\in A_{\ell}^{(j_{\ell})}$ such that if we look up $x_{i_{\ell}}$ in $D_{\ell-1}(R_{\ell-1}(\mu(S_1^*),\dots,\mu(S_{\ell-2}^*),S_{\ell-1}))$ for some $S_{\ell-1}$ containing $\mu(S_{\ell-1}^*)$, we get back a non-empty set of values $(i_1,\dots,i_{\ell-1},x_{i_{\ell}})$. For any such value found, we add $(i_1,\dots,i_{\ell},x_{i_{\ell}})$ to $D_{\ell}(R_{\ell}(\mu(S_1^*),\dots,\mu(S_{\ell-1}^*),S_{\ell-1}))$. This increments the value of $\bar{d}_{R'}^{\rightarrow}(i_{\ell-1})$, and so if $j_{\ell-1}\in R_{\ell-1}$ is such that $i_{\ell-1}\in A_{\ell-1}^{(j_{\ell-1})}$ (we can compute $j_{\ell-1}$ from $i_{\ell-1}$ in $\widetilde{O}(1)$, see \sec{k-dist-assumptions}), we have incremented the forward collision degree of $j_{\ell-1}$, $d_{R'}^{\rightarrow}(j_{\ell-1})$. We must therefore update the entry in $C^{\rightarrow}_{\ell-1}$ for $j_{\ell-1}$. We look up $j_{\ell-1}$, and if nothing is returned, insert $(j_{\ell-1},0)$. If $(j_{\ell-1},N)$ is returned, remove it and insert $(j_{\ell-1},N+1)$. 
We repeat this quantum search procedure, which is uniform, until we find $p_{\ell}={\sf polylog}(n)$ values $i_{\ell}$, or no new $i_{\ell}$ is returned for $\kappa\log n$ times. Since  we are assuming that the number of such collisions is less than $p_{\ell}$, since $(u,v)\in E_{\ell}\setminus\tilde{E}_{\ell}$, this finds all collisions with error $O(n^{-\kappa})$, in complexity 
	$\widetilde{O}\left(\sqrt{|A_{\ell}^{(j_{\ell})}|}\right)=\widetilde{O}(\sqrt{n/m_{\ell}})$.
\end{proof}

\noindent We have the following corollary of the results in \sec{k-dist-tail}.
\begin{corollary}\label{cor:ktilde-E}
	For any constant $\kappa$, there exists a choice of constants $\{c_{\ell}\}_{\ell=2}^{k-2}$ in the definition of $\tilde{E}_{\ell}'$ (\eq{k-dist-tilde-E-prime}) and polylogarithmic functions $\{p_{\ell}\}_{\ell=1}^{k-1}$ in the definition of $\tilde{E}_{\ell}$ (\eq{k-dist-tilde-E-ell}) large enough such that 
	$$\widetilde{\cal W} := \sum_{e\in \tilde{E}}\w_e = O(n^{-\kappa}{\cal W}(G)).$$ 
\end{corollary}
\begin{proof}
Fix $\ell\in\{2,\dots,k-1\}$.
Let $v_R^{\ell-1}$ be uniform random on $V_{\ell-1}$ (so $R$ is uniform on its support, see \eq{V-ell-Ss} and \eq{V-ell-union}). Then by \lem{setup-probability}, for any $j_{\ell}\in [m_{\ell}]$, if $p_{\ell}\in{\sf polylog}(n)$ is sufficiently large, 
\begin{equation}
\Pr[|{\cal K}(\overline{R}_1,\dots,\overline{R}_{\ell-1},A_{\ell}^{(j_{\ell})})|\geq p_{\ell}]\leq n^{-\kappa}.\label{eq:k-dist-pl}
\end{equation}
Referring to \eq{k-dist-tilde-E-ell}, this implies that 
$$|\tilde{E}_{\ell}|\leq n^{-\kappa}|\{(u,v):u\in V_{\ell-1}^+, v\in L^+(u)\}|=n^{-\kappa}|V_{\ell-1}^+|.$$

For $\ell\in\{2,\dots,k-2\}$, by \cor{fault-prob}, if $c_{\ell}$ is a sufficiently large constant,  
$$\Pr\left[|{\cal K}(\overline{R}_1,\dots,\overline{R}_{\ell-1},A_{\ell}^{(j_{\ell})},\overline{R}_{\ell+1})| \geq c_{\ell} \right]\leq n^{-\kappa},$$
which implies that $|{\cal I}(v_{R,j_{\ell}})|\geq c_{\ell}$ with probability at most $n^{-\kappa}$ (see \eq{k-dist-I-i} and \eq{k-dist-cal-I}). 
Referring to \eq{k-dist-tilde-E-prime}, this implies that
$$|\tilde{E}_{\ell}'|\leq n^{-\kappa}|\{(u,v):u\in V_{\ell-1}^+, v\in L^+(u)\}|=n^{-\kappa}|V_{\ell-1}^+|.$$
Since $\tilde{E}_{k-1}=\emptyset$, the above also holds for $\ell=k-1$.

Combining these, and using
the definition of $\tilde{E}$ in \eq{k-dist-tilde-E},
and that $|V_{\ell-1}^+|=\Theta(|E_{\ell}|+|\tilde{E}_{\ell}'|)$, since each vertex in $V_{\ell-1}^+$ has constant out-degree,
we have:
\begin{align*}
\widetilde{\cal W} &= \sum_{\ell=2}^{k-1}\w_{\ell}|V_{\ell-1}^+|
\leq 2n^{-\kappa}\sum_{\ell=2}^{k-1}\w_{\ell}O(|E_{\ell}|+|\tilde{E}_{\ell}'|) 
=O\left(n^{-\kappa}\sum_{e\in\overrightarrow{E}(G)}\w_{e}\right)
= O(n^{-\kappa}{\cal W}(G)).\qedhere
\end{align*}
\end{proof}

\begin{lemma}\label{lem:k-dist-final}
	There is a uniform subroutine ${\cal S}_k$ such that for all $(u,v)\in E_k$ with $i=f_u^{-1}(v)$ and $j=f_v^{-1}(u)$, ${\cal S}_k$ maps $\ket{u,i}$ to $\ket{v,j}$ with error 0 in complexity ${\sf T}_{u,v}={\sf T}_k=\widetilde{O}(1)$.
\end{lemma}
\begin{proof}
The proof is identical to that of \lem{3-dist-T1}.
\end{proof}

\noindent We combine the results of this section into the following.
\begin{corollary}\label{cor:k-dist-transitions}
Let $\kappa$ be any constant. There is a quantum subroutine (\defin{variable-time}) that implements the full transition map with errors $\epsilon_e\leq n^{-\kappa}$ for all $e\in \overrightarrow{E}(G)\setminus\tilde{E}$, and times: ${\sf T}_e={\sf T}_1=\widetilde{O}(1)$ for all $e\in E_1$; ${\sf T}_e={\sf T}_{\ell}^+=\widetilde{O}(1)$ for all $e\in E_{\ell}^+$, for all $\ell\in\{0,\dots,k-2\}$; ${\sf T}_e={\sf T}_{\ell}=\widetilde{O}(\sqrt{n/m_{\ell}})$ for all $e\in E_{\ell}$, for all $\ell\in\{2,\dots,k-1\}$; and ${\sf T}_e={\sf T}_k=\widetilde{O}(1)$ for all $e\in E_k$. 
\end{corollary}
\begin{proof}
This follows from combining \lem{k-dist-T1}, \lem{k-dist-T2}, \lem{k-dist-k1} and \lem{k-dist-final} using  \lem{uniform-combine}.
\end{proof}

\subsubsection{Initial State and Setup Cost}

The initial state will be the uniform superposition over $V_0$: 
\begin{equation*}
	\ket{\sigma} := \sum_{v^0_{R} \in V_0}\frac{1}{\sqrt{|V_0|}}\ket{v^0_{R}}.
\end{equation*}

\begin{lemma}\label{lem:setupk}
	The state $\ket{\sigma}$ can be generated with error $O(n^{-\kappa})$ for any constant $\kappa$ in complexity
	\begin{equation*}
		{\sf S}=\widetilde{O}\left(t_1+t_2\sqrt{\frac{n}{t_1}} + \cdots + t_{k-1}\sqrt{\frac{n}{t_{k-2}}}\right).
	\end{equation*}
\end{lemma}
\begin{proof}
Fix $p\in {\sf polylog}(n)$ and a constant $c$.
	We start by taking a uniform superposition over all $R_1\in \binom{[m_1]}{t_1^{(2^{c_{1}}-1)}}$ and querying each $\overline{R}_1$ to get $D_1(R)$, which costs $\widetilde{O}(t_1)$ (with log factors coming from the cost of inserting everything into data structures as in \sec{data}). For $\ell\in\{2,\dots,k-1\}$, we take a uniform superposition over all sets $R_{\ell} \in \binom{[m_{\ell}]}{t_{\ell}^{(c_1\cdots c_{\ell-1}(2^{c_{\ell}}-1))}}$. The total cost so far is $\widetilde{O}(t_1+t_2 \cdots + t_{k-1})$. 

\noindent Next, we need to populate the rest of the data structure:
\begin{enumerate}
	\item[] For each $\ell\in\{2,\dots,k-1\}$, do the following. 
	\begin{enumerate}
		\item[] For each $(s_1,\dots,s_{\ell-1},S_{\ell})\in [c_1]\times\dots\times[c_{\ell-1}]\times (2^{[c_{\ell}]}\setminus\{\emptyset\})$, do the following. 
		\begin{enumerate}
			\item[] Repeat until $p t_{\ell}$ values $i_{\ell}$ have been found, or $c\log n$ repetitions have passed in which no $i_{\ell}$ was found:
			\begin{enumerate}
				\item[] Search for a new value $i_{\ell}\in R_{\ell}(s_1,\dots,s_{\ell-1},S_{\ell})$ such that there exists 
$$(i_1,\dots,i_{\ell-1},x_{i_{\ell}})\in D_{\ell-1}(R_{\ell-1}(s_1,\dots,s_{\ell-2},S_{\ell-1}))$$
 for some $S_{\ell-1}$ containing $s_{\ell-1}$.  If such an $i_{\ell}$ is found, insert the tuple $(i_1,\dots,i_{\ell},x_{i_{\ell}})$ into $D_{\ell}(R_{\ell}(s_1,\dots,s_{\ell-1},S_{\ell})$, and increment the forward collision degree of $j_{\ell-1}$ such that $i_{\ell-1}\in A_{\ell-1}^{(j_{\ell-1})}$ stored in $C^{\rightarrow}_{\ell-1}(R)$, as described in the proof of \lem{k-dist-k1}.
			\end{enumerate}
		\end{enumerate}
	\end{enumerate}
\end{enumerate}
If the inner loop finds $Y\in [pt_{\ell}]$ values, so $Y=\widetilde{O}(t_{\ell})$, it costs at most (up to polylogarithmic factors):
$$\sum_{y=0}^{Y-1}\sqrt{\frac{|R_{\ell}|}{Y-y}} = \sqrt{\frac{t_{\ell}n}{m_{\ell}}}\sum_{y=1}^Y\frac{1}{\sqrt{y}} = \Theta\left(\sqrt{\frac{t_{\ell}n Y}{m_{\ell}}}\right)
=\widetilde{O}\left(t_{\ell}\sqrt{\frac{n }{m_{\ell}}}\right),$$
since $|R_{\ell}|=\Theta(t_{\ell}n/m_{\ell})$.
Since $k$, $c_1,\dots,c_{k-1}$ are all constant, there are $\widetilde{O}(1)$ loops in total, so the total cost of this procedure is
$O(\sqrt{n/m_{\ell}})$
for a total cost of:
$$\widetilde{O}\left(\sum_{\ell=1}^{k-1}t_{\ell}+\sum_{\ell=2}^{k-1}t_{\ell}\sqrt{\frac{n}{m_{\ell}}}\right)=\widetilde{O}\left(t_1+\sum_{\ell=1}^{k-2}t_{\ell+1}\sqrt{\frac{n}{t_{\ell}}}\right)$$
since $t_{\ell}=\Theta(m_{\ell+1})$ and for all $\ell>1$, $t_{\ell}=o(t_1)$ (see \tabl{k-dist-setsizes}).

In parts of the superposition in which there are more than $pt_{\ell}$ collisions to be found in some inner loop, we have failed to correctly populate the data $D(R)$, and so the state is not correct. We now argue that this represents a very small part of the state. For uniform random sets $\overline{R}_1,\dots,\overline{R}_{k-1}$, we could argue that the expected number of $\ell$-collisions in ${\cal K}(\overline{R}_1,\dots,\overline{R}_{\ell})$ is $\Theta(t_{\ell})$, and use a hypergeometric tail inequality to upper bound the proportion of $R$ for which this failure occurs. Things are more complicated, since the sets $\overline{R}_{\ell}$ for $\ell>1$ are not uniform on all possible sets -- they are composed instead of blocks. However, by \cor{col-prob}, for every $\ell\in \{2,\dots,k-1\}$, if $v_R^{\ell}$ is uniform random on $V_{\ell}$, meaning $R_1,\dots,R_{k-1}$ are uniform random sets, but $\overline{R}_1,\dots,\overline{R}_{k-1}$ have limited support, we still have the necessary tail bound, when $c'$ is a sufficiently large constant: 
	$$\Pr\left[|{\cal K}(\overline{R}_1,\dots,\overline{R}_{\ell})|\geq t_{\ell}c'\log^{2^{\ell-2}}(n)\right]\leq n^{-\kappa}.$$
Thus, choosing $p=c'\log^{2^{\ell-1}}(n)$, the state we generate is $O(n^{-\kappa})$-close to $\ket{\sigma}$.
\end{proof}

 \subsubsection{Positive Analysis}

For the positive analysis, we must exhibit a flow (see \defin{flow}) on $G$ whenever $M\neq \emptyset$. 

\begin{lemma}\label{lem:k-dist-positive}
	There exists some ${\cal R}^{\sf T}=O(\abs{V_0}^{-1})$ such that the following holds. Whenever there is a unique $k$-collision $(a_1,\dots,a_k)\in A_1\times \cdots \times A_k$, there exists a flow $\theta$ on $G$ that satisfies conditions \textbf{P1}-\textbf{P5} of \thm{full-framework}. Specifically:
	\begin{enumerate}
		\item For all $e\in\tilde{E}$, $\theta(e)=0$.
		\item For all $u\in V(G)\setminus (V_0\cup V_k)$ and $\ket{\psi_\star(u)}\in \Psi_\star(u)$, 
		$$\sum_{i\in L^+(u)}\frac{\theta(u,f^+_u(i))\braket{\psi_\star(u)}{u,i}}{\sqrt{\w_{u,i}}}-\sum_{i\in L^-(u)}\frac{\theta(u,f^-_u(i))\braket{\psi_\star(u)}{u,i}}{\sqrt{\w_{u,i}}}=0.$$
		\item $\sum_{u\in  V_0}\theta(u)=1$.
		\item $\sum_{u\in V_0}\frac{|\theta(u)-\sigma(u)|^2}{\sigma(u)}\leq 1$.
		\item ${\cal E}^{\sf T}(\theta)\leq {\cal R}^{\sf T}$.
	\end{enumerate}
\end{lemma}
\begin{proof}
	Recall the definition of $M$ from \eq{k-dist-M}. For $\ell \in \{2,\dots,k-1\}$, let $j^*_{\ell}\in [m_{\ell}]$ be the unique block label such that $a_{\ell}\in A_{\ell}^{(j_{\ell}^*)}$. Then $a_{\ell}\in \overline{R}_{\ell}(s_1,\dots,s_{\ell-1},S_{\ell})$ if and only if $j_{\ell}^*\in R_{\ell}(s_1,\dots,s_{\ell-1},S_{\ell})$. 
	
	Assuming $M\neq\emptyset$, we define a flow $\theta$ on $G$ with all its sinks in $M$. It will have sources in both $V_0$ and $M$, but all other vertices will conserve flow. This will imply \textbf{Item 2} for all \emph{correct} star states of $G$, but we take extra case to ensure that \textbf{Item 2} is satisfied for the additional star states in $\Psi_\star(u):u\in \bigcup_{\ell=0}^{k-2}V_{\ell}^+$. We define $\theta$ from $V_0$ to $V_k$ as follows.

\vskip7pt
\noindent\textbf{${\cal R}_{0}^+$, Item 3, and Item 4:} We define $M_0$ as the set of $v^{0}_{R}\in V_0$ such that for all $\ell \in \{2,\dots,k-1\}$, we have $|{\cal K}(\overline{R}_1,\dots,\overline{R}_{\ell-1},A_{\ell}^{(j^*_{\ell})})|< p_{\ell}$, where $p_{\ell}$ is as in \cor{ktilde-E}, and for all $\ell\in\{2,\dots,k-2\}$, $|{\cal K}(\overline{R}_1,\dots,\overline{R}_{\ell-1},A_{\ell}^{(j_{\ell}^*)},\overline{R}_{\ell+1})| = 0$. We define the flow $\theta$ over the edges in $E_0^{+}$ as 
		$$ \theta\left(v^{0}_{R},v^{0}_{R,j_1}\right) = \begin{cases}
			\frac{1}{\abs{M_0}} &\text{ if } v_{R}^0 \in M_0 \text{ and } A^{(j_1)} = \{a_1\},\\
			0 &\text{ else}.
		\end{cases}$$
		That is, each vertex in $M_0$ has a unique outgoing edge with flow, and the flow is uniformly distributed. From this construction we immediately satisfy \textbf{Item~3}. 
		
		By \cor{fault-prob}, we know that the proportion of vertices $v_R^0\in V_0$ that are excluded from $M_0$ because 
		$|{\cal K}(\overline{R}_1,\dots,\overline{R}_{\ell-1},A_{\ell}^{(j_{\ell}^*)},\overline{R}_{\ell+1})| \geq 1$ is $o(1)$. 
By \eq{k-dist-pl}, the proportion of vertices excluded because $|{\cal K}(\overline{R}_1,\dots,\overline{R}_{\ell-1},A_{\ell}^{(j_{\ell}^*)})|\geq p_{\ell}$ is also $o(1)$. 
Hence, we can compute:
		\begin{equation*}
			\frac{\abs{V_0}}{\abs{M_0}} = \left(1+O\left(\frac{t_1}{n}\right)\right)\Pi_{\ell=2}^{k-1}\left(1+O\left(\frac{t_\ell}{m_\ell}\right)\right)\left(1 + o(1)\right) = 1+o(1).
		\end{equation*}
		Since $\sigma(u)=\frac{1}{\abs{V_0}}$, we can conclude with \textbf{Item 4} of the theorem statement:
		\begin{equation*}
			\begin{split}
				\sum_{u\in V_0}\frac{|\theta(u)-\sigma(u)|^2}{\sigma(u)} 
				&= \abs{V_0}^2\left(\frac{1}{\abs{M_0}}-\frac{1}{\abs{V_0}}\right)^2
				= \left(\frac{\abs{V_0}}{\abs{M_0}}-1\right)^2=o(1).
			\end{split}
		\end{equation*}
Recall we want to compute ${\cal E}^{\sf T}(\theta)={\cal E}(\theta^{\sf T})$ (see \defin{nwk-length}), which treats an edge $e$ as a path of length ${\sf T}_{e}$.
Using ${\sf T}_0^+=\widetilde{O}(1)$ and $\w_0^+=1$ (refer to \tabl{k-dist-edges}), the contribution of the edges in $E_0^{+}$ to the energy of the flow can be computed as:
		\begin{equation}
			{\cal R}_0^+ = \sum_{e\in E_0^+}{\sf T}_0^+\frac{\theta(e)^2}{\w_0^+} = \widetilde{O}\left(\sum_{u\in M_0}\frac{1}{\abs{M_0}^2}\right)=\widetilde{O}\left(\frac{1}{\abs{M_0}}\right),\label{eq:k-R0+}
		\end{equation}
		since each vertex in $M_0$ has a unique outgoing edge with flow and the flow is uniformly distributed.

\vskip7pt		
\noindent\textbf{${\cal R}_1$ and Item 2 (partially):} Let $M_0^+$ be the set of $v^{0}_{R,j_1} \in V_0^+$ such that $v^{0}_{R}\in M_0$ and $A^{(j_1)}=\{a_1\}$, so $\abs{M_0^+}=\abs{M_0}$. These are the only vertices in $V_0^+$ that have incoming flow, which is equal to $\frac{1}{\abs{M_0}}$. Note that no fault can occur when we add $a_1$ to $R_1$ because we have ensured that $a_2\not\in \overline{R}_2$; that is ${\cal I}(v^0_{R,a_1})=\emptyset$, and so by \lem{k-dist-E1}, $\w_{v^0_{R,a_1},S_1}=\w_1=1$ for all $S_1\in 2^{[c_1]}\setminus\{\emptyset\}$. 
To ensure that we satisfy \textbf{Item 2} we define the flow as
		$$\theta(v^0_{R,j_1},v^1_{R'}) = \begin{cases}
			(-1)^{\abs{S_1}+1}\frac{1}{\abs{M_0}} &\text{ if } v^0_{R,j_1} \in M_0^+\mbox{ and }v_{R'}^1=f^+_{v_{R,j_1}^0}(S_1),\\
			0 &\text{ else}
		\end{cases}$$
where we recall that $v_{R'}^1=f^+_{v_{R,j_1}^0}(S_1)$ if and only if $R'$ is obtained from $R$ by inserting $j_1$ into $R_1(S_1)$. 
		We verify that indeed for each $u=v_{R,a_1}^0 \in M_0^+$ and $\ket{\psi^{{\cal I}_1}_\star(u)}\in \Psi_\star(u)$ (see \eq{k-Psi-v-0-plus}) Item 2 holds:
\begin{equation}
\begin{split}
\Theta_\star({\cal I}_1,u):={}&\sum_{i\in L^+(u)}\theta(u,f^+_{u}(i))\frac{\braket{\psi_\star^{{\cal I}_1}(u)}{u,i}}{\sqrt{\w_1}}-\sum_{i\in L^-(u)}\theta(u,f^-_{u}(i))\frac{\braket{\psi_\star^{{\cal I}_1}(u)}{u,i}}{\sqrt{\w_0^+}}\\
={}& \sum_{S_1 \in 2^{[c_1]\setminus {\cal I}_1} \setminus \{\emptyset\}}\theta(u,f^+_{u}(S_1))\frac{\sqrt{\w_1}}{\sqrt{\w_1}}-\theta(u,f^-_{u}(\leftarrow))\frac{-\sqrt{\w_0^+}}{\sqrt{\w_0^+}}.
\end{split}\label{eq:k-dist-pos-flow-1}
\end{equation}
We have $f^-_u(\leftarrow)=v_R^0\in V_0$, and $\theta(v_{R,a_1}^0,v_R^0)=-\theta(v_R^0,v_{R,a_1}^0)=-|M_0|^{-1}$, and $\theta(u,f^+_u(S_1))=(-1)^{|S_1|}|M_0|^{-1}$, so we continue from above:
\begin{equation}
\begin{split}
\Theta_\star({\cal I}_1,u)&= \sum_{S_1 \in 2^{[c_1]\setminus {\cal I}_1} \setminus \{\emptyset\}}(-1)^{|S_1|+1}|M_0|^{-1}-(-|M_0|^{-1})(-1)\\
&= -|M_0|^{-1}\left(\sum_{S_1\in 2^{[c_1]\setminus {\cal I}_1}}(-1)^{|S_1|}-1+1 \right)=0,
\end{split}\label{eq:k-dist-pos-flow-2}
\end{equation}
since $\sum_{S_1\in 2^{[c_1]\setminus {\cal I}_1}}(-1)^{|S_1|}=0$ (i.e.~for any set $S$, exactly half of its subsets have even size). 
Using ${\sf T}_1=\widetilde{O}(1)$ and ${\sf w}_1=1$, the contribution of the edges in $E_1$ to the energy of the flow can be upper bounded as:
		\begin{equation}
			{\cal R}_1 = \sum_{e\in E_1}{\sf T}_1\frac{\theta(e)^2}{\w_1} = \widetilde{O}\left(\sum_{\substack{u \in M_0^+, S_1 \in 2^{[c_1]} \setminus \{\emptyset\}}}\frac{1}{\abs{M_0}^2}\right)=\widetilde{O}\left(\frac{1}{\abs{M_0}}\right).\label{eq:k-R1}
		\end{equation}
		
\vskip7pt
\noindent\textbf{${\cal R}_{\ell}^+$ for $\ell \in \{1,\dots,k-2\}$:} Let $M_{\ell}(S_1,\dots,S_{\ell})$ be the set of $v^{\ell}_{R}\in V_{\ell}(S_1,\dots,S_{\ell})$ (see \eq{V-ell-Ss}) such that $a_1\in R_1(S_1)$, for all $\ell'\in\{2,\dots,\ell\}$, $j_{\ell}^*\in R_{\ell}(\mu(S_1),\dots,\mu(S_{\ell-1}),S_{\ell})$, and 
$$v^0_{R_1\setminus\{a_1\},R_2\setminus\{j_1^*\},\dots,R_{\ell}\setminus\{j_{\ell}^*\},R_{\ell+1},\dots,R_{k-1}}\in M_0.$$
Then letting $M_{\ell}$ be the union of all $M_{\ell}(S_1,\dots,S_{\ell})$, we have $|M_{\ell}|=\Theta(|M_0|)$. 
We will define $\theta$ so that $M_{\ell}$ are exactly the vertices of $V_{\ell}$ that have non-zero flow coming in from $V_{\ell-1}^+$, and specifically, we will ensure that the amount of incoming flow for each $v^{\ell}_R\in M_{\ell}(S_1,\dots,S_{\ell})$ is $(-1)^{\abs{S_1} + \cdots +\abs{S_{\ell}}+\ell}{\abs{M_0}}^{-1}$. So far this can only be verified for $\ell=1$ due to the flow that we constructed on $E_1$, but it will follow for all $\ell \in \{2,\dots,k-1\}$ when we define the flow on $E_{\ell}$ (see \eq{k-dist-theta-later}). For now, we define the flow $\theta$ over the edges in $E_{\ell}^+$ as 
		\begin{align*} 
			\theta(v^{\ell}_{R},v^{\ell}_{R,j_{\ell}})
			&= \begin{cases}
				\left(-1\right)^{\abs{S_1} + \cdots + \abs{S_{\ell}}+\ell}\frac{1}{\abs{M_0}} &\text{ if } v^{\ell}_R \in M_{\ell}(S_1,\dots,S_{\ell})\text{ and } j_{\ell} = j^*_{\ell},\\
				0 &\text{ else},
			\end{cases}
		\end{align*}
so we are just forwarding all flow from $v_R^{\ell}$ to a unique neighbour $v_{R,j_{\ell}^*}^{\ell}$.		
Using ${\sf T}_{\ell}^+=\widetilde{O}(1)$ and $\w_{\ell}^+=1$, the contribution of the edges in $E_{\ell}^+$ to the energy of the flow can be upper bounded as:
		\begin{equation}
			{\cal R}_{\ell}^+ = \sum_{e\in E_{\ell}^{+}}{\sf T}_{\ell}^+\frac{\theta(e)^2}{\w_{\ell}^{+}} = \widetilde{O}\left(\sum_{u\in M_{\ell}^{+}}\frac{1}{\abs{M_0}^2}\right)=\widetilde{O}\left(\frac{1}{\abs{M_0}}\right).\label{eq:k-Rell+}
		\end{equation}

\vskip7pt	
\noindent\textbf{${\cal R}_{\ell}$ for $\ell \in \{2,\dots,k-1\}$ and Item 2 (continued):} Let $M_{\ell-1}^{+}(S_1,\dots,S_{\ell-1})$ be the set of $v^{\ell-1}_{R,j^*_{\ell}}\in V_{\ell-1}^+(S_1,\dots,S_{\ell-1})$ such that $v^{\ell-1}_{R}\in M_{\ell-1}(S_1,\dots,S_{\ell-1})$, so letting $M_{\ell-1}^+$ be the union over all the sets $M_{\ell-1}^+(S_1,\dots,S_{\ell-1})$, $\abs{M_{\ell-1}^{+}} = O\left(\abs{M_{\ell-1}}\right) =  O\left(\abs{M_0}\right)$. $M_{\ell-1}^+(S_1,\dots,S_{\ell-1})$ are exactly the vertices of $V_{\ell-1}^+$ that have non-zero flow coming in from $M_{\ell-1}(S_1,\dots,S_{\ell})$. For any $v^{\ell-1}_{R,j^*_{\ell}}\in M_{\ell-1}^+(S_1,\dots,S_{\ell-1})$, this flow is equal to $(-1)^{\abs{S_1} + \cdots +\abs{S_{\ell-1}}+(\ell-1)}{\abs{M_0}}^{-1}$. Note that no fault can occur when we add $j_{\ell}^*$ to $R$, because we have ensured in our definition of $M_0$ that ${\cal K}(\overline{R}_1,\dots,\overline{R}_{\ell-1},A_{\ell}^{(j_{\ell}^*)},\overline{R}_{\ell+1})=\emptyset$, so we have ${\cal I}(v_{R,j_{\ell}^*}^{\ell-1})=\emptyset$, so by \eq{k-dist-w-ell-cond}, there is an edge for each $S_{\ell}\in 2^{[c_{\ell}]}\setminus\{\emptyset\}$ to which we can assign flow. 
To ensure that we satisfy \textbf{Item 2} we define the flow as
		\begin{align}\label{eq:k-dist-theta-later}
			 \theta(v^{\ell-1}_{R,j_{\ell}},v_{R'}^{\ell}) 
= \begin{cases}
			\left(-1\right)^{\abs{S_1} + \cdots +\abs{S_{\ell}}+\ell}\frac{1}{\abs{M_0}} &\text{ if } v_{R,j_{\ell}} \in M_{\ell-1}^+(S_1,\dots,S_{\ell-1}) \text{ and } v_{R'}^{\ell}=f^+_{v^{\ell-1}_{R,j_{\ell}}}(S_{\ell}),\\
			0 &\text{ else},
			\end{cases}
		\end{align}
where we recall that for $v_{R,j_{\ell}}^{\ell-1}\in V_{\ell-1}^+(S_1,\dots,S_{\ell-1})$, $v_{R'}^{\ell}=f^+_{v_{R,j_{\ell}}^{\ell-1}}(S_{\ell})$ if and only if $R'$ is obtained from $R$ by inserting $j_{\ell}$ into $R_{\ell}(\mu(S_1),\dots,\mu(S_{\ell-1}),S_{\ell})$. Note that this is consistent with the incoming flow we assumed when defining $\theta$ on the edges in $E_{\ell-1}^+$, above. 
We verify that for each $u=v_{R,j_{\ell}^*}^{\ell-1}\in M_{\ell-1}^+(S_1,\dots,S_{\ell-1})$ and $\ket{\psi^{{\cal I}_{\ell}}_\star(u)}\in \Psi_\star(u)$ (see \eq{k-Psi-v-ell-plus}), Item 2 holds. By a computation nearly identical to \eq{k-dist-pos-flow-1} and \eq{k-dist-pos-flow-2}, we obtain:
\begin{align*}
			&\sum_{i\in L^+(u)}\theta(u,f^+_{u}(i))\frac{\braket{\psi_\star^{{\cal I}_{\ell}}(u)}{u,i}}{\sqrt{\w_{\ell}}}-\sum_{i\in L^-(u)}\theta(u,f^- _{u}(i))\frac{\braket{\psi_\star^{{\cal I}_{\ell}}(u)}{u,i}}{\sqrt{\w_{\ell-1}^+}}\\
			={}& (-1)^{|S_1|+\dots+|S_{\ell-1}|+{\ell}}\left(\sum_{S_\ell\in 2^{[c_{\ell}]\setminus{\cal I}_{\ell}}}(-1)^{|S_{\ell}|}-1-(-1)\right)=0.
		\end{align*}
Using ${\sf T}_{\ell}=\widetilde{O}(\sqrt{n/m_{\ell}})$ and $\w_{\ell}=\sqrt{n/m_{\ell}}$ (see \tabl{k-dist-edges}), we can upper bound the contribution of the edges in $E_{\ell}$ to the energy of the flow:
\begin{equation}
\begin{split}
			{\cal R}_{\ell} &= \sum_{e\in E_{\ell}}{\sf T}_{\ell}\frac{\theta(e)^2}{\w_{\ell}} = \widetilde{O}\left(\sum_{\substack{u \in M_{\ell-1}^+, S_{\ell} \in 2^{[c_{\ell}]} \setminus \{\emptyset\}}}\frac{1}{\abs{M_0}}\right)=\widetilde{O}\left(\frac{1}{\abs{M_0}}\right).
\end{split}\label{eq:k-Rell}
\end{equation}
	
\vskip7pt	
\noindent\textbf{${\cal R}_{k}$:} Finally, let $M_{k-1}(S_1,\dots,S_{k-1})$ be the set of $v^{k-1}_{R}\in V_{k-1}(S_1,\dots,S_{k-1})$ (see \eq{V-ell-Ss}) such that, firstly, 
$a_1\in R_1(S_1)$; secondly, for all $\ell\in\{2,\dots,k-1\}$, $j_{\ell}^*\in R_{\ell}(\mu(S_1),\dots,\mu(S_{\ell-1}),S_{\ell})$; and finally,
$v^{0}_{R_1\setminus\{a_1\},R_2\setminus\{j_2^*\},\dots,R_{k-1}\setminus\{j_{k-1}^*\}}\in M_0.$
We let $M_{k-1}$ be the union of all $M_{k-1}(S_1,\dots,S_{k-1})$.
These are exactly the vertices of $V_{k-1}$ that have non-zero incoming flow, with the amount of incoming flow equal to $(-1)^{\abs{S_1} + \cdots +\abs{S_{k-1}}+(k-1)}{\abs{M_0}}^{-1}$. We define the flow $\theta$ on the edges in $E_{k}$ as 
		\begin{align*} 
			\theta(v^{k-1}_{R},v^k_{R,i_{k}})= \begin{cases}
				\left(-1\right)^{\abs{S_1} + \cdots + \abs{S_{k-1}}+\left(k-1\right)}\frac{1}{\abs{M_0}} &\text{if } v^k_R \in M_{k-1}(S_1,\dots,S_{k-1})\text{ and } i_{k} = a_k\\
				0 &\text{else}.
			\end{cases}
		\end{align*}
It is easy to verify that the only vertices $v^k_{R,i_k}\in V_k$ that have non-zero flow are those in $M$, and thus all sources and sinks are in $V_0\cup M$ ($M$ contains some sources, because some vertices have negative flow coming in). Using ${\sf T}_k=\widetilde{O}(1)$ and $\w_k=1$, the contribution of the edges in $E_k$ to the energy of the flow is:
\begin{equation}
			{\cal R}_k = \sum_{e\in E_k}{\sf T}_k\frac{\theta(e)^2}{\w_k} = \widetilde{O}\left(\sum_{u\in M_k}\frac{1}{\abs{M_k}^2}\right)=\widetilde{O}\left(\frac{1}{\abs{M_0}}\right)\label{eq:k-Rk}.
		\end{equation}

\vskip7pt
\noindent\textbf{Item 1:} Recall that $\tilde{E}:=\bigcup_{\ell=2}^{k-1}(\tilde{E}_{\ell}\cup\tilde{E}_{\ell}')$ (\eq{k-dist-tilde-E}), where $\tilde{E}_{\ell}\subset E_{\ell}$ (\eq{k-dist-tilde-E-ell}). By ensuring that there is only flow on $v^{\ell}_{R_1,\dots,R_{k-1}}$ whenever ${\cal K}(\overline{R}_1,\dots,\overline{R}_{\ell-1},A_{\ell}^{(j_{\ell}^*)})$ is not too big, we have ensured that the flow on the edges in $\bigcup_{\ell=2}^{k-1}\tilde{E}_{\ell}$ is 0, and by only sending flow down edges that are part of $E_{\ell}$, which is disjoint from $\tilde{E}_{\ell'}$ (\eq{k-dist-tilde-E-prime}), the flow on $\bigcup_{\ell=2}^{k-1}\tilde{E}_{\ell'}$ is 0 as well, which implies that the flow on all of $\tilde{E}$ is 0.

\vskip7pt	
\noindent\textbf{Item 5:} It remains only to upper bound the energy of the flow by adding up the contributions in \eq{k-R0+}, \eq{k-R1}, \eq{k-Rell+}, \eq{k-Rell} and \eq{k-Rk}:
		\begin{align*}
			{\cal E}^{\sf T}(\theta) &={\cal R}_0^++{\cal R}_1+\sum_{\ell=1}^{k-2}{\cal R}_{\ell}^++\sum_{\ell=2}^{k-1}{\cal R}_{\ell}+{\cal R}_k
=
\widetilde{O}\left(\frac{1}{|M_0|}\right).
		\end{align*}
		Substituting $|M_0|=\Theta(|V_0|)$ yields the desired upper bound.
\end{proof}

\subsubsection{Negative Analysis}

For the negative analysis, we need to upper bound the total weight of the graph, taking into account the subroutine complexities: ${\cal W}^{\sf T}(G)$.
\begin{lemma}\label{lem:k-dist-negative}
There exists ${\cal W}^{\sf T}$ such that
\begin{align*}
{\cal W}^{\sf T}(G) \leq  {\cal W}^{\sf T} 
\leq \widetilde{O}\left(\left(n  + \sum_{\ell=1}^{k-1} \frac{n^2}{t_{\ell}}\right)\abs{V_0}\right).
\end{align*}
\end{lemma}
\begin{proof}
Recall that ${\cal W}^{\sf T}(G)={\cal W}(G^{\sf T})$ is the total weight of the graph $G^{\sf T}$, where we replace each edge $e$ of $G$, with weight $\w_e$, by a path of ${\sf T}_e$ edges of weight $\w_e$, where ${\sf T}_e$ is the complexity of the edge transition $e$ (see \defin{nwk-length} and {\bf TS1-2} of \thm{full-framework}). Thus, ${\cal W}^{\sf T}(G)=\sum_{e\in E(G)}{\sf T}_e\w_e$. By \cor{k-dist-transitions} (see also \tabl{k-dist-edges}) ${\sf T}_e=\widetilde{O}(1)$ for all $e\in E_1 \cup E_k \cup \bigcup_{\ell \in \{0,\dots,k-2\}} E_{\ell}^{+}$ and ${\sf T}_e=\widetilde{O}(\sqrt{n/m_{\ell}})$ for all $e\in E_{\ell}$ for $\ell \in \{2,\dots,k-1\}$. We have defined the weight function (see \tabl{k-dist-edges}) so that $\w_e=1$ for all $e\in E_1 \cup E_k \cup \bigcup_{\ell \in \{0,\dots,k-2\}} E_{\ell}^{+}$ and $\w_e=\sqrt{n/m_{\ell}}$ for all $e\in E_{\ell}$ for $\ell \in \{2,\dots,k-1\}$.
Thus, using \eq{kE2}, the total contribution to the weight from the edges in $E_0^+$ is:
	\begin{equation}
		{\cal W}_0^+ := {\sf T}_0^+\w_0^+\abs{E_0^+}=\widetilde{O}\left(n\abs{V_0}\right).\label{eq:kW1}
	\end{equation}
For $\ell \in \{1,\dots,k-2\}$, we can use \eq{kEk0} to compute the total contribution to the weight from the edges in $E_{\ell}^+$:
	\begin{equation}
		{\cal W}_{\ell}^+ := {\sf T}_{\ell}^+\w_{\ell}^+\abs{E_{\ell}^+}=\widetilde{O}\left(n\abs{V_0}\right).\label{eq:kWk0}
	\end{equation}
Using \eq{kE2}, the total contribution from the edges in $E_1$ is:
	\begin{equation}
		{\cal W}_1 := {\sf T}_1\w_1\abs{E_1}=\widetilde{O}\left(n\abs{V_0}\right).\label{eq:kW2}
	\end{equation}
For $\ell \in \{2,\dots,k-1\}$ using \eq{kEk1} the total contribution from the edges in $E_{\ell}$ is:
	\begin{equation}
		{\cal W}_{\ell}:= {\sf T}_{\ell}\w_{\ell}\abs{E_{\ell}} =\widetilde{O}\left( \frac{n}{m_{\ell}} \right)|E_{\ell}|
		=\widetilde{O}\left(\frac{n^2}{m_{\ell}}\abs{V_0}\right).\label{eq:kWk1}
	\end{equation}
Finally, using \eq{kEk}, the total contribution from the edges in $E_k$ is:
	\begin{equation}
		{\cal W}_k := {\sf T}_k\w_k\abs{E_k}=\widetilde{O}\left(\frac{n^2}{t_{k-1}}\abs{V_0}\right).\label{eq:kWk}
	\end{equation}
	Combining \eq{kW1} to \eq{kWk}, we get total weight:
\begin{align*}
{\cal W}^{\sf T}(G) = \widetilde{O}\left(\left(n +\sum_{\ell=1}^{k-2}n+ n + \sum_{\ell=2}^{k-1} \frac{n^2}{m_{\ell}} + \frac{n^2}{t_{k-1}}\right)\abs{V_0}\right)
= \widetilde{O}\left(\left(n  + \sum_{\ell=1}^{k-1} \frac{n^2}{t_{\ell}}\right)\abs{V_0}\right),
\end{align*}
using $m_{\ell}=\Theta(t_{\ell-1})$ for all $\ell\in \{2,\dots,k-1\}$. 
\end{proof}

\subsubsection{Conclusion of Proof of Theorem~\ref{thm:k-dist}}\label{sec:conclusion}

We can now conclude with the proof of \thm{k-dist}, showing an upper bound of $\widetilde{O}\left(n^{\frac{3}{4}-\frac{1}{4}\frac{1}{2^k-1}}\right)$ on the bounded error quantum time complexity of $k$-distinctness. 

\begin{proof}[Proof of Theorem \ref{thm:k-dist}]
	We apply \thm{full-framework} to $G$ (\sec{k-dist-G-edges} and \sec{k-dist-G-vertices}), $M$ (\eq{k-dist-M}),  $\sigma$ the uniform distribution on $V_0$ (\eq{k-dist-V1}), and $\Psi_\star$ (\sec{k-dist-star-states}), with 
	$${\cal W}^{\sf T}=\widetilde{O}\left(\left(n  + \sum_{\ell=1}^{k-1} \frac{n^2}{t_{\ell}}\right)\abs{V_0}\right)
	\mbox{ and }
	{\cal R}^{\sf T}=\widetilde{O}\left(\abs{V_0}^{-1} \right).$$
	Then we have, referring to \tabl{k-dist-setsizes},
	$${\cal W}^{\sf T}{\cal R}^{\sf T} = \widetilde{O}\left(n  + \sum_{\ell=1}^{k-1} \frac{n^2}{t_{\ell}}\right)=o(n^2).$$
	
	\noindent We have shown the following:
	\begin{description}
		\item[Setup Subroutine:] By \lem{setupk}, the state $\ket{\sigma}$ can be generated in cost 
		$${\sf S}=\widetilde{O}\left(t_1+\sum_{\ell=1}^{k-2}t_{\ell+1}\sqrt{\frac{n}{t_{\ell}}}\right).$$
		\item[Star State Generation Subroutine:] By \lem{k-dist-star-states}, the star states ${\Psi_\star}$ can be generated in $\widetilde{O}(1)$ complexity.
		\item[Transition Subroutine:] By \cor{k-dist-transitions}, there is a quantum subroutine that implements the transition map with errors $\epsilon_{u,v}$ and costs ${\sf T}_{u,v}$, such that
		\begin{description}
			\item[TS1] For all $(u,v)\in \overrightarrow{E}(G)\setminus \tilde{E}$ (defined in \eq{k-dist-tilde-E}), taking $\kappa>2$ in \lem{k-dist-k1}, we have $\epsilon_{u,v}
			=O(n^{-\kappa})=o(1/({\cal R}^{\sf T}{\cal W}^{\sf T}))$.
			\item[TS2] By \cor{ktilde-E}, \lem{k-dist-negative} and using $\kappa > 2$:
			\begin{align*}
				\widetilde{\cal W} = O(n^{-\kappa}{\cal W}^{\sf T}(G)) = o(1/{\cal R}^{\sf T}).
			\end{align*}
			
		\end{description}
		
		\item[Checking Subroutine:] By \eq{k-dist-C}, for any $u\in V_{\sf M}=V_k$, we can check if $u\in M$ in cost $\widetilde{O}(1)$. 
		\item[Positive Condition:] By \lem{k-dist-positive}, there exists a flow satisfying conditions \textbf{P1}-\textbf{P5} of \thm{full-framework}, with ${\cal E}^{\sf T}(\theta)\leq {\cal R}^{\sf T}=\widetilde{O}\left(\abs{V_0}^{-1} \right)$. 
		\item[Negative Condition:] By \lem{k-dist-negative}, ${\cal W}^{\sf T}(G)\leq {\cal W}^{\sf T}=\widetilde{O}\left(\left(n  + \sum_{\ell=1}^{k-1} \frac{n^2}{t_{\ell}}\right)\abs{V_0}\right)$. 
	\end{description}
	Thus, by \thm{full-framework}, there is a quantum algorithm that decides if $M=\emptyset$ in bounded error in complexity:
	\begin{align*}
		\widetilde{O}\left({\sf S}+\sqrt{{\cal R}^{\sf T}{\cal W}^{\sf T}}\right) &
=\widetilde{O}\left(t_1+\sum_{\ell=1}^{k-2}t_{\ell+1}\sqrt{\frac{n}{t_{\ell}}}+\sqrt{n} + \sum_{\ell=1}^{k-1}\frac{n}{\sqrt{t_{\ell}}}\right)
=\widetilde{O}\left(t_1+\sum_{\ell=1}^{k-2}t_{\ell+1}\sqrt{\frac{n}{t_{\ell}}}+\sqrt{n} + \frac{n}{\sqrt{t_{k-1}}}\right)
	\end{align*}
since $t_1>t_2>\dots>t_{k-1}$.
Choosing the optimal values of $t_{\ell} = n^{\frac{3}{4}-\frac{1}{4}\frac{1}{2^k-1} - \sum_{\ell'=2}^{\ell}\frac{2^{k-1-\ell'}}{2^k-1}}$ for $\ell \in \{1,\dots,k-1\}$, as in \tabl{k-dist-setsizes}, we get an upper bound of $\widetilde{O}\left(n^{\frac{3}{4}-\frac{1}{4}\frac{1}{2^k-1}}\right)$.
Since $M\neq\emptyset$ if $x$ has a unique $k$-collision, and $M=\emptyset$ if $x$ has no $k$-collision, the algorithm distinguishes these two cases. By \lem{unique-to-multiple}, this is enough to decide $k$-distinctness in general.
\end{proof}

\subsection*{Acknowledgements}

We thank Arjan Cornelissen and Maris Ozols for discussions on the early ideas of this work; and Simon Apers, Aleksandrs Belovs, Shantanav Chakraborty, Andrew Childs, Fr\'ed\'eric Magniez
and Ronald de~Wolf for helpful comments and discussions about these results.

\newpage 

\printbibliography

@preamble{ "\newcommand{\lName}{1}" }

@preamble{ "\newcommand{\arxiv}[1]{arXiv: \href{https://arxiv.org/abs/#1}{\ttfamily{#1}}\removefirstdot}" }

@preamble{ "\newcommand{\arXiv}[1]{arXiv: \href{https://arxiv.org/abs/#1}{\ttfamily{#1}}\removefirstdot}" }

@preamble{ "\def\removefirstdot#1{\if.#1{}\else#1\fi}" }

@preamble{ "\providecommand{\multiletter}[1]{#1}\renewcommand{\multiletter}[1]{#1}" }

@preamble{ "\DeclareRobustCommand{\dutchPrefix}[2]{#2}" }

@preamble{ "\providecommand{\dutchPrefix}[2]{#2}\renewcommand{\dutchPrefix}[2]{#2}" }

@preamble{ "\newcommand{\skp}[3]{#2}" }

@preamble{ "\newcommand{\focs       }[1]{\if\lName1\skp{  }{Proceedings of the #1 {IEEE} Symposium on Foundations of Computer Science ({FOCS})}{                          }\else{FOCS}\fi}" }

@preamble{ "\newcommand{\stoc       }[1]{\if\lName1\skp{  }{Proceedings of the #1 {ACM} Symposium on the Theory of Computing ({STOC})}{                                   }\else{STOC}\fi}" }

@preamble{ "\newcommand{\soda       }[1]{\if\lName1\skp{  }{Proceedings of the #1 {ACM-SIAM} Symposium on Discrete Algorithms ({SODA})}{                                  }\else{SODA}\fi}" }

@preamble{ "\newcommand{\stacs      }[1]{\if\lName1\skp{  }{Proceedings of the #1 Symposium on Theoretical Aspects of Computer Science ({STACS})}{                        }\else{STACS}\fi}" }

@preamble{ "\newcommand{\itcs       }[1]{\if\lName1\skp{  }{Proceedings of the #1 Innovations in Theoretical Computer Science Conference (ITCS)}{                         }\else{ITCS}\fi}" }

@preamble{ "\newcommand{\fsttcs     }[1]{\if\lName1\skp{  }{Proceedings of the #1 International Conference on Foundations of Software Technology and Theoretical Computer Science (FSTTCS)}{ }\else{FSTTCS}\fi}" }

@preamble{ "\newcommand{\mfcs       }[1]{\if\lName1\skp{  }{Proceedings of the #1 International Symposium on Mathematical Foundations of Computer Science ({MFCS})}{      }\else{MFCS}\fi}" }

@preamble{ "\newcommand{\ccc        }[1]{\if\lName1\skp{  }{Proceedings of the #1 {IEEE} Conference on Computational Complexity ({CCC})}{                                 }\else{CCC}\fi}" }

@preamble{ "\newcommand{\isit       }[1]{\if\lName1\skp{  }{Proceedings of the #1 {IEEE} International Symposium on Information Theory ({ISIT})}{                         }\else{ISIT}\fi}" }

@preamble{ "\newcommand{\colt       }[1]{\if\lName1\skp{  }{Proceedings of the #1 Conference On Learning Theory (COLT)}{                                                  }\else{COLT}\fi}" }

@preamble{ "\newcommand{\nips       }[1]{\if\lName1\skp{  }{Advances in Neural Information Processing Systems #1 ({NIPS})}{                                               }\else{NIPS}\fi}" }

@preamble{ "\newcommand{\aistats    }[1]{\if\lName1\skp{  }{Proceedings of the #1 International Conference on Artificial Intelligence and Statistics ({AISTATS})}{        }\else{AISTATS}\fi}" }

@preamble{ "\newcommand{\icml       }[1]{\if\lName1\skp{  }{Proceedings of the #1 International Conference on Machine Learning (ICML)}{                                   }\else{ICML}\fi}" }

@preamble{ "\newcommand{\icalp      }[1]{\if\lName1\skp{  }{Proceedings of the #1 International Colloquium on Automata, Languages, and Programming (ICALP)}{              }\else{ICALP}\fi}" }

@preamble{ "\newcommand{\esa        }[1]{\if\lName1\skp{  }{Proceedings of the #1 Annual European Symposium on Algorithms (ESA)}{                                         }\else{ESA}\fi}" }

@preamble{ "\newcommand{\tqc        }[1]{\if\lName1\skp{  }{Proceedings of the #1 Conference on the Theory of Quantum Computation, Communication, and Cryptography (TQC)}{}\else{TQC}\fi}" }

@preamble{ "\newcommand{\jacm          }{\if\lName1\skp{    }{Journal of the ACM}{                             }\else{J. ACM}\fi}" }

@preamble{ "\newcommand{\acmta         }{\if\lName1\skp{    }{ACM Transactions on Algorithms}{                 }\else{{ACM} Tr. Alg}\fi}" }

@preamble{ "\newcommand{\acmtct        }{\if\lName1\skp{    }{ACM Transactions on Computation Theory}{         }\else{ACM Tr. Comp. Th.}\fi}" }

@preamble{ "\newcommand{\jams          }{\if\lName1\skp{    }{Journal of the AMS}{                             }\else{J. AMS}\fi}" }

@preamble{ "\newcommand{\pams          }{\if\lName1\skp{    }{Proceedings of the AMS}{                         }\else{Proc. AMS}\fi}" }

@preamble{ "\newcommand{\linalgappl    }{\if\lName1\skp{    }{Linear Algebra and its Applications}{            }\else{Lin. Alg. \& App.}\fi}" }

@preamble{ "\newcommand{\jalgo         }{\if\lName1\skp{    }{Journal of Algorithms}{                          }\else{J. Alg.}\fi}" }

@preamble{ "\newcommand{\jcss          }{\if\lName1\skp{    }{Journal of Computer and System Sciences}{        }\else{J. Comp. Sys. Sci.}\fi}" }

@preamble{ "\newcommand{\cc            }{\if\lName1\skp{    }{Computational Complexity}{                       }\else{Comp. Comp.}\fi}" }

@preamble{ "\newcommand{\algor         }{\if\lName1\skp{    }{Algorithmica}{                                   }\else{Alg.}\fi}" }

@preamble{ "\newcommand{\comb          }{\if\lName1\skp{    }{Combinatorica}{                                  }\else{Comb.}\fi}" }

@preamble{ "\newcommand{\cacm          }{\if\lName1\skp{    }{Communications of the ACM}{                      }\else{Comm. ACM}\fi}" }

@preamble{ "\newcommand{\sigart        }{\if\lName1\skp{    }{SIGART Bulletin}{                                }\else{SIGART Bull.}\fi}" }

@preamble{ "\newcommand{\sigactn       }{\if\lName1\skp{    }{SIGACT News}{                                    }\else{SIGACT News}\fi}" }

@preamble{ "\newcommand{\eatcsbul      }{\if\lName1\skp{    }{Bulletin of the {EATCS}}{                        }\else{Bull. {EATCS}}\fi}" }

@preamble{ "\newcommand{\siamrev       }{\if\lName1\skp{    }{SIAM Review}{                                    }\else{SIAM Rev.}\fi}" }

@preamble{ "\newcommand{\siamjc        }{\if\lName1\skp{    }{SIAM Journal on Computing}{                      }\else{SIAM J. Comp.}\fi}" }

@preamble{ "\newcommand{\siamjo        }{\if\lName1\skp{    }{SIAM Journal on Optimization}{                   }\else{SIAM J. Opt.}\fi}" }

@preamble{ "\newcommand{\siamjdm       }{\if\lName1\skp{    }{SIAM Journal on Discrete Mathematics}{           }\else{SIAM J. Disc. Math.}\fi}" }

@preamble{ "\newcommand{\siamjnum      }{\if\lName1\skp{    }{SIAM Journal on Numerical Analysis}{             }\else{SIAM J. Num. Anal.}\fi}" }

@preamble{ "\newcommand{\siamjmathanal }{\if\lName1\skp{    }{SIAM Journal on Mathematical Analysis}{          }\else{SIAM J. Math. Anal.}\fi}" }

@preamble{ "\newcommand{\discmath      }{\if\lName1\skp{    }{Discrete Mathematics}{                           }\else{Disc. Math.}\fi}" }

@preamble{ "\newcommand{\das           }{\if\lName1\skp{    }{Discrete Applied Mathematics}{                   }\else{Disc. App. Math.}\fi}" }

@preamble{ "\newcommand{\amatstat      }{\if\lName1\skp{    }{Annals of Mathematical Statistics}{              }\else{Ann. Math. Stat.}\fi}" }

@preamble{ "\newcommand{\rms           }{\if\lName1\skp{    }{Russian Mathematical Surveys}{                   }\else{Russ. Math. Surv.}\fi}" }

@preamble{ "\newcommand{\invmath       }{\if\lName1\skp{    }{Inventiones Mathematicae}{                       }\else{Inv. Math.}\fi}" }

@preamble{ "\newcommand{\jnumber       }{\if\lName1\skp{    }{Journal of Number Theory}{                       }\else{J. Num. Th.}\fi}" }

@preamble{ "\newcommand{\toc           }{\if\lName1\skp{    }{Theory of Computing}{                            }\else{Th. Comp.}\fi}" }

@preamble{ "\newcommand{\cjtcs         }{\if\lName1\skp{    }{Chicago Journal of Theoretical Computer Science}{}\else{Chic. J. Th. Comp. Sci.}\fi}" }

@preamble{ "\newcommand{\tocsys         }{\if\lName1\skp{    }{Theory of Computing Systems}{}\else{Theory Comput. Syst.}\fi}" }

@preamble{ "\newcommand{\quantum       }{\if\lName1\skp{    }{{Quantum}}{                                          }\else{Quant.}\fi}" }

@preamble{ "\newcommand{\cmp           }{\if\lName1\skp{    }{Communications in Mathematical Physics}{             }\else{Comm. Math. Phys.}\fi}" }

@preamble{ "\newcommand{\jmp           }{\if\lName1\skp{    }{Journal of Mathematical Physics}{                    }\else{J. Math. Phys.}\fi}" }

@preamble{ "\newcommand{\rspa          }{\if\lName1\skp{    }{Proceedings of the Royal Society A}{                 }\else{Proc. Roy. Soc. A}\fi}" }

@preamble{ "\newcommand{\qic           }{\if\lName1\skp{    }{Quantum Information and Computation}{                }\else{Quant. Inf. \& Comp.}\fi}" }

@preamble{ "\newcommand{\physrev       }{\if\lName1\skp{    }{Physical Review}{                                    }\else{Phys. Rev.}\fi}" }

@preamble{ "\newcommand{\pra           }{\if\lName1\skp{    }{Physical Review A}{                                  }\else{Phys. Rev. A}\fi}" }

@preamble{ "\newcommand{\prb           }{\if\lName1\skp{    }{Physical Review B}{                                  }\else{Phys. Rev. B}\fi}" }

@preamble{ "\newcommand{\pre           }{\if\lName1\skp{    }{Physical Review E}{                                  }\else{Phys. Rev. E}\fi}" }

@preamble{ "\newcommand{\prr           }{\if\lName1\skp{    }{Physical Review Research}{                           }\else{Phys. Rev. Research}\fi}" }

@preamble{ "\newcommand{\prx           }{\if\lName1\skp{    }{Physical Review X}{                                  }\else{Phys. Rev. X}\fi}" }

@preamble{ "\newcommand{\prl           }{\if\lName1\skp{    }{Physical Review Letters}{                            }\else{Phys. Rev. Lett.}\fi}" }

@preamble{ "\newcommand{\njp           }{\if\lName1\skp{    }{New Journal of Physics}{                             }\else{New J. Phys.}\fi}" }

@preamble{ "\newcommand{\prapp         }{\if\lName1\skp{    }{Physical Review Applied}{                            }\else{Phys. Rev. Appl.}\fi}" }

@preamble{ "\newcommand{\physrep       }{\if\lName1\skp{    }{Physics Reports}{                                    }\else{Phys. Rep.}\fi}" }

@preamble{ "\newcommand{\rmp           }{\if\lName1\skp{    }{Reviews of Modern Physics}{                          }\else{Rev. Mod. Phys. }\fi}" }

@preamble{ "\newcommand{\phystoday     }{\if\lName1\skp{    }{Physics Today}{                                      }\else{Phys. Today}\fi}" }

@preamble{ "\newcommand{\physics       }{\if\lName1\skp{    }{Physics}{                                            }\else{Phys.}\fi}" }

@preamble{ "\newcommand{\nature        }{\if\lName1\skp{    }{Nature}{                                             }\else{Nat.}\fi}" }

@preamble{ "\newcommand{\natcomm       }{\if\lName1\skp{    }{Nature Communications}{                              }\else{Nat. Comm.}\fi}" }

@preamble{ "\newcommand{\natphys       }{\if\lName1\skp{    }{Nature Physics}{                                     }\else{Nat. Phys.}\fi}" }

@preamble{ "\newcommand{\npjqi         }{\if\lName1\skp{    }{npj Quantum Information}{                            }\else{npj Quant. Inf.}\fi}" }

@preamble{ "\newcommand{\scirep        }{\if\lName1\skp{    }{Scientific Reports}{                                 }\else{Sci. Rep.}\fi}" }

@preamble{ "\newcommand{\science       }{\if\lName1\skp{    }{Science}{                                            }\else{Sci.}\fi}" }

@preamble{ "\newcommand{\jpa           }{\if\lName1\skp{    }{Journal of Physics A: Mathematical and Theoretical}{ }\else{J. Phys. A}\fi}" }

@preamble{ "\newcommand{\ijtp          }{\if\lName1\skp{    }{International Journal of Theoretical Physics}{       }\else{Int. J. Th. Phys.}\fi}" }

@preamble{ "\newcommand{\jmo           }{\if\lName1\skp{    }{Journal of Modern Optics}{                           }\else{J. Mod. Opt.}\fi}" }

@preamble{ "\newcommand{\jstatph       }{\if\lName1\skp{    }{Journal of Statistical Physics}{                     }\else{J. Stat. Phys.}\fi}" }

@preamble{ "\newcommand{\pnas          }{\if\lName1\skp{    }{Proceedings of the National Academy of Sciences}{    }\else{PNAS}\fi}" }

@preamble{ "\newcommand{\lncs          }{\if\lName1\skp{    }{Lecture Notes in Computer Science}{                  }\else{L. Notes Comp. Sci.}\fi}" }

@preamble{ "\newcommand{\lnai          }{\if\lName1\skp{    }{Lecture Notes in Artificial Intelligence}{           }\else{L. Notes Art. Int.}\fi}" }

@preamble{ "\newcommand{\lnm           }{\if\lName1\skp{    }{Lecture Notes in Mathematics}{                       }\else{L. Notes Math.}\fi}" }

@preamble{ "\newcommand{\tams          }{\if\lName1\skp{    }{Transactions of the American Mathematical Society}{  }\else{Trans. AMS}\fi}" }

@preamble{ "\newcommand{\ieeetit        }{\if\lName1\skp{    }{{IEEE} Transactions on Information Theory}{          }\else{{IEEE} Trans. Inf. Th.}\fi}" }

@preamble{ "\newcommand{\iscs          }{\if\lName1\skp{    }{International Series in Computer Science}{           }\else{Int. Ser. Comp. Sci.}\fi}" }

@preamble{ "\newcommand{\tocl          }{\if\lName1\skp{    }{Theory of Computing Library}{                        }\else{Th. Comp. Lib.}\fi}" }

@STRING{cryptology	= "Journal of Cryptology"}

@inproceedings{JefferyZ23,
  author       = {Stacey Jeffery and
                  Sebastian Zur},
  editor       = {Barna Saha and
                  Rocco A. Servedio},
  title        = {Multidimensional Quantum Walks},
  booktitle    = {Proceedings of the 55th Annual {ACM} Symposium on Theory of Computing,
                  {STOC} 2023, Orlando, FL, USA, June 20-23, 2023},
  pages        = {1125--1130},
  publisher    = {{ACM}},
  year         = {2023},
  url          = {https://doi.org/10.1145/3564246.3585158},
  doi          = {10.1145/3564246.3585158},
  bibsource    = {dblp computer science bibliography, https://dblp.org}
}

@article{aaronson2004QLowerBndCollisionAndElementDistinct,
	title={Quantum lower bounds for the collision and the element distinctness problems},
	author={Aaronson, Scott and Shi, Yaoyun},
	journal={\jacm},
	volume={51},
	number={4},
	pages={595--605},
	year={2004},
	doi = {10.1145/1008731.1008735},
}

@article{ajtai2005EDlowerbound,
	title = {A non-linear time lower bound for {B}oolean branching programs},
	author = {Mikl\'os Ajtai},
	journal = {\toc},
	volume = {1},
	pages = {149--176},
	year = {2005},
        doi = {10.1109/SFFCS.1999.814578}}

@article{ambainis2004QWalkForElementDist,
	author = {Ambainis, Andris},
	title = {Quantum Walk Algorithm for Element Distinctness},
	journal = {\siamjc},
	year = {2007},
	volume = {37},
	number = {1},
	pages = {210--239},
	doi = {10.1137/S0097539705447311},
eprinttype ={extra-open},
eprint  = {https://arxiv.org/abs/quant-ph/0311001},
	note = {Earlier version in FOCS'04.},	
}

@article{ambainis2010VTSearch,
	author    = {Andris Ambainis},
	title     = {Quantum Search with Variable Times},
	journal = {\tocsys},
	volume = {47},
	pages     = {786--807},
	year      = {2010},
	doi       = {10.1007/s00224-009-9219-1},
eprinttype ={extra-open},
eprint  = {https://arxiv.org/abs/quant-ph/0609168}
}

@inproceedings{ambainis2019QuadSpeedupFindingMarkedQW,
	author = {Ambainis, Andris and Gily\'{e}n, Andr\'{a}s and Jeffery, Stacey and Kokainis, Martins},
	title = {Quadratic Speedup for Finding Marked Vertices by Quantum Walks},
	year = {2020},
	booktitle = {\stoc{52nd}},
	pages = {412–424},
	numpages = {13},
	doi = {10.1145/3357713.3384252},
eprinttype ={extra-open},
eprint  = {https://arxiv.org/abs/1903.07493},
}

@inproceedings{apers2019UnifiedFrameworkQWSearch,
	title	={A Unified Framework of Quantum Walk Search},
	author  ={Apers, Simon and Gily{\'e}n, Andr{\'a}s and Jeffery, Stacey},
	year	={2020},
	booktitle = {\stacs{38th}},
	pages = {6:1--6:13},
	doi = {10.4230/LIPIcs.STACS.2021.6},
eprinttype ={extra-open},
eprint  = {https://arxiv.org/abs/1912.04233},
}

@article{apers2022quadratic,
	title={Quadratic speedup for spatial search by continuous-time quantum walk},
	author={Apers, Simon and Chakraborty, Shantanav and Novo, Leonardo and Roland, J{\'e}r{\'e}mie},
	journal={Physical review letters},
	volume={129},
	number={16},
	pages={160502},
	year={2022},
	publisher={APS},
        doi = {10.1103/PhysRevLett.129.160502}
}

@article{atia2021welded,
	title = {Improved Upper Bounds for the Hitting Times of Quantum Walks},
	author = {Atia, Yosi and Chakraborty, Shantanav},
	year = {2021},
	journal = {\pra},
	volume = {104},
	pages = {032215},
	doi = {10.1103/PhysRevA.104.032215},
eprinttype ={extra-open},
eprint  = {https://arxiv.org/abs/2005.04062}
}

@inproceedings{belovs2012LG,
	author = {Belovs, Aleksandrs},
	title = {Span programs for functions with constant-sized 1-certificates},
	booktitle = {\stoc{44th}},
	year = {2012},
	pages = {77--84},
	doi = {10.1145/2213977.2213985},	
}

@inproceedings{belovs2012kDist,
	Author = {Belovs, Aleksandrs},
	Booktitle = {\focs{53rd}},
	pages = {207--216},
	Title = {Learning-graph-based quantum algorithm for $k$-distinctness},
	Year = {2012},
	doi = {10.1109/FOCS.2012.18},
eprinttype ={extra-open},
eprint  = {https://arxiv.org/abs/1205.1534},
}

@inproceedings{belovs2013TimeEfficientQW3Distintness,
	 author       = {Aleksandrs Belovs and
                  Andrew M. Childs and
                  Stacey Jeffery and
                  Robin Kothari and
                  Fr{\'{e}}d{\'{e}}ric Magniez},
	title = {Time-Efficient Quantum Walks for 3-Distinctness},
	booktitle = {\icalp{40th}},
	year = {2013},
	pages = {105--122},
	numpages = {18},
	doi = {10.1007/978-3-642-39206-1_10},
}

@misc{belovs2013ElectricWalks,
	author = {Belovs, Aleksandrs},
	title = {Quantum walks and electric networks},
	year = {2013},	
    DOI = {10.48550/arXiv.1302.3143}
}

@article{bernstein1997BValg,
	author = {Bernstein, Ethan and Vazirani, Umesh},
	year = {1997},
	title = {Quantum Complexity Theory},
	journal = {\siamjc},
	volume = {26},
	issue = {5},
	pages = {1411--1473},
	doi = {10.1137/S0097539796300921}
}

@INCOLLECTION{brassard2002AmpAndEst,
	AUTHOR    = {Gilles Brassard and Peter H{\o}yer and Michele Mosca and Alain Tapp},
	TITLE     = {Quantum Amplitude Amplification and Estimation},
	BOOKTITLE = {Quantum Computation and Quantum Information: A Millennium Volume},
	SERIES    = {Contemporary Mathematics Series},
	PUBLISHER = {AMS},
	VOLUME    = {305},
	PAGES     = {53--74},
	YEAR      = {2002},
	doi	= {10.1090/conm/305/05215},
eprinttype ={extra-open},
eprint  = {https://arxiv.org/abs/0005055},
}

@article{brassard1997collision,
	Author = {Brassard, Gilles and H{\o}yer, Peter and Tapp, Alain},
	Journal = {ACM SIGACT News},
eprinttype ={extra-open},
eprint  = {https://arxiv.org/abs/9705002},
	Pages = {14--19},
	Title = {Quantum Algorithm for the Collision Problem},
	Volume = {28},
	Year = {1997},
        doi = {10.1007/BFb0054319}}

@article{buhrman2001ElementDistinctness,
	author = {Buhrman, Harry and D{\"u}rr, Christoph and Heiligman, Mark and H{\o}yer, Peter and Magniez, Fr\'ed\'eric and Santha, Miklos and {\dutchPrefix{Wolf}{d}}e Wolf, Ronald},
	title = {Quantum Algorithms for Element Distinctness},
	journal = {\siamjc},
	volume = {34},
	number = {6},
	pages = {1324--1330},
	year = {2005},
	doi = {10.1137/S0097539702402780},
eprinttype ={extra-open},
eprint  = {https://arxiv.org/abs/quant-ph/0007016},
	note = {Earlier version in CCC'01.},
}

@inproceedings{buhrman2022limits,
    title={Limits of Quantum Speed-Ups for Computational Geometry and Other Problems: Fine-Grained Complexity via Quantum Walks},
    author={Buhrman, Harry and Loff, Bruno and Patro, Subhasree and Speelman, Florian},
    booktitle={13th Innovations in Theoretical Computer Science Conference (ITCS 2022)},
    year={2022},
    organization={Schloss Dagstuhl-Leibniz-Zentrum f{\"u}r Informatik},
    doi = {10.4230/LIPIcs.ITCS.2022.31}
}

@inproceedings{bun2018PolyMethodStrikesBack,
	title={The polynomial method strikes back: Tight quantum query bounds via dual polynomials},
	author={Bun, Mark and Kothari, Robin and Thaler, Justin},
	booktitle={\stoc{50th}},
	year={2018},
	doi = {10.1145/3188745.3188784},
eprinttype ={extra-open},
eprint  = {https://arxiv.org/abs/1710.09079},
}

@article{chandra1996ElectricalResAndCommute,
	title={The electrical resistance of a graph captures its commute and cover times},
	author={Chandra, Ashok K. and Raghavan, Prabhakar and Ruzzo, Walter L. and Smolensky, Roman and Tiwari, Prasoon},
	journal={Computational Complexity},
	volume={6},
	number={4},
	pages={312--340},
	year={1996},
	doi = {10.1007/BF01270385},
}

@inproceedings{childs2003ExpSpeedupQW,
	title={Exponential algorithmic speedup by a quantum walk},
	author={Childs, Andrew M. and Cleve, Richard and Deotto, Enrico and Farhi, Edward and Gutmann, Sam and Spielman, Daniel A.},
	booktitle={\stoc{35th}},
	pages={59--68},
	year={2003},
	doi={10.1145/780542.780552},
	 eprinttype ={extra-open},
eprint  = {https://arxiv.org/abs/quant-ph/0209131}
}

@article{childs2005quantum,
    title={Quantum algorithms for subset finding},
    author={Childs, Andrew M and Eisenberg, Jason M},
    journal={Quantum Information \& Computation},
    volume={5},
    number={7},
    pages={593--604},
    year={2005},
    publisher={Rinton Press, Incorporated Paramus, NJ},
    doi = {10.5555/2011656.2011663}
}

@misc{childs2013arXivTimeEfficientQW3Distintness,
	  author       = {Andrew M. Childs and
                  Stacey Jeffery and
                  Robin Kothari and
                  Fr{\'{e}}d{\'{e}}ric Magniez},
	title = {A Time-Efficient Quantum Walk for 3-Distinctness Using Nested Updates},
	year = {2013},
	DOI = {10.48550/arXiv.1302.7316}
}

@inproceedings{cornelissen2020SpanProgramTime,
	author = {Cornelissen, Arjan and Jeffery, Stacey and Ozols, Maris and Piedrafita, Alvaro},
	title = {Span programs and quantum time complexity},
	year = {2020},
	booktitle = {\mfcs{45th}},
	pages = {21:1--26:14},
eprinttype ={extra-open},
eprint  = {https://arxiv.org/abs/2005.01323},
	doi = {10.4230/LIPIcs.MFCS.2020.26},
}

@inproceedings{costello2016isogeny,
	author = {Costello, Craig and Longa, Patrick and Naehrig, Michael},
	title = {Efficient algorithms for supersingular isogeny Diffie-Hellman},
	year = {2016},
	booktitle = {Advances in Cryptology (CRYPTO 2016)}, 
	pages = {572--601},
        doi = {10.1007/978-3-662-53018-4_21}
}

@book{Janson2011RandomGraphs,
	Author = {Janson, Svante. and Luczak, Tomasz. and Rucinski, Andrzej},
	Isbn = {9781118030967},
	Publisher = {John Wiley \& Sons},
	Series = {Wiley Series in Discrete Mathematics and Optimization},
	Title = {Random Graphs},
	Url = {http://books.google.ca/books?id=RjnqVoB4VmUC},
	Year = {2011},
	doi = {10.1002/9781118032718}}

@phdthesis{jeffery2014thesis,
	author = {Jeffery, Stacey},
	title = {Frameworks for Quantum Algorithms},
	year = {2014},
	school = {University of Waterloo},	
	url = {http://uwspace.uwaterloo.ca/handle/10012/8710},
}

@article{kaplan2009derandomized,
  title={Derandomized constructions of k-wise (almost) independent permutations},
  author={Kaplan, Eyal and Naor, Moni and Reingold, Omer},
  journal={\algor},
  volume={55},
  number={1},
  pages={113--133},
  year={2009},
  doi = {10.1007/s00453-008-9267-y},
  publisher={Springer}
}

@article{kitaev1996PhaseEst,
	author = {Kitaev, Alexei Y.},
	title = {Quantum measurements and the {A}belian stabilizer problem},
	journal= {ECCC},
	year= {1996},
	volume= {TR96-003},
   DOI = {10.48550/arXiv.quant-ph/9511026},
}

@article{kitaev1999quantum,
	title={Quantum np},
	author={Kitaev, Alexei},
	journal={Talk at AQIP},
	volume={99},
	year={1999}
}

@inproceedings{lee2011QQueryCompStateConv,
	author = {Lee, Troy and Mittal, Rajat and Reichardt, Ben W. and {\v{S}}palek, Robert and {\multiletter{Sz}}egedy, M\'{a}ri\'{o}},
	title = {Quantum Query Complexity of State Conversion},
	booktitle = {\focs{52nd}},
	year = {2011},
	pages = {344--353},
	doi = {10.1109/FOCS.2011.75},
eprinttype ={extra-open},
eprint  = {https://arxiv.org/abs/1011.3020}
}

@article{magniez2006SearchQuantumWalk,
	author = {Frédéric Magniez and Ashwin Nayak and Jérémie Roland and Miklos Santha},
	title = {Search via Quantum Walk},
	journal = {\siamjc},
	volume = {40},
	number = {1},
	pages = {142-164},
	year = {2011},
	doi = {10.1137/090745854},
eprinttype ={extra-open},
eprint  = {https://arxiv.org/abs/quant-ph/0608026},
	note = {Earlier version in STOC'07.},
}

@inproceedings{mande2020kDistLB,
	author =	{Nikhil S. Mande and Justin Thaler and Shuchen Zhu},
  	title =	{{Improved Approximate Degree Bounds for k-Distinctness}},
 	 booktitle =	{\tqc{15th}},
 	 pages =	{2:1--2:22},
  	year =	{2020},
  	volume =	{158},
  	doi =		{10.4230/LIPIcs.TQC.2020.2},
eprinttype ={extra-open},
eprint  = {https://arxiv.org/abs/2002.08389}
}

@inproceedings{szegedy2004QMarkovChainSearch, 
	author={Mario Szegedy}, 
	booktitle={\focs{45th}}, 
	title={Quantum speed-up of {M}arkov chain based algorithms}, 
	year={2004}, 
	pages={32-41}, 
	doi={10.1109/FOCS.2004.53}, 
eprinttype ={extra-open},
eprint  = {https://arxiv.org/abs/quant-ph/0401053},
}

@article{tani2009claw,
	author = {Tani, Seiichiro},
	title = {Claw finding algorithms using quantum walk},
	journal = {Theoretical Computer Science},
	volume = {410},
	number = {50},
	year = {2009},
	pages = {5285--5297},
        doi = {10.1016/j.tcs.2009.08.030}
}

\end{document}